\documentclass[11pt]{article}

\usepackage[utf8]{inputenc}
\usepackage[letterpaper,margin=1.00in]{geometry}
\usepackage{amsmath, amssymb, amsthm, amsfonts}
\usepackage{xspace}
\usepackage[linktocpage]{hyperref}
\usepackage[dvipsnames]{xcolor}

\usepackage[utf8]{inputenc}
\usepackage{amsmath,amsfonts,amssymb,amsthm}
\usepackage[sans]{dsfont}
\usepackage{mathtools}
\usepackage{url}
\usepackage{appendix}
\usepackage{algorithm}
\usepackage{fullpage}
\usepackage{bbm}
\hypersetup{
    unicode=false,          % non-Latin characters in Acrobatӳ bookmarks
    colorlinks=true,        % false: boxed links; true: colored links
    linkcolor=red,          % color of internal links (change box color with linkbordercolor)
    citecolor=OliveGreen,        % color of links to bibliography
    filecolor=magenta,      % color of file links
    urlcolor=cyan           % color of external links
}
\usepackage{thm-restate}
\usepackage[capitalize, nameinlink]{cleveref}

\newtheorem{theorem}{Theorem}[section]

\newtheorem{lemma}[theorem]{Lemma}

\newtheorem{claim}[theorem]{Claim}\crefname{claim}{claim}{claims}

\newtheorem{corollary}[theorem]{Corollary}

\newtheorem{definition}[theorem]{Definition}

\newcounter{note}[section]

\newcommand{\future}[1]{}

\newcommand{\eps}{\varepsilon}

\newcommand{\mcC}{\mathcal{C}}
\newcommand{\mcD}{\mathcal{D}}

\newcommand{\mcN}{\mathcal{N}}

\newcommand{\mcP}{\mathcal{P}}

\newcommand{\mcS}{\mathcal{S}}
\newcommand{\mcT}{\mathcal{T}}

\newcommand{\mcV}{\mathcal{V}}

\newcommand{\dset}{{\mathcal D}}

\newcommand{\nset}{{\mathcal N}}

\newcommand{\pset}{{\mathcal P}}

\newcommand{\congest}{$\mathsf{CONGEST}$\xspace}

\newcommand{\poly}{\operatorname{\text{{\rm poly}}}}
\newcommand{\floor}[1]{\lfloor #1 \rfloor}

\newcommand{\batch}{b}
\newcommand{\diam}{\operatorname{diam}}
\newcommand{\spa}{\operatorname{spars}}
\newcommand{\sep}{\operatorname{sep}}
\newcommand{\load}{\operatorname{load}}

\newcommand{\ball}{\operatorname{ball}}
\newcommand{\vol}{\operatorname{vol}}

\global\long\def\poly{\mathrm{poly}}%
\global\long\def\vol{\mathrm{vol}}%
\global\long\def\eps{\epsilon}%
\global\long\def\cN{{\cal N}}%
\global\long\def\cS{{\cal S}}%
\global\long\def\l{\ell}%
\global\long\def\congest{\mathrm{cong}}%
\global\long\def\step{\mathrm{step}}%
\global\long\def\spars{\mathrm{spars}}%
\global\long\def\sep{\mathrm{sep}}%
\global\long\def\val{\mathrm{val}}%
\global\long\def\disperse{\mathrm{disperse}}%
\global\long\def\deg{\mathrm{deg}}%
\global\long\def\diam{\mathrm{diam}}%
\global\long\def\cov{\mathrm{cov}}%
\global\long\def\ball{\mathrm{ball}}%
\global\long\def\supp{\mathrm{supp}}%
\newcommand{\U}{U}

\renewcommand{\l}{\ell}

\newcommand{\copies}{\text{copies}}

\renewcommand{\poly}{\operatorname{poly}}
\newcommand{\znote}[1]{\textcolor{red}{\sc{[ZT: #1]}}}

\newcommand{\alg}{\mathcal{A}}

\newcommand{\myparskip}{3pt}
\parskip \myparskip

\usepackage{caption}
\usepackage{subcaption}
\usepackage{algorithm}
\usepackage[noend]{algpseudocode}
\usepackage{xfrac}

\newcommand{\nrd}{\ensuremath{D_{A,k,k'}}\xspace}
\newcommand{\barnrd}{\ensuremath{\bar{D}_{A,k,k'}}\xspace}

\newcommand{\qLC}{\text{LC}}
\newcommand{\qLSC}{\text{LSC}}
\newcommand{\qLEC}{\text{LEC}}
\newcommand{\qSED}{\text{SED}}
\newcommand{\qLWSC}{\text{LWSC}}
\newcommand{\qLDSC}{\text{LDSC}}
\newcommand{\qLDSCS}{\text{LDSCS}}

\DeclareMathOperator*{\argmax}{arg\,max}

\newcommand{\dil}{\text{dil}}

\newcommand{\epoch}{\mathsf{epoch}}

\newcommand{\wsparseCut}{\textnormal{\textsf{W}}_{\textnormal{sparse-cut}}}
\newcommand{\dsparseCut}{\textnormal{\textsf{D}}_{\textnormal{sparse-cut}}}
\newcommand{\wcutStrat}{\textnormal{\textsf{W}}_{\textnormal{cut-strat}}}
\newcommand{\dcutStrat}{\textnormal{\textsf{D}}_{\textnormal{cut-strat}}}
\newcommand{\wED}{\textnormal{\textsf{W}}_{\textnormal{ED}}}
\newcommand{\dED}{\textnormal{\textsf{D}}_{\textnormal{ED}}}

\newcommand{\stepOne}{\hyperref[step:1]{step 1}\xspace}
\newcommand{\stepTwo}{\hyperref[step:2]{step 2}\xspace}
\newcommand{\stepThree}{\hyperref[step:3]{step 3}\xspace}

\title{New Structures and Algorithms for\\ Length-Constrained Expander Decompositions}
\author{
\begin{tabular}[t]{c@{\extracolsep{3em}}cc} 
        Bernhard Haeupler\thanks{Partially funded by the European Union's Horizon 2020 ERC grant 949272.} &    D Ellis Hershkowitz & Zihan Tan\thanks{Supported by a grant to DIMACS from the Simons Foundation (820931)} \\
        \small ETH Z\"urich \&  & \small   Brown University & \small DIMACS, \\
        \small  Carnegie Mellon University &  & \small  Rutgers University \\
\end{tabular}
}

\date{}
\begin{document}

\maketitle

\begin{abstract}

Expander decompositions form the basis of one of the most flexible paradigms for close-to-linear-time graph algorithms. Length-constrained expander decompositions generalize this paradigm to better work for problems with lengths, distances and costs. Roughly, an $(h,s)$-length $\phi$-expander decomposition is a small collection of length increases to a graph so that nodes within distance $h$ can route flow over paths of length $hs$ with congestion at most $1/\phi$.

\smallskip

In this work, we give a close-to-linear time algorithm for computing length-constrained expander decompositions in graphs with general lengths and capacities. Notably, and unlike previous works, our algorithm allows for one to trade off off between the size of the decomposition and the length of routing paths: for any $\eps > 0$ not too small, our algorithm computes in close-to-linear time an $(h,s)$-length $\phi$-expander decomposition of size $m \cdot \phi \cdot n^\eps$ where $s = \exp(\poly(1/\eps))$. The key foundations of our algorithm are: (1) a simple yet powerful structural theorem which states that the union of a sequence of sparse length-constrained cuts is itself sparse and (2) new algorithms for efficiently computing sparse length-constrained flows. 
\end{abstract}

\thispagestyle{empty}

\newpage
\thispagestyle{empty}
\tableofcontents
\thispagestyle{empty}

\newpage
\setcounter{page}{1}

\section{Introduction}\label{sec:intro}

% \enote{Kick this out, turn into like paragraph in abstract; don't need to define expander etc}

% Expanders are among the most well-studied graph families \cite{goldreich2011candidate,hoory2006expander,sipser1996expander,goldreich2011basic,lubotzky2012expander}. A graph $G = (V,E)$ is a $\phi$-expander iff for any subset $S \subseteq V$ we have that the sparsity of $S$ is at least $\phi$, i.e.\  $\frac{|E(S, V\setminus S)|}{\min(\vol(S), \vol(V \setminus S))} \geq \phi$ where $E(S, V\setminus S)$ is all edges with exactly one endpoint in $S$ and $\vol(S) := \sum_{v \in S} \deg(v)$ is the volume of $S$. 
% Much of the motivation for studying expanders comes from the fact that  (so-called  $\ell_{\infty}$) quantities like flows, congestion and cuts have nice structure in expanders. For example, any unit demand in a constant-degree $\phi$-expander can be routed by a multi-commodity flow with $\tilde{O}(\frac{1}{\phi})$ congestion \cite{leighton1999multicommodity}. Here, a unit demand is a function $D : V \times V \to \mathbb{Z}_{\geq 0}$ satisfying $\sum_{u} D(v,u) \leq \deg(v)$ and $\sum_{u} D(u,v) \leq \deg(v)$ for every $v \in V$ and the congestion of a multi-commodity flow is the maximum flow sent along an edge.

% \cite{ghaffari2017distributed,ghaffari2018new,goranci2021expander,broder1992existence,peleg1987constructing,kleinberg1996short,valiant1990general,leighton1999multicommodity,leighton1994packet}

Over the past few decades, expander decompositions have come to form the foundation of one of the most powerful and flexible paradigms for close-to-linear-time algorithms \cite{chuzhoy2020deterministic,van2020bipartite,li2021deterministic,saranurak2019expander}. Informally, expander decompositions separate a graph into expanders, allowing algorithms to make use of the structure of expanders on arbitrary graphs. One of the key properties of expanders---and, indeed, a way in which they are often defined--- is that any (reasonable) multi-commodity flow demand can be routed with low congestion \cite{ghaffari2017distributed,ghaffari2018new,goranci2021expander,broder1992existence,peleg1987constructing,kleinberg1996short,valiant1990general,leighton1999multicommodity,leighton1994packet}. %Formally, a $\phi$-expander decomposition is a collection of $(\kappa \cdot \phi m)$-many edges $C \subseteq E$ such that each connected component of $G - C := (V, E \setminus C)$ is a $\phi$-expander; $\kappa$ is called the \emph{cut slack}. 
% Expander decompositions have paved the way for numerous recent breakthroughs in close-to-linear-time algorithms \cite{chuzhoy2020deterministic,van2020bipartite,li2021deterministic,saranurak2019expander}.
However, while this property makes the expander decomposition paradigm very useful for algorithms dealing with flows, congestion and cuts, it is less useful for quantities like lengths, distances and costs.% since these are not as well-behaved in expanders.

% Much of the power of expander decompositions is that expanders and so expander decompositions give principled ways of reasoning about fundamental structures like .
% allow \cite{goranci2021expander}. However, lengths and costs.

\subsection{Length-Constrained Expanders and Expander Decompositions}
Motivated by this, \cite{haeupler2022expander} introduced the idea of length-constrained expanders. Informally, a length-constrained expander is a graph in which any (reasonable) multi-commodity flow can be routed over \emph{short} paths. More formally, a demand $D : V \times V \to \mathbb{Z}_{\geq 0}$ is \emph{$h$-length} if $D(u,v) > 0$ implies that $u$ and $v$ are at distance at most $h$ and \emph{unit} if no vertex sends or receives more than its degree in demand. Then, an $(h,s)$-length $\phi$-expander is a graph where any $h$-length unit demand can be routed by a multi-commodity flow with congestion $\tilde{O}(\frac{1}{\phi})$ \emph{over length $hs$-length paths}. $s$ is called the \emph{length slack}. If $hs \ll O(\frac{\log n}{\phi})$, a graph which is an $(h,s)$-length $\phi$-expander may not be a $\phi$-expander. For example, a path is an $(h, 1)$-length $\Omega(1)$-expander for constant $h$ but not an $\Omega(1)$-expander; henceforth, we use \emph{classic expander} to refer to (non-length-constrained) expanders. 

In order to better bring the expander decomposition machinery to bear on problems that deal with lengths, distances and costs, \cite{haeupler2022expander} introduced the idea of an $(h,s)$-length $\phi$-expander decomposition: a collection of $(\kappa \cdot \phi m)$ total edge length increases that make the input graph an $(h,s)$-length $\phi$-expander. Here, $\kappa$ is called the \emph{cut slack}. %They proved the existence of such decompositions for any $h$ and $\phi$ with cut and length slack $\kappa = s = O(\log n)$. More recent work showed that it is (existentially) possible to trade off between cut and length slack with $\kappa = n^{O(1/s)} \cdot \log n$ for any $s \geq 100$ \cite{ghaffari2022cut}.

Length-constrained expander decompositions greatly extend the problems for which the expander decomposition paradigm is suitable. For instance, length-constrained expander decompositions give a simple tree-like way of routing demands that is $n^{o(1)}$-competitive with respect to \emph{both} congestion and flow path length \cite{haeupler2022expander}. Furthermore, these routings are oblivious: the flow for each pair is fixed without knowledge of the demand. In turn, these routings give (1) compact routing tables that allow nodes to perform $n^{o(1)}$-competitive point-to-point communication %with node $v$ storing only $\deg(v) \cdot n^{o(1)}$ bits 
and (2) universally-optimal distributed algorithms in the CONGEST model of distributed computation, bypassing $\Omega(\sqrt{n} / \log n)$ lower bounds \cite{peleg2000near,das2011distributed} on networks with $n^{o(1)}$ time algorithms.

\subsection{Our Contributions}
In this work, we provide a deeper theory of length-constrained expanders and significantly more powerful close-to-linear time algorithms for computing length-constrained expander decompositions. %The latter is the first such efficient algorithm which works for the weighted case and the first algorithm that allows for a trade-off between length and cut slack.

% we significantly deepen, generalize and extend both the theory of and algorithms for length-constrained expander decompositions. The culmination of this is significantly more general algorithm for computing length-constrained expander decompositions as as well as new structural results for length-constrained expander decompositions.

% This work provides a more mature theory of length-constrained expanders. A main theme of our work is that many powerful theorems from the classic setting carry over to the length-constrained setting despite the significant complications of length constraints. The result is a more unified theory of length-constrained expanders and a suite of new tools for length-constrained expander decompositions. The culmination of this is a simple and fast algorithm for computing length-constrained expander decompositions with exponentially-improved cut slack and length slack, nearly matching what is existentially-possible.

% \subsubsection{A Simple Candidate Algorithm for Length-Constrained Expander Decompositions}

We begin by describing our new algorithm for computing length-constrained expander decompositions. Our algorithm is based on the notion of $(h,s)$-length $\phi$-sparse cuts which generalizes the classic notion of $\phi$-sparse cuts (a.k.a.\ moving cuts) \cite{haeupler2022expander}. An $(h,s)$-length $\phi$-sparse cut is a collection of length increases such that there is a large $h$-length unit ``witness'' demand whose support pairs are made at least $hs$-far apart by these length increases. Our algorithm will simply repeatedly cut length-constrained sparse cuts where ``cutting'' such a cut just consists of applying its length increases to the graph.\footnote{For the case of length-constrained expanders, ``cutting'' necessarily involves increasing lengths (possibly fractionally) rather than (integrally) deleting edges since the length-constrained setting is known to exhibit large flow-cut integrality gaps \cite{haeupler2020network}.}

It is known that a graph is a length-constrained expander iff it contains no length-constrained sparse cuts \cite{haeupler2022expander} and so, if one repeatedly cuts  $(h,s)$-length $\phi$-sparse cuts until none exist, the union of all such cuts gives a length-constrained expander decomposition. Of course, if these cuts are always very small in size (e.g.\ only a single edge has its length increased) then such an algorithm would have no hope of running in close-to-linear time. In order to avoid this, we require that our cuts be (approximately) largest among all length-constrained sparse cuts (for a suitable notion of ``large''). Summarizing, we have the following (conceptually simple) rough outline of our algorithm:
\begin{quote}
    \centering \textit{Until none exist:\\ Cut a $(\approx h, \approx s)$-length $(\approx\phi)$-sparse (approximately) largest cut.}
\end{quote}
The only technical caveat to the above outline is what is hidden by the ``$\approx$''s above. Specifically, our algorithm proceeds in epochs in order to deal with the slacks that occur in the length-constrained setting where after each epoch we relax the cuts we look for by appropriately increasing the length slack $s$ and decreasing the length $h$ and sparsity $\phi$. A more formal version of the above algorithm is given in \Cref{alg:EDsfromCuts}.

Appropriately instantiating the above approach gives us the main result of our work. We state the formal result below; some of the precise definitions are left to Sections \ref{sec:conventions} and \ref{sec:prelim}. Informally, though, the ``node-weighting'' $A$ specifies the subset of our graph we would like to be $(h,s)$-length expanding and our decomposition is ``witnessed'' in the sense that it comes with an embedding which shows how to route with low congestion over low length paths in the graph after applying the decomposition. We discuss $(\leq h, s)$-length below and make use of standard work-depth models of parallel computation (see, e.g.\ \cite{blelloch1996programming}); for sequential algorithms work is equivalent to time.

\begin{restatable}{thm}{mainAlgThm}
\label{thm:expdecomp exist} There is a constant $c >1$ such that given graph $G$ with edge lengths and capacities, $\eps \in \left(\frac{1}{\log^{1/c} N},1 \right)$, node-weighting $A$, length bound
$h \geq 1$, and conductance $\phi>0$, one can compute a witnessed $(\leq h,s)$-length $\phi$-expander
decomposition for $A$ in $G$
with cut and length slack respectively
\begin{align*}
    \kappa = n^{\eps}   \qquad \qquad s= \exp(\poly(1/\eps))
\end{align*}
and work and depth respectively
\begin{align*}
    \wED(A, m) \leq m \cdot  \tilde{O}\left(n^{\poly(\eps)}\cdot \poly(h) \right) \qquad  \dED(A, m) \leq \tilde{O}\left(n^{\poly(\eps)}\cdot \poly(h) \right).
\end{align*}
\end{restatable}
% Similarly to the classic setting, computing such a cut is made possible by (length-constrained variants of) cut-matching games. Likewise, the ``cutting'' operation here is not the removal of edges but simply the more general procedure of increasing their lengths; in short, this relaxation to length increases is necessary in the length-constrained setting since the length-constrained setting is known to exhibit large flow-cut integrality gaps \cite{haeupler2020network}.
\noindent The above improves over the previous  algorithms for length-constrained expander decompositions of  \cite{haeupler2022expander} in several major ways.\footnote{We also note that our algorithm is considerably conceptually simpler than that of \cite{haeupler2022expander} which relied on an intricate ``expander gluing'' framework. }
\begin{enumerate}
    \item \textbf{Trading Off Between Slacks:} Our algorithm allows one to trade off between length slack and cut slack. In particular, by increasing the cut slack $\kappa$, one can decrease the length slack $s$ and vice-versa. Notably, for suitably small constant parameter, the above allows us to achieve length slack $s = O(1)$ and cut slack $n^{\eps}$ in work $m^{1+\eps} \cdot \poly(h)$ and depth $n^{\eps} \cdot \poly(h)$ for any small constant $\eps$. The previous approach of \cite{haeupler2022expander} computed length-constrained expander decompositions with $\kappa = s = \exp(\log^{1-\delta} n)$ for a fixed $\delta$ slightly less than $1$ and so could not produce such decompositions. Crucially, all of the applications of our results (discussed in \Cref{sec:appAndRelated}) will make use of a constant length slack of $s = O(1)$.
    
    % worked for fixed with cut and length slack $2^{O(\log^{1-\delta} n)} = n^{o(1)}$ for fixed $\delta$ slightly less than $1$
    
    % , appropriately instantiated, we show that the above approach gives efficient algorithms trade off between length and cut slack in a way that \textbf{nearly matches existential bounds}. Specifically, we show that for any $\eps > 0$ (that is not much smaller than about $1/\log n$) one can compute.
    \item \textbf{General Lengths and Capacities:} Our algorithm is the first (close-to-linear time) algorithm for computing length-constrained expander decompositions on graphs with general lengths and general edge lengths. The algorithm of \cite{haeupler2022expander} only worked if one assumed that both the capacity and length of every edge is $1$. Even in the classic setting, efficient algorithms for expander decompositions for general capacities are significantly more difficult than those in the unit capacity setting \cite{li2021deterministic}. All of our applications will make use of general lengths and capacities and implementing the above approach for general lengths and capacities---particularly, for general capacities---presents significant difficulties; discussed below.
    \item \textbf{Stronger Routing Guarantees:} The routing guarantees provided by our algorithm are significantly stronger than those of \cite{haeupler2022expander} in two ways. First, our decomposition is what we call a $(\leq h, s)$-length expander decomposition (as opposed to an $(h, s)$-length expander decomposition). In particular, after applying our decomposition we provide the guarantee that any $h$-length demand can be routed by a low congestion multi-commodity flow where the flow paths between vertices $u$ and $v$ at distance $d \leq h$ have length at most $d\cdot s$ as opposed to $h \cdot s$ as in $(h,s)$-length expander decompositions. Notably, if $d \ll h$ (i.e.\ the two vertices are very close initially) then we route between $u$ and $v$ over paths whose length is proportional to $d$ rather than proportional to $h$.
    
    Second, the expander decomposition output by our algorithm is ``strong'' in the sense that after applying the decomposition to graph $G$ to get $G'$, any $h$-length demand in $G'$ can be routed with low congestion over low-length paths in $G'$; that of \cite{haeupler2022expander} was ``weak'' in that $h$-length demands could only be routed over low-length paths \emph{in $G$}.
\end{enumerate}
\noindent Showing that the above approach works requires overcoming two significant challenges:
\begin{quote}
     \centering \centering\textbf{Challenge 1:} \textit{How can we show that cutting large length-constrained sparse cuts\\ quickly yields a length-constrained expander decomposition?}

\end{quote}
and
\begin{quote}
     \centering \textbf{Challenge 2:} \textit{How can we efficiently compute large length-constrained sparse cuts?}
\end{quote}
In what follows, we discuss these challenges and how we overcome them.

\subsubsection{Overcoming Challenge 1: Union of Sparse Length-Constrained Cuts is Sparse}
We discuss how we show that cutting large length-constrained sparse cuts quickly yields a length-constrained expander decomposition.

\paragraph*{The Classic Approach.}
Repeatedly cutting large sparse cuts in order to compute classic expander decompositions is a well-studied approach \cite{saranurak2019expander}. In the classic setting the ``largness'' of a cut is its balance, namely, the volume of the side of the cut with smaller volume.\footnote{The volume of a set of vertices is the sum of their degrees.}

In the classic setting, this approach hinges on the fact that the ``union of sparse cuts is itself a sparse cut.'' In particular, if $C_1, C_2, \ldots$ are a series of cuts where each $C_i$ is a $\phi$-sparse cut after $C_j$ for $j < i$ has been cut, then $\bigcup_i C_i$ is itself a $O(\log n \cdot \phi)$-sparse cut. 

Thus, if each $C_i$ is a $\phi$-sparse cut whose size is within an $\alpha$ factor of the largest $O(\log n \cdot \phi)$-sparse cut and we cut about $\alpha$ of these cuts then we know that after cutting all of these cuts we must have substantially reduced the size of the largest $O(\log n \cdot \phi)$-sparse cut (otherwise $\bigcup_i C_i$ would be a $O(\log n \cdot \phi)$-sparse cut whose size is larger than the largest such cut).

The above union of cuts fact in the classic setting is easily shown. In particular, because we are measuring the size of the cut in terms of the smaller volume side and each time we apply a cut the smaller volume side has at most half of the total volume, we get a depth of $O(\log n)$, leading to the $O(\log n)$ in the above sparsity. In the interest of completeness, we give a proof of this fact in \Cref{sec:unionOfNHCCuts}.

\paragraph*{Issues in the Length-Constrained Setting.}
However, showing a comparable fact for the length-constrained setting is significantly more challenging. In particular, there is no clear notion of the ``side'' of a cut in the length-constrained setting since we are applying length increases, not fully deleting edges. Thus, a proof of the sparsity of the union of sparse length-constrained cuts cannot appeal to a tidy recursion where each time one recurses, one side of the cut reduces by a constant. Further, it is not even clear what the appropriate notion of the ``size'' of a cut is in this setting since, again, there is no ``smaller side'' whose volume we can measure.

In fact, given the definition of an $(h,s)$-length $\phi$-sparse cut, one might not think that the union of a sequence of $(h,s)$-length $\phi$-sparse cuts should be sparse. In particular, recall that an $(h,s)$-length cut $C_i$ is $\phi$-sparse iff there is some large witnessing demand $D_i$ that is $h$-length and unit such that after $C_i$ is applied every pair in the support of $D_i$ is at least $hs$-far apart. However, if we take the union of $(C_1, C_2, \ldots)$ as $C_1+ C_2+ \ldots$, then we cannot witness the sparsity of $C_1+ C_2+ \ldots$ by $D_1 + D_2 + \ldots$: the resulting demand can be arbitrarily far from being unit!

\paragraph*{Result.} Nonetheless, for the length-constrained setting we show that, somewhat surprisingly, the union of a sequence of $(h,s)$-length $\phi$-sparse cuts is itself an $(h,s)$-length $O(\phi \cdot n^{O(1/s)})$-sparse cut (with some slack in $h$ and $s$). More generally, we show that, the union of cuts in a sequence of length-constrained sparse cuts is at least as sparse as the (appropriately-weighted) average sparsity of its constituent cuts. The following gives our formal result (we again defer some of the technical definitions to later sections).

\begin{restatable}[Union of Sparse Moving Cuts is a Sparse Moving Cut]{thm}{unionOfCuts}
	\label{thm:unionOfMovingCuts} Let $(C_1, \ldots, C_k)$ be a sequence of moving cuts where $C_i$ is an $(h,s)$-length $\phi_i$-sparse cut in $G - \sum_{j < i} C_j$ w.r.t.\ node-weighting $A$. Then the moving cut $\sum_i C_i$ is an $(h',s')$-length $\phi'$-sparse cut w.r.t.\ $A$ where $h' = 2h$, $s' = \frac{(s-2)}{2}$ and $\phi' = s^3 \cdot \log^3 n \cdot n^{O(1/s)} \cdot \frac{\sum_i |C_i|}{\sum_i |C_i|/\phi_i}$.
\end{restatable}

\paragraph*{Techniques.} Our result for the length-constrained setting is based on an intriguing connection to parallel algorithms for greedy spanner constructions. In particular, we show that the witnessing demands in a sequence of $(h,s)$-length $\phi$-sparse cuts are analogous to a parallel process for greedily computing a spanner. Recent work \cite{haeupler2023parallel} showed that such parallel processes produce a graph with low---namely about $n^{O(1/s)}$---arboricity. We then make use of the low arboricity of such graphs to  decompose the demands of a sequence of length-constrained sparse cuts into forests and then use each tree in this forest to ``disperse'' the corresponding demand. The result is a demand that can be used to witness the sparsity of the union of $(C_1, C_2, \ldots)$: in particular, it is of approximately the same size as the sum of the $D_i$ and still separated by the union of $(C_1, C_2, \ldots)$ but it is actually unit, unlike the sum of the $D_i$. These results are discussed in \Cref{sec:unionOfHCCuts}.

By making use of the above union of cuts fact for the length-constrained setting and defining an appropriate notion of the size of a length-constrained cut---a notion we call the ``demand-size'' of a cut---we are able to extend the above approach to the length-constrained setting. In particular, we show that that repeatedly cutting an $(h,s)$-length $\phi$-sparse cut which is approximately demand-size-largest quickly yields a length-constrained expander decomposition. 

\subsubsection{Overcoming Challenge 2: Sparse Flows for the Spiral}
Having discussed how it suffices to show that repeatedly cutting large and sparse length-constrained cuts quickly yields a length-constrained expander, we now discuss how to compute these cuts. 

\paragraph*{Classic Approach.} In the classic setting, a well-studied means of cutting large sparse cuts is by what we call ``the spiral.'' In particular, it is known that one can compute large sparse cuts using ``cut matching games.'' Cut matching games, in turn, can be efficiently computed by computing expander decompositions. By the above-mentioned arguments, expander decompositions can be computed using large sparse cuts. In order to avoid a cycle of dependencies, one turns this cycle into a ``spiral.'' In particular, the algorithm is set up so that each time one goes around the cycle of dependencies the input size significantly decreases. 

\paragraph*{Issues in the Length-Constrained Setting.}
A recent work of \cite{ghaffari2022cut} provided a cut matching game that is suitable for our purposes and so, in light of our discussion above, one might hope to implement a similar approach in the length-constrained setting. However, here, the fact that we are interested in general capacities significantly complicates our problem. 

In particular, in the spiral we recurse on graphs produced by cut matching games. Each edge of these graphs, in turn, corresponds to a flow path in the flow decomposition of a flow we computed on our input graph. Thus, if we want to guarantee that when we recurse the input size has significantly gone down, it must be the case that the flows we construct for our cut matching game have low support size; i.e.\ the flow can be decomposed into a small number of flow paths.

This is easy to do in the classic setting but, to our knowledge, prior to our work no such result was known for the length-constrained setting for general capacities. In the length-constrained setting, the relevant notion of flow is $h$-length flows (i.e.\ flows whose flow paths have length at most $h$). Prior work \cite{haeupler2021fast} showed that one can compute $h$-length multi-commodity flows but with support size $\tilde{O}\left(b \cdot \poly(h) \cdot m \right)$ where, roughly, $b$ is the number of commodities. Using this result would lead to a multiplicative blowup of $\poly(h)$ in the total number of edges of our graphs which, unfortunately, cannot be made to work with the spiral.

\paragraph*{Result.} To solve the above issue, we give the first sparse flow algorithms for the length-constrained setting, improving the above sparsity to $\tilde{O}(|E| + b)$ and, notably, so that it does not depend on $h$. Specifically, we show the following (again, see later sections for relevant technical definitions).
\begin{restatable}{thm}{sparseMultiFlows}\label{thm:sparseMultiFlows}
Given a graph $G=(V,E)$ with capacities $\U$, lengths $\l$, length constraint $h \geq 1$, $0 < \eps < 1$ and $\batch$-batchable source, sink pairs $\{(S_i, T_i)\}_i$, one can compute a feasible $h$-length flow cut pair $(F, C)$ of $\{(S_i, T_i)\}_i$ that is $(1 \pm \epsilon)$-approximate in (deterministic) depth $\tilde{O}(\batch \cdot \poly(\frac{1}{\eps}, h))$ and work $m \cdot \tilde{O}(\batch \cdot \poly(\frac{1}{\eps}, h))$ where 
\begin{align*}
    |\supp(F)| \leq \tilde{O}(|E| + b).
\end{align*}
Furthermore, $F = \eta \cdot \sum_{j=1}^k F_j$ where $\eta = \tilde{\Theta}(\epsilon^2)$, $k = \tilde{O}\left(\kappa \cdot \frac{h}{\epsilon^4} \right)$ and $F_j$ is an integral $h$-length $S_i$-$T_i$ flow for some $i$.
\end{restatable}

\paragraph*{Techniques.} In order to show the above result, we make use of a novel ``blaming'' argument. The rough idea is to construct a near-optimal flow where each path in the support of the flow can uniquely point to an edge whose capacity was mostly used up when flow along this flow path was added to our solution. We discuss and prove the above result in \Cref{sec:sparseFlows}.
 
Combining the above flow algorithms with the spiral and our union of cuts fact shows that the above strategy quickly yields a length-constrained expander decomposition. As a corollary, we get efficient algorithms for large length-constrained sparse cuts (since these algorithms form a sub-routine of the spiral). While we feel the above are the main contributions of our work, we note that building up to them requires developing several new ideas and techniques for the length-constrained setting, including proofs of the robustness of length-constrained expanders to edge deletions and the equivalence of several notions of a graph's ``distance'' from being a length-constrained expander.

\section{Applications of Our Work}\label{sec:appAndRelated}

In this section we discuss applications of our work; both those in subsequent work and some corollaries of our results. We note that the first two results directly use the algorithms from this work.

\paragraph*{Application 1: $O(1)$-Approx.\ MC Flow in Close-To-Linear Time (Subsequent Work).} A recent work of \cite{lowStep} gave the first close-to-linear time algorithms to compute an $O(1)$-approximate $k$-commodity flow in almost-linear-time. These algorithms run in time $O((m+k)^{1+\epsilon})$ for arbitrarily small constant $\eps > 0$ \cite{lowStep}. Roughly speaking, these algorithms make use of the ``boosting'' framework of Garg and K{\"o}nemann \cite{garg2007faster} wherein solving multi-commodity flow is, by way of multiplicative-weights-type techniques, reduced to problems on graphs with arbitrary capacities and lengths. Thus, the fact that our algorithms work for arbitrary capacities and lengths are crucial for this later work. Likewise, this work also makes use of our sparse flow algorithms. Lastly, we note that this work uses the fact that our algorithms can output length-constrained expander decompositions with $s=O(1)$ to compute ``low-step multi-commodity flow emulators'' which are, roughly, low-diameter graphs which represent all multi-commodity flows.

% \enote{In order to do Garg konemann we need capacities (for capacitated multi commodity flows) and arbitrary lengths / costs (bc they come from multiplicative weights); none of this exists without sparse multi flows (a double dependency); dynamic paper also needs sparse flows; dynamic paper does length-constrained expander pruning; constant approx for distance oracles and for multi flows come from constant length-constrained expanders; on algorithmic part everything is much harder than Raecke paper}
% \enote{In related work just emphasize that all the cool results for length-constrained expanders use the stuff developed in this paper; skip the routing stuff; would kick out current stuff (other than current stuff on expander decomps)}
% \enote{Chuzhoy paper never gets down to constant / doesn't do capacities}

\paragraph*{Application 2: Distance Oracles (Subsequent Work)} Another recent work \cite{haeupler2024dynamic} makes use of our algorithms for length-constrained expander decompositions to give new distance oracle results. Specifically, this work shows that in a graph with general edge lengths one can maintain a data structure with $n^\eps$ worst-case update time which can answer distance queries between vertices that are $\exp(1/\eps)$-approximate in $\poly(1/\eps) \cdot \log \log n$ time. Crucially, this work makes use of the fact that our algorithms for length-constrained expander decompositions can trade off between cut and length slack (the $\eps$ in their work and ours are roughly analogous). The previous best result along these lines is that of  Chuzhoy and Zhang \cite{chuzhoy2023new} which gave a $(\log \log n)^{2^{O(1/\eps^3)}}$-approximate fully-dynamic deterministic all-pair-shortest-path distance oracle with \emph{amortized} $n^{\eps}$ update time using the below-mentioned well-connected graphs.

\paragraph*{Application 3: Capacitated Length-Constrained Oblivious Routing.} Lastly, we note that we obtain the first close-to-linear time algorithms for length-constrained oblivious routing on graphs with general capacities (with constant length slack). % One of the most successful applications of length-constrained expander decompositions is that of length-constrained routing. Traditional routing techniques (such as classic expander decompositions) provide a means of efficiently routing demands while minimizing congestion but provide no guarantees on the length of paths over which routing is performed. Such guarantees are inadequate if both congestion and path length determine routing speed as is the case in many models of distributed computation \cite{leighton1988universal, peleg2000near}.
Specifically, in length-constrained oblivious routing the goal is to fix for each pair of vertices a flow over $h'$-length paths so that for any demand the induced flow is always congestion-competitive with the minimum congestion flow over $h$-length paths routing this demand. We refer to $h'/h$ as the length competitiveness of such a routing scheme. \cite{ghaffari2021hop} proved the existence of length-constrained oblivious routing schemes that simultaneously achieve $\poly \log n$ length and congestion competitiveness (not using length-constrained expander decompositions). However, this result did not provide an efficient algorithm for computing such a scheme. \cite{haeupler2022expander} addressed this by observing that one can use length-constrained expander decompositions to compute length-constrained oblivious routing schemes in time $m^{1+o(1)}$ that simultaneously achieve $n^{o(1)}$ length and congestion competitiveness.

Applying the techniques of \cite{haeupler2022expander} and our algorithms for length-constrained expander decompositions, it follows that for any $\eps > 0 $ (not too small as in \Cref{thm:expdecomp exist}) one can compute in time $m^{1 + \eps}$ an oblivious length-constrained routing scheme that achieves length and congestion slack $\exp(\poly(1/\eps))$ and $n^{O(\eps)}$ respectively. Note that setting $\eps = 1/\sqrt{\log n}$ generalizes the \cite{haeupler2022expander}. More importantly, since our algorithms for length-constrained expander decompositions work for the general capacities case, so too do our oblivious routing schemes unlike those of \cite{haeupler2022expander}. Furthermore, if $\eps=O(1)$ then we achieve constant length competitiveness with sub-linear congestion competitiveness. Not only is this the first efficient algorithms for such a routing scheme, but even the existence of routing schemes with the stated competitiveness was not known prior to our work.
\future{Actually write this stuff up}

\section{Additional Related Work}

We give a brief overview of additional related work. 

\subsection{Applications of Expander Decompositions}

We start by describing some additional work on the applications of expander decompositions.

Areas of use include linear systems \cite{spielman2004nearly}, unique games \cite{arora2015subexponential,trevisan2005approximation,raghavendra2010graph}, minimum cut \cite{kawarabayashi2018deterministic}, and dynamic algorithms \cite{nanongkai2017dynamic}. Some of the long-standing-open questions which have recently been solved thanks to expander decompositions include: deterministic approximate balanced cut in near linear time (with applications to dynamic connectivity and MST) \cite{chuzhoy2020deterministic}, subquadratic time algorithms for bipartite matching \cite{van2020bipartite}, and deterministic algorithms for global min-cut in almost linear time \cite{li2021deterministic}.

\subsection{Parallel Work on Well-Connected Graphs}

A parallel and independent series of works  by Chuzhoy \cite{CSODA2023} and Chuzhoy  and Zhang \cite{chuzhoy2023new} develops notions similar to $(h,s)$-length $\phi$-expanders, sparse moving cuts, and cut-matching games for length-constrained expanders.

In particular, the concept of $(\nu,d)$-well-connected graphs (proposed after the length-constrained expanders of \cite{haeupler2022expander}) is similar in spirit to $(h,s)$-hop $\phi$-expanders when restricted to graphs with diameter less than $h$. By way of their parameters, both types of graphs provide separate control over congestion ($\nu$ for well-connected graphs and $\frac{1}{\phi}$ for length-constrained expanders) and length ($d$ for well-connected graphs and $hs$ for length-constrained expanders) of routing paths, but there are some technical differences. Like the focus on small or constant $s$ in this paper and \cite{haeupler2022expander}, the well-connected graphs of \cite{CSODA2023} are particularly interesting when guaranteeing routing via very short sub-logarithmic paths. %\enote{Having trouble parsing this sentecne}

One important difference between well-connected graphs and length-constrained expanders seems to be that $(h,s)$-length expanders are an expander-like notion that applies to arbitrary graphs with (potentially) large diameter, unlike well-connected graphs. In particular, they provide low congestion ($\tilde{O}(\frac{1}{\phi})$) and low length ($hs$) routing guarantees for every unit demand between $h$-close nodes where $h$ and $hs$ are both independent and potentially much smaller than the diameter of the entire graph. On the other hand, well-connected graphs provide routing guarantees between all nodes and, as such, seem to correspond more closely to what we call routers in this work; see \Cref{def:router}. Overall, it is unclear if the notion of a length-constrained expander decomposition of a general graph that stands at the center of this paper relates in an immediate way to the notions of \cite{CSODA2023} and \cite{chuzhoy2023new}.

% In~\cite{chuzhoy2023new} Chuzhoy and Zhang furthermore give a $(\log \log n)^{2^{O(1/\eps^3)}}$-approximate fully-dynamic deterministic all-pair-shortest-path distance oracle with amortized $n^{\eps}$ update time based on well-connected graphs, the RecDynNC reduction of~\cite{chuzhoy2021decremental}, and the distance-matching game of~\cite{CSODA2023}. 

% \subsection{Subsequent Works and Applications of our Results}

% \paragraph*{Applications of our Results.}

\future{Talk about dynamic paper: \cite{haeupler2024dynamic}}

% \subsection{Expander Decompositions}

% We give a brief overview of the applications of expander decompositions.

% Areas of use include linear systems \cite{spielman2004nearly}, unique games \cite{arora2015subexponential,trevisan2005approximation,raghavendra2010graph}, minimum cut \cite{kawarabayashi2018deterministic}, and dynamic algorithms \cite{nanongkai2017dynamic}. Some of the long-standing-open questions which have recently been solved thanks to expander decompositions include: deterministic approximate balanced cut in near linear time (with applications to dynamic connectivity and MST) \cite{chuzhoy2020deterministic}, subquadratic time algorithms for bipartite matching \cite{van2020bipartite}, and deterministic algorithms for global min-cut in almost linear time \cite{li2021deterministic}.

\subsection{Routing in Expanders}
As much of our work deals with finding good routes in (length-constrained) expanders, we briefly review some work on routing in expanders.

As mentioned earlier, expanders admit low congestion good multi-commodity flow solutions \cite{leighton1999multicommodity}. A closely related problem is that of finding edge-disjoint paths. \cite{peleg1987constructing} showed given any $\Omega(1)$-expander and $n^\eps$ pairs for some small constant $\eps > 0$, it is possible to find edge-disjoint paths between these pairs in polynomial time. This was later improved by \cite{frieze2001edge} to $\Theta(\frac{n}{\log n})$ edge-disjoint paths in regular expanders. Many other works have studied this problem \cite{broder1992existence,kleinberg1996short,valiant1990general,leighton1994packet}.

Along similar lines, \cite{ghaffari2017distributed} introduced the notion of ``expander routing'' which allows each node $v$ to exchange $\deg(v)$ messages with nodes of its choosing in the CONGEST model of distributed computation in about $n^{o(1)}/\phi$ time on $\phi$-expanders. Using this approach, \cite{ghaffari2017distributed} showed that a minimum spanning tree (MST) can be constructed in $\poly(\phi^{-1}) \cdot n^{o(1)}$ distributed time in $\phi$-expanders, bypassing the earlier-mentioned $\Omega(\sqrt{n} / \log n)$ lower bound \cite{peleg2000near,das2011distributed} for small $\phi$ networks. This was later extended by \cite{ghaffari2018new} to a much wider class of optimization problems.

What use are expander decompositions for routing in graphs that are not expanders? 
\cite{goranci2021expander} showed that expander decompositions can be used to route arbitrary demands in a tree-like, oblivious and $n^{o(1)}$-congestion-competitive manner. %Specifically, they introduced the idea of an expander hierarchy computed as follows: given an arbitrary graph $G$ compute a $\phi$-expander decomposition $C$\footnote{Really, a stronger version of expander decompositions (boundary-linked expander decomposition) is used.} where $G-C$ has connected components $\{S_i\}_i$; contract each $S_i$ in $G$; repeat until $G$ is a single vertex. Using appropriate $\phi$ and routing according to this hierarchy provides a tree-like and oblivious way of routing any $D$ in the input graph with $n^{o(1)}$-competitive congestion. 
There is also a great deal of related work on related ``tree flow sparsifiers'; see, e.g.\ \cite{racke2002minimizing}.

\section{Notation and Conventions}\label{sec:conventions}
Before moving on to a more formal description of our results we introduce the notation and conventions that we use throughout this work.

% \subsection{Graphs and General Notations}
\paragraph*{Graphs.}
Let $G=(V,E)$ be a graph with $n := |V|$ vertices and $m:= |E|$ edges. %integral positive capacities, at most one capacitated self-loop at any vertex. %We use $V_G$ to denote the vertex set of $G$ and $n = |V_G|$ for the number of vertices. We use $E_G \subseteq \binom{V}{2}$ for the edge set.
%We denote by $u_{G}:E\rightarrow\naturalnumbers$ the edge capacity function of $G$. 
%The degree of a vertex $v$ is $\deg_{G}(v)=\sum_{(v,w)\in E}u_{G}(v,w)$. Note that a self-loop $(v,v)$ contributes $u_{G}(v,v)$ to the degree of $v$. 
By default, $G$ is undirected, and allowed to have self-loops but not parallel edges.
For each vertex $v\in V(G)$, we denote by $\deg_{G}(v)$ the
degree of $v$ in $G$. For $S\subset V$ we let the volume of $S$ be $\vol(S)=\sum_{v\in S} \deg_G(v)$. $E(S, V\setminus S)$ gives all edges with exactly one endpoint in $S$. We drop $G$ subscript when it is clearly implied.
We use standard graph terminology like adjacency, connectivity, connected components, as defined in, e.g.\  \cite{west2001introduction}.

\paragraph*{Edge Values and Path Lengths}
We will associate two functions with the edges of graph $G$. We clarify these here.
\begin{enumerate}
    \item \textbf{Capacities:} We will let $\U = \{\U_e\}_e$ be the capacities of edges of $E$. These capacities will specify a maximum amount of flow (either length-constrained or not) that is allowed over each edge. Throughout this work we imagine each $\U_e$ is in $\mathbb{Z}_{\geq 0}$.
    \item \textbf{Lengths:} We will let $ \l = \{\l_e\}_e$ be the \emph{lengths} of edges in $E$. These lengths will be input to our problem and determine the lengths with respect to which we are computing length-constrained expanders and length-constrained expander decompositions. Throughout this work we imagine each $\l_e$ is in $\mathbb{Z}_{> 0}$. We will let $d_\l(u,v)$ or $d_G(u,v)$ or just $d(u,v)$ when $G$ or $\l$ are clear from context give the minimum value of a path in $G$ that connects $u$ and $v$ where the value of a path $P$ is $\l(P) := \sum_{e \in P} \l(e)$. We let $\ball(v,h) := \{u : d_G(v,u) \leq h\}$ be all vertices within distance $h$ from $v$ according to $\l_G$. Prior works primarily used length-constrained expanders in the context of unit-length graphs and talked about $h$-hop expanders and $h$-hop expander decompositions. We deal with general lengths and use ``length'' instead of ``hop'' where appropriate.
    % \item \textbf{Cuts:} We will let $\w = \{\w_e\}_e$ stand for the weights of edges in $A$. These weights will be given by our length-constrained cut solutions. Throughout this work each $\w_a$ will be in $\mathbb{R}_{> 0}$. 
\end{enumerate}
\future{Make $\l$ vs $G$ consistent}
% \enote{Cuts vs weights??}

% In general we will treat a path $P = ((v_1, v_2), (v_2, v_3), \ldots)$ as series of consecutive edges in $E$ (all oriented consistently towards one endpoint). For any one of these weighting functions $\phi \in \{\l, \U, \w\}$, we will let $d_\phi(u,v)$ give the minimum value of a path in $G$ that connects $u$ and $v$ where the value of a path $P$ is $\phi(P) := \sum_{e \in P} \phi(e)$. That is, we think of $d_\phi(u,v)$ as the distance from $u$ to $v$ with respect to $\phi$. If $G$ also has weights $\w$ then we let $d_\w^{(h)}(u,v) := \min_{P \in \mcP_h(u,v)}\w(P)$ give the minimum weight of a length at most $h$ path connecting $u$ and $v$. For vertex sets $W, W' \subseteq V$ we define $d_\w^{(h)}(W,W') := \min_{P \in \mcP_h(W, W')}\w(P)$ analogously. %We will refer to paths which minimize $\w$ as lightest paths (so as to distinguish them from e.g.\ shortest paths with respect to $\l$). We let $\ball(v,h) := \{u : d_G(v,u) \leq h\}$ be all vertices within distance $h$ from $v$ according to $\l_G$. Prior works primarily used length-constrained expanders in the context of unit-length graphs and talked about $h$-hop expanders and $h$-hop expander decompositions. We deal with general lengths and use ``length'' instead of ``hop'' where appropriate.

%
%{\bfseries Remark on Polynomially Bounded Object Sizes: }\\

\paragraph*{Polynomial Size Objects ($N$).} All objects (e.g. graphs with self-loops) defined in this paper are assumed to be of size polynomial in $n$, i.e., for a fixed large enough constant $c_{\max}$ we assume that all objects are of size at most $N < n^{c_{\max}}$ where $n$ is the number of vertices in the underlying graph. 
%In particular we only consider graphs\footnote{While our constructions and algorithms graphs sometimes grow an input graph or other objects a bit (e.g., by adding self-loops) such a growth is never by more than a polynomial factor overall.} with $m < N$. \enote{$m_g$?} 
This polynomial bound on object sizes in this paper also allows us to treat logarithmic upper bounds in the sizes of these objects as essentially interchangeable, e.g., for any constant bases $b,b'$ we have that $\log_b n \leq \log_{b'} n^{c_{\max}} = \Theta(\log_{b'} N)$. Throughout the paper we therefore use $O(\log N)$ without any explicitly chosen basis to denote such quantities. That is, all $O$-notation depends on $c_{\max}$. We use $\tilde{O}$, $\tilde{\Omega}$ and $\tilde{\Theta}$ notation to hide $\poly(N)$ factors.

\paragraph{Graph Arboricity.} Give graph $G$, a \emph{forest cover} of $G$ consists of sub-graphs $F_1, F_2, \ldots, F_k$ of $G$ which are forests such that for $j\neq i$ we have $F_i$ and $F_j$ are edge-disjoint and every edge of $G$ occurs in some $F_i$. $k$ is called the size of the forest cover and the \emph{arboricity} $\alpha$ of graph $G$ is the minimum size of a forest cover of $G$.

\paragraph{Flows.} A \emph{(multicommodity) flow} $F$ in $G$ is a function that assigns to each simple path $P$ in $G$ a flow value $F(P)\ge0$. We say $P$ is a \emph{flow-path} of $F$ if $F(P)>0$. 
%$P$ is a \emph{$(v,w)$-flow-path} of $F$ if $P$ is both a $(v,w)$-path and a flow-path of $F$. The $(v,w)$-flow $f_{(v,w)}$ of $F$ is such that $f_{(v,w)}(P)=F(P)$ if $P$ is a $(v,w)$-path , otherwise $f_{(v,w)}(P)=0$. The \emph{value}of $f_{(v,w)}$ is $\val(f_{(v,w)})=\sum_{P\text{ is a }(v,w)\text{-path}}F(P)$. 
The value of $F$ is $\val(F) = \sum_{P} F(P)$.
We let $F(e) := \sum_{P \ni e} F(P)$ be the total flow through edge $e$.
The \emph{congestion of $F$ on an edge $e$} is defined to be 
\begin{align*}
    \congest_{F}(e):=\sum_{P:e\in P}F(P) / \U_e,
\end{align*}
i.e.,\ the total flow value of all paths going through $e$ divided by the capacity of $e$. The \emph{congestion} of $F$ is 
\begin{align*}
    \congest(F):=\max_{e\in E(G)}\congest_{F}(e).
\end{align*}
The \emph{length} (a.k.a.\ dilation) of $F$ is 
\begin{align*}
    \dil (F)=\max_{P:F(P)>0}\ell(P),
\end{align*}
i.e.,\ the maximum length of all flow-paths of $F$. The \emph{(maximum) step} of $F$ is 
\begin{align*}
    \step_{F}=\max_{P:F(P)>0}|P|,
\end{align*}
i.e.,\ the maximum number of edges in all flow-paths of $F$. Given $S, T \subseteq V$ we say that $F$ is an $S$-$T$ flow if each path in its support is from a vertex in $S$ to a vertex in $T$. Given source, sink pairs $\{S_i, T_i\}_i$, we say that $F$ is an $\{S_i, T_i\}_i$ flow if each path in its support is from a vertex in some $S_i$ to a vertex in $T_i$. We say that $F$ is feasible (with respect to capacities $\U$) if $\congest(F) \leq 1$.
%We sometimes write $\congest_{G,F},\length_{G,F},\step_{G,F}$ to emphasize that they are with respect to $G$.

%\enote{Is $F(u,v)$, the flow between 2 vertices defined anywhere?}

%\subsection{}
%All demands in this paper are directed, non-negative, and real-valued demands between nodes.

\paragraph{Demands.} A \emph{demand} $D:V\times V\rightarrow\mathbb{R}_{\ge0}$ assigns
a non-negative value $D(v,w) \ge 0$ to each ordered pair of vertices in $V$. The size of a demand is written as $|D|$ and is defined as $\sum_{v,w} D(v,w)$. 
A demand $D$ is called \emph{$h$-length} constrained (or simply \emph{$h$-length} for short) if it assigns positive values only to pairs that are within
distance at most $h$, i.e., $D(v,w)>0$ implies that $d_{G}(v,w)\le h$. We call a demand integral if all $D(v,w)$ are integers and empty if $|D|=0$.
Given a flow $F$, the \emph{demand routed by $F$} is denoted
by $D_{F}$ where, for each $u,v\in V$,\emph{ }$D_{F}(u,v)=\sum_{P\text{ is a }(u,v)\text{-path}}F(P)$
is the value of $(u,v)$-flow of $F$. We say that a \emph{demand $D$ is routable in $G$ with congestion $\eta$ and dilation $h$} iff there exists a flow $F$ in $G$ where $D_{F}=D$, $\congest_{F}\le\eta$ and $\dil(F) \le h$. We say that demand $D$ is routable with dilation inflation $s$ if there is a flow $F$ that routes $D$ and for every $u$ and $v$ in the support of $F$, every path in the support of $F$ has length at most $s \cdot d(u,v)$. We say that an $\alpha$ fraction of $D$ is routable if there exists a flow $F$ with $\val(F) \geq \alpha \cdot |D|$ where the flow $F$ sends from $u$ to $v$ is at most $D(u,v)$ for every $u,v \in V$. %\enote{Can't be any flow; should be one routing the input demand; easier if we define things in terms of sub-demands and then argue they're large and routable}

\paragraph{Graph Embeddings.} We will adopt the convention of an embedding of one graph into another being a mapping of each edge to a flow. Specifically, an $h$-length embedding of edge-capacitated graph $H = (V, E)$ into edge-capacitated graph $G$ is defined as follows. Let $D_H$ be the demand that for each edge $e = \{u,v\}$ with capacity $\U_e$ sends $\U_e$ demand from $u$ to $v$ and vice versa. Then an $h$-length embedding of $H$ into $G$ is an $h$-length flow in $G$ that routes $D_H$.

%\enote{Can we force node weightings to never assign more than degree to a node?}

%
%DEFINE THE DEGREE NODE-WEIGHTINGS FOR (CAPACITATED) EDGE SETS, FOR MOVING CUTS (OR DO THIS LATER?), AND FOR A GRAPH G (SIMPLY THE NODE-WEIGHTING OF ALL THE GRAPHS EDGES). %For any moving cut $C$, the \emph{degree with respect to $C$} of a vertex $v$ is defined as 
%\[
%\deg_{C}(v)=\sum_{e\text{ incident to }v}u(e)\cdot C(e).
%\]
%

\section{Preliminaries}\label{sec:prelim}

In this section we review key definitions and theorems from previous work of which we make use.

\subsection{Classic Expanders and Expander Robustness}\label{sec:regularExpanders}

We summarize some (mostly standard) definitions of cut sparsity and classic expanders. As this work will mostly deal with cuts understood as a collection of edges (rather than a collection of vertices), we provide edge-centric definitions. We begin with such a definition for cuts and cut sparsity.

\future{Sparsity should be conductance}

\begin{definition}[Classic Cut]\label{dfn:cutSparsity}
    Given connected graph $G = (V,E)$, a (classic) cut is a set of edges $C \subseteq E$ such that $G - C := (V, E \setminus C)$ contains at least $2$ connected components
\end{definition}

\begin{definition}[(Classic Edge) Cut Sparsity]\label{dfn:normalSpars}
     Given graph $G$, the sparsity of (classic) cut $C$ is
    \begin{align*}
        \phi(C) := 
        % \begin{cases}
            % 0 & \text{if $|\mathcal{S}_C| = 1$}\\
            |C| / \sum_{S_C \in \mathcal{S}_C} \vol(S_C)
        % \end{cases}
    \end{align*}
        where $\mcS_C$ is all connected components of $G-C$ except for the connected component of maximum volume. We refer to $\mathcal{S}_C$ as the witness components of $C$.
\end{definition}

\noindent Notice that the above definition of cut sparsity is equivalent to the vertex cut definition provided in \Cref{sec:intro} provided $C$ separates the graph into two components. 

The following formalizes classic expanders.

\begin{definition}[Classic Expander]\label{def:normal-expander}
A graph $G=(V, E)$ is a $\phi$-expander if the sparsity of every cut $C \subseteq E$ is at least $\phi$.
% and $S\subset V$, we define the sparsity of $S \subseteq V$ as $\frac{|E(S, V\setminus S)|}{\min\{\vol(S), \vol(V\setminus S)\}}$. The conductance of $G$ is the smallest sparsity of any $S \subseteq V$. We say $S$ is $\phi$-expanding if its sparsity is at least $\phi$ and that $G$ is a $\phi$-expander if every $S \subseteq V$ is $\phi$-expanding.
\end{definition} 

% \begin{definition}[Sparsity, Conductance, Expanders]\label{def:normal-expander}
% Given a graph $G=(V, E)$ and $S\subset V$, we define the sparsity of $S \subseteq V$ as $\frac{|E(S, V\setminus S)|}{\min\{\vol(S), \vol(V\setminus S)\}}$. The conductance of $G$ is the smallest sparsity of any $S \subseteq V$. We say $S$ is $\phi$-expanding if its sparsity is at least $\phi$ and that $G$ is a $\phi$-expander if every $S \subseteq V$ is $\phi$-expanding.
% \end{definition} 

For the sake of comparison to our results in the length-constrained setting, we give the following result summarizing the robustness of expanders to edge deletions.

%\subsubsection{(Normal) Expander Pruning}
\begin{theorem}[Theorem 1.3 of \cite{saranurak2019expander}]\label{thm:expPrun} Let $G = (V, E)$ be a $\phi$-expander with $m$ edges, let $D \subseteq E$ be a collection of edges and let $G' = (V, E \setminus D)$ be $G$ with $D$ deleted. Then there is a set $P \subseteq V$ such that:
\begin{enumerate}
    \item $G'[V \setminus P]$ is $\frac{\phi}{6}$-expanding;
    \item $\vol(P) \leq \frac{8}{\phi} \cdot |D|$.
\end{enumerate}
\end{theorem}

\future{expander pruning even works online. do we (not) need this?}

\subsection{Length-Constrained Cuts (Moving Cuts)}
We now recall formal definitions of length-constrained cuts from \cite{haeupler2020network} which will allow us to define length-constrained expanders and length-constrained expander decompositions.

The following is the length-constrained analogue of a cut.
\begin{definition}[Length-Constrained Cut (a.k.a.\ Moving Cut)~\cite{haeupler2020network}]\label{def:movingcut}
An $h$-length moving cut (a.k.a. $h$-length cut) $C: E \mapsto \{0,\frac{1}{h},\frac{2}{h},\dots,1\}$ assigns to each edge $e$ a fractional cut value between zero and one which is a multiple of $\frac{1}{h}$. The \emph{size} of $C$ is defined
as $|C|=\sum_{e} \U_e \cdot C(e)$. The length increase associated with the $h$-length moving cut $C$ is denoted with $\l_{C,h}$ and defined as assigning an edge $e$ the length $\l_{C,h}(e) = h\cdot C(e)$. Any moving cut which
only assigns cut values equal to either $0$ or $1$ is called a pure moving cut. 
\end{definition}

We can understand length-constrained cuts as cutting in one of two ways. First, we can consider them as cutting apart vertex sets for which they cover all $h$-length paths between these vertex sets.
\begin{definition}[Length-Constrained $S$-$T$ and $\{S_i,T_i\}_i$ Cuts]
    Given vertex subsets $S, T \subseteq V$, we say that $h$-length cut $C$ is an $h$-length $S$-$T$ cut if each $h$-length $S$-$T$ path $P$ satisfies $C(P) \geq 1$. Likewise, given vertex subset pairs $\{S_i, T_i\}_i$ we say that $C$ is an $h$-length $\{S_i, T_i\}_i$ cut if it is an $S_i$-$T_i$ cut for each $i$.
\end{definition}
\noindent By strong duality the size of the minimum size $h$-length $\{S_i,T_i\}$ cut is equal to the value of the maximum value feasible $\{S_i,T_i\}$ flow; see e.g.\ \cite{haeupler2021fast}.\footnote{Really, strong duality requires that $C$ assigns general values to edges (not just values that are multiples of $\frac{1}{h}$); this nuance can be ignored in this work as we are only interested in approximately optimal cuts.} As such, we will say that a pair of $h$-length $\{S_i, T_i\}_i$ flow and cut $(F, C)$ is $(1 \pm \eps)$-approximate for $\eps \geq 0$ if the cut certifies the value of the length-constrained flow up to a $(1 - \epsilon)$; i.e.\ if $(1-\epsilon)\cdot|C| \leq \val(F)$. 

% \bnote{scale by capacity}

Second, if we interpret length-constrained cuts as length increases, then we can understand them as cutting apart vertices that are made sufficiently far apart (i.e.\ separated). Specifically, consider the following is the result of applying a length-constrained cut in a graph.
\begin{definition}[$G-C$]
For a graph $G$ with length function $l_G$ and moving cut $C$, we denote with $G-C$ the graph $G$ with length function $l_{G-C} =  l_G + \l_{C,h}$. We refer to $G-C$ as the \emph{graph $G$ after cutting}
$C$ or \emph{after applying the moving cut} $C$.
\end{definition}

\noindent Then, the following gives the appropriate length-constrained analogue of disconnecting two vertices or a demand by making the demand pairs sufficiently far apart (i.e.\ separated).
\begin{definition}[$h$-Length Separation]
Let $C$ be an $h$-length moving cut. We say two node $v,v' \in V$ are $h$-length separated by $C$ if their distance in $G-C$ is larger than $h$, i.e., if $d_{G-C}(v,v')>h$.
\end{definition}
Observe that if $C$ is an $h$-length $\{u\}$-$\{v\}$ cut then it always $h$-length separates $u$ and $v$. However, $C$ might $h$-length separate nodes $u$ and $v$ even if it is not an $h$-length $\{u\}$-$\{v\}$ cut; e.g.\ in the case where $u$ and $v$ are nearly $h$-far in $G$.

\begin{definition}[$h$-Length Separated Demand]\label{dfn:sepDem}
For any demand $D$ and any $h$-length moving cut $C$, we define the amount of $h$-length separated demand as the sum of demands between vertices that are $h$-length separated by $C$. We denote this quantity with $\sep_{h}(C,D)$,
i.e., $$\sep_{h}(C,D) = \sum_{u,v : d_{G-C}(u,v)>h} D(u,v).$$
\end{definition}

\noindent Using demand separation, we can carry over cut sparsity to the length-constrained setting.
\begin{definition}[$h$-Length Sparsity of a Cut $C$ for Demand $D$]\label{dfn:CDSparse}
For any demand $D$ and any $h$-length moving cut $C$ with $\sep_{h}(C,D)>0$, the $h$-length sparsity of $C$ with respect to $D$ is the ratio of $C$'s size to how much demand it $h$-length separates i.e., 
$$\spa_{h}(C,D) = \frac{|C|}{\sep_{h}(C,D)}.$$
\end{definition}

\noindent Note that if a demand $D$ has any demand between vertices that have length-distance exceeding $h$ then the empty cut has an $h$-length sparsity for $D$ which is equal to zero. For all other demands, i.e., for any non-empty $h$-length demand $D$ with $h < h'$, the $h'$-length sparsity of any cut $C$ for $D$ is always strictly positive. 
%We define the $h'$-length sparsity of any $h'$-length moving cut $C$ for a demand $D$ to be $\infty$ if $\sep_{h'}(C,D)=0$. This includes the empty cut.

% \begin{definition}[$h$-Length Sparsity for a Demand]
% We say $C$ is a sparsest $h$-length cut for $D$ if $C$ minimizes $\spa_{h}(C,D)$
% among all $h$-length moving cuts. We say the demand $D$ has $h$-length sparsity $\phi$ if the sparsest $h$-length cut $C$ for $D$ has
% $\spa_{h}(C,D)=\phi$.
% \end{definition}

\subsection{Length-Constrained Expanders}\label{sec:h-length-expander-def}
We now move on to formally defining length-constrained expanders. Informally, they are graphs with no sparse length-constrained cuts.

We begin by introducing the notion of node-weightings which will give us a formal way of defining what it means for a subset of a graph to be a length-constrained $\phi$-expanding.

\begin{definition}[Node-Weightings]
    A \emph{node-weighting} $A:V\rightarrow\mathbb{R}_{\ge0}$ of $G$ assigns a value $A(v)$ to a vertex $v$. The \emph{size} of $A$ is denoted by $|A|=\sum_{v}A(v)$. For two node-weightings $A,A'$ we define $\min(A,A')$, $A-A'$ and $A+A'$ as pointwise operations and let $\supp(A) := \{v : A(v) > 0\}$.
\end{definition}

\noindent The following summarizes the demands we consider for a particular node-weighting.
\begin{definition}[Demand Load and Respect]
The load of a demand $D$, denoted with $\load(D)$, is the node-weighting which assigns the node $v$ the weight $\max\{\sum_{w \in V}D(v,w),\sum_{w \in V}D(w,v)\}$.  We write $A \prec A'$ if $A$ is pointwise smaller than $A'$. We say demand $D$ is \emph{$A$-respecting} if $\load(D) \prec A$. We say that $D$ is \emph{degree-respecting} if $D$ is $\deg_{G}$-respecting. 
\end{definition}

\noindent Having defined node-weightings and their corresponding demands, we can now define their sparsity.

\begin{definition}[$(h,s)$-Length Sparsity of a Cut w.r.t.\ a Node-Weighting]\label{def:sparsity}
The $(h,s)$-length sparsity of any $hs$-length moving cut $C$ with respect to a node-weighting $A$ is defined as:
$$\spa_{(h,s)}(C,A) = \min_{A\text{-respecting h-length demand}\ D} \spa_{s \cdot h} (C,D).$$
\end{definition}
\noindent We refer to the minimizing demand $D$ above as the demand \emph{witnessing} the sparsity of $C$ with respect to $A$. We can analogously define the length-constrained conductance of a node-weighting.

\begin{definition}[$(h,s)$-Length Conductance of a Node-Weighting]\label{def:conductance}
The $(h,s)$-length conductance of a node-weighting $A$ in a graph $G$ is defined as the $(h,s)$-length sparsity of the
sparsest $hs$-length moving cut $C$ with respect to $A$, i.e.,
$$\operatorname{cond}_{(h,s)}(A) = \min_{hs\text{-length moving cut } C} \spa_{(h,s)}(C,A).$$ 
\end{definition}

When no node-weighting is mentioned then $(h,s)$-length sparsity and $(h,s)$-length conductance are defined with respect to the node-weighting $\deg_G$ which gives vertex $v$ value $\deg_G(v)$. In other words, the $(h,s)$-length sparsity of an $hs$-length moving cut $C$ in $G$ is defined as $\spa_{(h,s)}(C,G) = \min_{D} \spa_{hs} (C,D)$ where the minimum is taken over all degree-respecting $h$-length demands. Similarly, the $(h,s)$-length conductance of $G$ is defined as $\operatorname{cond}_{(h,s)}(G) = \min_{C} \spa_{(h,s)}(C,\deg_G)$ where the minimum is taken over all $hs$-length moving cuts.

With the above notion of conductance, we can now define when $A$ is length-constrained expanding and when $G$ is length-constrained expander.
\begin{definition}[$(h,s)$-Length $\phi$-Expanding Node-Weightings]
We say a node-weighting $A$ is $(h,s)$-length $\phi$-expanding in $G$ if the $(h,s)$-length conductance of $A$ in $G$ is at least $\phi$.
\end{definition}
\noindent Equivalently to the above, we will sometimes say that $G$ is an $(h,s)$-length $\phi$-expander for $A$. Applying the above definition to $\deg_G$ gives our formal definition of length-constrained $\phi$-expanders.
\begin{definition}[$(h,s)$-Length $\phi$-Expanders]
A graph $G$ is an $(h,s)$-length $\phi$-expander if $\deg_G$ is $(h,s)$-length $\phi$-expanding in $G$.
\end{definition}

The above definition of length-constrained expanders characterizes them in terms of conductance. The below fact from \cite{haeupler2022expander} (see their Lemma 3.16) exactly characterizes length-constrained expanders as those graphs that admit both low congestion and low dilation routings.

\begin{theorem}
[Routing Characterization of Length-Constrained Expanders, \cite{haeupler2022expander}]\label{thm:flow character} Given graph $G$ and node-weighting $A$, for any $h \geq 1$, $\phi < 1$ and $s \geq 1$ we have:
\begin{itemize}
\item \textbf{Length-Constrained Expanders Have Good Routings} If $A$ is $(h,s)$-length $\phi$-expanding in $G$, then every $h$-length $A$-respecting demand can be routed in $G$ with congestion at most $O(\log(N)/\phi)$ and dilation at most $s\cdot h$.
\item \textbf{Not Length-Constrained Expanders Have a Bad Demand} If $A$ is not $(h,s)$-length $\phi$-expanding in $G$, then some $h$-length $A$-respecting demand cannot be routed in $G$ with congestion at most $1/2\phi$ and dilation at most $\frac{s}{2}\cdot h$.
\end{itemize}
\end{theorem}

%B: Look under Routers: An expander is a $O(log n/phi)$-step $O(log n/phi)$-router
%\enote{I need a lemma about routing in normal expanders a la the below; I just picked a length slack that I assume is large enough:}
%\begin{lemma}[???]\label{lem:regExpEqualslengthExp}
%Suppose $G$ is a (normal) $\phi$-expander. Then $G$ is a $\left(\frac{\log N}{\phi}, 100 \right)$-length $\phi$-expander.
%\end{lemma}

\subsection{At Most Length-Constrained Expanders}\label{sec:proportional-h-length-expanders}

\begin{definition}[$(\leq h,s)$-Length $\phi$-Expanding Node-Weightings]
We say a node-weighting $A$ is $(\leq h,s)$-length $\phi$-expanding in $G$ if the $(h',s)$-length sparsity of $A$ in $G$ is at least $\phi$ for every $h' \leq h$.
\end{definition}

\future{Note in intro that we get stronger results of at most length expanders}

\begin{definition}[$(\leq h,s)$-Length $\phi$-Sparsity of Cut]
We say a node-weighting $A$ is $(\leq h,s)$-length $\phi$-expanding in $G$ if the $(h',s)$-length sparsity of $A$ in $G$ is at least $\phi$ for every $h' \leq h$.
\end{definition}

\subsection{Length-Constrained Expander Decompositions}

Having defined length-constrained expanders, we can now define length-constrained expander decompositions. Informally, these are simply moving cuts whose application renders the graph a length-constrained expander.

\begin{definition}[Length-Constrained Expander Decomposition] \label{def:LCED}
Given graph $G$, an \emph{$(h,s)$-length (resp.\ $(\leq h,s)$-length) $\phi$-expander decomposition} for a node-weighting $A$ with cut slack $\kappa$ and length slack $s$ is an $hs$-length cut $C$ of size at most $\kappa \cdot \phi|A|$ such that $A$ is $(h,s)$-length (resp.\ $(\leq h,s)$-length) $\phi$-expanding in $G-C$.
\end{definition}

We will make use of a strengthened version of length-constrained expander decompositions called ``linked'' length-constrained expander decompositions. Informally, this is a length-constrained expander decomposition which renders $G$ length-constrained expanding even after adding many self-loops. This is a strengthened version because adding self-loops only makes it harder for a graph to be a length-constrained expander. The following definition gives the self-loops we will add for a length-constrained expander decomposition $C$.

\begin{definition}[Self-Loop Set $L^{\ell}_C$]
Let $C$ be an $h$-length moving cut of a graph $G = (V,E)$ and let $\ell$ be a positive integer divisible by $h$. For any vertex $v$, define $C(v)=\sum_{e \ni v}C(e)$. The self-loop set $L^{\ell}_C$ consists of $C(v)\cdot \ell$ self-loops at $v$. We let $G+L^{\ell}_C := (V, E \cup L^{\ell}_C)$.
\end{definition}
\noindent Using the above self-loops, we can now define linked length-constrained expander decompositions.

\begin{definition}[Linked Length-Constrained Expander Decomposition]
Let $G$ be a graph. An $\ell$-linked $(h,s)$-length (resp.\ $(\leq h, s)$-length) $\phi$-expander decomposition of a node-weighting $A$ on $G$ with cut slack $\kappa$ is an $hs$-length moving cut $C$ such that $|C|\le \kappa \cdot \phi |A|$ and $A\cup L^{\ell}_C$ is $(h,s)$-length (resp.\ $(\leq h, s)$-length) $\phi$-expanding in $G + L^{\ell}_C - C$.
\end{definition}
\noindent As before, we say that moving cut $C$ is an $(h,s)$-length $\phi$-expander decomposition (linked or not) for $G$ if it is an $(h,s)$-length $\phi$-expander decomposition for the node-weighting $\deg_G$.

\future{I.e. State linkedness in terms of node-weightings instead of self-loops}

\subsection{Routers, Power Graphs and Expander Power Graph Robustness}\label{sec:routers}
Having introduced length-constrained expander decompositions in the previous section, we introduce and discuss the closely related notion of graph routers, which will be useful for our characterizations of length-constrained expanders in terms of (classic) expanders.

%\begin{itemize}
%\item Use $\kappa_{\mathrm{router}}$ and %$\kappa_{\mathrm{decomp}}$ instead.
%\end{itemize}
%\enote{Should the following require that the node weightings are degree respecting or should we, alternatively, define node weights as never assigning a value larger than degree?}
%Next, we define the notion of \emph{routers}.
\begin{definition}
[Routers]\label{def:router} Given graph $G$ and node-weighting $A$, we say that $G$ is a \emph{$t$-step $\kappa$-router for $A$} if every $A$-respecting demand can be routed in $H$ via a $t$-step flow with congestion at most $\kappa$.
\end{definition}
\noindent As with length-constrained expanders, we say $G$ is just a $t$-step $\kappa$-router if it is a $t$-step $\kappa$-router for $\deg_G$.

Observe that, by the flow characterization of length-constrained expanders (\Cref{thm:flow character}), routers are essentially the same object as $(h,s)$-length $\phi$-expanders with unit length $t=\Theta(hs)$ long routing paths, congestion $\kappa = \tilde{\Theta}(\frac{1}{\phi})$ congestion, and a diameter smaller than $h$, which guarantees that any demand is $h$-length and therefore any degree-respecting demand must be routable. The following formalizes this.

\begin{lemma}
Let $G$ be a graph with unit length edges. Let $A$ be a node-weighting where $\diam_{G}(\supp(A))\le h$. Then:
\begin{itemize}
\item \textbf{Expander Implies Router:} If $G$ is a $(h,s)$-length $\phi$-expander for $A$, then $G$
is an $hs$-step $O(\frac{\log N}{\phi})$-router for $A$;
\item \textbf{Router Implies Expander:} If $G$ is not a $(h,s)$-length $\phi$-expander for $A$, then $G$
is not an $(\frac{hs}{2})$-step $\frac{1}{2\phi}$-router for $A$. 
\end{itemize}
\end{lemma}

\begin{proof}
Since $\diam_{G}(\supp(A))\le h$, the set of $A$-respecting demands
and the set of $h$-length $A$-respecting demands are identical.
%Also, the notion of dilation and maximum steps of a flow in $G$ become the same as edges have unit length. 
Therefore, by \Cref{thm:flow character},
if $G$ is a $(h,s)$-length $\phi$-expander for $A$, then every
$A$-respecting demand can be routed in $G$ with congestion $O(\log(N)/\phi)$
and $sh$ steps. Otherwise, some $A$-respecting demand cannot be
routed in $G$ with congestion $1/2\phi$ and $(sh)/2$ steps. This
completes the proof by \Cref{def:router}.
\end{proof}

Of particular interest to us will be routers constructed by taking power graphs, defined as follows.
\begin{definition}[$k$th Power Graph]
Given a graph $G = (V,E)$, we let $G^k = (V, E^k)$ be the graph that has an edge for each path of length at most $k$. In particular, $\{u,v\} \in E^k$ iff $d_G(u,v) \leq k$.
\end{definition}

The following are easy-to-verify routers. %\enote{Add citations}

\begin{lemma}[Star is a Router]
Let $G$ be a star where rooted at $r$ with node-weighting $A$ where the set of leaves of $G$ is $\supp(A)$. The capacity of each star-edge $(r,v)\in E(G)$ is $A(v)$. Then,
$G$ is a $2$-step $1$-router for $A$.
\end{lemma}

\begin{lemma}[Complete Graph is a Router]
Let $G$ be a complete graph with node-weighting $A$ where $V(G)=\supp(A)$ and where for each $v,w\in\supp(A)$ and the capacity of $(v,w)$ is $\frac{A(v)\cdot A(w)}{|A|}$. Then, $G$
is a $2$-step $2$-router for $A$. 
\end{lemma}

\begin{lemma}[Expander Power is a Router]\label{lem:expanPower}
Let $G$ be a constant-degree $\Omega(1)$-expander. Then $G$ is a $O(\log n)$-step $O(\log n)$-router and $G^k$ is a $O(\log n/k)$-step $O(\log n/k)$-router.
\end{lemma}

% Next we generalize the expander example to any node-weighting. 

% \begin{lemma}\label{lem:router}
% For any vertex set $V$, any $t$, and any integral node-weighting $A$ on $V$ there exists a $\Delta = |A|^{O(1/t)}$ and a $t$-step router $G$ with $\deg_G = \Delta \cdot A$ and congestion $O(\frac{\log |A|}{t})$.
% \end{lemma}
% \begin{proof}
% Let $H$ be a $\Delta$-regular expander with constant $\Delta$ and constant conductance on a vertex set of size $|A|$. This is a $t'=O(\log |A|)$-router with congestion $O(\log |A|)$. Let $k = \ceil{\frac{t'}{t}}$. The graph $H^k$ is a $t$-length router with congestion $\Theta(t)$ and regular degree $\Delta'^k$. Partitioning nodes of this graph into groups according to $A$ and contracting these groups results in the desired router. 
% \end{proof}

Lastly, we observe that power graphs when used as routers are robust to edge deletions. In particular, the following is immediate from \Cref{lem:expanPower} and well-known results for expanders.
\begin{lemma}[Robustness of Expander Power Routers, Adaptation of Theorem 2.1 of \cite{saranurak2019expander}.]\label{lem:prunePowerToRouter}
Suppose $H = (V, E_H)$ is an $\Omega(1)$-expander with maximum degree $\Delta$ and $G := H^k = (V, E_G)$ is the $k$th power graph of $H$. Then for any subset $D \subseteq E_G$, let $G' := (V, E_G \setminus D)$ be $G$ with $D$ deleted. Then, there exists a subset of vertices $V' \subseteq V$ connected in $G'[V']$ where:
\begin{enumerate}
    \item $G'[V']$ is a $O(\frac{\log N}{k})$-step $O(\frac{\log N}{k})$-router;
    % \item $\vol(P) \leq \frac{8}{\phi} \cdot |D|$.
    \item $|V'| \geq |V| - O(\frac{k \cdot \Delta^{2k}}{\phi}) \cdot |D|$.
\end{enumerate}
\end{lemma}

\future{cite algorithmic results (too); again, we get this online and maybe (or not) should mention this?}

\subsection{Neighborhood Covers}

Our characterizations of length-constrained expander decompositions in terms of (classic) expander decompositions will make use of so-called neighborhood covers. Neighborhood covers are a special kind of clustering.

\begin{definition}[Clustering]
    Given a graph $G$ with edge lengths $\l$, a \emph{clustering} $\cS$ in $G$ with diameter $h_{\diam}$ is a collection of
disjoint vertex sets $S_1,\ldots,S_{|\cS|}$, called clusters, where every cluster has diameter at most $h_{\diam}$ in $G$.  A clustering has \emph{absolute separation} $h_{\sep}$ or \emph{separation factor} $s$ if the distance between any two clusters is at least $h_{\sep}$ or $s
\cdot h_{\diam}$, respectively.
\end{definition}
\future{Is above weak or strong diameter?}

\future{Clarify this; separation is with respect to diameter of one cluster; separation proportional to radius of any two clusters in any clustering}
\begin{definition}[Neighborhood Cover]
    A neighborhood cover $\cN$ with width $\omega$ and covering radius $h_{\cov}$ is a collection of $\omega$ many clusterings $\cS_1, \ldots, \cS_\omega$ such that for every node $v$ there exists a cluster $S$ in some $\mcS_i$ containing $\ball(v,h_{\cov})$. %We denote by $\load_{\cN}$ the maximum number of clusters in which any node $v$ is contained. 
\end{definition}
\noindent We will say that a neighborhood cover $\cN$ has diameter $h_{\diam}$ if every clustering has diameter at most $h_{\diam}$ and absolute separation $h_{\sep}$ or separation-factor $s$ if this applies to each of its clusterings. Note, however, that two clusterings in a given neighborhood cover may have different diameters.

We use the following result for neighborhood covers,  which can be proved in a similar way to Lemma 8.15 of \cite{haeupler2022expander}.

%To eliminate an extra $\log n$ factor one can define the min cover frequency as the smallest number of times any neighborhood is covered and compare this to the load. We call this ratio the $\loadfactor_{\cN}$. 

% \enote{Diameter should be exponential in $1+O(s)$; would be more comfortable iwth should be $s^{O(1/\eps)}$}
\begin{theorem}[\cite{haeupler2022expander}]
\label{thm:cover-separation-factor-existential}
For any $h_{\cov}, s>1$ and $\eps \in (0,1)$ and graph $G$ there exists a neighborhood cover with covering radius $h_{\cov}$, separation-factor $s$, diameter $h_{\diam} = \frac{1}{\eps}\cdot O(s)^{O(1/\eps)} \cdot h_{\cov}$ and width $\omega = N^{O(\eps)} \log N$.
Moreover, there exists an algorithm that computes such a neighborhood cover in $h_{\diam} \cdot N^{O(\eps)}$ depth and $m \cdot h_{\diam} \cdot N^{O(\eps)}$ work with high probability.
% Moreover, there exists a deterministic CONGEST algorithm that computes such a neighborhood cover in $h_{\diam} \cdot N^{O(\eps + \eps')}$ rounds.
\end{theorem}
% \begin{theorem}[\cite{haeupler2022expander}]
% \label{thm:cover-separation-factor-existential}
% For any $h_{\cov}, s>1$ and $\eps , \eps' \in (0,1)$ and graph $G$ there exists a neighborhood cover with covering radius $h_{\cov}$, separation-factor $s$, diameter $h_{\diam} = O(s)^{O(1/\eps)} \cdot \frac{1}{\eps'} \cdot h_{\cov}$, width $\omega = N^{O(\eps)} \log N$, $\load_{\cN} = N^{O(\eps')} \log N$.
% Moreover, there exists an algorithm that computes such a neighborhood cover in $h_{\diam} \cdot N^{O(\eps + \eps')}$ depth and $m \cdot h_{\diam} \cdot N^{O(\eps + \eps')}$ work with high probability.
% % Moreover, there exists a deterministic CONGEST algorithm that computes such a neighborhood cover in $h_{\diam} \cdot N^{O(\eps + \eps')}$ rounds.
% \end{theorem}

\future{write/cite fatldd paper}

\subsection{Length-Constrained Expansion Witnesses}

In this section we define the notion of a witness of $(\leq h,s)$-length $\phi$-expansion. Below, for node-weighting $A$ and nodes $S$, we let $A_S$ be $A$ restricted to nodes in $S$.
\begin{definition}[Length-Constrained Expansion Witness]\label{def:LCExpWitness}
    A $(\leq h,s)$-length $\phi$-expansion witness for graph $G$ and node-weighting $A$ consists of the following where $s_0 \cdot s_1 \leq s$ and $\kappa_0 \cdot \kappa_1 \leq 1/\phi$:
    \begin{itemize}
        \item \textbf{Neighborhood Cover:} a neighborhood cover $\mcN_{h'}$ of $G$ with covering radius $h'$ for each $h' \leq h$ a power of $2$;
        \item \textbf{Routers:} an $s_0$-step and $\kappa_0$ congestion router $R_S$ of $A_S$ for each $S \in \bigcup_{h'} \mcN_{h'}$;
        \item \textbf{Embedding of Routers:} an $(h' \cdot s_1)$-length embedding $F_S$ of $R_S$ into $G$ for each $S \in \mcN_{h'}$ for each $\mcN_{h'}$ such that $\sum_{h'}\sum_{S \in \mcN_{h'}} F$ has congestion at most $\kappa_1$.
    \end{itemize}
\end{definition} 
\noindent A witness of $(h,s)$-length $\phi$-expansion is defined identically to the above but there is only a single neighborhood cover $\mcN_{h}$ with covering radius $h$.

It is easy to verify that such a witness indeed witnesses that an input graph is a length-constrained expander.
\begin{lemma}
    If a graph $G$ and node-weighting $A$ have an $(h,s)$-length (resp.\ $(\leq h,s)$-length) $\phi$-expansion witness then $G$ is an $(h,s)$-length (resp.\ $(\leq h,s)$-length) $\phi$-expander with respect to $A$.
\end{lemma}
\future{write proof}
The following gives our definition of a witnessed length-constrained expander decomposition.
\begin{definition}[Witnessed Length-Constrained Expander Decompositions]
A witnessed $(h,s)$-length (resp.\ $(\leq h,s)$-length) $\phi$-expander decomposition consists of an $(h,s)$-length (resp.\ $(\leq h,s)$-length) $\phi$-expander decomposition $C$ along with an $(h,s)$-length (resp.\ $(\leq h,s)$-length) $\phi$-expansion witness for $G - C$.
\end{definition}

\section{Main Result Restated: Length-Constrained ED Algorithm}\label{sec:algGuarantees}
Having formally defined length-constrained expander decompositions in the previous section, we now restate the guarantees of our algorithm for computing these decompositions for convenience.

The following gives the main result of our work.
\mainAlgThm*
\noindent The proof of \Cref{thm:expdecomp exist} uses several new notions and techniques that we introduce in the subsequent sections. The final proof of \Cref{thm:expdecomp exist} is in \Cref{sec:spiral}. We show that we can additionally achieve linkedness in \Cref{sec:graphProps}.

\future{Put the algorithm based on balanced cuts here and then say will give algo later}

\section{Length-Constrained Expanders as Embedded Expander Powers}\label{sec:expanderPowersCharacterization}

Towards understanding whether large sparse length-constrained cuts quickly yield length-constrained expander decompositions, we begin by providing a new characterization of length-constrained expander decompositions in terms of classic expanders. One of the great benefits of such a characterization is that it will allow us to bring well-studied tools from the classic expander setting to bear on length-constrained expanders. 

\paragraph*{Challenge.} Achieving such a characterization may seem impossible. Length-constrained expanders are fundamentally different objects from classic expanders: as earlier mentioned, many graphs which are length-constrained expanders are not classic expanders and adding length constraints destroys the structure of many otherwise well-behaved objects.
 \paragraph*{Basic Version of Result.} We show that, nonetheless, length-constrained expander decompositions are exactly those graphs that admit low congestion \emph{and dilation} embeddings of regular expanders in every local neighborhood. More formally, if $H$ is a constant-degree $\Omega(1)$-expander, we show that a graph $G$ is an $(h,s)$-length $\phi$-expander iff it is possible to embed $H$ into every $h$-neighborhood in $G$ with congestion about $\frac{1}{\phi}$ using paths of length at most $hs/\Theta(\log n)$.

\paragraph*{Full Result.}
The above characterization has two downsides. First, assuming integer lengths, such an embedding only makes sense when $hs \geq \Omega(\log n)$. However, $(h,s)$-length $\Omega(1)$-expanders are only interesting and distinct from classic $\Omega(1)$-expanders when $hs \ll O(\log n)$. Second, the embedding is brittle in the sense that if even a small part of the graph is not embeddable then the entire graph is declared to be not length-constrained expanding.

We address these issues by significantly strengthening this result. First, we allow for $hs = o(\log n)$ by considering embeddings of the ``$k$th power graph'' $H^k$ of classic expander $H$ rather than $H$ itself. In particular, this allows uxs to trade off between congestion and dilation, showing $G$ is an $(h,s)$-length $\phi$-expander iff $H^k$ can be embedded into every neighborhood with congestion at most about $\frac{1}{\phi \cdot 2^k}$ using paths of length at most $hs \cdot k / \Theta(\log n)$. Observe that, assuming integer lengths, such an embedding makes sense as long as $hs \cdot k \geq \Theta(\log n)$. Second, we show that if such embeddings are \emph{mostly} possible in $G$, then \emph{most} of $G$ is an $(h,s)$-length $\phi$-expander. 

\paragraph*{Techniques.} The main proof idea for this result is as follows. First, we argue that we can route in $H^k$ with low congestion over paths of length $O(\log n / k)$. Then, if our embedding exists, we can embed these routes in the expander powers into the original graph with low dilation and congestion. This gives low congestion and dilation routes in the original graph, showing it is a length-constrained expander. We formalize these embeddings with the idea of the ``neighborhood router demand.'' %The key requirement of such an embedding is that it embeds each edge of the expander power into \emph{short} paths.

\future{BH: To get this algorithmic group node-weighting valuers for powers of 2; embed each group as is with capacities proportional to node weighting value; now just need to route between node weightings; then add a capacitied edge between each blob with capacity the smaller of the two sides; unpack this edge into a matching; now number of edges is proportional to support size of node weighting; main thing this is needed for is algorithmic pruning; would put the algorithmic version of this in its own section (try and set up proofs so not too much needed to rehash)}

\future{Parallel to regular expanders: there have sense of witnessing graph; question becomes what is sense of witness graph of length-bounded expanders}

\future{BH: name change: What's to change neighborhood routing demand to a witness graph then it's demand of witness}

% In this section we characterize length-constrained expanders as exactly the graphs that admit low congestion and dilation embeddings of expander powers. Specifically, we will formalize this embedding using what we call the neighborhood routing demand, $\nrd$. The basic idea for showing our characterization is to argue that embeddings of expander powers facilitate low dilation and congestion routing; furthermore, these routings are only interrupted proportionally to the amount of neighborhood router demand that is separated. The rest of this section formalizes the neighborhood router demand and this argument in order to show the following.
\begin{restatable}{thm}{embedding}\label{lem:neighRouting}
Suppose we are given a graph $G$, node-weighting $A$, $h \geq 1$ and parameters $k, k' \geq 1$. Then:
\begin{itemize}
    \item \textbf{Length-Constrained Expander Implies Embedding:} If $G$ is an $(hk',s)$-length $\phi$-expander for $A$ then \future{ANY NEIGHBORHOOD ROUTING DEMAND (up to loss parameters, no?)} $\nrd$ is routeable with congestion $O(\frac{\log N}{\phi})$ and dilation $2k' \cdot hs$.
    \item \textbf{Embedding Implies Length-Constrained Expander:} Let \future{D be ANY neighborhood routing demand and let } $G'$ be $G$ with a moving cut applied. If a $(1-\epsilon)$ fraction of  $\nrd$ is routable in $G'$ with congestion $\frac{1}{\phi}$ and dilation $hs$ then there is a node-weighting $A' \preccurlyeq A$ of size $|A'| \geq |A| \cdot \left(1 - \epsilon'  \right)$ such that $A'$ is $\left(h, s' \right)$-length $\phi'$-expanding in $G'$ where $\epsilon' =  O\left(\epsilon \cdot 2^{O(k)} \cdot N^{O(1/k')} \log N\right)$, $s' = O\left(s \cdot \frac{\log N }{ k} \right)$ and $\phi' = \Omega\left(\phi \cdot 2^{-O(k)} \cdot N^{-O(1/k')} \cdot \log^{-2} N \right)$.
\end{itemize}
\end{restatable}
\noindent Notice that the above states that if one can embed expander powers into \emph{most} of a graph then \emph{most} of the graph is a length-constrained expander. $k$ and $k'$ above correspond to the power of the expanders that we take and the diameter of neighborhood covers that we use respectively.

\subsection{Formalizing the Embedding via the Neighborhood Router Demand}\label{sec:NRD}
We begin by formalizing the aforementioned embedding with what we call the neighborhood router demand, $\nrd$.

\begin{figure}
    \centering
    \begin{subfigure}[b]{0.32\textwidth}
        \centering
        \includegraphics[width=\textwidth,trim=0mm 0mm 0mm 0mm, clip]{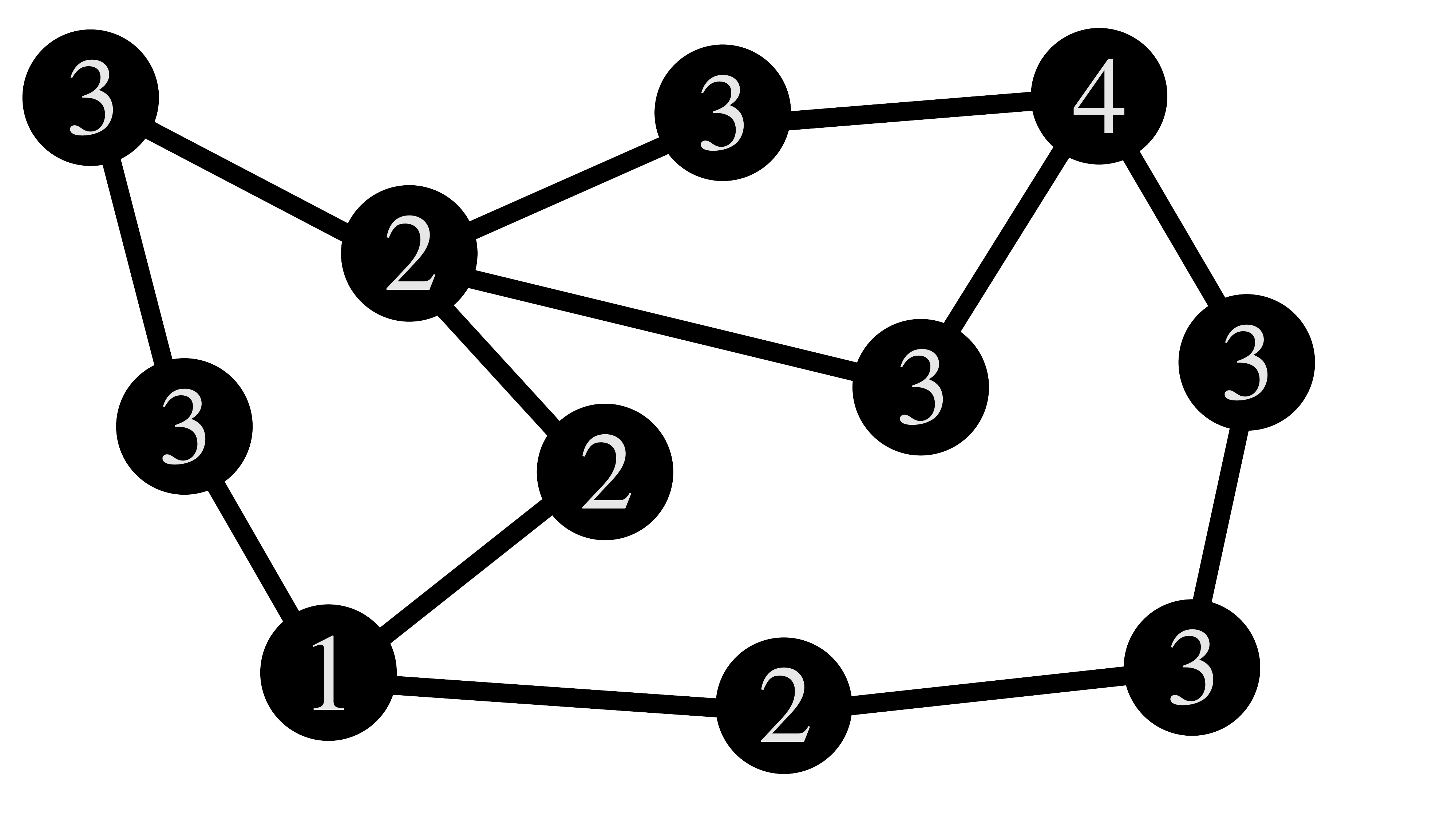}
        \caption{Cluster $S_{i,j}$.}\label{sfig:NRD1}
    \end{subfigure}    \hfill
    \begin{subfigure}[b]{0.32\textwidth}
        \centering
        \includegraphics[width=\textwidth,trim=0mm 0mm 0mm 0mm, clip]{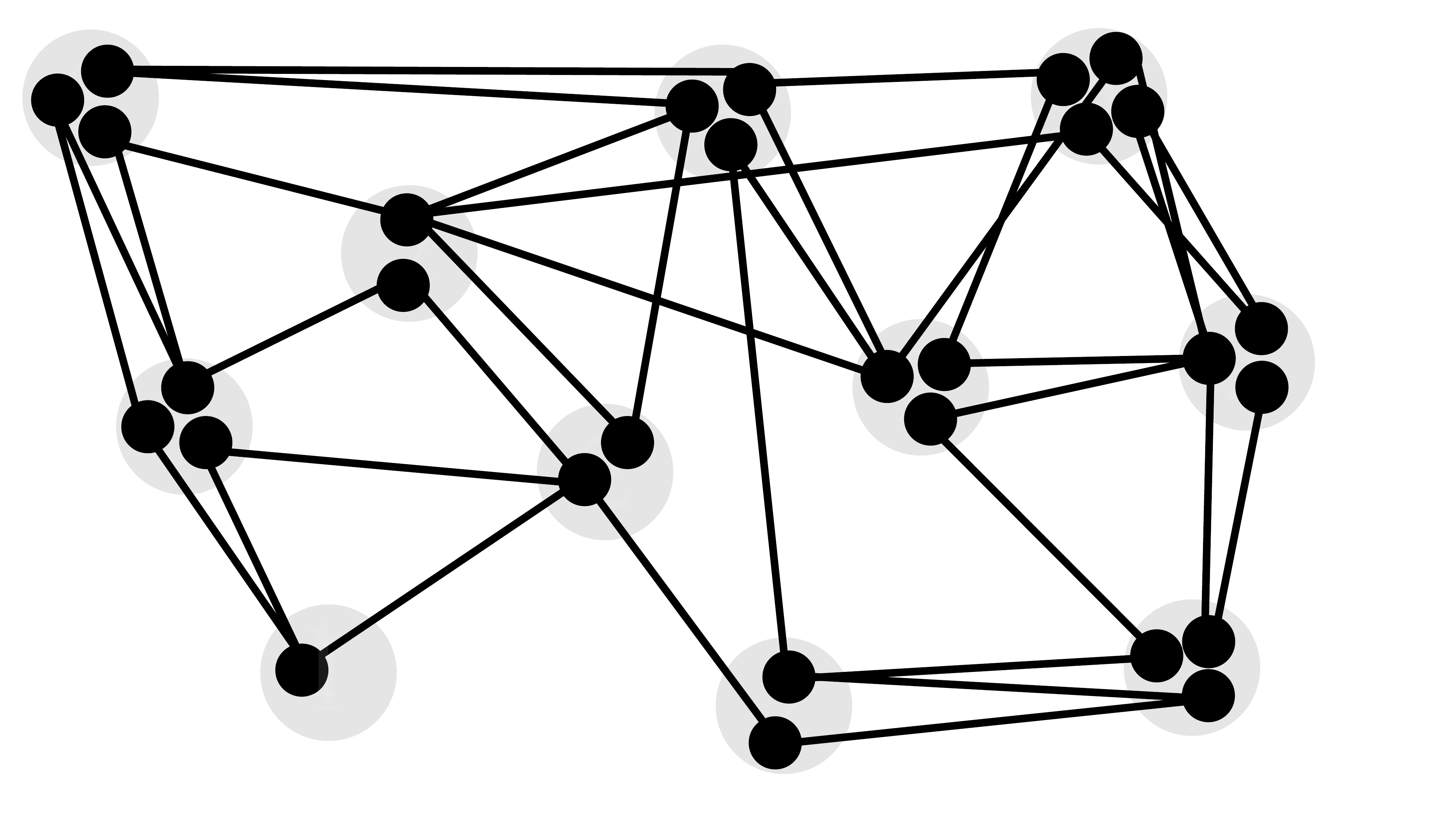}
        \caption{Graph $H_{i,j}^k$.}\label{sfig:NRD2}
    \end{subfigure}    \hfill
    \begin{subfigure}[b]{0.32\textwidth}
        \centering
        \includegraphics[width=\textwidth,trim=0mm 0mm 0mm 0mm, clip]{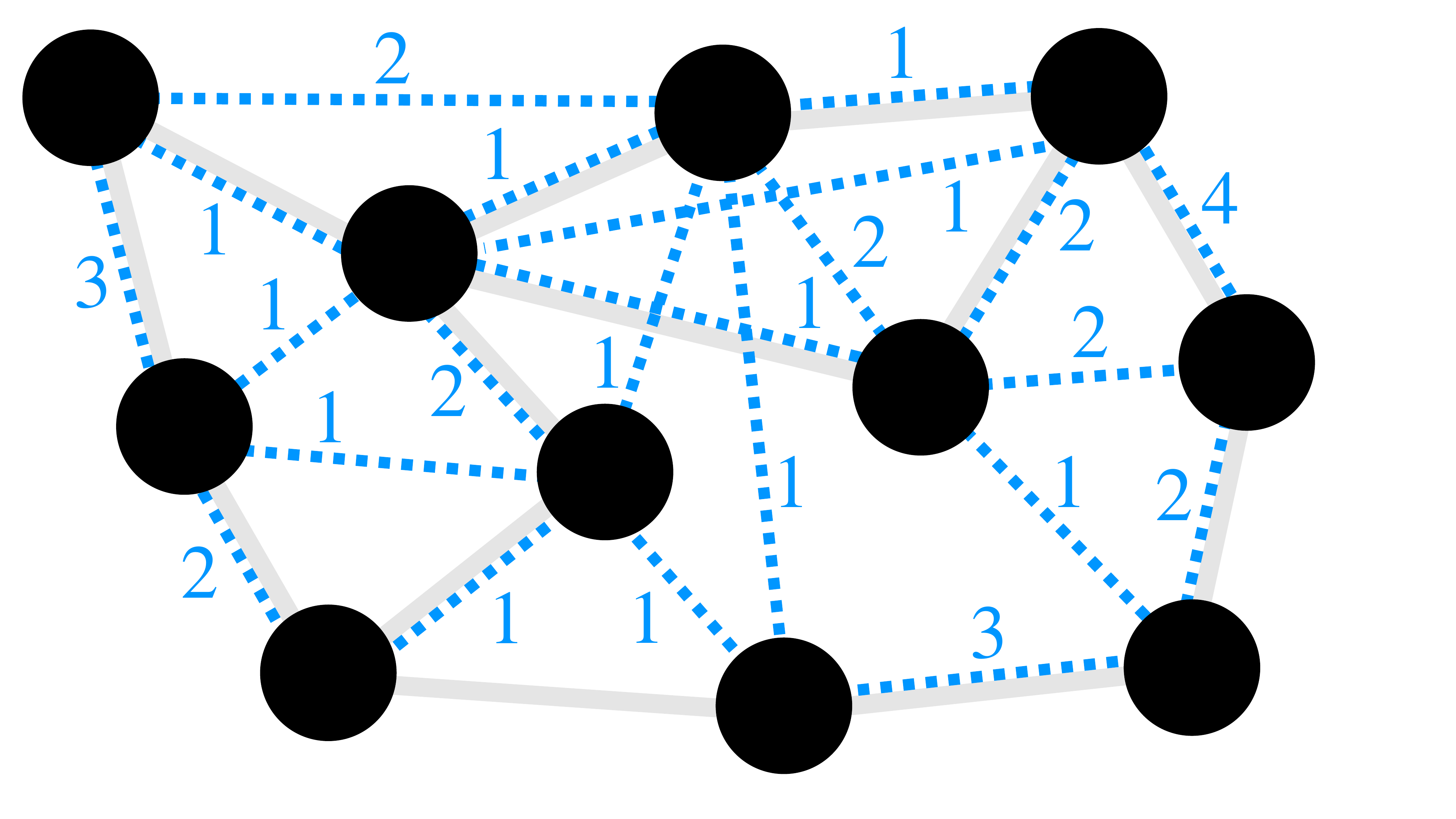}
        \caption{Demand $\Delta^k \cdot D_{i,j}$.}\label{sfig:NRD3}
    \end{subfigure}
    \caption{Our neighborhood router demand for one cluster $S_{i,j}$. \ref{sfig:NRD1} gives the cluster $S_{i,j}$ with node $v$ labeled with $A(v)$. \ref{sfig:NRD2} gives the graph $H_{i,j}^k$ we want to embed into $S_{i,j}$. \ref{sfig:NRD3} gives the demand $D_{i,j}$ corresponding to this embedding (times $\Delta^k$) in dashed blue. Our neighborhood router demand is then computed by doing this for all such clusters and scaling down appropriately.}\label{fig:nrd}
\end{figure}

Formally, suppose we are given a graph $G = (V, E)$ and a node-weighting $A$ of $G$ as well as a length bound $h$ and parameters $k, k' \geq 1$. To compute the neighborhood router demand we first compute a neighborhood cover $\mcN = \{\mcS_1, \mcS_2 \ldots \}$ with covering radius $h$, diameter $k' \cdot h$, load $N^{O(1/k')} \log N$ and separation factor $2$; such a neighborhood cover exists by \Cref{thm:cover-separation-factor-existential}. 

We now construct a graph $H_{i,j}$ associated with each cluster $S_{i,j} \in \bigcup_i \mcS_i$. The vertex set of $H_{i,j}$ consists of $A(u)$ copies for each vertex $u \in V$ for $\sum_{u \in S_{i,j}} A(u)$ total vertices in $H_{i,j}$. We let $\copies_{i,j}(u)$ denote these copies of $u$ in $H_{i,j}$ and let $V_{i,j} := \bigcup_{u \in S_{i,j}}\copies_{i,j}(u)$ be all such copies. Then, we let $H_{i,j} = (V_{i,j}, E_{i,j})$ be a fixed but arbitrary $\Omega(1)$-expander with max degree $\Delta = O(1)$ with vertex set $V_{i,j}$; such graphs are well-known to exist by e.g.\  \cite{alon2008elementary}.

Likewise, we let $H_{i,j}^k = (V_{i,j}, E_{i,j}^k)$ be the $k$th power graph of $H_{i,j}$ as defined in \Cref{sec:routers}. Lastly, for $u, v \in V$ we let $$E_{i,j}^k(u,v) := \{ \{u', v'\} \in E': u' \in \copies_{i,j}(u), v' \in \copies_{i,j}(v)\}$$ be edges of $H^{k}_{i,j}$ between copies of $u$ and $v$. We then define the neighborhood router demand associated with cluster $S_{i, j}$ on vertices $u,v \in V$ as
\begin{align*}
    D_{i,j}(u, v) := \frac{1}{\Delta^k}  \cdot |E_{i,j}^k(u,v)|.
\end{align*}
\noindent See \Cref{fig:nrd} for an illustration of $D_{i,j}$.

Similarly, we define the entire neighborhood router demand as
\begin{align*}
    \nrd := \frac{1}{\load_\mcN} \cdot \sum_{i,j} D_{i,j}.
\end{align*}

We verify the basic properties of \nrd. 
\begin{lemma}\label{lem:nrdProps}
\nrd is $A$-respecting and $(hk')$-length. Furthermore $|\nrd| \geq \frac{1}{N^{O(1/k')} \cdot \log N} \cdot |A|$.
\end{lemma}
\begin{proof}
We prove that \nrd is $hk'$-length and then that it is $A$-respecting.
\begin{itemize}
    \item $\nrd$ is $hk'$-length since each $D_{i,j}(u,v) > 1$ only if $u,v \in S_{i,j}$ for some $i$ and $j$ and each $S_{i,j}$ has diameter at most $hk'$ by our assumption that $\mcN$ has diameter $hk'$.
    \item $\nrd$ is $A$-respecting: recall that to be $A$-respecting we must show that for every $v$ we have $\max\left\{\sum_u \nrd(u,v), \sum_u \nrd(v, u) \right\} \leq A(v)$; consider such a $v$; we will show $\sum_u \nrd(v,u) \leq A(v)$ (showing $\sum_u \nrd(u,v) \leq A(v)$ is symmetric); observe that since each vertex in $H$ has degree at most $\Delta$, we know that each vertex in $H^k_{i,j}$ has degree at most $\Delta^k$ and so for each $i$ and $j$ we have
    \begin{align}
        \sum_u D_{i,j}(v,u) &= \sum_{u} \sum_{u' \in \copies_{i,j}(u)} \sum_{v' \in \copies_{i,j}(v)} \frac{1}{\Delta^k} \cdot \mathbbm{1}(\{u', v'\} \in E(H^k_{i,j}))\nonumber\\\nonumber
        &= \sum_{v' \in \copies_{i,j}(v)} \frac{1}{\Delta^k} \cdot \deg_{H^k_{i,j}}(v')\\\nonumber
        &\leq \sum_{v' \in \copies_{i,j}(v)} 1\\ 
        &\leq A(v)\label{eq:eb1}
    \end{align}
    The claim then follows from the fact that no vertex is in more than $\load \mcN$-many clusters. Namely, applying \Cref{eq:eb1} and the definition of $\load{\mcN}$:
    \begin{align*}
        \nrd(v,u) &= \frac{1}{\load{\mcN}} \cdot \sum_{i,j} D_{i,j}(v, u)\\
        & \leq \frac{1}{\load{\mcN}} \cdot \sum_{i,j} A(v)\\
        & \leq A(v).
    \end{align*}
\end{itemize}

Lastly, we argue the lower bound on $|\nrd|$. Observe that each time a vertex $u$ is in $S_{i,j}$, it is responsible for $A(u)$ copies in each $H_{i,j}^k$ where each copy has degree $\Omega(1)$ and so contributes at least $\Omega(A(u)/\load_{\mcN})$ to $|\nrd|$.
\end{proof}

\subsection{Proving Length-Constrained Expanders are Embedded Expander Powers}

Having formally defined the neighborhood router demand, we proceed to show the main theorem of this section (\Cref{lem:neighRouting}): the extent to which a graph is a length-constrained expander is more or less captured by how well the neighborhood router demand can be routed and, in particular, how well expander powers embed into the input graph.

\embedding*
\begin{proof}
The fact that if $G$ is an $(hk',s)$-length $\phi$-expander for $A$ then $\nrd$ can be routed with the stated path length and congestion is immediate from \Cref{thm:flow character} and \Cref{lem:nrdProps}

We now turn to the second point; namely we argue that if a $(1-\epsilon)$ fraction of $\nrd$ can be routed with congestion $\frac{1}{\phi}$ and dilation $hs$ in $G'$ then there is a node-weighting $A' \preccurlyeq A$ of size at least $|A| \cdot (1-\epsilon')$ that is $(h,s')$-length $\phi'$-expanding in $G'$. To argue that such an $A'$ exists, by \Cref{thm:flow character} it suffices to show that there is an $A'$ of size at least $|A| \cdot (1 - \epsilon')$ where every $A'$-respecting $h$-length demand can be routed with congestion at most $O(\frac{1}{\phi'})$ and dilation at most $O(h s')$. We proceed to argue this $A'$ exists.

We let $F$ be the flow in $G'$ that routes at least a $(1-\epsilon)$ fraction of \nrd with congestion $\frac{1}{\phi}$ and dilation $hs$. The basic idea is to map $F$ into each $H^k_{i,j}$, perform (normal) expander pruning in $H^k_{i,j}$ and then map the large expanding subset in $H^k_{i,j}$ back into $G'$.

% For each $u,v \in S_{i,j}$ we will now define $E_{i,j}$ which, informally, is all edges in $H^k_{i,j}$ which $F$ successfully ``embeds''. For each $u,v \in S_{i,j}$, we let $E_{i,j}(u,v)$ be all of the pairs in $S_{i,j}$ corresponding to the pair $\{u,v\}$ mostly embedded pairs as follows:
% \begin{align*}
% E_{i,j}(u,v) := \{\{u', v'\} \in E(H^k_{i,j}) : u' \in \copies_{i,j}(u), v' \in \copies_{i,j}(v)\},
% \end{align*}
% where as a convention we will imagine that $E_{i,j}(u,v) = \emptyset$ if either $u$ or $v$ is not in $S_{i,j}$.
We say that the pair $\{u,v\} \subseteq V$ is \emph{unsuccessfully embedded} if $F(u,v) < .5 \cdot \nrd(u,v)$ and let $\bar{E}$ be all unsuccessfully embedded pairs. For a fixed $S_{i,j}$, this allows us to define $\bar{E}_{i,j}^k$ as the edges of $H^k_{i,j}$ corresponding to all unsuccessfully embedded pairs as follows:
\begin{align*}
\bar{E}_{i,j}^k := \bigcup_{\{u,v\} \in \bar{E}} E_{i,j}^k(u,v).
\end{align*}
% Likewise we let $\bar{E}_{i,j}$ be all other edges of $H^k_{i,j}$:
% \begin{align*}
% \bar{E}_{i,j} = E(H^k_{i,j}) \setminus E_{i,j}
% \end{align*}

% be all edges of $H^k_{i,j}$ that have one endpoint in $\copies_{i,j}(u)$ and one endpoint in $\copies_{i,j}(v)$.  For each $i$ and $j$, we let $\hat{E}_{i,j}$ consists of all $\{u,v\}$ pairs where BLAH

% We then say that the pair  let $\hat{E}_{i,j}(u,v)$ be an arbitrary but fixed subset of $E_{i,j}(u,v)$ of size $\left\lfloor \frac{F(u,v)}{\nrd(u,v)}\cdot |E_{i,j}(u,v)| \right\rfloor$. Likewise, we let $\hat{E}_{i,j} := \bigcup_{u,v \in S_{i,j}}\hat{E}_{i,j}(u,v)$. Complementarily, we let $\bar{E}_{i,j} = E(H^k_{i,j}) \setminus \hat{E}_{i,j}$ be all edges of $H^k_{i,j}$ that $F$ does not successfully embed.

We proceed to define our node weights $A' \preccurlyeq A$ of size at least $(1-\epsilon') \cdot |A|$. In short, we will find this subset by applying expander pruning on $H^k_{i,j}$ where $\bar{E}_{i,j}^k$ gives us the edges we deleted as input to expander pruning. Specifically, letting $\hat{H}^k_{i,j}$ be $H^k_{i,j}$ with $\bar{E}_{i,j}$ deleted, we know by \Cref{lem:prunePowerToRouter} that there is a subset of vertices $V'_{i,j} \subseteq V_{i,j}$ connected in $\hat{H}^k_{i,j}[V'_{i,j}]$ satisfying
\begin{enumerate}
    \item $\hat{H}^k_{i,j}[V'_{i,j}]$ is a $O(\frac{\log N}{k})$-step $O(\frac{\log N}{k})$-router;
    \item $|V'_{i,j}| \geq |V_{i,j}| - k \cdot 2^{O(k)} \cdot |\bar{E}_{i,j}|$.
\end{enumerate}
This subset $V'_{i,j}$ naturally corresponds to a node-weighting $A_{i,j}'$. In particular, we let $A_{i,j}'$ on $v$ be:
\begin{align*}
    A_{i,j}'(v) := |\copies_{i,j}(v) \cap V'_{i,j}|
\end{align*}
and then let $A'$ on $v$ be defined as:
\begin{align*}
    A'(v) := \min_{i,j} A_{i,j}'(v).
\end{align*}
First, we claim that $A' \preccurlyeq A$. Observe that for each $v$ we have:
\begin{align*}
    A'(v) &=  \min_{i,j} A_{i,j}'(v)\\
    & \leq  \max_{i,j} A_{i,j}'(v)\\
    & =  \max_{i,j} |\copies_{i,j}(v) \cap V'_{i,j}|\\
    % & \leq \load{\mcN} \cdot |\copies_{i,j}|\\
    & \leq A(v).
\end{align*}

Next, we claim that $|A'| \geq |A| \cdot \left(1 - \eps \cdot 2^{O(k)} \cdot N^{O(1/k')} \cdot \log N\right)  $. 
% To see this, first observe that we can lower bound $|V'_{i,j}|$ for each $i$ and $j$ as
% \begin{align*}
%     |V'_{i,j}| &\geq  |V(H^k_{i,j})| - O\left(\frac{k \cdot \Delta^{2k}}{\phi} \right) \cdot |\bar{E}_{i,j}|\\
%     % &\geq  |V(H^k_{i,j})| - O\left(\frac{k \cdot \Delta^{2k}}{\phi} \right) \cdot |E(H^k_{i,j}) \setminus \hat{E}_{i,j}|\\
%     % &\geq  |V(H^k_{i,j})| - O\left(\frac{k \cdot \Delta^{2k}}{\phi} \right) \cdot \left( |E(H^k_{i,j})|- \sum_{u,v \in S_{i,j}}\left \lfloor \frac{F(u,v)}{\nrd(u,v)} \right \rfloor \cdot |E_{i,j}(u,v)| \right)\\
%     % &\geq  |V(H^k_{i,j})| - O\left(\frac{k \cdot \Delta^{2k}}{\phi} \right) \cdot \left( \Delta^k \cdot A(S_{i,j})- \sum_{u,v \in S_{i,j}}\left \lfloor \frac{F(u,v)}{\nrd(u,v)} \right \rfloor \cdot |E_{i,j}(u,v)| \right) 
% \end{align*}

Let $\barnrd$ be all the demand that $F$ does not satisfy to extent at least $.5$; that is, $\barnrd$ on $(u, v)$ is $\sum_{\{u,v\} \in \bar{E}} \nrd(u,v)$. Observe that by an averaging argument and our assumption that at least a $(1-\epsilon)$ fraction of \nrd is routed we know that $|\barnrd| \leq 2 \epsilon \cdot |\nrd|$. On the other hand, each pair of vertices $\{u,v\}$ in $\bar{E}$ (i.e.\ each unsuccessfully embedded pair) corresponds to at most $|E_{i,j}(u,v)|$ edges in $H^k_{i,j}$; and so
\begin{align}
    \sum_{i,j} |\bar{E}_{i,j}| & = \sum_{i,j} \sum_{\{u, v\} \in \bar{E}} |E_{i,j}(u,v)|\nonumber\\
    &= \sum_{i,j} \Delta^k \sum_{\{u, v\} \in \bar{E}} D_{i,j}(u,v)\nonumber\\
    &= \Delta^k \cdot \load_{\mcN} \cdot \sum_{\{u, v\} \in \bar{E}} \nrd(u,v) \nonumber\\
    &= \Delta^k \cdot \load_{\mcN} \cdot |\barnrd| \nonumber\\
    &\leq \Delta^k \cdot \load_{\mcN} \cdot 2 \epsilon \cdot |\nrd|\nonumber\\
    & \leq \epsilon \cdot 2^{O(k)} \cdot N^{O(1/k')} \cdot \log N \cdot |\nrd| .\label{eq:unembeddedEdges}
\end{align}
where in the last line we used our choice of neighborhood cover with load $N^{O(1/k')} \log N$ and the fact that $\Delta = O(1)$.

% across all clusters in our neighborhood cover by the definition of load we know that
% \begin{align*}
%     \sum_{i,j} \nrd(\bar{E}_{i,j}) &\leq \load_{\mcN} \cdot \nrd(\bar{E})\\
%     & \leq 2 \epsilon \cdot \load_{\mcN} \cdot |\nrd|
% \end{align*}

% And so it follows that
% \begin{align*}
%     \sum_{i,j} |\bar{E}_{i,j}| & = BLAH
% \end{align*}

Thus, combining \Cref{eq:unembeddedEdges}, the fact that $|V'_{i,j}| \geq |V_{i,j}| - k \cdot 2^{O(k)} \cdot |\bar{E}_{i,j}|$ for each $i$ and $j$ and the fact that $|\nrd| \leq |A|$ since $\nrd$ is $A$-respecting (as proved in \Cref{lem:nrdProps}), we have
\begin{align*}
   |A'| &= \sum_v A'(v)\\
   & = \sum_v \min_{i,j} |\copies_{i,j}(v) \cap V_{i,j}'|\\
   & = \sum_v A(v)-\max_{i,j} |\copies_{i,j}(v) \setminus V_{i,j}'|\\
    & \geq \sum_v \left[ A(v)-\sum_{i,j} |\copies_{i,j}(v) \setminus V_{i,j}'| \right]\\
    &= |A| - \sum_{i,j}\left[|V(H^k_{i,j})| - |V_{i,j}'| \right]\\
   & \geq |A| - k \cdot 2^{O(k)} \cdot \sum_{i,j}  |\bar{E}_{i,j}|\\
    & \geq |A| \cdot \left(1 - \epsilon \cdot 2^{O(k)} \cdot N^{O(1/k')} \cdot \log N \right)
%   & \geq 
\end{align*}

% observe that
% \begin{align*}
%     |A'| &= \sum_v A'(v)\\
%     &= \frac{1}{\load{\mcN}} \cdot \sum_v \sum_{i,j} A_{i,j}'(v)\\
%     &= \frac{1}{\load{\mcN}} \cdot \sum_{i,j} \sum_v  |\copies_{i,j}(v) \cap V'_{i,j}|\\
%     &= \frac{1}{\load{\mcN}} \cdot \sum_{i,j} |V_{i,j}'|\\
%     &\geq \frac{1}{\load{\mcN}} \cdot \sum_{i,j} \left[ |V(H^k_{i,j})| - O\left(\frac{k \cdot \Delta^{2k}}{\phi} \right) \cdot |\bar{E}_{i,j}| \right]\\
%     &\geq \frac{1}{\load{\mcN}} \cdot \sum_{i,j} \left[ |V(H^k_{i,j})| - O\left(\frac{k \cdot \Delta^{2k}}{\phi} \right) \cdot |E(H^k_{i,j}) \setminus \hat{E}_{i,j}| \right]\\
%     &\geq \frac{1}{\load{\mcN}} \cdot \sum_{i,j} \left[ |V(H^k_{i,j})| - O\left(\frac{k \cdot \Delta^{2k}}{\phi} \right) \cdot \left( |E(H^k_{i,j})|- \sum_{u,v \in S_{i,j}}\left \lfloor \frac{F(u,v)}{\nrd(u,v)} \right \rfloor \cdot |E_{i,j}(u,v)| \right) \right]
% \end{align*}
% $\bar{E}_{i,j} = E(H^k_{i,j}) \setminus \hat{E}_{i,j}$

% $ \left \lfloor \frac{F(u,v)}{\nrd(u,v)} \right \rfloor \cdot |E_{i,j}(u,v)|$

Lastly, we claim that $A'$ is $(h,s')$-length $\phi'$-expanding in $G'$. Recall that $s' = O\left(s \cdot \frac{\log N}{k}\right)$ and $\phi' = \Omega\left(\phi \cdot 2^{-O(k)} \cdot N^{-O(1/k')} \log^{-2} N \right)$ By \Cref{thm:flow character} it suffices to argue that every $h$-length $A'$-respecting demand can be routed in $G'$ with congestion at most $O\left(\frac{1}{\phi'}\right)$ and dilation at most $O(hs')$. 

Consider an $A'$-respecting $h$-length demand $D$ in $G'$. As above, let $F$ be our flow which certifies that at least a $(1-\eps)$ fraction of $\nrd$ can be routed in $G'$. The basic idea is to route $D$ by treating $F$ as an embedding of each power graph $H^k_{i,j}$ into $G'$ and then routing in each $H^k_{i,j}$. We will use this strategy to construct a flow $\hat{F}$ which routes $D$ by routing it in each $H^k_{i,j}$ and then using $F$ to project this routes back into $G'$.

More formally, fix an $S_{i,j}$. We proceed to define a demand $D'_{i,j}$ on $\hat{H}^k_{i,j}[V'_{i,j}]$. In particular, for each pair of vertices $u', v' \in V'_{i,j}$ where $u' \in \copies(u)$ and $v' \in \copies(v)$ we let 
\begin{align*}
   D'_{i,j}(u',v'):= 
       \begin{cases}
           \frac{D(u,v)}{A_{i,j}'(u) \cdot A_{i,j}'(v)}  & \text{if $A'(u) \cdot A'(v) \neq 0$}\\
           0 &\text{otherwise}
       \end{cases}
\end{align*}
and so by definition of $A_{i,j}'(u)$ we have
\begin{align*}
    \sum_{u' \in \copies_{i,j}(u)}\sum_{v' \in \copies_{i,j}(v)}D'_{i,j}(u',v') = D(u,v).
\end{align*}
We next claim that $D_{i,j}'$ is degree-respecting in $\hat{H}^k_{i,j}[V'_{i,j}]$. To do so we will argue that for any $u' \in V_{i,j}'$ we have $\sum_{v'}D'_{i,j}(u',v') \leq 1$; showing a symmetric upper bound on $\sum_{v'}D'_{i,j}(v', u')$ is symmetric. $u'$ must have degree at least $1$ since we know that $\hat{H}^k_{i,j}[V'_{i,j}]$ is connected and so this shows that $D_{i,j}'$ is degree-respecting. To see this observe that applying the definition of $D_{i,j}'$ and the fact that $D$ is $A'$-respecting we have:
\begin{align*}
    \sum_{v'}D'_{i,j}(u',v') &= \sum_{v \in S_{i,j}} \sum_{v' \in \copies_{i,j}(v)} D'_{i,j}(u',v')\\
    &= \frac{1}{A'(u)} \cdot \sum_{v} D(u,v)\\
    &\leq 1;
\end{align*}
 %and therefore is routable on $\hat{H}^k_{i,j}[V'_{i,j}]$ with congestion $O(\log N/k)$ and $O(\log N/k)$ steps.

Since $D_{i,j}'$ is degree respecting and since $\hat{H}^k_{i,j}[V'_{i,j}]$ is a $O(\log N/k)$-step $O(\log N/k)$-router, $D'_{i,j}$ can be routed on $\hat{H}^k_{i,j}[V'_{i,j}]$ with congestion $O(\log N/k)$ and $O(\log N/k)$ steps. Let $F'_{i,j}$ be the flow on $\hat{H}^k_{i,j}[V'_{i,j}]$ that routes $D'_{i,j}$ with congestion $O(\log N/k)$ and $O(\log N/k)$ dilation.

By using $F$ as an embedding, we can see that $F'_{i,j}$ naturally corresponds to a flow $F_{i,j}$ on $G$. Specifically, by assumption an edge of $\hat{H}^k_{i,j}[V'_{i,j}]$ is of the form $\{u',v'\}$ where $u' \in \copies_{i,j}(u)$, $v' \in \copies_{i,j}(v)$ and $F(u,v) \geq .5 \cdot \nrd(u,v)$. We can therefore project a flow path in the support of $F_{i,j}'$ into $G$ by concatenating the projection of each of its incident edges $\{u',v'\}$ to the flow paths given by $F$ in $G$ between $u$ and $v$. Below, we formalize this idea.

Consider a path $P' = (v_1', v_2', \ldots, v_l')$ in the support of $F_{i,j}'$ where $v_x' \in \copies_{i,j}(v_x)$ for each $x \in [l]$. Let $\mcP_x(P')$ be all paths between $v_x$ and $v_{x+1}$ in the support of $F$. By adding paths with multiplicity to $\mcP_x(P')$, we may assume that $F$ sends some equal flow amount $\rho_{P'}$ along each path in $\mcP_x(P')$ for every $x$. Then, letting $z$ be $\min_x |\mcP_x(P')|$ (i.e.\ the number of paths corresponding to the $v_x, v_{x+1}$ pair that has the least flow sent by $F$) we construct a collection of $z$ paths $\mcP(P')$ in $G$ gotten by selecting one path from each $x$; that is if use an arbitrary fixed ordering to the paths of $\mcP_x(P')$ and we imagine that the first $z$ paths of $\mcP_x(P')$ are $\{P_x^{(1)}, P_x^{(2)},\ldots\, P_x^{(z)}\}$ then $\mcP(P')$ consists of paths in $G$ and is
\begin{align*}
    \mcP(P') := \left\{P_1^{(y)} \oplus P_2^{(y)} \oplus \ldots \oplus P_{l-1}^{(y)} : y \in [z] \right\}
\end{align*}
where $\oplus$ is path concatenation. Observe that each path in $\mcP(P')$ is a path in $G$ of length at most $O(h \cdot \log N / k)$. Let $F_{i,j}^{(P')}$ be the flow that sends $F_{i,j}(P')$ flow by sending equal flow along each path in $\mcP(P')$ in $G$; that is
\begin{align*}
	F_{i,j}^{(P')}(P) := 
	\begin{cases}
		\frac{1}{z} \cdot F_{i,j}'(P')& \text{if $P \in \mcP(P')$}\\
		0 & \text{otherwise}
		\end{cases}
\end{align*}	

%\begin{align*}
%	\sum_{P \in \mcP(P') U TO V} F_{i,j}^{(P')}(P) = \sum_{P \in \mcP(P')} \min_{e \in P} F_{i,j}^{(P')}(P)
%\end{align*}	

Lastly, we can construct $F_{i,j}$ as the result of doing this for all flows in the support of $F'_{i,j}$; that is,
\begin{align*}
F_{i,j} := \sum_{P' \in \supp(F'_{i,j})} F_{i,j}^{(P')}
\end{align*}

We let our final flow $\hat{F}$ to route the demand $D$ that we started with be the result of doing this for every cluster in our neighborhood cover:
\begin{align*}
	\hat{F} := \sum_{i,j} F_{i,j}
\end{align*}
%Observe that  for any $u$ and $v$ in $S_{i,j}$
%\begin{align*}
%\sum_{u' \in \copies_{i,j}(u)}\sum_{v' \in \copies_{i,j}(v)} F_{i,j}'(u', v') = \sum_{P'} \sum_{u' \in \copies_{i,j}(u)}\sum_{v' \in \copies_{i,j}(v)} F_{i,j}'(u', v')
%\end{align*}	

It remains to argue that $\hat{F}$	routes our $A'$-respecting $h$-length demand $D$ with congestion $O\left(\frac{1}{\phi'}\right)$ and dilation $O(hs')$ in $G'$. 

We begin by arguing that $\hat{F}$ indeed routes $D$. To see this consider a pair of vertices $u$ and $v$ in $G'$ with $D(u,v) > 0$. Since $D$ is $h$-length, $\mcN$ has covering radius $h$ and applying a moving cut only increases the distances between nodes, it follows that there is some $i$ and $j$ for which $u,v \in S_{i,j}$. For this $i$ and $j$ we know that each $F_{i,j}^{(P')}$ sends $F_{i,j}'(P')$ from a copy of $u$ to a copy of $v$ and so summing across all $P'$ from copies of $u$ to copies of $v$ we can see that
\begin{align*}
		\hat{F}(u,v) &\geq F_{i,j}(u,v) \\
			&= \sum_{u' \in \copies_{i,j}(u)} \sum_{v' \in \copies_{i,j}(v)}  \sum_{P' =(u', \ldots, v')} F_{i,j}^{(P')}(u,v)\\
					&= \sum_{u' \in \copies_{i,j}(u)} \sum_{v' \in \copies_{i,j}(v)}  \sum_{P' =(u', \ldots, v')} F'_{i,j}(P') \\
		&= \sum_{u' \in \copies_{i,j}(u)} \sum_{v' \in \copies_{i,j}(v)}  F'_{i,j}(u', v') \\
		&= \sum_{u' \copies_{i,j}(u)} \sum_{v' \in \copies_{i,j}(v)}D'_{i,j}(u', v') \\
		&= D(u,v)
\end{align*}	

Our bound on the dilation of $\hat{F}$ is immediate from the fact that every path in $\mcP(P')$ for every $P' \in \supp(F_{i,j}')$ for every $i$ and $j$ has length at most $O(h  \cdot \log N / k) = O(hs')$.

Lastly, we now argue that $\hat{F}$ has congestion at most $O\left(\frac{1}{\phi'}\right) = O\left(\frac{1}{\phi} \cdot 2^{O(k)} \cdot N^{O(1/k')} \log^{2} N \right)$. For each $e' = \{u', v'\} \in E(H^k_{i,j})$ where $u' \in \copies_{i,j}(u)$ and $v' \in \copies_{i,j}(v)$, let
\begin{align*}
F_{e' \to e} := \frac{1}{|E_{i,j}^k(u,v)|} \cdot  \sum_{P = (u, \ldots, v) \ni e} F(P)
\end{align*}	
be the flow induced on edge $e$ resulting from embedding $e'$ into $e$. Applying the fact that each $F_{i,j}'$ has congestion $O(\log N / k)$, the fact that $F(u,v) \geq .5 \cdot \nrd(u,v)$ if $\{u', v'\} \in \hat{H}_{i,j}^k[V_{i,j}']$, the separation factor of our neighborhood cover $\mcN$ is $2$ (and so our absolute separation is at least $h$) and $F$ has congestion $\frac{1}{\phi}$, we have that the congestion of $\hat{F}$ on a given edge $e$ in $G$ is
\begin{align*}
	\sum_{P \ni e} \hat{F}(P) &= \sum_{i,j}  \sum_{e' \in E(H^k_{i,j}[V_{i,j}'])} F_{i,j}'(e') \cdot F_{e' \to e}\\
	&\leq O(\log N/k) \cdot  \sum_{i,j}   \sum_{e' \in E(H^k_{i,j}[V_{i,j}'])} F_{e' \to e}\\
	& = O(\log N/k) \cdot \Delta^k \cdot \load_{\mcN} \cdot F(e)\\
	& \leq O\left(\frac{1}{\phi} \cdot \frac{\Delta^k \cdot \load_{\mcN}}{k}\right)\\
	& \leq O\left(\frac{1}{\phi} \cdot 2^{O(k)} \cdot N^{O(1/k')} \cdot \log^2 N\right)
\end{align*}
as desired.
\end{proof}

\section{Robustness of Length-Constrained Expanders}\label{sec:robustness}

In this section we prove that, like (classic) expanders, length-constrained expanders are robust to edge deletions.  In particular, we show that if we begin with a length-constrained expander and delete a small number of edges, then, up to various slacks, the remaining graph is mostly still a length-constrained expander. 

The basic strategy for proving this fact will be to use the characterization of length-constrained expanders in terms of expander power embeddings from previous section. In particular, we will use this characterization and then apply the robustness of each of the embedded expander powers.

Formally, we show the following theorem.
\begin{restatable}[Robustness of Length-Constrained Expanders]{thm}{HCExpPru}
	\label{thm:HCExpPru} Suppose that $G = (V,E, \l)$ is an $(hk',s)$-length $\phi$-expander with respect to node weighting $A$ for $k' \geq 1$, let $C$ be a moving cut and let $G':= (V, E,  \l + C)$ be $G$ with $C$ applied. Then, for any integer $k \geq 1$ there exists a node weighting $A' \preccurlyeq A$ that is $(h,s')$-length $\phi'$-expanding in $G'$ such that $|A'| \geq |A| - |C| \cdot O\left(\frac{2^{O(k)} \cdot  N^{O(1/k')} \cdot \log N}{\phi} \right)$ where $\phi' = \Omega\left(\phi \cdot 2^{-O(k)} \cdot N^{-O(1/k')} \cdot \log^{-2} N \right)$ and $s' = O\left(s \cdot \frac{k'}{k}  \cdot \log N  \right)$.
\end{restatable}
\begin{proof}
Our proof sketch is as follows. Consider the neighborhood router demand. Then, observe that $C$ can increase the length of only a constant fraction of the paths in the support of the flow which routes the neighborhood router demand and conclude by \Cref{lem:neighRouting} that there exists a large expanding subset.

More formally, let \nrd be our neighborhood routing demand as defined in \Cref{sec:NRD}. %By \Cref{lem:neighRouting} in order to show that the stated $A'$ exists, it suffices to show that at least a $1-\epsilon$ fraction of $\nrd$ can be routed with dilation $O(hs')$ and congestion $O(\frac{1}{\phi'})$ where $\eps =$ BLAH.
%By \Cref{lem:neighRouting}, to show that the desired $A'$ exists, it suffices to argue that at least a $1-\epsilon$ fraction of $\mc\nrd$ can be routed in $G'$ with congestion at most BLAH and dilation at most BLAH.

By \Cref{lem:neighRouting} since $G$ is an $(hk',s)$ $\phi$-expander, it follows that $\nrd$ can be routed in $G$ with congestion at most $O(\frac{1}{\phi})$ and dilation $2k' \cdot hs$. Let $F$ be the flow that witnesses this routing. Since $F$ has congestion at most $O(\frac{1}{\phi})$, it follows by an averaging argument that the total flow across paths in the support of $F$ which have their length increased to at least $4k' \cdot hs$ is at most $O(\frac{1}{\phi}) \cdot |C|$. That is,
\begin{align*}
	\sum_{P \in \supp(F): \l_{G'}(P) > 4k' \cdot hs } F(P) \leq O\left(\frac{1}{\phi} \right) \cdot  |C|.
\end{align*}	
Let $F'$ be $F$ restricted to all paths in the support of $F$ with length at most $4k' \cdot hs$ in $G'$ so that $\val(F') \geq \val(F) - O\left(\frac{1}{\phi} \right) \cdot  |C|$. Observe that by construction the dilation of $F'$ is $O(k' \cdot hs)$.

Furthermore, observe that it follows from the existence of $F'$ that at least a $1-\epsilon$ fraction of $\nrd$ can be routed in $G'$ for $\eps = \Omega\left(\frac{|C|}{\val(F)} \cdot \frac{1}{\phi} \right) \geq \Omega\left(\frac{|C|}{|A|} \cdot \frac{1}{ \phi} \right)$ with congestion at most $O\left(\frac{1}{\phi} \right)$ and dilation at most $O(h \cdot sk')$. Applying \Cref{lem:neighRouting} tells us that there is a node weighting $A' \preceq A$ where 
\begin{align*}
	|A'| &\geq |A| \cdot \left(1 - O\left( \eps \cdot 2^{O(k)} \cdot N^{O(1/k')} \cdot \log N \right) \right)\\
	& \geq |A| - |C| \cdot O\left(\frac{2^{O(k)} \cdot N^{O(1/k')} \cdot \log N}{\phi} \right)
\end{align*}
and $A'$ is $(h,s')$-length $\phi'$-expanding in $G'$ where $s' = O(s \cdot \frac{ k'}{k} \cdot \log N)$ and $\phi' = \Omega\left(\phi \cdot 2^{-O(k)} \cdot N^{-O(1/k')} \cdot \log^{-2} N \right)$ as required to show our theorem.
\end{proof}

\section{Union of Sparse (Classic) Cut Sequence is Sparse}\label{sec:unionOfNHCCuts}
In this section we prove that the union of a sequence of sparse (classic) cuts is still a sparse (classic) cut with only an $O(\log n)$ loss in sparsity. Before proceeding, it may be useful for the reader to recall the definition of the sparsity of a classic cut in \Cref{sec:regularExpanders} and, in particular, the definition of the witness components of a cut.

Formally, we consider the following notion of a sequence of sparse (classic) cuts.
\begin{definition}[Sequence of (Classic) Cuts]
    Given graph $G = (V,E)$, a sequence of (classic) cuts is a sequence of (classic) cuts $(C_1, C_2, \ldots)$ where each $C_i$ is a cut all of whose edges are internal to one connected component $H_i$ of $G - \sum_{j < i} C_j$. We refer to $H_i$ as the source component of $C_i$ and the sparsity of $C_i$ in $H_i$ as its sparsity in the sequence.
\end{definition}

The following formalizes the idea that the union of (classic) sparse cuts is sparse. Observe that, in fact, it proves a slightly stronger statement, namely, that the union of cuts has a sparsity at most the (appropriately weighted) average of the cuts in the sequence.
\begin{restatable}[Union of Sparse (Classic) Cut Sequence is a Sparse (Classic) Cut]{thm}{unionOfNHCCuts} \label{thm:unionOfNHCCuts} 
Given graph $G = (V,E)$, let $(C_1, C_2, \ldots)$ be a sequence of (classic) cuts where $C_i$ has sparsity $\phi_i$ in this sequence. Then the cut $\bigcup_i C_i$ has sparsity $\phi'$ in $G$ where $\phi' \leq O(\log n) \cdot \frac{\sum_i |C_i|}{\sum_i |C_i|/\phi_i}$.
\end{restatable}
\noindent Observe that if each $C_i$ has sparsity $\phi$ in the sequence then we get that $\bigcup_i C_i$ has sparsity $O(\phi \cdot \log n)$. Intuitively, we get $O(\log n)$ overhead because each time we take a cut we charge to the smaller side and so an edge can only be charged at most $O(\log n)$ times until its side is empty. The remainder of this section makes this proof idea formal.

\subsection{Warmup: When the Source Component is Not a Witness Component}

As a simple warmup and helper lemma we observe that the union of two sparse cuts is still sparse assuming the second cut is not in a witness component of the first cut.
\begin{lemma}\label{lem:pairOfCuts}
    Let $(C_1, C_2)$ be a sequence of two (classic) cuts where $C_i$ has sparsity $\phi_i$ in the sequence. Then if the source component of $C_2$ is not a witness component of $C_1$ we have $C_1 \cup C_2$ has sparsity at most $\phi'$ in $G$ where $\phi' \leq \frac{\sum_i |C_i|}{\sum_i |C_i|/\phi_i}$.
\end{lemma}
\begin{proof}
    Let $C = C_1 \cup C_2$. Let $M_1$ be the connected component of $G - C_1$ of maximum volume in $G$, (i.e.\ the one component that is not a witness component for $C_1$ but which is a source component for $C_2$). Similarly, let $M$ be the connected component of $G-C$ of maximum volume in $G$ (i.e.\ the one component that is not a witness component for $C$).

   Observe that $\vol(M_2) \leq \vol(M)$ since $M_2$ is a subset of $M$. Letting $m$ be $|E|$, o, it follows that
    \begin{align}
        m - \vol(M) &= m - \vol(M_1) + \vol(M_1) - \vol(M) \nonumber \\
        &\geq m - \vol(M_1) + \vol(M_1) - \vol(M_2).\label{eq:aa}
    \end{align}

    On the other hand, observe that by definition of sparsity and the fact that the source component of $C_2$ is $M_1$ we have
    \begin{align}\label{eq:ab}
        \frac{|C_1|}{\phi_1} = m - \vol(M_1)
    \end{align}
    and
    \begin{align}\label{eq:ac}
        \frac{|C_2|}{\phi_2} = \vol(M_1) - \vol(M_2)
    \end{align}
    and so combining Equations \ref{eq:aa}, \ref{eq:ab} and \ref{eq:ac} we get 
    \begin{align*}
        m - \vol(M) \geq \sum_i \frac{|C_i|}{\phi_i}.
    \end{align*}

    Applying this and the definition of sparsity we get
    \begin{align*}
        \phi(C) \leq \frac{\sum_i |C_i|}{\sum_i |C_i| / \phi_i},
    \end{align*}
    as required.
\end{proof}

As a simple implication of this fact we have that the union of an arbitrary-length sequence of (classic) sparse cuts is still sparse provided each cut in the sequence is never applied to the witness component of any previous cut.
\begin{lemma}\label{cor:sparseCutCorr}
   Let $(C_1, C_2, \ldots )$ be a sequence of (classic) cuts where $C_i$ has sparsity $\phi_i$ in the sequence. Then if we have that the source component of $C_j$ is not a witness component of $C_i$ for all $i$ and $j > i$ then $\bigcup_i C_i$ is a $\phi'$-sparse cut in $G$ where $\phi' \leq \frac{\sum_i |C_i|}{\sum_i |C_i|/\phi_i}$.
\end{lemma}
\begin{proof}
The proof is by a simple induction on the number of cuts in the sequence $k$ and repeated application of \Cref{lem:pairOfCuts}.

The base case of $k=1$ is trivial and the base case of $k = 2$ is immediate from \Cref{lem:pairOfCuts}. 

Consider the case of $k > 3$. By \Cref{lem:pairOfCuts} we know that $C := C_{k-1} \cup C_k$ is a cut with sparsity $\phi_C$ which is at most $\frac{|C_{k-1}| + |C_k|}{|C_{k-1}|/\phi_{k-1 } + |C_{k}|/\phi_{k}}$. Now consider the cut sequence resulting from taking the union of the last two cuts, namely $(C_1, C_2, \ldots, C_{k-2}, C)$. The source component of $C$ is the same as the source component of $C_{k-1}$ and so the precondition of our induction holds and, in particular, by induction we know that $\bigcup_i C_i$ has sparsity at most
\begin{align}\label{eq:ad}
    \frac{\sum_{i} |C_i|}{|C|/\phi_C + \sum_{i \leq k-2}|C_i|/\phi_i}.
\end{align}
Further, observe that by our upper bound on $\phi_C$ and the fact that $|C| = |C_{k-1}| + |C_{k}|$ we know that $\frac{|C|}{\phi_C}\geq \frac{|C_{k-1}|}{\phi_{k-1}} + \frac{|C_{k}|}{\phi_{k}}$. Combining this with \Cref{eq:ad} gives that $\bigcup_i C_i$ has sparsity at most
\begin{align*}
    \frac{\sum_{i} |C_i|}{\sum_{i}|C_i|/\phi_i},
\end{align*}
as required.
\end{proof}

\subsection{Proving the General Case for (Classic) Cuts}

We now conclude this section with our proof that the union of (classic) sparse cuts is a sparse cut (\Cref{thm:unionOfNHCCuts}). Roughly, our proof arranges the sequence of cuts into a ``cut sequence tree'' and then decomposes this tree into paths that satisfy the preconditions of \Cref{cor:sparseCutCorr} and so can be unioned together to get cuts of the same sparsity. After contracting such paths in our cut sequence tree we can argue that the result is a depth $O(\log n)$ tree and so by averaging over layers of this tree we can argue that some layer induces sparsity within an $O(\log n)$ of the sparsity for which we are aiming. For readers familiar with heavy-light decompositions \cite{sleator1981data}: this decomposition can be understood as a special sort of heavy-light decomposition whose heavy paths correspond satisfy the preconditions of \Cref{cor:sparseCutCorr}.

%Our proof will arrange the sequence of cuts into a tree, apply \Cref{cor:sparseCutCorr} to ``contract'' paths of cuts that can be combined without any overhead sparsity, observe that the contracted tree has depth at most $O(\log n)$ and then conclude by an averaging argument over all $O(\log n)$ layers.

Formally, our proof will arrange our cuts into a cut sequence tree which corresponds to the natural laminar partition gotten by applying a sequence of sparse cuts. See \Cref{fig:cutSequenceTree} for an illustration.

\begin{figure}
    \centering
    \begin{subfigure}[b]{0.32\textwidth}
        \centering
        \includegraphics[width=\textwidth,trim=0mm 0mm 0mm 0mm, clip]{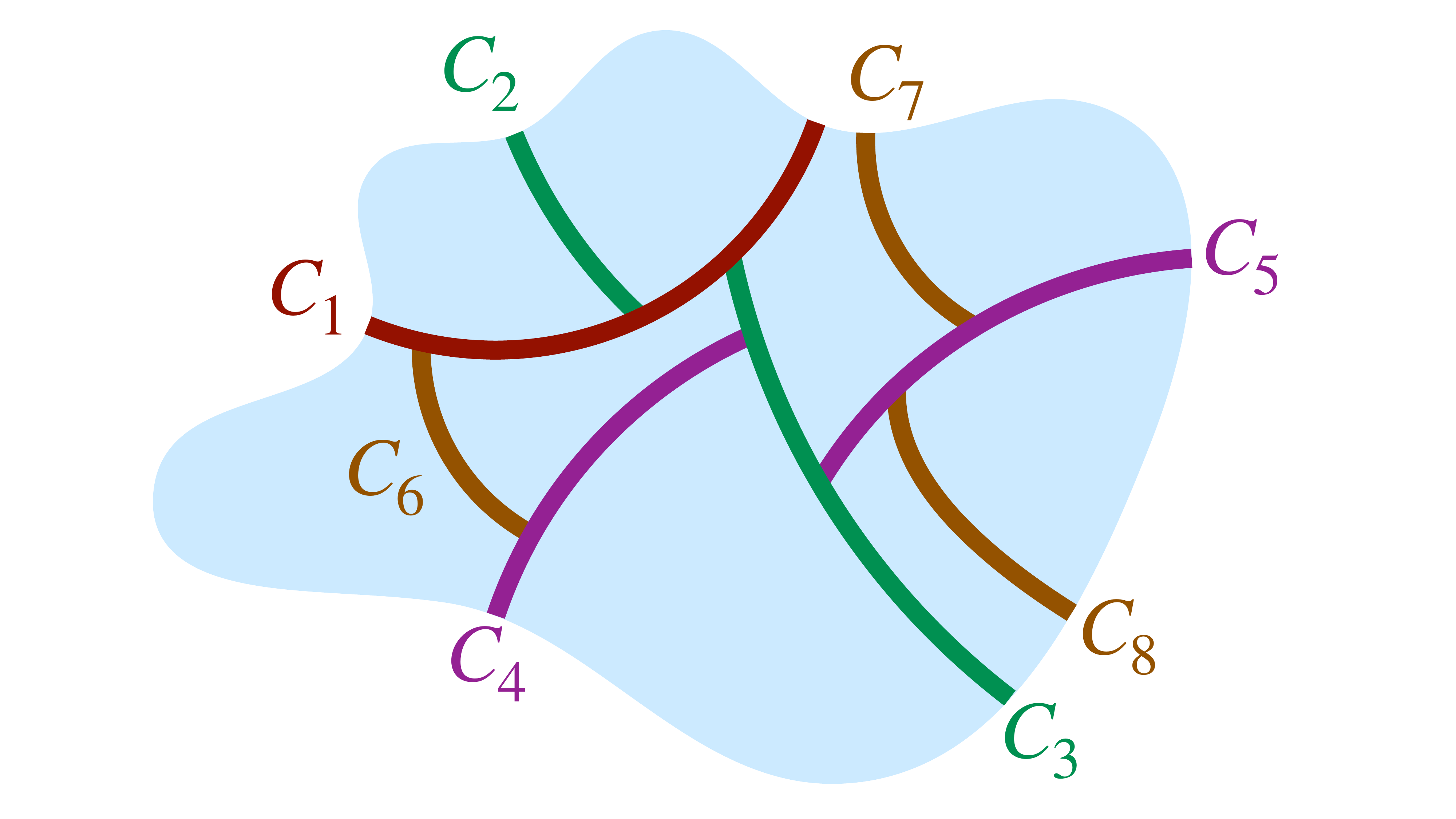}
        \caption{Cut sequence on $G$.}\label{sfig:cutTree1}
    \end{subfigure}    \hfill
    \begin{subfigure}[b]{0.32\textwidth}
        \centering
        \includegraphics[width=\textwidth,trim=0mm 0mm 0mm 0mm, clip]{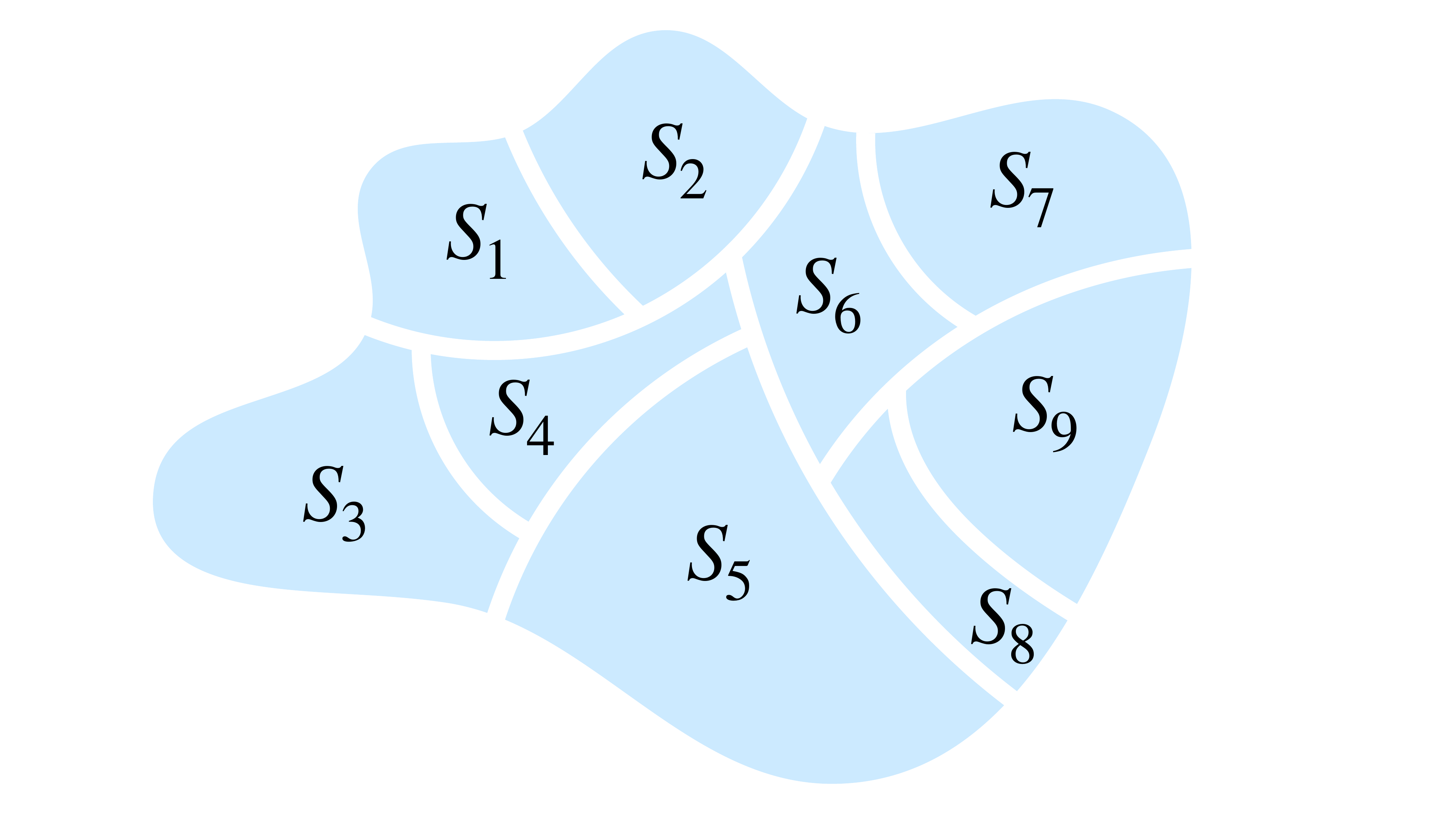}
        \caption{Connected components.}\label{sfig:cutTree2}
    \end{subfigure}    \hfill
    \begin{subfigure}[b]{0.32\textwidth}
        \centering
        \includegraphics[width=\textwidth,trim=0mm 0mm 0mm 0mm, clip]{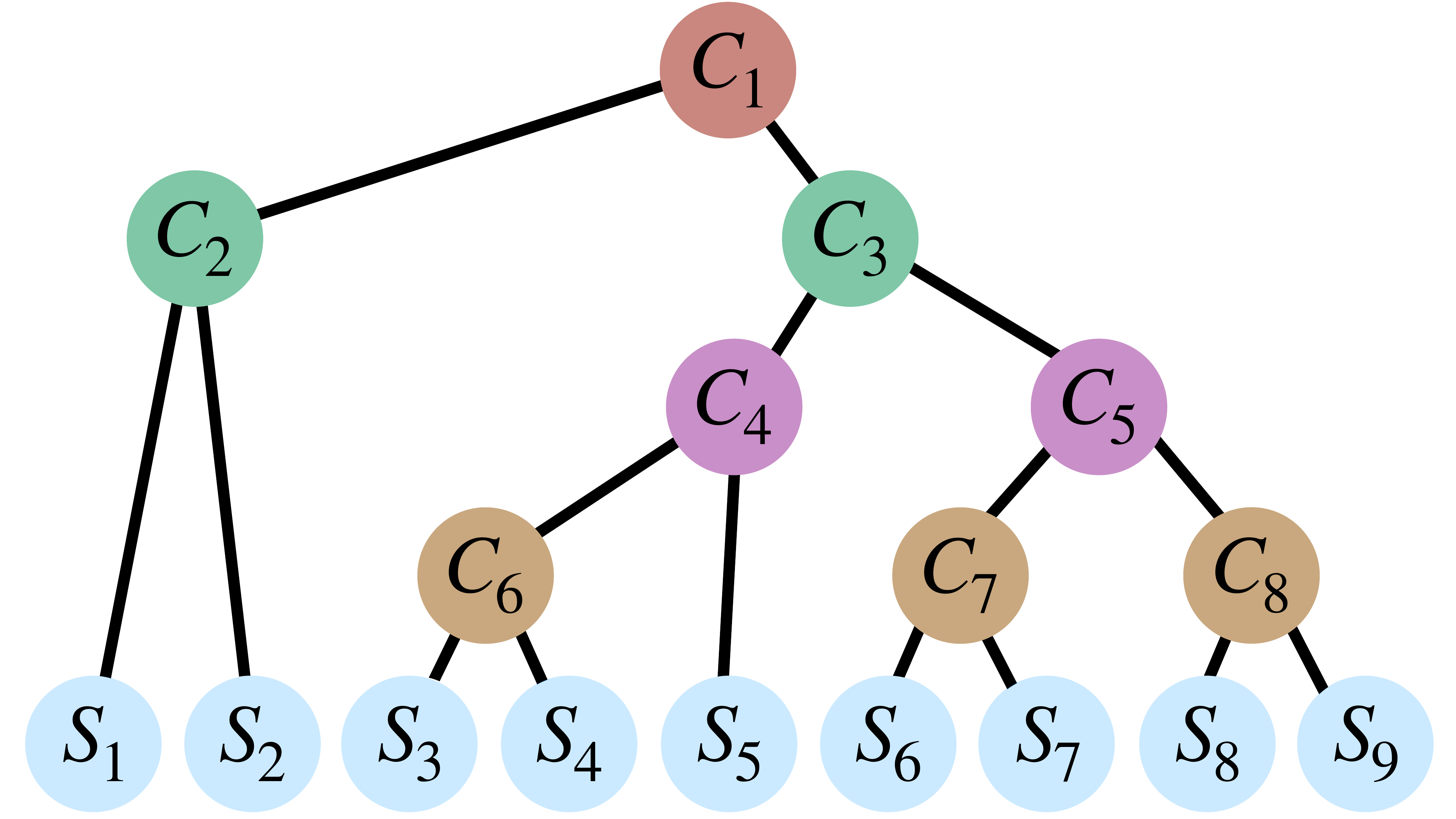}
        \caption{Cut sequence tree.}\label{sfig:cutTree3}
    \end{subfigure}
    \caption{A cut sequence (\ref{sfig:cutTree1}), the resulting connected components from applying all cuts in the sequence (\ref{sfig:cutTree2}) and the corresponding cut sequence tree (\ref{sfig:cutTree3}). Cuts colored to correspond to the depth of their nodes in the cut sequence tree. Internal nodes in cut sequence tree labeled according to their corresponding cut and leaves labeled according to their connected component.}\label{fig:cutSequenceTree}
\end{figure}

\begin{definition}[Cut Sequence Tree]\label{dfn:cutSequenceTree}
Given graph $G= (V,E)$, let $\mcC := (C_1, C_2, \ldots, C_k)$ be a sequence of (classic) cuts. Then, the cut tree $T_\mcC$ of $\mcC$ is recursively defined as follows. 
\begin{itemize}
    \item Suppose $k = 1$. Then $T_{\mcC}$ is a star with one leaf for every connected component of $G - C_1$.
    \item Suppose $k > 1$. Then let $T_{\mcC'}$ be the cut sequence tree for $\mcC' = (C_1, C_2, \ldots, C_{k-1})$ and let $H_k$ be the source component of $C_k$.  $T_{\mcC}$ is the result of taking $T_{\mcC'}$ and adding one child for each connected component of $H_k - C_k$ to the leaf of $T_{\mcC'}$ corresponding to $H_k$.
\end{itemize}
\end{definition}
\noindent Observe that in the above tree each internal vertex corresponds to a cut of $\mcC$ and every vertex corresponds to a connected component in $G$ after applying some prefix of the cuts $\mcC$.

% \begin{enumerate}
%     \item \textbf{Vertices:} The internal vertices of $T_\mcC$ are $v_1, v_2, \ldots, v_k$. The root of $T_\mcC$ is $v_1$. We say that $v_i$'s corresponding cut is $C_i$ and that $v_i$'s corresponding component is the source component of $C_i$ in $\mcC$. We have a leaf vertex $u_1, u_2, \ldots$ for each connected component of $G - \bigcup_i C_i$ where we say that each such component is the corresponding component of $u_i$.
%     \item \textbf{Edges:} We have an edge from parent $u$ to child $v$ if the component corresponding to $v$
% \end{enumerate}

We proceed to prove the our main theorem for the section.
\unionOfNHCCuts*
\begin{proof}
Let $T_{\mcC}$ be the cut sequence tree of $\mcC = (C_1, C_2, \ldots)$ as defined in \Cref{dfn:cutSequenceTree}. 

Call an edge from parent $u$ to child $v$ \emph{heavy} if:
\begin{enumerate}
    \item $u$ and $v$ are internal in $T_{\mcC}$ and;
    \item the component corresponding to $v$ is not a witness component of the cut corresponding to $u$.
\end{enumerate}
Otherwise, call an edge \emph{light}.

Observe that as each vertex has at most one heavy child the collection of heavy edges induces vertex-disjoint paths in $T_{\mcC}$, the union of which contain all vertices corresponding to a cut in $\mcC$ (i.e.\ all internal vertices of $T_{\mcC}$). Also, observe that any any root to leaf path intersects at most $O(\log n)$ light edges since each time a light edge is traversed the corresponding component's number of vertices is decreased by at least $\frac{1}{2}$.

Consider one such path of heavy edges $P = (u_1, u_2, \ldots)$ where $u_1$ is the vertex closest to the root in $T_{\mcC}$. Call such a path a heavy path. Observe that by definition of a heavy edge the cut sequence gotten by taking the cuts corresponding $(C_1', C_2', \ldots )$ to these vertices satisfies the conditions of \Cref{cor:sparseCutCorr} and so their union $\mcC' := \bigcup_i C_i'$ has sparsity at most
\begin{align}
    \frac{\sum_i |C_i'|}{\sum_i |C_i'|/\phi_i'}.\label{eq:cc}
\end{align}
where $\phi_i'$ is the sparsity of $C_i'$.

Next, consider the cut sequence from taking the cuts corresponding to all such heavy paths. More formally order the heavy paths of $T_{\mcC}$ in a fixed but arbitrary order $(P_1, P_2, \ldots)$ where if $i < j$ then no vertex of $P_j$ is an ancestor of a vertex in $P_i$ in $T_{\mcC}$. Next, consider the cut sequence $\tilde{\mcC} := (\tilde{C}_1, \tilde{C}_2, \ldots )$ where $\tilde{C}_i$ is the result of taking the union of all cuts corresponding to vertices in path $P_i$. 

Observe that the cut tree $T_{\tilde{\mcC}}$ for this sequence $\tilde{\mcC}$ is just the result of contracting all heavy edges in $T_{\mcC}$. Furthermore, observe that this tree has depth at most $O(\log n)$ since any root to leaf path in $T_{\mcC}$ intersects at most $O(\log n)$ light edges.

Consider a fixed layer of $T_{\tilde{\mcC}}$, namely a subset of vertices all of who are equal distance in $T_{\tilde{\mcC}}$ from the root of $T_{\tilde{\mcC}}$. Let $\tilde{L}$ be the cuts of $\tilde{\mcC}$ corresponding to these vertices and let $L$ be all cuts of $\mcC$ that are a subset of some cut in $\tilde{\mcC}$. 

Further, assume that $\tilde{L}$ is the layer with corresponding $L$ maximizing $\sum_{C_i \in L} |C_i| / \phi_i$. Observe that by averaging we know that 
\begin{align}\label{eq:avg}
    \sum_{C_i \in L} |C_i| / \phi_i \geq \Omega(1 / \log n) \cdot \sum_{i} |C_i| / \phi_i
\end{align}

% We claim that 
% \begin{align*}
%     \phi' \leq \frac{\sum_i |\tilde{C}_i'|}{\sum_i |\tilde{C}_i'|/\phi_i'}
% \end{align*}
% where $\phi_i'$ is the sparsity of $\tilde{C}_i'$ in its source component. 

Letting $C = \bigcup_i C_i$ we know that
\begin{align}\label{eq:bb}
    \sum_{S_C \in \mathcal {S}_C} \vol(S_C) \geq \sum_{\tilde{C} \in \tilde{L}} \sum_{S_{\tilde{C}} \in \mathcal {S}_{\tilde{C}}} \vol(S_{\tilde{C}})
\end{align}
since the witness components of all cuts in $\tilde{L}$ are disjoint and the one component of $G - C$ that is not a witness component of $C$ can only be smaller in volume than any of the non-witness components of cuts in $\tilde{L}$.

Furthermore, observe that by the definitions of sparsity (\Cref{dfn:normalSpars}), $L$ and $\tilde{L}$ as well as \Cref{eq:cc} we have
\begin{align}\label{eq:bc}
    \sum_{\tilde{C} \in \tilde{L}} \sum_{S_{\tilde{C}} \in \mathcal {S}_{\tilde{C}}} \vol(S_{\tilde{C}}) \geq \sum_{C_i \in L} |C_i| / \phi_i.
\end{align}
Combining Equations \ref{eq:avg}, \ref{eq:bb} and \ref{eq:bc} we get 
\begin{align*}
    \Omega(1 / \log n) \cdot \sum_{i} |C_i| / \phi_i \leq \sum_{S_C \in \mathcal {S}_C} \vol(S_C).
\end{align*}
which when combined with the definition of sparsity and fact that $|C| = \sum_i |C_i|$ gives our claim.
\end{proof}

\section{Union of Sparse Moving Cut Sequence is Sparse}\label{sec:unionOfHCCuts}
In the previous section we saw that the union of sequence of classic sparse cuts is itself a sparse cut. In this section, we prove this fact in the much more challenging length-constrained setting. In particular, we show that taking the union of moving cuts preserves sparsity up to an $N^{O(1/s)}$ factor. 

Formally, we consider a sequence of moving cuts, defined as follows.
\begin{definition}[Sequence of Moving Cuts]\label{dfn:movingCutSequence}
    Given graph $G = (V,E)$ and node-weighting $A$, a sequence of moving cuts is a sequence of moving cuts $(C_1, C_2, \ldots)$. We refer to the $(h,s)$-length sparsity of $C_i$ with respect to $A$ in $G - \sum_{j < i} C_j$ as its $(h,s)$-length sparsity in the sequence. We say the sequence $(C_1, C_2, \ldots)$ is $(h,s)$-length $\phi$-sparse if each $C_i$ is $(h,s)$-length $\phi$-sparse in the sequence.
\end{definition}

The following summarizes the main theorem of this section: that the union of sparse length-constrained moving cuts is sparse.
\unionOfCuts*
\noindent Observe that if every cut $C_i$ above is $(h,s)$-length, then $\sum_i C_i$ is $(h',s')$-length $\phi$-sparse. However, if some of the cuts are even sparser then this lowers the sparsity of $\sum_i C_i$.
$\frac{\phi}{\spa_{(h,s)}(C,A)} \cdot |C|$.

To prove the above theorem we must demonstrate the existence of some unit demand $D$ that witnesses the sparsity of $\sum_i C_i$. However, if $D_i$ is the demand which witnesses the sparsity of $C_i$, we cannot just use $\sum_i D_i$ as $D$ since the result need not be unit. The main idea is to argue that $\sum_i D_i$ can be understood as, more or less, greedily constructing a spanner in parallel and as such induces a graph with arboricity about $N^{O(1/s)}$ where $s$ is the length slack (see \Cref{sec:conventions} for a definition of arboricity). We can then decompose $\sum_i D_i$ into trees and use each of these trees to ``disperse'' the load of $\sum_i D_i$ so that the resulting demand $D$ is unit (after scaling down by about the arboricity). The rest of this section formalizes this argument.

\subsection{Low Arboricity Demand Matching Graph via  Parallel Greedy Spanners}
We begin by formalizing the graph induced by the witnessing demands of our cut sequence. We call this graph the demand matching graph. Informally, this graph simply creates $A(v)$ copies for each vertex $v$ and then matches copies to one another in accordance with the witnessing demands.
\begin{definition}[Demand Matching Graph]\label{def:demandMatching}
Given a graph $G = (V,E)$, a node-weighting $A$ and $A$-respecting demands $\mcD = (D_1, D_2, \ldots )$ we define the demand matching (multi)-graph $G(\mcD) = (V', E')$ as follows:
\begin{itemize}%[leftmargin=10mm]
    \item  \textbf{Vertices:} $H$ has vertices $V' = \bigsqcup_v \copies(v)$ where $\copies(v)$ is $A(v)$ unique ``copies'' of $v$.
    \item \textbf{Edges:} For each demand $D_i$, let $E_i$ be any matching where the number of edges between $\copies(u)$ and $\copies(v)$ for each $u,v \in V$ is $D_i(u,v)$. Then $E' = \bigcup_i E_i$.
\end{itemize}
\end{definition}

The key property of the demand matching graph that we use is that it induces a graph with low arboricity. We prove this by observing that it can be understood as performing a certain parallel greedy spanner construction. The following summarizes this.

\begin{lemma}[Bounded Arboricity of Demands of Sequence of Cuts]\label{lem:arbBound} Let $C_1, C_2, \ldots$ be a sequence of $(h, s)$-length cuts with witnessing demands $\mcD = (D_1, D_2, \ldots)$. Then $G(\mcD)$ has arboricity at most $s^3 \cdot \log^3 N \cdot N^{O(1/s)}$.
\end{lemma}
\begin{proof}
%\todo 
%\enote{Prove using: stuff from Zihan writeup}
\newcommand{\pg}{\textnormal{\textsf{pg}}\xspace}
We say that a sequence $(E_1,\ldots,E_k)$ of edge sets on $V$ is $s$-$\pg$ (abbreviation for $s$-parallel greedy) for some integer $s\ge 2$, iff for each $1\le i\le k$,
\begin{itemize}
    \item set $E_i$ is a matching on $V$; and
    \item if we denote by $G_{i-1}$ the graph on $V$ induced by edges in $\bigcup_{1\le j\le i-1}E_j$, then for every edge $(u,v)\in E_i$, $d_{G_{i-1}}(u,v)> s$.
\end{itemize}
Equivalently, a sequence $(E_1,\ldots,E_k)$ is $s$-$\pg$ iff for each $1\le i\le k$, every cycle in $G_i$ of length at most $s+1$ contains at least two edges in $E_i$.
We say that a graph $G$ is $s$-$\pg$ iff its edge set $E(G)$ is the union of some $s$-$\pg$ sequence on $V(G)$. \cite{haeupler2023parallel} proves that every $s$-\pg graph on $n$ vertices has arboricity $s^3 \cdot \log^3 n \cdot n^{O(1/s)}$.

We now show that this fact implies \Cref{lem:arbBound} by showing that $G(\dset)$ is a $s$-$\pg$ graph. By definition, $G(\dset)$ is the union of a sequence of matchings. 
Denote $D=(D_1,\ldots,D_k)$. For each $i$, let $G_{i}(\dset)$ be the union of all matchings corresponding to $D_k,\ldots,D_i$. Consider an edge in $E_{i-1}$, the matching corresponding to demand $D_{i-1}$. By definition, it suffices to show that for each $(u,v)\in E_{i-1}$, there is no length-at-most-$s$ path in $G_i$ containing $u,v$. Assume for contradiction that there exists a path $P=(u,x_1,\ldots,x_{s-1},v)$ in $G_{i}$. 
%Assume that $(x_j,x_{j+1})$ is the first edge of this path added to graph $G_{i}$. This means that $(x_j,x_{j+1})$ is some pair in a demand $D_{i'}$ with $i'<j$, and so the moving cut $C_{i'}$ should have separates $x_j$ and $x_{j+1}$ to distance at least $hs$, namely 
This means that in graph $G-\sum_{1\le j\le i-1}C_{j}$, every pair in $(u,x_1),(x_1,x_2),\ldots,(x_{s-1},v)$ is at distance at most $h$. %However, every other edge in path $P$ corresponds to some pair with distance at most $h$ in $G-\sum_{1\le j'\le j}C_{j'}$. Therefore
By triangle inequality, this implies that the distance between $u,v$ in $G-\sum_{1\le j\le i-1}C_{j}$ is less than $hs$. However, as the moving cut $C_{i-1},C_{i-1},\ldots,C_1$ separates all pairs in $D_{i-1},D_{i-1},\ldots,D_1$ to distance more than $hs$, as an edge in $E_{i-1}$ (which is a pair in $\bigcup_{1\le j\le i-1}D_{j}$), $u$ and $v$ should be at distance more than $hs$ in $G-\sum_{1\le j\le i-1}C_{j}$, a contradiction.
\end{proof}

\subsection{Matching-Dispersed Demand}
In the previous section we formalized the graph induced by the witnessing demands $(D_1, D_2, \ldots)$ of a sequence of sparse moving cuts and argued that this graph has low arboricity. We now discuss how to use the forest decomposition of this graph to disperse $(D_1, D_2, \ldots)$ so that the result is a unit demand.

\begin{figure}
    \centering
    \begin{subfigure}[b]{0.48\textwidth}
        \centering
        \includegraphics[width=\textwidth,trim=0mm 0mm 0mm 0mm, clip]{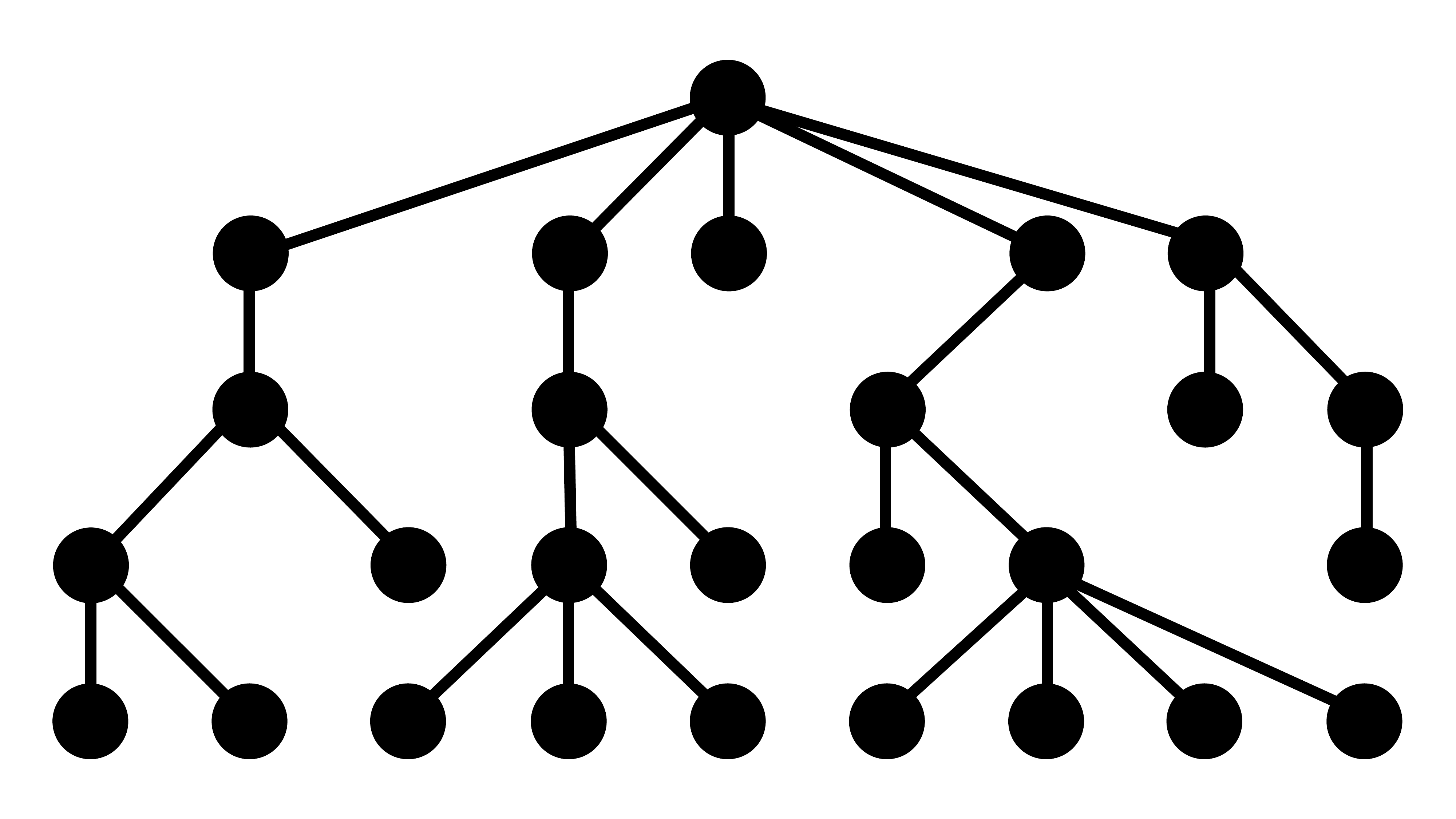}
        \caption{Input tree $T$.}\label{sfig:dispTree1}
    \end{subfigure}    \hfill
    \begin{subfigure}[b]{0.48\textwidth}
        \centering
        \includegraphics[width=\textwidth,trim=0mm 0mm 0mm 0mm, clip]{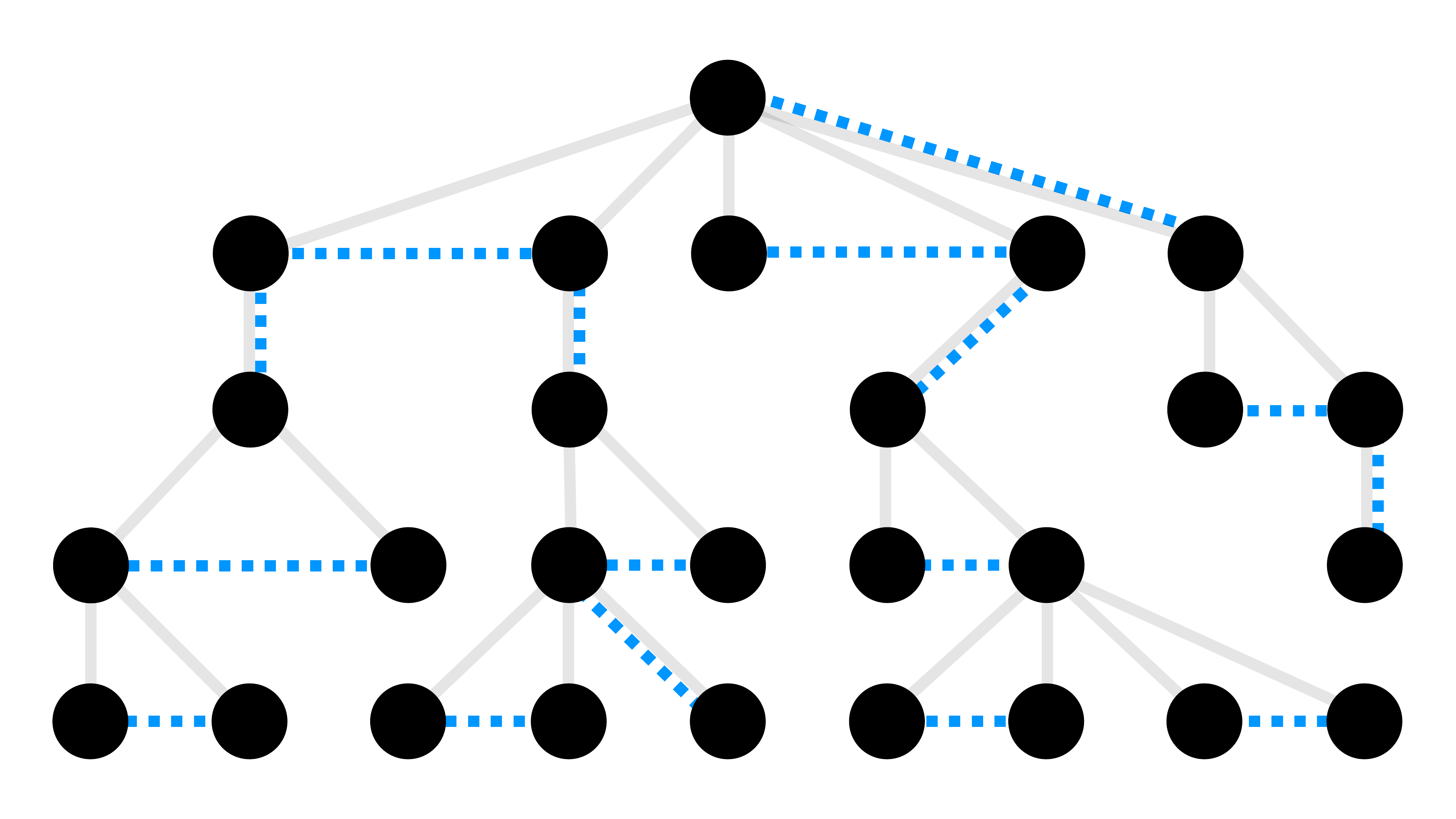}
        \caption{$\disperse_T$.}\label{sfig:dispTree2}
    \end{subfigure}
    \caption{How we disperse demand given a tree $T$ (\ref{sfig:dispTree1}). \ref{sfig:dispTree2} gives the support of $\disperse_T$ dashed in blue; notice that each vertex has degree at most $2$.}\label{fig:disperseTree}
\end{figure}

The following notion of tree matching demand formalizes how we disperse the demand in each tree of the forest decomposition of demand matching graph. Informally, given a tree this demand simply matches siblings in the tree to one another. If there are an odd number of siblings the leftover child is matched to its parent. See \Cref{fig:disperseTree} for an illustration.

\begin{definition}[Tree Matching Demand]\label{def:treeMatchDemand}
Given tree $T = (V,E)$ we define the tree-matching demand on $T$ as follows. Root $T$ arbitrarily. For each vertex $v$ with children $C_v$ do the following. If $|C_v|$ is odd let $U_v = C_v \cup \{v\}$, otherwise let $U_v = C_v$. Let $M_v$ be an arbitrary perfect matching on $U_v$ and define the demand associated with $v$ as 
\begin{align*}
    D_v(u_1, u_2) := 
    \begin{cases}
        1 & \text{if $\{u_1,u_2\} \in M_v$}\\
        0 & \text{otherwise}.
    \end{cases}
\end{align*}
where each edge in $M_v$ has an arbitrary canonical $u_1$ and $u_2$. Then, the tree matching demand for $T$ is defined as
\begin{align*}
    \disperse_{T} := \sum_{v \text{ internal in }T} D_v
\end{align*}
\end{definition}

We observe that a tree matching demand has size equal to the input size (up to constants).
\begin{lemma} \label{lem:treeMatchSize}
Let $T$ be a tree with $n-1$ vertices. Then $|\disperse_T| \geq \frac{n-1}{2}$.
\end{lemma}
\begin{proof}
    For each $v$ that is internal in $T$ let the vertices $U_v$ and the perfect matching $M_v$ on $U_v$ be as defined in \Cref{def:treeMatchDemand}. Then, observe that $\sum_v |U_v| \geq n-1$ since every vertex except for the root appear in at least one $U_v$. On the other hand, for each $v$ since $M_v$ is a perfect matching on $U_v$ we have $|M_v| = \frac{1}{2} |U_v|$ and since $|\disperse_T| = \sum_{v \text{ internal in } T} |M_v|$, it follows that $|\disperse_T| \geq \frac{n-1}{2}$ as required.
\end{proof}

Having formalized how we disperse a demand on a single tree with the tree matching demand, we now formalize how we disperse an arbitrary demand by taking a forest cover, applying the matching-dispersed demand to each tree and then scaling down by the arboricity.

\begin{definition}[Matching-Dispersed Demand]\label{dfn:matchingDemand}
Given graph $G$, node-weighting $A$ and demands $\mcD = (D_1, D_2, \ldots)$, let $G(\mcD)$ be the demand matching graph (\Cref{def:demandMatching}), let $T_1, T_2, \ldots$ be the trees of a minimum size forest cover with $\alpha$ forests of $G(\mcD)$ (\Cref{def:demandMatching}) and let $\disperse_{T_1}, \disperse_{T_2}, \ldots$ be the corresponding tree matching demands (\Cref{def:treeMatchDemand}). Then, the matching-dispersed demand on nodes $u,v \in V$ is
\begin{align*}
    \disperse_{\mcD, A}(u,v) := \frac{1}{2\alpha} \cdot \sum_i \sum_{u' \in \copies(u)}  \sum_{v' \in \copies(v)}\disperse_{T_i}(u', v')
\end{align*}
\end{definition}

We begin with a simple helper lemma that observes that the matching-dispersed demand has size essentially equal to the input demands (up to the arboricity).
\begin{lemma}\label{lem:matchingDemandSize}
    Given graph $G$, node-weighting $A$ and and demands $\mcD = (D_1, D_2, \ldots)$ where $G(\mcD)$ has arboricity $\alpha$, we have that the matching-dispersed demand $\disperse_{\mcD, A}$ satisfies $|\disperse_{\mcD, A}| \geq \frac{1}{4 \alpha} \sum_i |D_i|$
\end{lemma}
\begin{proof}
    Observe that the number of edges in $G(\mcD)$ is exactly $\sum_i |D_i|$ and so summing over each tree $T_j$ in our forest cover and applying \Cref{lem:treeMatchSize} gives
\begin{align*}
    \sum_j |\disperse_{T_j}| \geq \frac{1}{2} \cdot \sum_i |D_i|
\end{align*}
Combining this with the definition of $\disperse_{\mcD, A}$ (\Cref{dfn:matchingDemand}) gives the claim.
\end{proof}

We now argue the key properties of the matching-dispersed demand which will allow us to argue that it can be used as a witnessing demand for $\sum_i D_i$.
\begin{lemma}[Properties of Matching-Dispersed Demand]\label{lem:sparseOfMatching}
    Given graph $G = (V, E)$ and node-weighting $A$, let $C_1, C_2, \ldots$ be a sequence of moving cuts where $C_i$ is $(h, s)$-length $\phi_i$-sparse in $G - \sum_{j < i} C_i$ w.r.t.\ $A$ with witnessing demands $\mcD = (D_1, D_2, \ldots)$. Then the matching dispersed demand $\disperse_{\mcD, A}$ is:
    \begin{itemize}
        \item a $2h$-length $A$-respecting demand;
        \item $h \cdot (s-2)$-separated by $\sum_i C_i$ and;
        \item of size $|\disperse_{\mcD,A}| \geq \frac{1}{s^3 \cdot \log^3 N \cdot N^{O(1/s)}} \sum_i \frac{|C_i|}{\phi_i} $.
    \end{itemize}
\end{lemma}
\begin{proof}
To see that $\disperse_{\mcD, A}$ is $2h$-length observe that a pair of vertices $u$ and $v$ have $\disperse_{\mcD, A}(u,v) > 0$ only if there is a path consisting of at most two edges between a node in $\copies(u)$ and a node in $\copies(v)$ in the demand matching graph $G(\mcD)$ (\Cref{def:demandMatching}). Furthermore, $u' \in \copies(u)$ and $v' \in \copies(v)$ have an edge in $G(\mcD)$ only if there is some $i$ such that $D_i(u, v) > 0$ and since each $D_i$ is $h$-length, it follows that in such a case we know $d_G(u,v) \leq h$. Thus, it follows by the triangle inequality that $\disperse_{\mcD, A}$ is $2h$-length.
    
To see that $\disperse_{\mcD, A}$ is $A$-respecting we observe that each vertex in $G(D)$ is incident to at most $2\alpha$ matchings across all of the tree matching demands we use to construct $\disperse_{\mcD, A}$ (at most $2$ matchings per forest in our forest cover). Thus, for any $u \in V$ since $|\copies(u)| = A(u)$ we know
\begin{align*}
    \sum_{u' \in \copies(u)} \sum_j  \sum_{v} \sum_{v' \in \copies(v)}\disperse_{T_j}(u', v') \leq \sum_{u' \in \copies(u)} 2\alpha \leq 2\alpha \cdot A(u).
\end{align*}

It follows that for any $u \in V$  we have
\begin{align*}
\sum_v \disperse_{\mcD, A}(u, v) = \sum_v \frac{1}{2\alpha}\sum_j \sum_{u' \in \copies(u)}  \sum_{v' \in \copies(v)}\disperse_{T_j}(u', v') \leq A(u)
\end{align*}
A symmetric argument shows that $\sum_v \disperse_{\mcD, A}(v, u) \leq A(u)$ and so we have that $\disperse_{\mcD, A}$ is $A$-respecting.

We next argue that $\sum_i C_i$ is a moving cut that $h(s-2)$-separates $\disperse_{\mcD, A}$. Consider an arbitrary pair of vertices $u$ and $v$ such that $\disperse_{\mcD, A}(u,v) > 0$; it suffices to argue that $\sum_i C_i$ $h(s-2)$-separates $u$ and $v$. As noted above, $\disperse_{\mcD, A}(u,v) > 0$ only if there is a path $(u', w', v')$ in $G(\mcD)$ where $u' \in \copies(u)$, $v' \in \copies(v)$ and for some $w \in V$ we have $w' \in \copies(w)$. But, $\{u', w'\}$ and $\{w', v'\}$ are edges in $G(\mcD)$ only if there is some $i$ and $j$ such that $D_i(u,w) > 0$ and $D_j(w,v) > 0$. 
%By definition of $C_1, C_2,\ldots$, we know that $\sum_i C_i$ $hs$-separates $u$ from $w$ and $w$ from $v$. Furthermore,

By definition of $G(\mcD)$ (\Cref{def:demandMatching}), each $D_i$ corresponds to a different matching in $G(\mcD)$ and so since $\{u', w'\}$ and $\{w', v'\}$ share the vertex $w'$, we may assume $i \neq j$ and without loss of generality that $i < j$. Let $G_{\leq i}$ be $G$ with $\sum_{l \leq i} C_l$ applied.

Since $D_i$ is $hs$-separated by $C_{\leq i}$ and $D_i(u,w) > 0$, we know that
\begin{align}
    d_{G_{\leq i}}(u, w) \geq hs.\label{eq:a}
\end{align}

On the other hand, since $D_j$ is an $h$-length demand, $j > i$ and $D_j(w,v) > 0$, we know that the distance between $w$ and $v$ in $G_{\leq i}$ is
\begin{align}
    d_{G_{\leq i}}(w, v) \leq h. \label{eq:b}
\end{align}

Thus, it follows that $C_{\leq i}$ must $h(s-2)$ separate $u$ and $v$ since otherwise we would know that $d_{G_{\leq i}}(u,w) \leq h(s-2)$ and so combining this with \Cref{eq:b} and the triangle inequality we get $d_{G_{\leq i}}(u, w) \leq hs - h$, contradicting \Cref{eq:a}. Thus, $\sum_{i} C_i$ must $h(s-2)$ vertices $u$ and $v$.

Lastly, we argue that $|\disperse_{\mcD,A}| \geq \frac{1}{s^3 \cdot \log^3 N \cdot N^{O(1/s)}} \sum_i \frac{|C_i|}{\phi_i}$. By \Cref{lem:matchingDemandSize} we know that 
\begin{align*}
    |\disperse_{\mcD, A}| \geq \frac{1}{4\alpha} \sum_i |D_i|
\end{align*}
where $\alpha$ is the arboricity of $G(\mcD)$; applying our bound of $s^3 \cdot \log^3 N \cdot N^{O(1/s)}$ on the arboricity of $G(\mcD)$ from \Cref{lem:arbBound} and the fact that since each $C_i$ is $\phi_i$-sparse, we know that $|D_i| \geq \frac{|C_i|}{\phi_i}$ for each $i$ gives us 
\begin{align*}
|\disperse_{\mcD,A}| \geq \frac{1}{4 \cdot s^3 \cdot \log^3 N \cdot N^{O(1/s)}} \sum_i \frac{|C_i|}{\phi_i},
\end{align*}
 as required.
\end{proof}

\subsection{Proving Union of Sparse Moving Cuts is a Sparse Moving Cut}
We conclude this section by arguing that the union of sparse moving cuts is itself sparse. Our argument does so by using the matching-dispersed demand as the witnessing demand for the union of sparse cuts.

\unionOfCuts*
\begin{proof}
Recall that to demonstrate that $\sum_i C_i$ is a $\phi'$-sparse $(h', s')$-length sparse cut, it suffices to argue that there exists an $h'$-length $A$-respecting demand $D$ that is $h's'$-separated by $\sum_i C_i$ where $|D| \geq  \frac{\sum_i |C_i|}{\phi'}$.

\Cref{lem:sparseOfMatching} demonstrates the existence of exactly such a demand---namely the matching dispersed demand as defined in \Cref{dfn:matchingDemand}---for $h' = 2h$, $s' = \frac{(s-2)}{2}$ and $\phi' = s^3 \cdot \log^3 N \cdot N^{O(1/s)} \cdot \frac{\sum_i |C_i|}{\sum_i |C_i| / \phi_i}$, giving the claim.
\end{proof}

\section{Equivalence of Distances to Length-Constrained Expander}\label{sec:equivalences}
We now use the tools we developed in the previous section to argue that several quantities related to length-constrained expansion are all equal (up to slacks in sparsity, $h$ and length slack). This equivalence will form the backbone of the analysis of our algorithm. Before proceeding, it may be useful for the reader to recall the definition of a sequence of moving cuts (\Cref{dfn:movingCutSequence}) and $\spa$ (\Cref{def:sparsity}). The

% With the above tools in place, we can prove the key fact upon which the analysis of our algorithm is based. In particular, up to some lower-order slacks, we prove the equivalence of the (informally defined) quantities below.

% \begin{enumerate}
%     \item \textbf{Largest $(h,s)$-Length $\phi$-Sparse Cut Size (\qLC)}: size of the largest $(h,s)$-length $\phi$-sparse cut.
%     \item \textbf{Largest $(h,s)$-Length $\phi$-Sparse Cut Sequence Size (\qLSC)}: sums of the sizes of the largest sequence of $(h,s)$-length $\phi$-sparse cuts, where each cut is sparse in the graph after the preceding cuts in the sequence are applied.
%     \item \textbf{Largest Weighted $(h,s)$-Length $\phi$-Sparse Cut Sequence Size (\qLWSC)}: sums of the sizes of the largest sequence of $(h,s)$-length $\phi$-sparse cuts, where \emph{on average} each cut is sparse in the graph after the preceding cuts in the sequence are applied.
%     \item \textbf{Largest $(h,s)$-Length $\phi$-Expander's Complement Size (\qLEC)}: how much of the graph is not $(h,s)$-length $\phi$-expanding. 
%     \item \textbf{Smallest $(h,s)$-Length Expander Decomposition Size ($\qSED$)}: size of the smallest $(h,s)$-length $\phi'$-expander decomposition (assuming $\phi' \approx \phi$).
% \end{enumerate}
%  As the smallest $(h,s)$-length $\phi'$-expander decomposition is, in some sense, a graph's distance from being an expander, this result shows that all of the above quantities give equivalent ways of measuring how for a graph is from being a length-constrained expander.

The following series of definitions provides the quantities we will argue are all equal up to slacks.

\begin{definition}\label{dfn:equantities}
Fix a graph $G$, node-weighting $A$ and parameters $h$, $s$ and $\phi$. %$ that is $(h_G,s_G)$-length $\phi_G$-expanding. 
Then, we define the following quantities:
\begin{enumerate}
    \item \textbf{Largest Sparse Cut Size:} $\qLC(\phi, h, s)$ is the size of the largest $(h,s)$-length $\phi$-sparse cut in $G$ w.r.t.\ $A$. That is 
    \begin{align*}
        \qLC(\phi, h, s) := |C_0|
    \end{align*}
    where $C_0$ is the moving cut of largest size in the set $\{C : \spa_{(h,s)}(C,A) \leq \phi \}$.\label{eqv: maximum sparse cut}
    
    \item \textbf{Largest Sparse Cut Sequence Size:}
    $\qLSC(\phi, h, s)$ is the size of the largest sequence of $\phi$-sparse moving cuts. Then 
    \begin{align*}
        \qLSC(\phi, h, s) := \sum_i |C_i|
    \end{align*}
    where $(C_1, C_2, \ldots)$ is the $(h,s)$-length $\phi$-expanding moving cut sequence maximizing $\sum_i |C_i|$.\label{eqv: maximum sparse cut sequence}
    
    \item \textbf{Largest Weighted Sparse Cut Sequence Size:} $\qLWSC(\phi, h, s)$ is the largest weighted size of a sparse cut sequence, namely
    \begin{align*}
        \qLWSC(\phi, h, s) := \sum_i \frac{\phi}{\spa_{(h,s)}(C_i,A)} \cdot |C_i|
    \end{align*}
    where $(C_1, C_2, \ldots)$ is the $(h,s)$-length $\phi$-expanding moving cut sequence maximizing $\sum_i \frac{\phi}{\spa_{(h,s)}(C_i,A)} \cdot |C_i|$ and each $\spa_{(h,s)}(C_i, A)$ is computed after applying $C_j$ for $j <i$.\label{eqv: weighted maximum sparse cut sequence}
    
    \item \textbf{Largest Expander's Complement Size:} $\qLEC(\phi, h, s)$ is $\phi$ times the size of the complement of the largest $(h,s)$-length $\phi$-expanding subset of $A$. That is, let $\hat{A}$ be the $(h,s)$-length $\phi$-expanding node-weighting on $G$ satisfying $\hat{A} \preceq A$ with largest size and let $\bar{A} = A - \hat{A}$ be its complement. Then 
    \begin{align*}
        \qLEC(\phi, h, s) := \phi \cdot |\bar{A}|.
    \end{align*}\label{eqv: node weighting}
    
    \item \textbf{Smallest Expander Decomposition Size:} $\qSED(\phi, h, s)$ is the size of the smallest expander decomposition. That is, 
    \begin{align*}
        \qSED(\phi, h, s) := |C^*|
    \end{align*}
    where $C^*$ is the moving cut of minimum size such that $A$ is $(h,s)$-length $\phi$-expanding in $G-C^*$.\label{eqv: ED size}
\end{enumerate}
\end{definition}

The following formalizes our main claim in this section, the equivalence of the above quantities.
\begin{restatable}{thm}{equivQuan}
\label{thm: equivalence}
Fix a graph $G$, parameters $k, k' \geq 1$ and $\phi, h, s$ and a node-weighting $A$. 
\begin{align*}
    \qLC(\phi, h, s) \leq \qLEC(\phi_1, h, s) \leq \qLSC(\phi_2, h_2, s_2) \leq \qLC(\phi_3, h_3, s_3)
\end{align*}
where 
\begin{align*}
    &\phi_1 = 3 \phi&&\\
    &\phi_2 = \phi_1 \cdot 2^{O(k)} \cdot N^{O(1/k')} \cdot \log^{5} N  \cdot s^3 \cdot N^{O(1/s)}, &h_2 = h \cdot 2k', &&&s_2 = s \cdot O\left(\frac{k}{ k'} \cdot \log N \right)\\
    &\phi_3 = \phi_2 \cdot s^3 \cdot \log^3 N \cdot N^{O(1/s)},& h_3 = h_2 \cdot 2, &&&s_3 = \frac{(s_2-2)}{2}~.
\end{align*}
Furthermore, if $A$ is $(h_4, s_4)$-length $\phi_G$-expanding then:
\begin{align*}
    \frac{\phi_G}{\phi_4} \cdot \qLEC(\phi, h, s) \leq \qSED(\phi_4, h_4,s_4) \leq \qLSC(\phi_4, h_4, s_4) \leq \qLC(\phi_5, h_5, s_5) \leq \qLEC(\phi_5, h_5, s_5)
\end{align*}
where
\begin{align*}
    &\phi_4 = \phi \cdot 2^{O(k)} \cdot N^{O(1/k')} \cdot \log^{2} N, &h_4 = h \cdot k', &&&s_4 = O\left(\frac{k}{ k'} \cdot \log N \right)\\
    &\phi_5 = \phi_4 \cdot s^3 \cdot \log^3 N \cdot N^{O(1/s)}, &h_5 = h_4 \cdot 2, &&&s_5 = \frac{(s_4 - 2)}{2}~.
\end{align*}
\end{restatable}

The remainder of this section is dedicated to providing proofs of a series of inequalities which can be combined to get the inequalities in \Cref{thm: equivalence}. \paragraph*{Techniques.} We prove the equivalence of these quantities by a series of inequalities. Four of these inequalities are non-trivial and rely on the above-established theory; intuition below.
\begin{itemize}
    \item \textbf{$\qLWSC \leq \qLC$:} Let $(C_1, C_2, \ldots)$ be the largest weighted sequence of $(h,s)$-length $\phi$-sparse cuts and let $C_0$ be the largest $(h,s)$-length $\phi$-sparse cut. Showing that the size of $(C_1, C_2, \ldots)$ is at most the size of $C_0$ follows from observing that (as discussed above), one can take the union of cuts in $(C_1, C_2, \ldots)$ to (essentially) get an $(h,s)$-length $\phi$-sparse cut of equal size. Since $C_0$ is the largest such cut, the inequality follows. Here, we also make use of the idea of ``padding out'' a sparse cut which forces said cut to have an exact desired sparsity.
    \item  \textbf{$\qLC \leq \qLEC$:} Consider the largest $(h,s)$-length $\phi$-sparse cut $C_0$. Intuitively, $C_0$ should not cut too much into any part of the graph that is already $(h,s)$-length $\phi$-expanding, otherwise it would not sparse. Thus, $C_0$ cannot have size much larger than the part of the graph that is not expanding. Formalizing this intuition relies on the idea of a ``projected down demand.''
    \item \textbf{$\qLEC \leq \qLWSC$ and $\qLEC \leq \qSED$:} Again, let $(C_1, C_2, \ldots)$ be the largest weighted sequence of $(h,s)$-length $\phi$-sparse cuts. Proving that the amount of the graph that is not expanding is at most the size of $(C_1, C_2, \ldots)$ can be done using the above-described characterization of $(h,s)$-length expanders in terms of expander power embeddings. In particular, the union of cuts in $(C_1, C_2, \ldots)$ must be an $(h,s)$-length $\phi$-expander decomposition, otherwise we could append another cut to it and contradict its maximum size. Thus, after applying this sequence, the resulting graph can embed expander powers into \emph{most} neighborhoods using short paths. It follows that in the original graph these parts must have been $(h,s)$-length expanding, which, in turn, upper bounds how much of the original graph is not $(h,s)$-length expanding. Arguing $\qLEC \leq \qSED$ is similar.
\end{itemize}
The remainder of the inequalities to be proven are mostly immediate from the relevant definitions. We conclude the section with the proof of this theorem by appropriately stringing together these inequalities.

% As discussed in \Cref{sec:introEquiv}, the four non-trivial inequalities we must show are: 
% \begin{enumerate}
%     \item $\qLWSC \leq \qLC$ using the fact that union of sparse moving cuts is a sparse moving cut as developed in \Cref{sec:unionOfHCCuts} and the idea of ``padding out'' sparse cuts  (see \Cref{sec:LWSCAtMostLC} for details)
%     \item $\qLC \leq \qLEC$ using the notion of projected down demand (see \Cref{sec:LCAtMostLEC} for details).
%     \item $\qLEC \leq \qLWSC, \qSED$ using our characterization of length-constrained expanders in terms of expander power embeddings as developed in \Cref{sec:expanderPowersCharacterization} (see \Cref{sec:LECAtMostWLWSC} and \Cref{sec:LECAtMostSED})
% \end{enumerate}
% \noindent See \Cref{fig:ineqOverview} for an overview of the inequalities we use and the tools we use for each inequality.

\begin{figure}
    \centering
        \centering
        \includegraphics[width=.8\columnwidth,trim=0mm 0mm 0mm 0mm, clip]{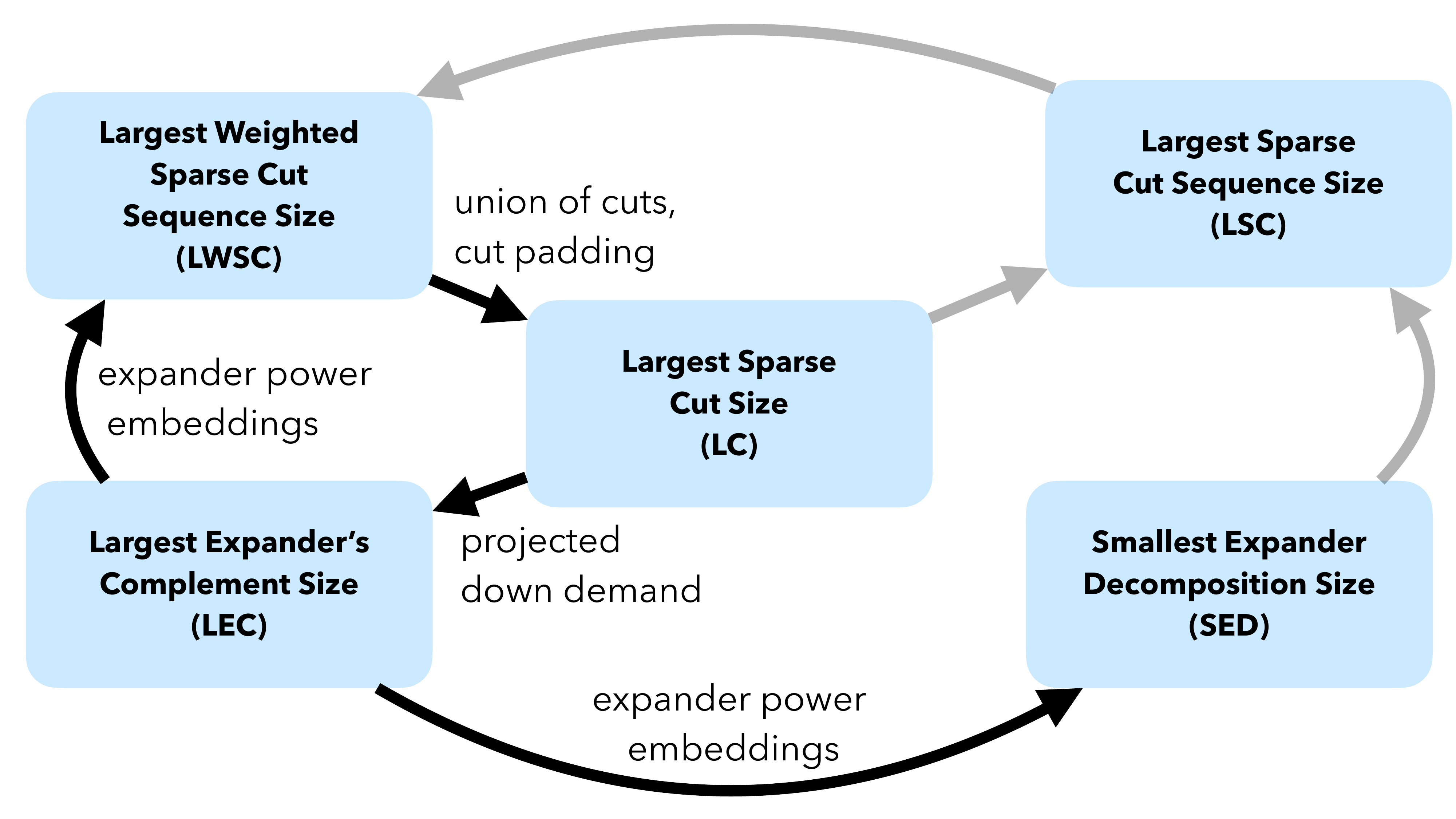}
    \caption{An overview of the inequalities we show. An arrow from $a$ to $b$ indicates $a \leq b$. Each non-trivial inequality opaque and labeled with the key idea of its proof.}\label{fig:ineqOverview}
\end{figure}

\future{This should all be largest demand-size sparse cut sequence}

\subsection{Weighted Sparse Cut Sequence (\ref{eqv: weighted maximum sparse cut sequence}) at Most Largest Cut (\ref{eqv: maximum sparse cut})}\label{sec:LWSCAtMostLC}

We show $\qLWSC \leq \qLC$ using the idea of padding out sparse cuts and the fact that the union of sparse moving cuts is sparse. In particular, padding out cuts allows us to increase the size of our cuts while forcing them to still have bounded sparsity. This allows us to pad out the cuts in our sequence and then take the union of this padded out sequence to observe that it results in a sparse cut which, by definition, can be no larger than the sparsest cut.

\begin{restatable}{thm}{LWSCAtMostLC}\label{lem:LWSCAtMostLC}
Given graph $G$ and node-weighting $A$, we have that 
\begin{align*}
\qLWSC(\phi, h, s) \leq \qLC(\phi', h', s')
\end{align*}
where $\phi' = \phi \cdot s^3 \cdot \log^3 N \cdot N^{O(1/s)}$, $h' = 2h$ and $s' = \frac{(s-2)}{2}$.
\end{restatable}

The following formalizes our notion of padding out sparse cuts.
\begin{lemma}[Padding Out Sparse Cut]\label{lem:paddingCuts}
Given graph $G$, let $C$ be a $(h,s)$-length $\phi$-sparse moving cut w.r.t.\ node-weighting $A$. Then there exists an $(h,s)$-length $\phi$-sparse moving cut $C'$ such that 
\begin{align*}
    \frac{\phi}{\spa_{(h,s)}(C,A)} \cdot |C| \leq |C'|.
\end{align*}
\end{lemma}
\begin{proof}
The basic idea is to simply arbitrarily add length increases to $C$ which increases its size while increasing its sparsity; doing so until its sparsity reaches $\phi$ allows us to make $C'$ of the appropriate size. We must take some slight care to make sure that there are enough length increases we can add to $C$ to make it appropriately large.

% In particular, since $C$ is $\phi$-sparse we know that there is an $A$-respecting demand $D$ that is $hs$-separated by $C$ where $|D| \geq \frac{|C|}{\spa_{(h,s)}(C,A)}$ and, in particular, it follows that 
% \begin{align*}
%     |A| \geq \frac{|C|}{\spa_{(h,s)}(C,A)}.
% \end{align*}

For each edge $e$ let $\bar{C}(e) = 1 - C(e)$ be the complement of $C$. It follows that %Since $|C|$ is $\spa_{(h,s)}(C,A)$-sparse we know that $|C| \leq \frac{|A|}{\spa_{(h,s)}(C,A)}$ and so 
\begin{align*}
    |\bar{C}| \geq |A| - |C|.    
\end{align*}
Furthermore, since there is some $A$-respecting demand $D$ witnessing the $(h,s)$-length $\spa_{(h,s)}(C,A)$-sparsity of $C$ where $\frac{|C|}{\spa_{(h,s)}(C,A)} = |D| \leq |A|$ so we have that
\begin{align*}
\frac{\phi}{\spa_{(h,s)}(C,A)} \cdot |C| \leq |A|.
\end{align*}

Thus, it follows that $|\bar{C}| \geq \frac{\phi}{\spa_{(h,s)}(C,A)} \cdot |C| - |C|$ and so we can arbitrarily increase the length of edges to turn $C$ into a $C'$ satisfying $\frac{\phi}{\spa_{(h,s)}(C,A)} \cdot |C| = |C'|$.\footnote{We ignore rounding to multiples of $\frac{1}{h}$ here for simplicity of presentation.} 

Lastly, any such $C'$ is $(h,s)$-length $\phi$-sparse since $C'$ $hs$-separates $D$ and so
\begin{align*}
    \spa_{(h,s)}(C',A) \geq \frac{|C'|}{|D|} = \frac{\phi \cdot |C|}{|D| \cdot \spars_{(h,s)}(C,A)} = \phi.
\end{align*}
\end{proof}

We now prove the main theorem of this section.
\LWSCAtMostLC*
\begin{proof}
The proof is immediate from \Cref{lem:paddingCuts} and the fact that the union of sparse moving cuts is itself a sparse cut as per \Cref{thm:unionOfMovingCuts} and so smaller than the largest sparse moving cut. 

More formally, let $(C_1, C_2, \ldots)$ be the largest $(h,s)$-length $\phi$-sparse moving cut sequence w.r.t.\ $A$ in $G$  of largest weighted size as defined in \Cref{dfn:equantities} and let $C_0$ be the $(h,s)$-length $\phi$-sparse cut of largest size w.r.t.\ $A$ in $G$. Our goal is to show 
\begin{align}
    \sum_{i \geq 1} \frac{\phi}{\spa_{(h,s)}(C_i,A)} \cdot |C_i| \leq |C_0|.\label{eq:x}
\end{align}

% By \Cref{lem:paddingCuts} we know that there exists an $(h,s)$-length $\phi$-sparse moving cut sequence $(C_1', C_2', \ldots)$ such that 
% \begin{align*}
%     \sum_i \frac{\phi}{\spa_{(h,s)}(C_i,A)} \cdot |C_i| \leq \sum_i |C_i'|
% \end{align*}
% and so it suffices to show that
% \begin{align*}
%     \sum_{i \geq 1} |C_i'| \leq |C_0|.
% \end{align*}

Let $\mcC := \sum_{i \geq 1} C_i$ be the union of our moving cut sequence. By \Cref{thm:unionOfMovingCuts} we know that
\begin{align}
    \spa_{(h',s')}(\mcC,A) \leq s^3 \cdot \log^3 N \cdot N^{O(1/s)} \cdot \frac{\sum_i |C_i|}{\sum_i |C_i|/\spa_{(h,s)}(C_i, A)}.\label{eq:y}
\end{align}

Furthermore, by \Cref{lem:paddingCuts} we know that there exists an $(h',s')$-length $\phi'$-sparse moving cut $\mcC'$ such that 
\begin{align}
    \frac{\phi'}{\spa_{(h',s')}(\mcC,A)} \cdot |\mcC| \leq |\mcC'|. \label{eq:z}
\end{align}

Thus, combining \Cref{eq:y} and \Cref{eq:z} we get
\begin{align*}
    \frac{\sum_i |C_i|/\spa_{(h,s)}(C_i, A)}{s^3 \cdot \log^3 N \cdot N^{O(1/s)} \cdot \sum_i |C_i|} \cdot \phi' \cdot |\mcC| \leq |\mcC'|
\end{align*}
and so using our definition of $\phi'$ and the fact that $|\mcC| = \sum_{i \geq 1} |C_i|$ we get
\begin{align*}
    \sum_{i \geq 1} \frac{\phi}{\spa_{(h,s)}(C_i,A)} \cdot |C_i| \leq |\mcC'|.
\end{align*}
However, since $C_0$ is the largest $(h',s')$-length $\phi'$-sparse cut we know that $|\mcC'| \leq |C_0|$, giving our desired inequality (\Cref{eq:x}).
\end{proof}

\subsection{Largest Cut (\ref{eqv: maximum sparse cut}) at Most Largest Expander's Complement (\ref{eqv: node weighting})}\label{sec:LCAtMostLEC}

We now show $\qLC \leq \qLEC$. The basic idea is to argue that if the largest length-constrained sparse cut were too large then it would cut into the length-constrained expanding part of our graph, contradicting its sparsity. We formalize this argument with the notion of projected down demand.

\begin{restatable}{thm}{LCAtMostLEC}\label{thm:lcAtMostLEC}
Given graph $G$ and node-weighting $A$ and parameters $h,s,\phi$, we have that 
\begin{align*}
\qLC(\phi, h, s) \leq \qLEC(\phi', h, s)  
\end{align*}
where $\phi' = 3 \phi$.
\end{restatable}

\subsubsection{Projected Down Demands}
The following formalizes the projected down demand.
\begin{definition}[Projected Down Demand]\label{def:projDemand}
Suppose we are given graph $G$, node-weighting $A$, $A$-respecting demand $D$ and $\hat{A} \preceq A$ where $\bar{A} := A - \hat{A}$ is the complement of $\hat{A}$. Then, let $D^+$ be any demand such that $\sum_v D^+(u, v) = \min(\bar{A}(u), \sum_v D(u,v))$ for every $u$ and  $D^+ \preceq D$. Symmetrically, let $D^-$ be any demand such that $\sum_v D(v,u) = \min(\bar{A}(u), \sum_v D(v,u))$ and $D^- \preceq D$. Then we define the demand $D$ projected down to $\hat{A}$ on $(u,v)$ as
\begin{align*}
    D^{\downharpoonright \hat{A}}(u,v) := \max(0, D(u,v) - D^+(u,v) - D^-(u,v)).
\end{align*}
\end{definition}

The following establishes the basic properties of $D^{\downharpoonright \hat{A}}$.
\begin{lemma}\label{lem:projDemProps} Given graph $G$, node-weighting $A$, $\hat{A} \preceq A$ where $\bar{A} := A - \hat{A}$, we have that $D^{\downharpoonright \hat{A}}$ is $\hat{A}$-respecting, $|D^{\downharpoonright \hat{A}}| \geq |D| - 2 |\bar{A}|$ and $D^{\downharpoonright \hat{A}} \preceq D$.
\end{lemma}
\begin{proof}
To see that $D^{\downharpoonright \hat{A}}$ is $\hat{A}$-respecting, fix a vertex $u$. Casing on the minimizer of $\min(\bar{A}(u), \sum_v D(u,v))$ we have the following.
\begin{itemize}
    \item If $\sum_v D^+(u,v) = \bar{A}(u)$ (where $D^+$ is defined in \Cref{def:projDemand}) then by the fact that $D$ is $A$-respecting we have
    \begin{align*}
        \sum_v D^{\downharpoonright \hat{A}}(u,v) &\leq \sum_v D(u,v) - \sum_{v} D^+(u,v) \\
        &\leq A(v) - \bar{A}(u) \\
        & = \hat{A}(u).
    \end{align*}
    \item On the other hand, if $\sum_v D^+(u,v) = \sum_v D(u,v)$ then by the non-negativity of node-weightings we have \begin{align*}
        \sum_v D^{\downharpoonright \hat{A}}(u,v) &\leq \sum_v D(u,v) - \sum_{v} D^+(u,v) \\
        &= 0 \\
        & = \hat{A}(u).
    \end{align*}
\end{itemize}

In either case we have $\sum_v D^{\downharpoonright \hat{A}}(u,v) \leq \hat{A}(u)$. A symmetric argument using  $D^-$ (where $D^-$ is defined in \Cref{def:projDemand}) shows that $\sum_v D^{\downharpoonright \hat{A}}(v,u) \leq \hat{A}(u)$ and so $D^{\downharpoonright \hat{A}}(u,v)$ is $\hat{A}$-respecting.

To see that $|D^{\downharpoonright \hat{A}}| \geq |D| - 2 |\bar{A}|$, observe that by definition, $|D^{\downharpoonright \hat{A}}| \geq |D|-|D^+| - |D^-|$. But, also by definition, $|D^+|, |D^-| \leq |\bar{A}|$, giving the claim.

Lastly, observe that $D^{\downharpoonright \hat{A}} \preceq D$ trivially by construction.
\end{proof}

\subsubsection{Proof Of Largest Cut At Most Largest Expander's Complement}
Having formalized the projected down demand, we can now formally prove the main theorem of this section.
\LCAtMostLEC*
\begin{proof}
The basic idea of the proof is to argue that the largest sparse cut cannot be larger than the size of the complement of the largest expanding subset because if it were then it would cut into the the largest expanding subset itself; this contradicts the fact that no sparse cut can cut too much into an expanding subset. The projected down demand (\Cref{def:projDemand}) allows us to formalize this idea.

More formally, let $C_0$ be the $(h,s)$-length $\phi$-sparse cut of largest size w.r.t.\ $A$ in $G$ and let $\bar{A}$ be the complement of the largest $(h,s)$-length $\phi'$-expanding subset $\hat{A} \preceq A$ as in \Cref{dfn:equantities}.

Let $D$ be the demand that witnesses the $(h,s)$-length $\phi$-sparsity of $C_0$; that is, it is the minimizing $A$-respecting demand in \Cref{def:sparsity}. We may assume, without loss of generality, that $C_0$ $hs$-separates all of $D$; that is, $\sep_{hs}(C_0, D) = |D|$. Let $D^{\downharpoonright \hat{A}}$ be the projected down demand (as in \Cref{def:projDemand}). Recall that by \Cref{lem:projDemProps} we know that $D^{\downharpoonright \hat{A}}$ is $\hat{A}$-respecting, $|D^{\downharpoonright \hat{A}}| \geq |D| - 2 |\bar{A}|$ and $D^{\downharpoonright \hat{A}} \preceq D$.

However, since $C_0$ $hs$-separates all of $D$ and $D^{\downharpoonright \hat{A}} \preceq D$ we know that $C_0$ must $hs$-separate all of $D^{\downharpoonright \hat{A}}$ and so applying this and $|D^{\downharpoonright \hat{A}}| \geq |D| - 2 |\bar{A}|$ we have
\begin{align}\label{eq:ay}
    \spa_{s \cdot h}(C_0, D^{\downharpoonright \hat{A}}) = \frac{|C_0|}{\sep(C_0, D^{\downharpoonright \hat{A}})} = \frac{|C_0|}{|D^{\downharpoonright \hat{A}}|}\leq  \frac{|C_0|}{|D| - 2 |\bar{A}|}
\end{align}
where, as a reminder, $\spa$ is defined in \Cref{dfn:CDSparse} and \Cref{def:sparsity}.

On the other hand, since $D^{\downharpoonright \hat{A}}$ is $\hat{A}$-respecting and $\hat{A}$ is $(h,s)$-length $\phi'$-expanding by definition, we know that no cut can be too sparse w.r.t.\ $\hat{A}$ and, in particular, we know that
\begin{align}\label{eq:by}
    3\phi = \phi' \leq \spa_{(h,s)}(C_0, \hat{A}) \leq \spa_{s
    \cdot h}(C_0, D^{\downharpoonright \hat{A}})
\end{align}

Combining \Cref{eq:ay} and \Cref{eq:by} and solving for $\phi' \cdot |\bar{A}|$ we have
\begin{align*}
    \frac{3\phi \cdot |D| - |C_0|}{2} \leq \phi' \cdot |\bar{A}|.
\end{align*}

However, recall that $C_0$ is an $(h,s)$-length $\phi$-sparse cut witnessed by $D$ and, in particular, this means that $\frac{1}{\phi} |C_0|\geq |D|$. Applying this we conclude that 
\begin{align*}
    |C_0| \leq \phi' \cdot |\bar{A}|.
\end{align*}
as required.
\end{proof}

\subsection{Largest Expander's Complement (\ref{eqv: node weighting}) at Most Weighted Sparse Cut Sequence (\ref{eqv: weighted maximum sparse cut sequence})}\label{sec:LECAtMostWLWSC}
We now argue that $\qLEC \leq \qLWSC$. The basic idea is to use our characterization of length-constrained expanders in terms of expander power embeddings (as developed in \Cref{sec:expanderPowersCharacterization} and formalized by the neighborhood router demand). In particular, any demand of the neighborhood router demand not separated by the largest cut sequence must be efficiently routable after applying any expander decomposition since the resulting graph is a length-constrained expander. By \Cref{lem:neighRouting} this implies the existence of a large expanding subset and so the complement of the largest expanding expanding complement must be small. Formally we show the following.

\begin{restatable}{thm}{LECatMostLWSC}\label{lem:LECatMostLWSC}
Given graph $G$ and node-weighting $A$, we have that 
\begin{align*}
\qLEC(\phi, h, s) \leq \qLWSC(\phi', h', s')
\end{align*}
where $\phi' = \phi \cdot \Omega\left(2^{O(k)} \cdot N^{O(1/k')} \cdot \log^{2} N \right)$, $h' = h \cdot k'$ and $s' = s \cdot O\left(\frac{k}{ k'} \cdot \log N \right)$.
\end{restatable}

It will be useful for us to abstract out this argument as we will later use it to argue that $\qLEC \leq \qSED$. The following formalizes the fact which we abstract out.

\begin{lemma}\label{lem:EDLemma}
Given graph $G$, node-weighting $A$ and parameters $k,k' \geq 1$, let $\bar{A}$ be the largest $(h,s)$-length $\phi$-expanding subset's complement (as defined in \Cref{dfn:equantities}). Furthermore, let $C^*$ be an $(h',s')$-length $\phi'$-expander decomposition where $\phi' = \phi \cdot \Omega\left(2^{O(k)} \cdot N^{O(1/k')} \cdot \log^{3} N \right)$, $h' = h \cdot k'$ and $s' = s \cdot O\left(\frac{k}{ k'} \cdot \log N \right)$. Then
\begin{align*}
    \spa_{(h',s')}(C^*,A) \cdot |\bar{A}| \leq |C^*| \cdot O\left(2^{O(k)} \cdot N^{O(1/k')} \cdot \log^2 N \right).
\end{align*}

\end{lemma}
\begin{proof}
Let $G' := G - C^*$ be $G$ with $C^*$ applied. Observe that $G'$ must be an $(h',s')$-length $\phi'$-expander since $C^*$ is an expander decomposition.

Next, let $\nrd$ be the neighborhood router demand, as defined in \Cref{sec:NRD}. Recall that by \Cref{lem:nrdProps} we know $\nrd$ is $A$-respecting and $h'$-length. Let $\nrd'$ be $\nrd$ restricted to be $h's'$-length in $G'$. That is, $\nrd'$ on $u$ and $v$ is defined as
\begin{align*}
    \nrd'(u,v) &:= \begin{cases}
        \nrd(u,v) & \text{if $d_{G'}(u,v) \leq h's'$}\\
        0 & \text{otherwise}.
    \end{cases}
\end{align*}

Observe that $\nrd'$ is trivially $A$-respecting since \nrd is $A$-respecting. Also, $\nrd'$ is $h's'$-length by construction. Additionally, observe that by construction we have
\begin{align*}
    |\nrd'| \geq |\nrd| - \sep_{h's'}\left(C^*, A \right).
\end{align*}

Since $G'$ is an $(h', s')$-length $\phi'$-expander and $\nrd'$ is $A$-respecting and $h's'$-length, we know that $\nrd'$ can be routed in $G'$ with congestion at most $O(\frac{\log N}{\phi'})$ and dilation $O(h's')$ by \Cref{thm:flow character}. Letting $\eps = \frac{\sep_{h's'}(C^*, A)}{ |\nrd|}$, it follows that a $1-\eps$ fraction of \nrd can be routed with congestion at most $O(\frac{\log N}{\phi'})$ and dilation at most $O(h's')$. Applying \Cref{lem:neighRouting} and the fact that $|\nrd| \geq \Omega\left(\frac{1}{N^{O(1/k')} \cdot \log N} \cdot |A| \right)$ by \Cref{lem:nrdProps}, we have that that there is a node weighting $A' \preceq A$ of size at least
\begin{align*}
    |A'| \geq |A| \cdot \left(1 - \frac{\sep_{h's'}(C^*, A)}{ |\nrd|} \cdot O\left(2^{O(k)} \cdot N^{O(1/k')} \cdot \log N \right) \right)\\
    \geq |A| - \sep_{h's'}(C^*, A) \cdot O\left(2^{O(k)} \cdot N^{O(1/k')} \cdot \log^2 N \right)
\end{align*}
such that $A'$ is $(h, s)$-length $\phi$-expanding in $G'$. 

Let $\bar{A}$ be the complement of the largest $(h,s)$-length $\phi$-expanding subset of $A$ as in \Cref{dfn:equantities}.
Then, observing that $|\bar{A}|\leq |A| - |A'|$ and $\sep_{h's'}(\mcC, A) = \frac{|\mcC|}{\spa_{(h',s')}(\mcC,A)}$ then gives:
\begin{align*}
    |\bar{A}| \leq |\mcC| \cdot O\left(\frac{2^{O(k)} \cdot N^{O(1/k')} \cdot \log^2 N}{\spa_{(h',s')}(\mcC,A)} \right)
\end{align*}
as required.
\end{proof}

Applying the above helper lemma allows us to conclude the main fact of this section.
\LECatMostLWSC*
\begin{proof}
    The basic idea is to observe that the union of the largest sequence of sparse cuts is an expander decomposition and then apply \Cref{lem:EDLemma}.

    More formally, let $\bar{A}$ be the largest $(h,s)$-length $\phi$-expanding subset's complement (as defined in \Cref{dfn:equantities}). Also, let $(C_1, C_2, \ldots)$ be the $(h',s')$-length $\phi'$-sparse moving cut sequence of largest weighted size as defined in \Cref{dfn:equantities}, let $\mcC := \sum_i C_i$. 
    
    Observe that $\mcC$ must be an $(h',s')$-length $\phi'$-expander decomposition since otherwise there would be a cut that could be appended to $(C_1, C_2, \ldots)$ to increase its size, contradicting its maximality. It follows by \Cref{lem:EDLemma} that
    \begin{align}\label{eq:xa}
        \spa_{(h',s')}(\mcC,A) \cdot |\bar{A}| \leq |\mcC| \cdot O\left(2^{O(k)} \cdot N^{O(1/k')} \cdot \log^2 N \right).
    \end{align}

    Let $D$ be the $A$-respecting demand witnessing the sparsity of $\mcC$ so that 
    \begin{align*}
    \spa_{(h',s')}(\mcC,A) = \frac{|\mcC|}{\sep_{h's'}(\mcC, D)}.    
    \end{align*}
    Observe that by definition of $\mcC$ and since $\spa_{(h',s')}(C_i, D) \geq \spa_{(h',s')}(C_i, A)$ we get
    \begin{align*}
        \sep_{(h's')}(\mcC,D) = \sum_i \sep_{h's'}(C_i, D) = \sum_i \frac{|C_i|}{\spa_{(h',s')}(C_i, D)} \leq \sum_i \frac{|C_i|}{\spa_{(h',s')}(C_i, A)}
    \end{align*}
    and so it follows that 
    \begin{align*}
        \spa_{(h',s')}(\mcC,A) \geq \frac{|\mcC|}{ \sum_i |C_i| / \spa_{(h',s')}(C_i, A)}.
    \end{align*}
    Combining this bound on $\spa_{(h',s')}(\mcC,A)$ and \Cref{eq:xa} we get
    \begin{align*} 
        |\bar{A}| \leq \sum_i \frac{|C_i|}{\spa_{(h',s')}(C_i, A)} \cdot O\left(2^{O(k)} \cdot N^{O(1/k')} \cdot \log^2 N \right).
    \end{align*}
    Multiplying both sides by $\phi$ we get
    \begin{align*}
        \phi \cdot |\bar{A}| \leq \sum_i \frac{\phi'}{\spa_{(h',s')}(C_i, A)} \cdot |C_i|.
    \end{align*}
    as required.
\end{proof}

The remaining inequalities we show are mostly trivial.

\subsection{Largest Expander's Complement (\ref{eqv: node weighting}) at Most Smallest Expander Decomposition (\ref{eqv: ED size})}\label{sec:LECAtMostSED}
We now leverage the helper lemma from the previous section (\Cref{lem:EDLemma}) to show $\qLEC \leq \qSED$.

\begin{restatable}{thm}{LECAtMostSED}\label{lem:LECAtMostSED}
Given graph $G$ that is $(h',s')$-length $\phi_G$-expanding w.r.t.\  node-weighting $A$ and parameters $h,s,\phi$, we have that 
\begin{align*}
\frac{\phi_G}{\phi'} \cdot \qLEC(\phi, h, s) \leq \qSED(\phi', h', s').
\end{align*}
where $\phi' = \phi \cdot \Omega\left(2^{O(k)} \cdot N^{O(1/k')} \cdot \log^{2} N \right)$, $h' = h \cdot k'$ and $s' = s \cdot O\left(\frac{k}{ k'} \cdot \log N \right)$.
\end{restatable}
\begin{proof}
The proof is immediate from our previous helper lemma, \Cref{lem:EDLemma}.

In particular, let $C^*$ be the smallest $(h',s')$-length $\phi'$-expander decomposition and let $\bar{A}$ be the largest $(h,s)$-length $\phi$-expanding subset's complement (as in \Cref{dfn:equantities}). By \Cref{lem:EDLemma} we know that
    \begin{align*}
        \spa_{(h',s')}(C^*,A) \cdot |\bar{A}| \leq |\mcC| \cdot O\left(2^{O(k)} \cdot N^{O(1/k')} \cdot \log^2 N \right).
    \end{align*}
Since $G$ is $(h',s')$-length $\phi_G$-expanding we know $\spa_{(h',s')}(C^*,A) \geq \phi_G$ and so 
    \begin{align*}
        \phi_G \cdot |\bar{A}| \leq |\mcC| \cdot O\left(2^{O(k)} \cdot N^{O(1/k')} \cdot \log^2 N \right).
    \end{align*}
Multiplying both sides by $\phi$ and applying the definition of $\phi'$ then gives
    \begin{align*}
        \frac{\phi_G}{\phi'} \cdot \phi |\bar{A}| \leq |\mcC|.
    \end{align*}
as required.    
\end{proof}

\subsection{Largest Cut (\ref{eqv: maximum sparse cut}) at Most Largest Cut Sequence (\ref{eqv: maximum sparse cut sequence})}
Showing $\qLC \leq \qLSC$ is trivial since the first moving cut of any sparse length-constrained cut sequence cut could be the largest sparse length-constrained cut.

\begin{restatable}{thm}{LCAtMostLSC}\label{lem:LCAtMostLSC}
Given graph $G$ and node-weighting $A$ and parameters $h,s,\phi$, we have that 
\begin{align*}
\qLC(\phi, h, s) \leq \qLSC(\phi, h, s).  
\end{align*}
%where $\phi' = BLAH$, $h' = BLAH$ and $s' = BLAH$.
\end{restatable}
\begin{proof}
Any cut sequence $(C_1,C_2,\ldots)$ that begins with $C$ will satisfy $|C|\le \sum_{i}|C_i|$ and since $C$ can always be chosen as the first cut in an $(h,s,\phi)$-sequence this gives the theorem. 
\end{proof}

\subsection{Largest Cut Sequence (\ref{eqv: maximum sparse cut sequence}) at Most Largest Weighted Cut Sequence (\ref{eqv: weighted maximum sparse cut sequence})}

Likewise $\qLSC \leq \qLWSC$ is trivial since the largest weighted sequence of sparse length-constrained cuts always has as a candidate the largest (unweighted) sequence of sparse length-constrained cuts.

\begin{restatable}{thm}{LSCAtMostLWSC}\label{lem:LSCAtMostLWSC}
Given graph $G$ and node-weighting $A$ and parameters $h,s,\phi$, we have that 
\begin{align*}
\qLSC(\phi, h, s) \leq \qLWSC(\phi, h, s).
\end{align*}
%where $\phi' = BLAH$, $h' = BLAH$ and $s' = BLAH$.
\end{restatable}
\begin{proof}
The proof is immediate from the fact that the weighted size of a sequence of $(h,s)$-length $\phi$-sparse cuts is always larger than its actual size

Specifically, let $(C_1, C_2,)$ be the largest sequence of $(h,s)$-length $\phi$-sparse cuts and let $(C_1', C_2', \ldots)$ be the sequence of of $(h,s)$-length $\phi$-sparse cuts of largest weighted size (as in \Cref{dfn:equantities}). Observe that since each $C_i$ is $\phi$-sparse we know that for each $C_i$ we have $\spa_{(h,s)}(C_i,A) \leq \phi$ and since $(C_1', C_2', \ldots)$ is the $(h,s)$-length $\phi$-sparse cut sequence of largest size we have
\begin{align*}
\sum_i |C_i| \leq \sum_i \frac{\phi}{\spa_{(h,s)}(C_i,A)} \cdot |C_i| \leq \sum_i \frac{\phi}{\spa_{(h,s)}(C_i',A)} \cdot |C_i'|
\end{align*}
as required.
\end{proof}

\subsection{Smallest Expander Decomposition (\ref{eqv: ED size}) at Most Largest Cut Sequence (\ref{eqv: maximum sparse cut sequence})}

Lastly, $\qSED \leq \qLSC$ is trivial since the largest sequence of length-constrained sparse cuts is itself a length-constrained expander decomposition.

\begin{restatable}{thm}{SEDAtMostLSC}\label{lem:SEDAtMostLSC}
Given graph $G$ and node-weighting $A$ and parameters $h,s,\phi$, we have that 
\begin{align*}
\qSED(\phi, h, s) \leq \qLSC(\phi, h, s).
\end{align*}
%where $\phi' = BLAH$, $h' = BLAH$ and $s' = BLAH$.
\end{restatable}
\begin{proof}
 Let $(C_1, C_2, \ldots)$ be the largest sequence of $(h,s)$-length $\phi$-sparse cuts and let $C^*$ be the smallest $(h,s)$-length $\phi$-expander decomposition.

From the maximality of the sequence $(C_1, C_2,\ldots)$, it must hold that $G-\sum_{i}C_i$ is an $(h,s)$-length $\phi$-expander, as otherwise we can find another moving cut and append it to the sequence, contradicting the maximality of the sequence $(C_1, C_2, \ldots)$. Therefore, the moving cut $\mcC =\sum_{i} C_i$ is an $(h,s)$-length $\phi$-expander decomposition of $G$ for $A$. But since $C^*$ is the smallest such expander decomposition we then know that $|C^*| \leq |\mcC| = \sum_i |C_i|$ as required.
\end{proof}

\subsection{Proof of Equivalence of Distance Measures (\Cref{thm: equivalence})}
We conclude this section by stringing together our proven inequalities to show the equivalence of our various graph quantities.

\equivQuan*
\begin{proof}
We string together our inequalities as follows.

$\qLC(\phi, h, s) \leq \qLEC(\phi_1, h, s)$ by \Cref{thm:lcAtMostLEC}

$\qLEC(\phi_1, h_1, s_1) \leq \qLSC(\phi_2, h_2, s_2)$ by \Cref{lem:LECatMostLWSC}, \Cref{lem:LWSCAtMostLC} and \Cref{lem:LCAtMostLSC}.

$\qLSC(\phi_2, h_2, s_2) \leq \qLC(\phi_3, h_3, s_3)$ by \Cref{lem:LSCAtMostLWSC} and \Cref{lem:LWSCAtMostLC}.

For our second set of inequalities we additionally have the following.

$\frac{\phi_G}{\phi_4} \cdot \qLEC(\phi, h, s) \leq \qSED(\phi_4, h_4,s_4)$ by \Cref{lem:LECAtMostSED}.

$\qSED(\phi_4, h_4,s_4) \leq \qLSC(\phi_4, h_4, s_4)$ by \Cref{lem:SEDAtMostLSC}.

$\qLSC(\phi_4, h_4, s_4) \leq \qLC(\phi_5, h_5, s_5)$ again by \Cref{lem:LSCAtMostLWSC} and \Cref{lem:LWSCAtMostLC}.

$\qLC(\phi_5, h_5, s_5) \leq \qLSC(\phi_5, h_5, s_5)$ by \Cref{lem:LCAtMostLSC}.    
\end{proof}

\section{Algorithm: Sparse Flows and Cutmatches}\label{sec:sparseFlows}
Our algorithm to compute large length-constrained sparse cuts from cut strategies will be based on the previously-studied idea of a (length-constrained) cutmatch. Informally, a cutmatch matches two node sets over flow paths and finds a cut certifying that the unmatched nodes cannot be matched without significant additional congestion. In the rest of this section we give new algorithms for efficiently computing cutmatches with sparse flows; towards this, we give the first efficient algorithms for near-optimal length-constrained flows with support size $\tilde{O}(m)$. As our algorithms parallelize, we give our results as parallel algorithms in this section.

% \subsubsection{Length-Constrained Cutmatches}
%\enote{TODO for ZT: This needs to be updated to reflect our notation etc}
% Like the results of \cite{haeupler2022expander}, our results will make use of some recent results in length-constrained flows and cutmatches of \cite{haeupler2021fast}.

%\enote{For Zihan to edit / update}

To define our cutmatch algorithm guarantees we will make use of the following notion of batching.
\begin{definition}[$\batch$-Batchable]\label{dfn:batchable}
 Given a graph $G = (V,E)$ with edge lengths $\l$ and vertex subsets $V_1, V_2, \ldots \subseteq V$, we say that $\mcV = \{V_1, V_2, \ldots\}$ is $\batch$-batchable for length $h$ if $\mcV$ can be partitioned into ``batches'' $\mcV_1, \mcV_2, \ldots, \mcV_{\batch}$ so that if $u \in V_i \in \mcV_j$ and $v \in V_{i'} \in \mcV_{j}$ and $i \neq i'$ then $u$ and $v$ are at least $2h$ apart in $G$. We say that pairs of vertex subsets $\{(S_i, T_i)\}_i$ are $\batch$-batchable if $\{S_i \cup T_i\}_i$ is $\batch$-batchable.
\end{definition}

A cutmatch is defined as follows.

% \begin{definition}[Multi-Commodity $h$-Length Cutmatch, \cite{haeupler2021fast}]\label{def:cutmatch}
% Given a graph $G = (V,E)$ with lengths $\l$, an $h$-length $\phi$-sparse cutmatch of congestion $\gamma$ between disjoint node-weighting pairs $\{(A_i,B_i)\}_i$ consists of:
% \begin{itemize}
%     \item An integral $h$-length flow $F = \sum_i F_i$ in $G$ with lengths $\l$ of congestion $\gamma$ where each $F_i$ is an $A_i$-$B_i$ flow satisfying $F_i(e) \leq \U_e$ for each edge $e$ incident to $\supp( A_i \cup B_i)$; 
%     \item A moving cut $C$ in $G$ of size $|C|\leq \phi \cdot \left(\sum_i |A_i| - \val(F_i) \right)$ such that for every $i$ we have $\supp(A_i)$ and $\supp(B_i)$ are at least $h$ apart according to $\l_i$. Here, $\l_i$ on edge $e$ is
%     \begin{align*}
%             (\l_i)_e := 
%             \begin{cases}
%                 h+1 & \text{if $e \in \delta(\supp(A_i)) \cup \delta(\supp(B_i))$ and $F_i(e) = \U_e$}\\
%                 \l_e + h \cdot C(e) & \text{otherwise}
%             \end{cases}
%     \end{align*}
% \end{itemize}
% \end{definition}

% \future{Change to partition node weighting}
\begin{definition}[Multi-Commodity $h$-Length Cutmatch, \cite{haeupler2021fast}]\label{def:cutmatch}
Given a graph $G = (V,E)$ with lengths $\l$, an $h$-length $\phi$-sparse cutmatch of congestion $\gamma$ between disjoint and equal-size node-weighting pairs $\{(A_i,A_i')\}_i$ consists of, for each $i$, a partition of the support of the node-weightings $M_i \sqcup U_i = \supp(A_i)$ and $M_i' \sqcup U_i' = \supp(A_i')$ where $M_i, M_i'$ and $U_i, U_i'$ are the ``matched'' and ``unmatched'' parts respectively and
\begin{itemize}
    \item An integral $h$-length flow $F = \sum_i F_i$ in $G$ with lengths $\l$ of congestion $\gamma$ according to $\U$ where, for each $i$, each $u$ sends at most $A_i(u)$ according to $F_i$ (with equality iff $u \in M_i$) and each $u' \in M_i'$ receives at most $A_i'(u')$ flow according to $F_i$ (with equality iff $u'\in M_i'$);
    \item A moving cut $C$ in $G$ where $U_i$ and $U_i'$ are at least $h$-far according to lengths $\{\l_e + h \cdot C(e)\}_e$ and $C$ has size at most
    \begin{align*}
        |C|\leq \phi \cdot \left( \left(\sum_i |A_i| \right) - \val(F) \right).
    \end{align*}
\end{itemize}
\end{definition}

% \future{TODO: This needs to be weighted - also, unless convinced otherwise, I strongly vote against this interface / definition - sorry!}
% \future{I.e. use different definition of 2 node weightings; a matched part and an unmatched part; matching for the matched part and cut for the unmatched part}
\future{TODO: Update flow paper with this definition of cutmatch or discuss how this is implied}

The below summarizes our new cutmatch algorithms. Previously, \cite{haeupler2021fast} gave the same result but with support size $\tilde{O}(b \cdot \poly(h) \cdot m )$.

\begin{theorem}\label{thm:multiCutmatch}
Suppose we are given a graph $G = (V,E)$ on $m$ edges with lengths $\l$, $h \geq 1$ and $\phi \leq 1$. There is an algorithm that, given node-weighting pairs $\{(A_i,B_i)\}_i$ whose supports are $b$-batchable for length $h$, outputs a multi-commodity $h$-length $\phi$-sparse cutmatch $(F,C)$ of congestion $\gamma$ where $\gamma=\tilde{O}(\frac{1}{\phi})$. Furthermore, $|\supp(F)| \leq \tilde{O}(m + b + \sum_i |\supp(A_i \cup B_i)|)$ and this algorithm has depth $b \cdot \poly(h,\log N)$ and work $\tilde{O} \left(|\supp(F)| \cdot \poly(h) \right)$.
\end{theorem}
% \enote{BH not sure on above; should be sum; maybe ok to pay $n \cdot L$}

Our cutmatch algorithm will be based on new length-constrained flow algorithms whose guarantees are summarized by the below.
\sparseMultiFlows*

% \begin{definition}[$\kappa$-Batchable]
%  Given graph $G$, lengths $\l$, $h \geq 1$ and source, sink set pairs $\{(S_i, T_i)\}_i$, we say that a $\{(S_i, T_i)\}_i$ is $\kappa$-batchable if sets in $\{S_i, T_i\}_i$ can be partitioned into batches $\{\mcS_j, \mcT_j\}_j$ if 
%  \begin{enumerate}
%      \item \textbf{Covering:} For each $i$ there some $j$ such that $S_i \in \mcS_j$ and $T_i \in \mcT_j$;
%      \item \textbf{Well-Separated:} For each $i$ and $i'$, if $v \in S_i \cup T_i$ and $v' \in S_{i'} \cup T_{i'}$ and $S_i, S_{i'} \in \mcS_j$ for some $j$ then $d_{\l}(v,v') > 2h$.
%  \end{enumerate}
% \end{definition}
% \noindent Observe that if the number of commodities is $\kappa$ then the set of source, sink pairs is trivially $\kappa$-batchable.

\subsection{Rounding Flows to Blaming Flows}

In order to achieve sparse flows and cutmatches we introduce the following sense of blaming flows. The utility of the following sense of blaming is that each time an edge is $\gamma$-blamed, a $\gamma$-fraction of its capacity is used up so if we compute a series of $\gamma$-blaming flows the total support size of these flows should be at most (about) $m / \gamma$.

\begin{definition}[Blaming Flow]
Given flow $F$ we say that $F$ is $\gamma$-blaming if for each $P$ in the support of $F$ there is a unique edge $e \in P$ that $P$ ``blames'' such that $F(e) \geq \gamma \cdot \U_e$.
\end{definition}

The following algorithm shows how to convert arbitrary flows into blaming flows
\begin{lemma}\label{lem:makeBlaming}
    Given a feasible (possibly fractional) $h$-length flow $F$ on graph $G = (V,E)$ one can compute a feasible integral flow $\hat{F}$ where $\supp(\hat{F}) \subseteq \supp(F)$ and
    \begin{enumerate}
        \item \textbf{Blaming:} $\hat{F}$ is $\frac{1}{2}$-blaming;
        \item \textbf{Approximate:} $\val(\hat{F}) \geq \Omega \left(\frac{1}{h \cdot \log^2 N}\right) \cdot \val(F)$;
    \end{enumerate}
    in $\tilde{O}(|\supp(\hat{F})|/|E| + \log |\supp(\hat{F})|)$ parallel time with $m$ processors.
\end{lemma}
\begin{proof}
    The basic idea of the algorithm is as follows. By a standard bucketing trick we can assume that every flow path in the support of $F$ has the same value $2^{j^*}$ and that the minimum capacity edge used by every flow path has the same capacity of $2^{i^*}$. We then create an instance of maximal independent set (MIS) whose vertices are our flow paths to select a subset of flow paths $I \subseteq \supp(F)$ so that each capacity $2^{i^*}$ edge used by $F$ has at most one flow path in $I$ going over it; sending $2^{i^*}$ flow along each flow path of $I$ gives our result.

    % \enote{If integral and blameable can get rid of it}

    % We begin by decomposing $F$ into paths which can already blame an edge and those which cannot. Namely, let $F = F_b + F_{\bar{b}}$ where $F_b = $

    More formally, we begin by rounding our flow values and capacities down to powers of $2$. Specifically let $\U'$ be the capacity which gives edge $e$ value
    \begin{align*}
        \U_e' :=  2^{\lfloor\log_2(\U_e) \rfloor}
    \end{align*}
    Likewise, let $F'$ be the flow which gives path $P$ flow value
    \begin{align*}
        F'(P) :=  2^{\lfloor\log_2(F(P)) - 1 \rfloor}
    \end{align*}
    Observe that by the feasibility of $F$ for capacities $\U$ and the extra $-1$ in the exponent in our definition of $F'(P)$ we know that $F'$ is feasible in $G$ for capacities $\U'$. Furthermore, we know by definition of $F'$ that 
    \begin{align}
        \val(F') \geq \val(F)/2. \label{eq:qega}
    \end{align}

    Next, we partition paths of $F'$ by the smallest capacity edge that they use and their flow value. Specifically, for each $i$ and $j \leq i$, let $\mcP_{i,j} := \{P \in \supp(F') : \min_{e \in P} \U_e' = 2^i \text{ and } F'(P) = 2^j\}$ be all paths in the support of $F'$ with minimum capacity edge of capacity $2^i$ and flow value $2^j$. Let $F_{i,j}'$ be the flow which matches $F'$ on paths in $\mcP_{i,j}$, namely $F_{i,j}'$ on path $P$ is defined as
    \begin{align*}
        F_{i,j}'(P) := \begin{cases}
            F'(P) = 2^j & \text{if $P \in \mcP_{i,j}$}\\
            0 & \text{otherwise}.
        \end{cases}
    \end{align*}
    Let $i^*, j^* := \argmax_{i,j} \val(F_{i,j}')$ be the index of the maximum such flow. We let 
    \begin{align*}
        F^* := F_{i^*, j^*}'
    \end{align*}
    for ease of notation. By our assumption of polynomial capacities, the fact that $F' = \sum_{i,j} F_{i,j}'$ and \Cref{eq:qega} we know that 
    \begin{align}\label{eq:ASFas}
        \val(F^*) \geq \Omega \left(\val(F') / \log^2 N \right) \geq \Omega \left(\val(F) / \log^2 N \right).
    \end{align}
    Furthermore, we know that $F^*$ is feasible in $G$ with capacities $\U'$ and therefore feasible in $G$ with capacities $\U$.

    In what remains we will show how to round $F^*$ to another flow $\hat{F}$ which is both integral and blaming to at a negligible loss in cost. Each path in the support of $\hat{F}$ will blame some edge of capacity $2^{i^*}$ according to $\U'$. 
    
    We will construct $\hat{F}$ by solving an appropriate instance of maximal independent set (MIS) based on $F^*$. By construction every path in the support of $F^*$ has flow value $2^{j^*}$ and uses an edge of capacity $2^{i^*}$. Our goal will be to construct an instance of MIS which allows us to select paths so that for each edge of capacity $2^{i^*}$ that is used by some path in $\supp(F^*)$ we select exactly one path. We will the increase said paths flow value to $2^{i^*}$; likewise we will select $2^{i-{i^*}}$ many paths for an edge of capacity $i > {i^*}$ in $G$.
    Specifically, consider the following instance of MIS on graph $H = (V_H, E_H)$.
      \begin{itemize}
        \item \textbf{MIS Vertices:} For each path in $P \in \supp(F^*)$ we have $1$ vertex. In other words, $V_H := \supp(F^*)$.
        \item \textbf{MIS Edges:} For each edge $e \in G$ in our original graph such that $F^*(e) > 0$ we construct edges $E_e$ where $E_H := \sqcup_{e \in E} E_e$ is the union of all such edges. $E_e$ is constructed as follows. Let $\mcP_e = \{P_0, P_1, \ldots\}$ be all paths in $\supp(F^*)$ which include $e$ ordered arbitrarily and for $l \in \left[\lceil|\mcP_e|/2^{{i^*}-j^*} \rceil\right]$ let $\mcP_e^{(l)} := \{P_{(l-1) \cdot 2^{i^*-j^*}}, \ldots P_{(l)\cdot 2^{i^*-j^*}}\}$ be the $l$th set of contiguous $2^{i^*-j^*}$ such paths. For each $l$ and each $P, P' \in \mcP_e^{(l)}$ we include the edge $\{P, P'\}$ in $E_e$; in other words, we add a clique for each $\mcP_e^{(l)}$.
    \end{itemize}
    Let $I \subseteq \supp(F^*)$ be an MIS in $H$. Then, we define $\hat{F}$ as the flow corresponding to $I$ where each flow value is rounded up from $2^{j^*}$ to $2^{i^*}$; that is, $\hat{F}$ on path $P$ is defined as 
    \begin{align*}
        \hat{F}(P) := \begin{cases}
            2^{i^*} & \text{if $P \in I$}\\
            0 & \text{otherwise}.
        \end{cases} 
    \end{align*}

    We now argue that $\hat{F}$ satisfies the required properties. We have $\supp(\hat{F}) \subseteq \supp(F)$ since $\supp(\hat{F}) \subseteq \supp(F^*) \subseteq \supp(F') = \supp(F)$.
    
    $\hat{F}$ must be integral since $i^* \geq 1$ since we have assumed $\U_e$ is integral for every $e$.
    
    We next argue that $\hat{F}$ is feasible for capacities $\U$. Even stronger, we observe that $\hat{F}$ is feasible for $\U'$. In particular, applying the fact that $I$ includes at most one element from each $\mcP_e^{(l)}$ and $l \leq |\mcP_e|/2^{i^*-j^*}$ we have the total flow that $\hat{F}$ sends over edge $e$ is
    \begin{align*}
        \hat{F}(e) &= \sum_{l} \sum_{P \in I \cap \mcP_e^{(l)}} 2^{i^*}\\
        &\leq \sum_{l}  2^{i^*}\\
        & \leq |\mcP_e| \cdot 2^{j^*}\\
        & = F^*(e)\\
        & \leq F'(e)\\
        &\leq \U_e'.
    \end{align*}
    Thus, $\hat{F}$ is feasible for $\U'$ and since $\U_e' \leq \U_e$ for every edge $e$, $\hat{F}$ is also feasible for $\U$.

    Next, we claim that $\hat{F}$ is $\frac{1}{2}$-blaming. By definition of $F^*$ we know that, for each $P \in \supp (F^*)$, there is an edge $e \in P$ such that $\U_e' = 2^{i^*}$. As $\supp(\hat{F}) \subseteq \supp(F^*)$ and $\hat{F}$ sends $2^{i^*}$ over each path in its support, it follows that for each $P \in \supp(\hat{F})$ there is a some $e \in P$ such that $\U'_e = 2^{i^*}$. Since $\hat{F}(P) = 2^{i^*}$ and $\hat{F}$ is feasible for $\U'$, it follows that this edge for each $P \in \supp(\hat{F})$ is unique and $\hat{F}(e) = 2^{i^*}$. Thus, we have $P \in \supp(\hat{F})$ blame this unique edge $e$; since $\U'_e \leq \frac{1}{2}\cdot \U_e$, it follows that $\hat{F}$ is $\frac{1}{2}$-blaming.

    Lastly, we argue that $\hat{F}$ has approximately the same value as $F$. Since every path in $F^*$ consists of at most $h$-many edges, observe that the maximum degree in $H$ is at most $h \cdot 2^{i^*-j^*}$ so $|I| \geq \frac{1}{h \cdot 2^{i^*-j^*}} \cdot |\supp(F^*)|$. Combining this with the fact that $|\supp(F^*)| = \val(F^*) / 2^{j^*}$ we have
    \begin{align}
        |I| &\geq \frac{1}{h \cdot 2^{i^*-j^*}} \cdot |\supp(F^*)|\nonumber\\
        &= \frac{1}{h \cdot 2^{i^*}} \cdot \val(F^*). \label {eq:asfgas}
    \end{align}
    Applying \Cref{eq:asfgas} and the definition of $\hat{F}$ we get
    \begin{align}
        \val(\hat{F}) &= 2^{i^*} \cdot |I|\nonumber\\
        & = \frac{1}{h} \cdot \val(F^*)\label{eq:sgadda}
    \end{align}
    Combining \Cref{eq:sgadda} and \Cref{eq:ASFas} we get
    \begin{align*}
        \val(\hat{F}) \geq \Omega \left(\frac{1}{h \cdot \log^2 N} \right) \cdot \val(F)
    \end{align*}
    as required.
    
    It remains to argue the runtime of our algorithm. Computing $\hat{F^*}$ in the stated time is trivial to do by inspecting each path in the support of $F$ in parallel. Likewise, computing $\hat{F}$ from $I$ is trivial to do in the stated time. The only non-trivial step is to construct $H$ and compute $I$. Constructing $H$ can be done in the stated time as it consists of $|\supp(F^*)| \leq |\supp(F)|$-many vertices and each $\mcP_e^{(l)}$ and its corresponding clique can be computed in parallel. Computing $I$ can then be done by any number of a standard number of parallel MIS algorithms running in deterministic parallel time $O(\log |V_H|) = O(\log |\supp(F)|)$ rounds; see e.g.\ \cite{luby1985simple}.
\end{proof}

\subsection{Blaming Flow Sequences}
Our algorithm will ultimately compute a sequence of blaming flows, defined as follows.

\begin{definition}[Blaming Flow Sequence]
Given flow $F$ we say that $F$ is decomposable into a $\gamma$-blaming flow sequence $F_1, F_2, \ldots$ if $F$ can be expressed as $F = F_1 + F_2 + \ldots$ where $F_i$ is $\gamma$-blaming in $G$ with capacities $\U^{(i)} := \{\U_e - \sum_{j < i} F_j(e)\}_e$.
\end{definition}
Given a flow $F$ that can be decomposed into a blaming flow sequence we will refer to the number of times $F$ blames an edge $e$ by which we mean the number of flows in $F_1, F_2, \ldots$ that have in their support a path that blames $e$.

The following will imply the sparsity of a sequence of blaming flows.
\begin{lemma}\label{lem:blamingFlowSeq}
    Let $F$ be decomposable into a $\gamma$-blaming flow sequence $F_1, F_2, \ldots$ of integral flows on graph $G = (V,E)$ with capacities $\U$. Then each edge is blamed at most $O(\frac{\log N}{\gamma})$ times by $F$.
\end{lemma}
\begin{proof}
    Each time an edge is blamed its capacity is reduced by a $1-\gamma$ multiplicative factor and by our assumption of polynomial-size capacities such a reduction can happen at most $\frac{1}{\gamma}\cdot \log N$ times.
    
    More formally, fix an edge $e$ suppose (without loss of generality) that $F_1', F_2', \ldots, F_k'$ is a subsequence of $F_1, F_2, \ldots$ where each $F_i'$ has in its support a path that blames $e$. We have that the capacity of $e$ in the graph in which $F_{i+1}'$ contains a path blaming $e$ is at most
    \begin{align*}
        \U_e \cdot (1-\gamma)^{i}
    \end{align*}
    By our assumption that $\U_e \leq N$ we then have that $i \leq O(\log N)$ as required.
\end{proof}

\subsection{Blaming Near-Lightest Path Blockers}
The sequence of blaming flows computed by our algorithm will be a so-called ``near-lighted path blocker'' as previously introduced by \cite{haeupler2021fast}. Towards defining these, it will be useful to treat a moving cut $C$ as assigning ``weights'' to edges of our input graph. Given a moving cut $C$ and path $P$ we let 
\begin{align*}
    C(P): = \sum_{e \in P}C(e)
\end{align*} 
be the total weight of the path and let 
\begin{align*}
d_C^{(h)}(u,v) := \min_{u-v \text{ path }P : 
\l(P) \leq h} C(P)    
\end{align*}
give the minimum weight of a length at most $h$ path connecting $u$ and $v$. For vertex sets $W, W' \subseteq V$ we define $d_C^{(h)}(W,W') := \min_{w \in W} \min_{w' \in W'} d_C^{(h)}(w,w')$ analogously. Then, we have our definition of near-lightest path blockers below.

\begin{definition}[$h$-length $(1+\epsilon)$-Lightest Path Blockers, \cite{haeupler2021fast} Definition 11.1]\label{dfn:alphaPathBlocker}
Let $G = (V,E)$ be a graph with lengths $\l$, weights $C$ and capacities $\U$. Fix $\epsilon >0$, $h \geq 1$, $\lambda \leq d^{(h)}_{C}(S,T)$ and $S, T \subseteq V$. Let $F$ be an $h$-length integral $S$-$T$ flow. $F$ is an $h$-length $(1+\epsilon)$-lightest path blocker if:
\begin{enumerate}
    \item \textbf{Near-Lightest:} $P \in \supp(F)$ has weight at most $(1 + 2\epsilon) \cdot \lambda$;
    \item \textbf{Near-Lightest Path Blocking:} If $S$-$T$ path $P'$ has length at most $h$ and weight at most $(1+\epsilon) \cdot \lambda$ then there is some $e \in P'$ where $F(e) = \U_e$.
\end{enumerate}
\end{definition}

Previous work showed how to compute near-lightest path blockers.
\begin{restatable}{thm}{pathBlockerAlgOld}[\cite{haeupler2021fast} Theorem 11.1]
\label{thm:pathBlockerAlgOld}
One can compute $h$-length $(1+\epsilon)$-lightest path blocker $F$ in deterministic parallel time $\tilde{O}\left(\poly(\frac{1}{\eps}, h)\right)$ with $m$ processors where $|\supp(F)| \leq \tilde{O}\left(\poly(\frac{1}{\eps}, h) \cdot |E|\right)$.
\end{restatable}

By repeatedly making near-lightest path blockers blaming, we can compute a near-lightest path blockers which is also blaming.
\begin{restatable}{thm}{pathBlockerAlgNew}
	\label{thm:pathBlockerAlgNew}
	One can compute an $h$-length $(1+\epsilon)$-lightest path blocker $F$ in deterministic parallel time $\tilde{O}\left(\poly(\frac{1}{\eps}, h)\right)$ with $m$ processors which is decomposable into a $\frac{1}{2}$-blaming flow sequence.
\end{restatable}
\begin{proof}
    Our algorithm simply repeatedly takes lightest path blockers, rounds them to be blaming, reduces capacities and iterates.

    More formally, we do the following. We initialize our output $(1+\eps)$-lightest path blocker $F$ to be the empty flow. Then, we compute a $(1+\eps)$-lightest path blocker $F'$ using \Cref{thm:pathBlockerAlg}. We then apply \Cref{lem:makeBlaming} to round this flow to flow $\hat{F}$ which is $\frac{1}{2}$-blaming. We update $F$ to $F+\hat{F}$ and decrement the capacity of each edge $e$ by $\hat{F}(e)$. We iterate this until $F$ is near-lightest path blocking.

    We first claim that the above algorithm must only iterate $\tilde{O}(\poly(h))$ times. This is proven in Theorem 11.1 of \cite{haeupler2021fast}. \future{Unpack this if I get time} Our runtime and the fact that $\hat{F}$ is a $(1+\eps)$-lightest path blocker and decomposable into a $\frac{1}{2}$-blaming flow sequence is immediate by construction and the guarantees of \Cref{thm:pathBlockerAlg} and \Cref{lem:makeBlaming}.
\end{proof}

\subsection{Sparse Flows and Cutmatches via Blaming Near-Lightest Path Blockers}
We now use our near-lightest path blockers to compute sparse flows and cutmatches. Specifically, we adopt \Cref{alg:mwMulti} which was shown by \cite{haeupler2021fast} to compute a near-optimal flow.

\begin{restatable}{thm}{pathBlockerAlg}[\cite{haeupler2021fast} Theorem 11.1]
\label{thm:pathBlockerAlg}
\Cref{alg:mwMulti} returns a feasible  $h$-length $\{(S_i, T_i)\}_i$ flow, moving cut pair $(F, C)$ that is $(1 \pm \epsilon)$-approximate in deterministic parallel time $\tilde{O}\left(\poly(\frac{1}{\eps}, h\right)$ with $m$ processors. Also, $F = \eta \cdot \sum_{j=1}^k F_j$ where $\eta = \tilde{\Theta}(\epsilon^2)$, $k = \tilde{O}\left(\frac{h}{\epsilon^4} \right)$ and each $F_j$ is an integral $h$-length $S_i$-$T_i$ flow for some $i$.
\end{restatable}
We observe that if we use blaming flows for our near-lightest path blockers in  \Cref{alg:mwMulti} then the resulting flow is sparse.

\begin{lemma}\label{lem:blamingFlows}
    If each lightest path blocker in \Cref{alg:mwMulti} is decomposable into a  $\gamma$-blaming flow sequence then the flow $F$ returned by \Cref{alg:mwMulti} satisfies $\supp(F) \leq \tilde{O}\left(\frac{1}{\gamma} \cdot |E|\right)$.
\end{lemma}
\begin{proof}
    Consider one edge $e$. Next, consider one shortest path blocker $\hat{F}$ computed by \Cref{alg:mwMulti}. By \Cref{lem:blamingFlowSeq} if $\hat{F}$ blames $e$ at least once then it blames it at most $O\left(\frac{\log N}{\gamma}\right)$ times. Furthermore, if $\hat{F}$ $\gamma$-blames $e$ at least once then $e$ has its cut value increases by a $(1+\eps_0)^{\gamma}$ multiplicative factor. Since each edge has its cut value initialized to $\frac{1}{m^{O(1)}}$ it follows that the total number of computed shortest path blockers that blame $e$ is at most $\tilde{O}(1)$ and since each such shortest path blocker blames $e$ at most $O\left(\frac{\log N}{\gamma}\right)$ times it follows that the total number of paths in the support of all shortest path blockers computed by \Cref{alg:mwMulti} and therefore in the support of $F$ is $\supp(F) \leq \tilde{O}\left(\frac{1}{\gamma}\cdot |E| \right)$.
\end{proof}

\begin{algorithm}[ht]
    \caption{Multi-Commodity Length-Constrained Flows  and Moving Cuts}
    \label{alg:mwMulti}
    \begin{algorithmic}[0] % The number tells where the line numbering should start
        % \Procedure{Euclid}{$a,b$} \Comment{The g.c.d. of a and b}
            \State \textbf{Input:} graph $G = (V,E)$ with lengths $\l$, capacities $\U$, length constraint $h$ and $\kappa$-batchable source, sink pairs $\{(S_i, T_i)\}_i$ where $S_i, T_i \subseteq V$ for every $i$ and an $\eps \in (0,1)$.
            \State \textbf{Output:} $(1 \pm \epsilon)$-approximate $h$-length multi-commodity flow $F$ and moving cut $C$.
            \State Let $\epsilon_0 = \frac{\epsilon}{6}$, let $\zeta = \frac{1+2 \eps_0}{\eps_0} + 1$ and let $\eta = \frac{\eps_0}{(1 + \eps_0) \cdot \zeta} \cdot \frac{1}{\log m}$.
            \State Initialize $C(e) \gets \left(\frac{1}{m}\right)^{\zeta}$ for all $e \in E$.
            \State Initialize $\lambda \gets  \left(\frac{1}{m}\right)^{\zeta}$.
            \State Initialize $F(P) \gets 0$ for every path $P$.
            \While{$\lambda < 1$}:

                \For{$j \in [\kappa]$ and each batch $(\mcS_j, \mcT_j)$}
                    \For{each $(S_i, T_i)$ with $S_i \in \mcS_j$ and $T_i \in \mcT_j$ in parallel}
                        \For{$\Theta\left(\frac{h \log_{1+\epsilon_0} n}{ \epsilon_0} \right)$ repetitions}
                                \State Compute any $S_i$-$T_i$ $h$-length  $(1+\epsilon_0)$-lightest path blocker $\hat{F}$.
                                \State \textbf{Length-Constrained Flow (Primal) Update:} $F \gets F + \eta \cdot \hat{F}$.
                                \State \textbf{Moving Cut (Dual) Update:} $C(e) \gets (1+\epsilon_0)^{\hat{F}(e)/ \U_e} \cdot C(e)$ for every $e \in E$.
                        \EndFor
                    \EndFor
                \EndFor
                \State $\lambda \gets (1+\epsilon_0) \cdot \lambda$
            \EndWhile
            \State \Return $(F,C)$.
    \end{algorithmic}
\end{algorithm}

We conclude with our sparse flow and cutmatch algorithms.
\sparseMultiFlows*
\begin{proof}
    The proof is immediate from combining \Cref{thm:pathBlockerAlgNew}, \Cref{lem:blamingFlows} and \Cref{thm:pathBlockerAlg}.
\end{proof}

% \subsubsection{Sparse Cutmatches Using Sparse Flows}
\cite{haeupler2021fast} showed how to compute cutmatches using flow algorithms. Combining our sparse flows (\Cref{thm:sparseMultiFlows}) with the cutmatch algorithms of \cite{haeupler2021fast} (which just call batchable multicommodity length-constrained flow, cut algorithms as a blackbox) immediately gives our sparse cutmatches as described in \Cref{thm:multiCutmatch}.
\future{Describe this more / the extra loss here is from the super source trick}

\section{Algorithm: Demand-Size-Large Sparse Cuts from EDs}\label{sec:LCsFromStrats} In the previous section we developed the theory of length-constrained expander decompositions. We now put this theory to use by giving new algorithms for length-constrained expander decompositions. Our algorithms will make use of a well-studied ``spiral'' paradigm from the classic setting where we compute a length-constrained expander decomposition by repeatedly computing large sparse cuts \cite{saranurak2019expander}. In particular, we will show that one can compute expander decompositions from large length-constrained sparse cuts (\Cref{lem: from MSC to ED}) which one can compute from expander decompositions (\Cref{lem:cutsFromEDs}) and so on. In order to prevent this argument from becoming circular we argue that it ``spirals'' in that the expander decompositions we must compute get smaller and smaller each time we go around the circle of dependencies.

In this section we show how to compute large length-constrained sparse cuts using length-constrained expander decompositions.

The following is our notion of size which is analogous to the volume of a cut in the non-length-constrained setting.
\begin{definition}[$(h,s)$-Separated Demand-Size]\label{def:demandSize}
    Given length-constrained cut $C$ and node-weighting $A$, we define the $(h,s)$-length demand-size of $C$ with respect to $A$ as the size of the largest $A$-respecting $h$-length demand which is $(hs)$-separated by $C$. We denote this ``demand-size'' by $A_{(h,s)}(C)$.
\end{definition}

\begin{definition}[Demand-Size Largest Sparse Cut, \qLDSC]\label{def:LDSC}
    We call the $(h,s)$-length $\phi$-sparse cut $C$ the demand-size largest $(h,s)$-length $\phi$-sparse cut for node-weighting $A$ if it its demand-size is maximum among all $(h,s)$-length $\phi$-sparse cuts. We notate the size of this cut as \begin{align*}
        \qLDSC(\phi, h, s) := A_{(h,s)}(C)
    \end{align*}
\end{definition}

\future{change alpha phi to kappa}
\begin{definition}[Approximately Demand-Size-Largest Sparse Length-Constrained Cut]\label{dfn:apxLargeCut}
    Length-constrained cut $C$ is an $\alpha$-approximate demand-size-largest $(\leq h,s)$-length $\phi$-sparse cut for node-weighting $A$ with length approximation $\alpha_s$ and sparsity approximation $\alpha_{\phi}$ if it is an $(h'', s)$-length $\phi$-sparse cut for some $h'' \leq \alpha_s \cdot h $ and for all $h' \leq h/\alpha_s$ we have
    \begin{align*}
        A_{(h'', s)}(C) \geq \frac{1}{\alpha} \cdot \qLDSC(\phi/\alpha_\phi, h', s \cdot \alpha_s).
    \end{align*}
    % there does not exist a moving cut $C'$ of demand-size 
    % \begin{align*}
    %     A_{(h,s)}(C') \geq \alpha_C \cdot A_{(h,s)}(C)
    % \end{align*}
    % whose $(h'', \alpha_s \cdot s)$-length sparsity w.r.t.\ $A$ is at most $\phi/\alpha_\phi$ for some $h'' \leq h$.
\end{definition}
\future{Mention tricriteria tuple notation}

For the below result, recall the definition of a length-constrained expansion witness (\Cref{def:LCExpWitness}).
\begin{restatable}{lem}{largeCutsFromEDs}
\label{lem:cutsFromEDs}
For any parameter $\eps>0$, there exists an algorithm that, given a graph $G$ on $n$ vertices and $m$ edges, a node weighting $A$, a length bound and slack $h$ and $s$, a recursion size parameter $L$, and a conductance parameter $\phi$, computes a $\alpha$-approximate demand-size-largest $(\leq h, s)$-length $\phi$-sparse cut with sparsity approximation $\alpha_\phi$ and length slack approximation $\alpha_s$ with respect to $A$ where
\begin{align*}
    \alpha = \frac{\tilde{O}(N^{O(\eps)})}{\eps^3} \qquad \qquad \alpha_\phi = \frac{s^3 N^{O(1/s)}}{\eps} \qquad \qquad \alpha_s = \max\left(2,\frac{1}{\eps^2} \cdot (s)^{1+O(1/\eps)}\right)
\end{align*}
with work
\begin{align*}
    \wsparseCut(A,m)\leq m \cdot \tilde{O}\left( \frac{1}{\eps} \cdot (s)^{O(1/\eps)} \cdot N^{O(\eps)}   + \frac{1}{\eps} \cdot L \cdot N^{O(\eps)} \cdot \poly(h) \right) +\frac{\tilde{O}(1)}{\poly(\eps)}  \sum_i  \wED(A_i,m_i)
    %\frac{1}{\eps} \cdot \sum_i \tcutStrat(A_i, m_i) + m \cdot \tilde{O}\left((1/\eps)^{O(1/\eps')} \cdot N^{O(\eps + \eps')} +  \frac{L^3 }{\eps} \cdot N^{O(\eps)}\cdot \poly(h) \right) 
\end{align*}
and depth 
\begin{align*}
    \dsparseCut(A,m)\leq \tilde{O}\left( \frac{1}{\eps} \cdot (s)^{O(1/\eps)} \cdot N^{O(\eps)}   + \frac{1}{\eps} \cdot L \cdot N^{O(\eps)} \cdot \poly(h) \right) +\frac{\tilde{O}(1)}{\poly(\eps)}  \max_i    \dED(A_i,m_i)
    %\frac{1}{\eps} \cdot \sum_i \tcutStrat(A_i, m_i) + m \cdot \tilde{O}\left((1/\eps)^{O(1/\eps')} \cdot N^{O(\eps + \eps')} +  \frac{L^3 }{\eps} \cdot N^{O(\eps)}\cdot \poly(h) \right) 
\end{align*}
where $\wED(A_i, m_i)$ and $\dED(A_i, m_i)$ are the work and depth to compute an $(h,2^{1/\eps})$-length $\phi$-expander decomposition with cut slack $N^{\poly(\eps)}$ for node-weighting $A_i$ in an $m_i$-edge graph and for all $i$ each $A_i \preceq A$ and $|A_i| \leq \frac{|A|}{L}$ and $\{m_i\}_i$ are non-negative integers satisfying
\begin{align*}
    \sum_i m_i \leq \frac{1}{\eps} \cdot \tilde{O}(m + n^{1 + O(\eps)} + L^2 \cdot N^{O(\eps)}).
\end{align*}
Furthermore, if the graph is $( \leq h \cdot \frac{1}{\eps} \cdot s^{O(1/\eps)}, s)$-length $\phi$-expanding then the algorithm also returns a $(\leq h, s_w)$-length $\phi_w$-expansion witness where $s_w = \frac{1}{\eps^2} \cdot s^{O(1/\eps)}$ and $\phi_w = \tilde{O}(\phi \eps/N^{O(\eps)})$.
\end{restatable}

% $B_{h'}$ is $(h', \frac{1}{\eps^2} \cdot (s)^{1+O(1/\eps)})$-length $\tilde{O}(\phi \cdot \eps)$-expanding

Observe that applying our previous relations we can get a simple lower bound on the demand-size of the largest-demand-size length-constrained sparse cut.
\begin{lemma} \label{lem:LDSCAtMostLEC} Given graph $G$ and node-weighting $A$ and parameters $h,s,\phi$, we have that 
    \begin{align*}
        \qLDSC(\phi, h, s) \leq \qLEC(\phi', h', s)
    \end{align*}
    where $\phi' = \tilde{O}(\phi \cdot s^3 \cdot N^{O(1/s)})$ and $h' = 2h$.
\end{lemma}
\begin{proof}
\future{Put this stuff in the main quantities section}
    See \Cref{dfn:equantities} for a definition of the relevant graph quantities below. Let $C$ be the $(h,s)$-length $\phi$-sparse cut of largest demand-size. Observe that $\qLDSC(\phi, h, s) \leq \frac{1}{\phi}\qLWSC(\phi, h, s)$ since we can use $C$ in the largest weighted sparse cut sequence and so $\phi \cdot A_{(h,s)}(C)$ is a lower bound on $\qLWSC(\phi, h, s)$. Continuing, by \Cref{lem:LWSCAtMostLC} we have 
    \begin{align*}
        \qLWSC(\phi, h, s) \leq \qLC(\phi'', h', s')
    \end{align*}
    where $\phi'' = \tilde{O}(\phi \cdot s^3 \cdot N^{O(1/s)})$, $h' = 2h$ and $s' = (s-2)/2$. Lastly, by \Cref{thm:lcAtMostLEC} we have 
    \begin{align*}
        \qLC(\phi'', h', s') \leq \qLEC(3\phi'', h', s').
    \end{align*}
    Combining the above and observing that $\qLEC(3\phi'', h', s') \leq \qLEC(3\phi'', h', s)$ gives the lemma.
\end{proof}
 We will use the above lower bound to argue that our algorithm returns an approximately demand-size largest sparse cut as described below.

\subsection{(Preliminary) Algorithm: Cut Strategies from Expander Decompositions}\label{sec:stratsFromEDs}

We describe the cut matching games and describing prior work on computing expander decompositions from cut strategies. The cut-matching game was first proposed and studied in \cite{khandekar2007cut,khandekar2009graph}, and later on it has found a wide range of applications in graph algorithms. We use a slightly generalized version of it as follows.

\paragraph{Cut Strategies.} A cut strategy is an algorithm which given a graph $G$ and node-weighting $A$ produces a set of node-weightings $\{(A^{(j)}, B^{(j)})\}_j$ where for each $j$ we have $A^{(j)} + B^{(j)} \preccurlyeq A$ and $|A^{(j)}| = |B^{(j)}|$ and $A \preceq  \sum_{j} A^{(j)} + B^{(j)}$. 

\paragraph{Matching Strategies.} A matching strategy is an algorithm which given a graph $G$ and the node-weighting pairs $\{(A^{(j)}, B^{(j)})\}_j$ produced by a cut player outputs a set of edges $M^{(j)} \subseteq \supp(A^{(j)}) \times \supp(B^{(j)})$ for each $j$ between the vertices in the support of $A^{(j)}$ and $B^{(j)}$ and capacities $\U$ subject to the constraint that for every vertex $u$ we have $\U(\delta_{M^{(j)}}(u)) \leq A^{(j)}(u), B^{(j)}(u)$.

\paragraph{Cut Matching Games.} A cut matching game is a procedure for using cut and matching strategies to produce good routers by a sequence of interactions between cut and matching strategies. Namely, given a set of vertices $V$ and a node-weighting $A$ on $V$, it produces a series of graphs $G_0,\dots,G_r$ where $G_0 = (V, \emptyset)$ is the empty graph and we call $G_r$ the output of the cut matching game. The graph $G_{i}$ is $G_{i-1}$ plus the output of the matching player when given the output of the cut player when given $G_{i-1}$. That is, if $\{(A^{(j)}_{i}, B^{(j)}_{i})\}_j$ is the output of the cut player when given $G_{i-1} = (V, E_{i-1})$ and $\{M^{(j)}_i\}_j$ is the output of the matching player when given $\{(A^{(j)}_{i}, B^{(j)}_{i})\}_j$, then $G_{i} = (V, E_{i-1} \cup \bigcup_j M_{i}^{(j)})$. 

We will be interested in the following parameters of a cut matching game.
\begin{itemize}
    \item \textbf{Rounds of Interaction:} We call $r$ the number of rounds of interaction.
    \item \textbf{Cut Batch Size:} We call the maximum number of pairs the cut strategy plays in each round of interaction $\max_i |\{(A_{i}^{(j)}, B_i^{(j)})\}_j|$ the cut batch size of the cut matching game. In typical cut matching games \cite{khandekar2007cut,khandekar2009graph} the cut batch size is $1$; we will be interested in potentially larger batch sizes.
    \item \textbf{Matching Perfectness:} If each set of edges the matching player plays for a batch always has total capacity at least a $1-\alpha$ fraction of the total node-weighting then we say that the cut matching game is $(1-\alpha)$-perfect. That is, a cut matching game is $(1-\alpha)$-perfect if for every $i$ we have
    \begin{align*}
        \sum_j \U\left(M_i^{(j)} \right) \geq (1-\alpha) \cdot \sum_j|A_i^{(j)}| = (1-\alpha) \cdot \sum_j|B_i^{(j)}|.
    \end{align*}
    % \item \textbf{Router Quality:} We say that a cut player has $t$-step $\eta$-congestion router quality if there exists a $A' \preceq A$ for which $G_r$ and $|A'| \geq \beta \cdot |A|$.
\end{itemize}

\noindent % and we are interested in the (non-length-constrained) expansion of $G_r$, rather than its quality as a router.

We use the following result from \cite{ghaffari2022cut} which shows both the existence of high quality cut matching games and how to compute them assuming we can compute length-constrained expander decompositions.

\begin{theorem}[\cite{ghaffari2022cut}]\label{thm:cutStrat}
\label{final cut matching game} For every $\eps > 0$ there is a cut strategy with cut batch size $N^{O(\eps)}$ which when used in a cut matching game with $1/\eps$ rounds of interaction against any $(1-\alpha)$-perfect matching strategy results in a $G_r$ that is a $1/\eps$-step and $N^{O(\eps)}$-congestion router for some $A' \preceq A$ of size $|A'| \geq (1-O(\frac{\alpha}{\eps})) \cdot |A|$. 

This cut-strategy on a node-weighting $A$ in a graph with $m$ edges can be computed in work 
\begin{align*}
    \wcutStrat(A, m) \leq \frac{\tilde{O}(1)}{\poly(\eps)} \cdot \wED(A,m)
\end{align*}
and depth
\begin{align*}
    \dcutStrat(A, m) \leq \frac{\tilde{O}(1)}{\poly(\eps)} \cdot \dED(A,m)
\end{align*}
where $\wED(A,m)$ and $\dED$ are the work and depth respectively for computing an $(h,2^{1/\eps})$-length $\phi$-expander decomposition with cut slack $N^{\poly(\eps)}$ for node-weighting $A$ in an $m$-edge graph. Likewise, $A'$ can be computed in the same work and depth and is vertex induced: i.e.\ for each vertex $u$ if $A'(u) \neq 0$ then $A'(u) = A(u)$.
\end{theorem}

% \begin{theorem}[\cite{ghaffari2022cut}]
% \label{final cut matching game}
% Given graph $G = (V, E)$, node-weighting $A$ and $t$, there exists an $(r,t,\eta,\Delta)$-strategy and $r,\eta,\Delta \leq n^{\poly(1/t)}$. Moreover, such a cut-strategy can be computed in time 
% \begin{align*}
%     \tcutStrat(A, m) \leq \tilde O(\tED(A,m))/\poly(\eps)
% \end{align*}
% where $\tED(A,m)$ is the time needed for computing an $(h,2^{1/\eps})$-length $(\phi,n^{\poly(\eps)})$ expander decomposition for $A$ where $G$ has $m$ edges.
% \end{theorem}

\subsection{Algorithm: Demand-Size-Large Sparse Cuts from Cut Strategies}

The following is our main result for this section and shows how to compute large sparse length-constrained cuts using cut strategies. Below, we let $\wcutStrat$ and $\dcutStrat$ give the work and depth to compute the cut strategy given by \Cref{final cut matching game}.

% From \Cref{final cut matching game}, there exists $(r,t,\eta,\Delta)$ cut strategy $\alg$ with $t=1/\eps$, $b\le 1/\eps$, and $r,\eta,\Delta\le N^{\poly(\eps)}$ for node-weighting $A$. 

% on a node weighting $A$ in a graph with $b$-many edges.

 % \enote{It follows that $C_{h'}$ is a $\frac{\eps^3}{\tilde{O}(N^{O(\eps)})}$-approximate $(\leq h,s)$-length $\phi$-sparse cut with length slack approximation $\alpha_s = \max\left(2,(1/\eps)^{1+O(1/\eps')}\right)$ and sparsity approximation $\alpha_\phi = \frac{\eps}{s^3 N^{O(1/s)}}$ }

\future{Add intuition about recursion decreasing by $L$ on size of volume and total size of recursion of size basically $m$}
\begin{restatable}{theorem}{largeCutsFromCutStrats}
% \begin{lemma}
\label{lem: MBSC to CMG}
For any parameter $\eps>0$, there exists an algorithm that, given a graph $G$ on $n$ vertices and $m$ edges, a node weighting $A$, a length bound and slack $h$ and $s$, a recursion size parameter $L$, and a conductance parameter $\phi$, computes a $\alpha$-approximate demand-size-largest $(\leq h, s)$-length $\phi$-sparse cut with sparsity approximation $\alpha_\phi$ and length slack approximation $\alpha_s$ with respect to $A$ where
\begin{align*}
    \alpha = \frac{\tilde{O}(N^{O(\eps)})}{\eps^3} \qquad \qquad \alpha_\phi = \frac{s^3 N^{O(1/s)}}{\eps} \qquad \qquad \alpha_s = \max\left(2,\frac{1}{\eps^2} \cdot (s)^{1+O(1/\eps)}\right)
\end{align*}
with work
\begin{align*}
    \wsparseCut(A,m)\leq m \cdot \tilde{O}\left( \frac{1}{\eps} \cdot (s)^{O(1/\eps)} \cdot N^{O(\eps)}   + \frac{1}{\eps} \cdot L \cdot N^{O(\eps)} \cdot \poly(h) \right) +\frac{1}{\eps}  \sum_i  \wcutStrat(A_i, m_i)
    %\frac{1}{\eps} \cdot \sum_i \tcutStrat(A_i, m_i) + m \cdot \tilde{O}\left((1/\eps)^{O(1/\eps')} \cdot N^{O(\eps + \eps')} +  \frac{L^3 }{\eps} \cdot N^{O(\eps)}\cdot \poly(h) \right) 
\end{align*}
and depth 
\begin{align*}
    \dsparseCut(A,m)\leq \tilde{O}\left(\frac{1}{\eps} \cdot (s)^{O(1/\eps)} \cdot N^{O(\eps)}  + \frac{1}{\eps} \cdot L \cdot N^{O(\eps)} \cdot \poly(h) \right) +\frac{1}{\eps}  \max_i  \dcutStrat(A_i, m_i)
    %\frac{1}{\eps} \cdot \sum_i \tcutStrat(A_i, m_i) + m \cdot \tilde{O}\left((1/\eps)^{O(1/\eps')} \cdot N^{O(\eps + \eps')} +  \frac{L^3 }{\eps} \cdot N^{O(\eps)}\cdot \poly(h) \right) 
\end{align*}
where for all $i$ each $A_i \preceq A$ and $|A_i| \leq \frac{|A|}{L}$ and $\{m_i\}_i$ are non-negative integers satisfying
\begin{align*}
    \sum_i m_i \leq \frac{1}{\eps} \cdot \tilde{O}(m + n^{1 + O(\eps)} + L^2 \cdot N^{O(\eps)}).
\end{align*}
and $\wcutStrat$ and $\dcutStrat$ are the work and depth to compute the cut strategy described in \Cref{thm:cutStrat} on an $m_i$-edge graph for node-weighting $A_i$. Furthermore, if the graph is $( \leq h \cdot \frac{1}{\eps} \cdot s^{O(1/\eps)}, s)$-length $\phi$-expanding then the algorithm also returns a $(\leq h, s_w)$-length $\phi_w$-expansion witness where $s_w = \frac{1}{\eps^2} \cdot s^{O(1/\eps)}$ and $\phi_w = \tilde{O}(\phi \eps/N^{O(\eps)})$.
% \end{lemma}
\end{restatable}

Having shown in the previous section how to compute sparse cutmatches, we now use these cutmatches to compute large length-constrained cuts using the cut strategies from cut matching games. That is, we prove \Cref{lem: MBSC to CMG}. 

We begin by describing the algorithm for \Cref{lem: MBSC to CMG}.

\paragraph{Step 1: Create Clusters for Cut Matching Games.}\label{step:1} We do the following for each $h' \leq h$ which is a power of $2$.  First, apply \Cref{thm:cover-separation-factor-existential} 
to $G$ to compute a neighborhood cover $\nset_{h'}$ with covering radius $h_{cov}=h'$, separation factor $2s$, cluster diameter $h_{\diam} = \frac{1}{\eps} \cdot (s)^{O(1/\eps)} \cdot h'$ and width $\omega = N^{O(\eps)} \log N$. %We let $\nset$ be all clusters across all $\nset_{h'}$. %\enote{Need multiple covering radii; powers of $2$ up to $h$} 

Modify $\nset_{h'}$ as follows. Since we would like to only run our cut strategy on clusters whose node-weightings are a small fraction of the total size of $A$, we must further break up the node-weighting in each cluster in our neighborhood cover. 
Specifically, for each cluster $S\in \nset$ we let $A_S$ be the restriction of node-weighting $A$ on $S$ (i.e.\ $A_S(u)$ is $A(u)$ if $u \in S$ and $0$ otherwise). Then if $|A_S| \leq |A|/L$ we do nothing. However, if $|A_S| > |A|/L$ then we break $A_S$ into sub-node-weightings $A_S^{(1)}, A_S^{(2)}, \ldots, A_S^{(L)}$ so that $\sum_i A_S^{(i)} = A$ where each of these has equal size and size at most $|A|/L$. We remove $S$ from $\nset_{h'}$ and add a copy of $S$ for each of  $A_S^{(1)}, A_S^{(2)}, \ldots$ to $\mcN_{h'}$. If we don't break up $A_S$ then we say that $A_S$ corresponds to cluster $S$; if we do then we say that each of $A_S^{(1)}, A_S^{(2)}, \ldots$ correspond to each respective copy of $S$. For ease of notation, if we do not break up $A_S$ then we let $A_S^{(1)}:= A_S$. %We will let $\mcA$ be all node-weightings across all clusters and let $\mcA_S$ be the node-weightings which correspond to cluster $S$.

Observe that (by e.g.\ iterating over vertices and greedily constructing $A^{(i)}_S$), we can ensure that the total support size across the node-weightings of all clusters of $\mcN[h']$ is
\begin{align}\label{eq:brokenClusterSize}
    \sum_{S \in \mcN_{h'}} \sum_i |\supp(A_S^{(i)})| \leq \omega \cdot (n + L^2) \leq \tilde{O}(n^{1 + O(\eps)} + N^{O(\eps)}L^2).
\end{align}
\future{state this more formally}

We let $\mcN$ be the union of all clusters of all $\mcN_{h'}$. We next partition $\mcN$ on the basis of cluster diameter. Specifically, for $h'' \leq \frac{1}{\eps} \cdot (s)^{O(1/\eps)} \cdot h$ which is a power of $2$, we let $\mcN[h'']$ be all clusters of $\mcN$ whose diameter is in $(h''/2,h'']$. Observe that clusters of $\mcN[h'']$ have diameter at most $h''$ but may contain clusters in $\mcN_{h'''}$ for $h''' > h''$ since $\mcN_{h'''}$ may contain clusters with diameter much smaller than $h'''$. 
For each $S \in \mcN[h'']$, we let $S^+$ be all nodes within distance $s \cdot h''$ of some vertex in $S$. %Similarly, we let $G[S^+]:= G[S^+]$ be the graph induced by all such nodes.

% Next, we regroup our clusters as follows. For each $h'' \leq h$ which is a power of $2$, we let $\mcC_{h''}$ be all clusters of $\nset_{h'}$ that have diameter at most $h''$ and at least $h''/2$. Since $\nset_{h'}$ has width $\omega = \frac{N^{O(\eps')}\log N}{\eps'}$ and separation factor $s \geq 2$ and since we have broken each cluster into potentially $L$ more clusters, we have that $\mcC_{h''}$ is $\batch$-batchable (\Cref{dfn:batchable}) for
% \begin{align*}
%     \batch = L \cdot \frac{N^{O(\eps')}\log N}{\eps'}
% \end{align*}

\paragraph{Step 2: Run Cut Matching Games.}\label{step:2} 
First, let 
\begin{align}\label{eq:phiPrime}
    \phi' := \phi / \tilde{O}\left( N^{O(\eps)}\right)
\end{align}
be the (relaxed) sparsity with respect to which we will run compute our cutmatches. 

Next, we do the following for each $h'' \leq h \cdot \frac{1}{\eps} \cdot (s)^{O(1/\eps)}$ which is a power of $2$. More or less, we simultaneously implement a cut matching game for the node-weightings corresponding to clusters in $\mcN[h'']$. For each $S \in \mcN[h'']$ with corresponding node-weightings $A_S^{(1)}, A_S^{(2)}, \ldots$, we initialize graph $G_{Si} = (S, \emptyset)$ to the empty graph. Then, repeat the following $1/\eps$ times.
\begin{enumerate}
    \item \textbf{Run Cut Strategies:} For each $S \in \mcN[h'']$ and each node-weighting $A_S^{(i)}$ corresponding to $S$, apply the the cut strategy (from \Cref{thm:cutStrat}) to $G_{Si}$. Let $\{(A_{Si,k},B_{Si,k})\}_k$ be the output pairs of node-weightings from the cut strategy for cluster $S$.
    \item \textbf{Compute a Cutmatch:} For all pairs $\{(A_{Si,k},B_{Si,k})\}_{S,i,k}$ just computed, compute a $(h'' \cdot s)$-length $\phi'$-sparse cutmatch $(F, C)$ of congestion $\tilde O(1/\phi')$ by invoking \Cref{thm:multiCutmatch} (we will reason about the batch size in our analysis). We let $F = \sum_{S,i,k} F_{Si,k}$ be the relevant decomposition of this flow. %\enote{$hs$ should be about cluster diameter}
    %For clusters $S\in \nset$ with $|C^i_{S}|\ge |S|/(r\log N)$ for some $i$, we terminate the \cmg on it. 
    \item \textbf{Update Graphs:} For each pair $(A_{Si,k}, B_{Si,k})$, with corresponding flow $F_{Si,k}$, let $E_{Si,k}$ be the edge set which for each path $P$ in the support of $F_{Si,k}$ with flow value $F_{Si,k}(P)$ from node $u$ to node $v$ has an edge from $u$ to $v$ of capacity $F_{Si,k}(P)$. Add to $G_{Si}$ the edge set $\bigcup_k E_{Si,k}$.
    \future{Comb the new pair notation through}
    % \item \textbf{Update Graph:} Let $E_F$ be the edge set which for each path $P$ in the support of $F$ with flow value $F(P)$ from node $u$ to node $v$ has an edge from $u$ to $v$ of capacity $F(P)$. Additionally, let $E_F'$ be a set of ``fake edges'' with capacities to guarantee that the added edges form a (weighted) perfect matching; specifically  let $E_F'$ be such that each node $v \in \supp(B)$ has total weighted capacity $B(v)$ with respect to $E_F \cup E_F'$ (and the same symmetrically holds for nodes in $\supp(B')$). Add to $G_S$ the edge set $E_F \cup E_F'$.
\end{enumerate}

\paragraph{Step 3: Glue Broken Up Clusters.}\label{step:3} Lastly, we glue together our broken-up clusters. Specifically, we again do the following for each $h''\le h \cdot \frac{1}{\eps} \cdot (s)^{O(1/\eps)}$ which is a power of $2$. Let $S$ be a cluster we broke up with diameter in $(h''/2, h'']$ whose node-weighting we broke up into $A_S^{(1)}, A_S^{(2)}, \ldots$. Let $\{A_S^{(i)}, A_S^{(j)}\}_{i \neq j}$ be all relevant pairs for this cluster and let $\{A_S^{(i)}, A_S^{(j)}\}_{i \neq j, S}$ be all pairs across all clusters whose node-weightings we broke up; here, $S$ ranges over all clusters whose node-weightings we broke up. Then, we compute a $(\phi'/ L)$-sparse $h''s$-length cutmatch $(F,C)$ of congestion $\tilde{O}(L/\phi')$ on the pairs $\{A_S^{(i)}, A_S^{(j)}\}_{i \neq j, S}$ by invoking \Cref{thm:multiCutmatch}; here, the $S$ in these pairs again ranges over all clusters whose node-weightings we broke up.

% Let $\mcP_{S} := \{(A_S^{(1)}, A_S^{(i)})\}_{i > 1}$ be all relevant node-weighting pairs between the node-weightings we broke $A$ into on $S$. Then, letting $\mcP := \bigcup_{S} \mcP_S$ where the union is taken over all cluster we broke up, we invoke \Cref{thm:multiCutmatch} on $\mcP$ to compute a $\phi'$-sparse cutmatch $(F,C)$ of congestion $\tilde{O}(1/\phi')$ (we will reason about the batch size of $\mcP$ in our analysis).

% \enote{Have to do not stupid breaking up of node-weightings to control total number of nodes in recursion; need to add num nodes to recursion}
% \enote{Above sparsity should really be clique and smaller by $1/L$ because expander we're embedding is stronger than needed}
% \enote{Really use a hypercube / expander; use pruning to find expanding subset; really need it to be a router with router pruning (this is just existential)}

% We apply \Cref{thm:multiCutmatch} to this set of $O(L^2)$ many pairs to compute an $(D_s \cdot s)$-length $(\phi/L)$-sparse cutmatch for these pairs.

% \enote{What happens if splitting something up separates one of the neighborhoods it was supposed to cover?}
\paragraph{Algorithm Output.}
We return as our cut $C^*$ (and corresponding length $h'' \leq h \cdot \frac{1}{\eps} \cdot (s)^{O(1/\eps)}$) the largest cut (by demand-size) of any cutmatch we computed above in \stepTwo or \stepThree (among all $O(\log N/\eps)$ cutmatches). If the size of $C^*$ is $0$ then we return as our $(\leq h, s_w)$-length $\phi_w$-expansion witness the neighborhood covers $\{\mcN_{h'}\}_{h'}$, the roouter for cluster $S$ gotten by taking the union of the routers computed for $A_S^{(1)}, A_S^{(2)}, \ldots$ from \stepTwo and the matchings corresponding to \stepThree and the embedding given by all flows we compute for our cutmatches in \stepTwo and \stepThree. Recall that $s_w = \frac{1}{\eps^2} \cdot s^{O(1/\eps)}$ and $\phi_w = \tilde{O}(\phi \eps/N^{O(\eps)})$.

\future{demand-size not clearly easily computable; really have to just estimate based on the cutmatch output}
\future{comb through fact thare we're also returning $h''$}

% Then, return as our cut $C^*$ the cut $C_{h'}$ which maximizes $|C_{h'}|$.

% Let $C^*$ be the sum of all cuts ever computed by the above algorithm (in both \stepTwo and \stepThree). We return $C^*$ as our cut. 

We conclude with our proof of \Cref{lem: MBSC to CMG}.

\largeCutsFromCutStrats*
\begin{proof}
We use the algorithm described directly above.    
\paragraph{Runtime Analysis.}

We begin by analyzing the runtime of the above algorithm. We begin with \stepOne wherein we build our neighborhood covers. Since $\diam(S)\le h \cdot \frac{1}{\eps} \cdot (s)^{O(1/\eps)}$, by \Cref{thm:cover-separation-factor-existential} we can compute each of our $O(\log h) \leq O(\log N)$-many neighborhood covers which form $\mcN$ in work at most 
\begin{align}\label{eq:NCTime}
    m \cdot \frac{1}{\eps} \cdot (s)^{O(1/\eps)} \cdot N^{O(\eps)} \cdot \log N 
\end{align}
and depth at most
\begin{align}\label{eq:NCDepth}
    \frac{1}{\eps} \cdot (s)^{O(1/\eps)} \cdot N^{O(\eps)} \cdot \log N 
\end{align}
\noindent Likewise, since each of the clusterings of each $\mcN_{h'}$ are disjoint and each $\mcN_{h'}$ has width $N^{O(\eps)} \log N$, we can break up all of our clusters in work at most
\begin{align}\label{eq:breakupTime}
    m \cdot N^{O(\eps)} \log^2 N.
\end{align}
and depth at most 
\begin{align}\label{eq:breakupDepth}
    N^{O(\eps)} \log^2 N.
\end{align}\future{Doing the smart breaking up / greedy thing algorithmically here actually requires a little bit of thought similar to the parallel algos we have for sparse flow decompositions}
Thus, combining Equations \ref{eq:NCTime} and \ref{eq:breakupTime}, the total work of \stepOne is 
\begin{align}\label{eq:step1Time}
    m \cdot \tilde{O} \left( \frac{1}{\eps} \cdot (s)^{O(1/\eps)} \cdot N^{O(\eps)} \right)
\end{align}
and the total depth of \stepOne is 
\begin{align}\label{eq:step1Depth}
    \tilde{O} \left( \frac{1}{\eps} \cdot (s)^{O(1/\eps)} \cdot N^{O(\eps)} \right)
\end{align}

We now discuss \stepTwo. We first discuss how we compute our cutmatches in \stepTwo. Towards this, we first discuss the batch sizes used when invoking \Cref{thm:multiCutmatch} for a given $h'$ in a given one of the $1/\eps$-many iterations. First, observe that for a given $S \in \mcN[h']$, we have that $\{(A_{Si,k},B_{Si,k})\}_{i,k}$ is $N^{O(\eps)}$-batchable since, by \Cref{thm:cutStrat}, the batch size of our cut strategy is $N^{O(\eps)}$. Furthermore, since $\mcN_{h'}$ has width $N^{O(\eps)} \log N$ before we break up clusters and since we duplicate a given cluster at most $L$-many times when breaking up clusters, it follows that each $\mcN[h']$ is $N^{O(\eps)} L \log^2 N$-batchable. Thus, we therefore have that $\{(A_{Si,k},B_{Si,k})\}_{Si,k}$ is $\tilde{O}(L \cdot N^{O(\eps)})$-batchable.

% Observe that since $\mcN_{h'}$ has width $N^{O(\eps)} \log N$, the batch size of our cut strategy game is $N^{O(\eps)}$ (by \Cref{thm:cutStrat}) and each cluster is broken up into up to $L$ clusters, it follows that each $\mcN[h']$ is $\batch$-batchable for length $h'$ (\Cref{dfn:batchable}) with 
% \begin{align*}
% \batch = \tilde{O}\left(L \cdot N^{O(\eps)} \right).
% \end{align*}
 It follows by \Cref{thm:multiCutmatch} that in one iteration of step 2, we can compute all of our cutmatches for all clusters in $\mcN[h']$ in work
\begin{align*}
    m \cdot \tilde{O}\left(L \cdot N^{O(\eps)} \cdot \poly(h) \right)
\end{align*}
and depth at most 
\begin{align*}
    \tilde{O}\left(L \cdot N^{O(\eps)} \cdot \poly(h) \right)
\end{align*}
\noindent and so we can compute our cutmatches for all clusters in $\mcN$ across all $1/\eps$ iterations in work at most
\begin{align}\label{eq:cutMatchTime}
    m \cdot \tilde{O}\left(\frac{1}{\eps} \cdot L \cdot N^{O(\eps)} \cdot \poly(h)\right).
\end{align}
and depth at most
\begin{align}\label{eq:cutMatchDepth}
    \tilde{O}\left(\frac{1}{\eps} \cdot L \cdot N^{O(\eps)} \cdot \poly(h)\right).
\end{align}

Next, we analyze the time to compute our cut strategy cuts in \stepTwo. To do so, we first bound the total number of edges across all $G_S$. Specifically, observe that since we have constructed our clusters so that if $A_S^{(i)}$ is a node-weighting corresponding to cluster $S \in \mcN[h']$, then we have $\sum_{S \in \mcN[h'], i} \supp(A_S^{(i)}) \leq \tilde{O}(n^{1 + O(\eps)} + N^{O(\eps)}L^2)$ (see \Cref{eq:brokenClusterSize}). Furthermore, since in \stepTwo the pairs $\{(A_{Si,k},B_{Si,k})\}_{k}$ for fixed cluster $S$ and $i$ are $N^{O(\eps)}$-batchable (by \Cref{thm:cutStrat}), it follows that for a fixed $S \in \mcN[h']$ and fixed $i$ we have
\begin{align*}
    \sum_{k} |\supp(A_{Si,k} \cup B_{Si,k})) |\leq N^{O(\eps)} \cdot \supp(A_S^{(i)})
\end{align*}
and so 
\begin{align*}
    \sum_{S,i} \sum_{k} |\supp(A_{Si,k} \cup B_{Si,k})| \leq \sum_{S,i} N^{O(\eps)} \cdot \supp(A_S^{(i)}) \leq \tilde{O}(n^{1 + O(\eps)} + N^{O(\eps)}L^2)
\end{align*}
Thus, plugging this bound on $\sum_{S,i,k} |\supp(A_{Si,k} \cup B_{Si,k})|$ into the guarantees of \Cref{thm:multiCutmatch} and the fact that our pairs are $L \cdot N^{O(\eps)}$-batchable, we have that each time we compute a cutmatch in \stepTwo, the total number of edges we add across all $G_S$ for $S \in \mcN[h']$ for a fixed $h'$ is at most

\begin{align*}
    %\tilde{O}(|E(G[S^+])| + b + \sum_i |\supp(A_i \cup B_i)|) = 
    \tilde{O}(m + L \cdot N^{O(\eps)} + n^{1 + O(\eps)} + N^{O(\eps)}L) = \tilde{O}(m + n^{1 + O(\eps)} + L^2 \cdot N^{O(\eps)}).
\end{align*}
Since we have $1/\eps$ iterations, it follows that the number of edges across all $G_S$ for $S \in \mcN[h']$ is never more than
\begin{align*}
    \frac{1}{\eps} \cdot \tilde{O}(m + n^{1 + O(\eps)} + L^2 \cdot N^{O(\eps)}).
\end{align*}
It follows that the work and depth to compute all cut strategies for all $S \in \mcN[h']$ for all $1/\eps$-many iterations and all $h' \leq h \cdot \frac{1}{\eps} \cdot (s)^{O(1/\eps)}$ a power of $2$ in \stepTwo are respectively
\begin{align}\label{eq:cutStratTime}
    \frac{1}{\eps} \cdot  \sum_i \wcutStrat(A_i, m_i)
\end{align}
and
\begin{align}\label{eq:cutStratDepth}
    \frac{1}{\eps} \cdot  \max_i \dcutStrat(A_i, m_i)
\end{align}
where $|A_i| \leq |A| / L$ for all $i$ and $\sum_i m_i \leq \tilde{O}(m + n^{1 + O(\eps)} + L^2 \cdot N^{O(\eps)})$.

Combining the work and depth to compute cutmatches (work \Cref{eq:cutMatchTime} and depth \Cref{eq:cutMatchDepth}) and cut strategies (work \Cref{eq:cutStratTime} and depth \Cref{eq:cutStratDepth}) in \stepTwo, we have that the $1/\eps$-many iterations of \stepTwo for all $h' \leq h \cdot \frac{1}{\eps} \cdot (s)^{O(1/\eps)}$ a power of $2$ can be implemented in work
\begin{equation}\label{eq:step2Time}
m \cdot \tilde{O}\left(\frac{1}{\eps} \cdot L \cdot N^{O(\eps)} \cdot \poly(h)\right)+\frac{1}{\eps} \cdot  \sum_i \wcutStrat(A_i, m_i)
\end{equation}
and depth
\begin{equation}\label{eq:step2Depth}
\tilde{O}\left(\frac{1}{\eps} \cdot L \cdot N^{O(\eps)} \cdot \poly(h)\right) + \frac{1}{\eps} \cdot  \max_i \dcutStrat(A_i, m_i)
\end{equation}
where $|A_i| \leq |A| / L$ for all $i$ and $\sum_i m_i \leq \frac{1}{\eps} \cdot \tilde{O}(m + n^{1 + O(\eps)} + L^2 \cdot N^{O(\eps)})$.

Lastly, we analyze the runtime of \stepThree. Since each $\mcN[h']$ is $\tilde{O}(L \cdot N^{O(\eps)})$-batchable and by definition of how we broke up clusters, we have that all pairs $\{A_S^{(i)}, A_S^{(j)}\}_{i \neq j, S}$ in \stepThree are $\tilde{O}(L \cdot N^{O(\eps)})$-batchable. Thus, applying \Cref{thm:multiCutmatch}, we have that all cumatches of \stepThree can be computed in work 
\begin{align}\label{eq:step3Time}
    m \cdot \tilde{O}\left( L \cdot N^{O(\eps)} \cdot \poly(h) \right).
\end{align}
and depth
\begin{align}\label{eq:step3Depth}
    \tilde{O}\left( L \cdot N^{O(\eps)} \cdot \poly(h) \right).
\end{align}
The work and depth of our algorithm then follows by combining the running time of 
\begin{itemize}
    \item \stepOne work (\Cref{eq:step1Time}) and depth (\Cref{eq:step1Depth})
    \item \stepTwo work (\Cref{eq:step2Time}) and depth (\Cref{eq:step2Depth})
    \item \stepThree work (\Cref{eq:step3Time}) and depth (\Cref{eq:step3Depth}).
\end{itemize}

\textbf{Correctness Analysis: Upper Bound on Largest Sparse Cut Size.} The basic idea will be to argue that most vertices are successfully ``embedded'' which in turn gives us a large expanding subset which will allow us to upper bound the demand-size of the largest sparse cut by \Cref{lem:LDSCAtMostLEC}. Fix an $h'' \leq h \cdot \frac{1}{\eps} \cdot (s)^{O(1/\eps)}$. We fix a suitably large constant $c \in (0,1)$.

% We begin by defining the notion of \emph{$h'$-length embedded nodes}.

\textit{Cut Matching Game Success.}  Consider a node-weighting $A_S^{(i)}$ with corresponding cluster $S \in \mcN[h'']$. Let $B_S^{(i)}$ be the (expanding) node-weighting returned by our cut matching game in \stepTwo as described by \Cref{thm:cutStrat}. We say that a vertex $u$ succeeded for the cut matching game for $A_S^{(i)}$ if $B_S^{(i)}(u) = A_S^{(i)}(u)$. We say that $A_S^{(i)}$ succeeded for the cut matching game if $|B_S^{(i)}| \geq c \cdot |A_S^{(i)}|$. If $S$ is a cluster whose node-weighting we didn't break up then we say that $S$ succeeded for its cut matching games if $A_S^{(1)}$ succeeded. If $S$ is a cluster whose node-weighting we broke up into $A_S^{(1)}, A_S^{(2)}, \ldots$ then we say that $S$ succeeded for its cut matching games if at least a $c$ fraction of the cut matching games of $A_S^{(1)}, A_S^{(2)}, \ldots$ succeeded. We say that a vertex $u$ succeeds for the cut matching games of $S$ if $\sum_i B_S^{(i)}(u) \geq c \cdot A_S(u)$.

% Next, suppose $S \in \mcN[h'']$ is a cluster whose node-weighting $A_S$ we broke up into $A_S^{(1)}, A_S^{(2)}, \ldots$. We say that the cut matching game succeeded $S$ is 
%We say that a vertex $u$ is $h'$-length embedded for all cut matching games if for every $S \in \mcN[h'']$ which contains $u$ we have that $B_S(u) = A_S(u)$. We let $W_0$ be all nodes which are $h''$-length embedded for all cut matching games. We say that $A_S$ is $h''$-length cut matching game embedded if $A_S(W_0) \geq .9 |A_S|$.
%We let $\mcA_0$ corresponding to some $S \in \mcN[h'']$ be all node-weightings that are are $h''$-length embedded.

\textit{Cutmatching Success.} Next, suppose $S \in \mcN[h'']$ is a cluster whose node-weighting $A_S$ we broke up into $A_S^{(1)}, A_S^{(2)}, \ldots$ so that for each $i$ and $j$ the pair $(A_S^{(i)}, A_S^{(j)})$ is a pair for our cutmatching in \stepThree. We let $F_{Sij}$ be the flow returned for the pair $(A_S^{(i)}, A_S^{(j)})$ of the cutmatching returned in \stepThree. We say that the pair $(i,j)$ succeeded for the cutmatching if $\val(F_{Sij}) \geq c |A_S^{(i)}| = c |A_S^{(j)}|$. Likewise, we say that the cutmatching succeeded for $A_S^{(i)}$ if, among all $j$, at least a $c$ fraction of the $(i,j)$ succeeded for the cutmatching. We let $I_S$ be the indices of all $A_S^{(i)}$ which succeeded for the cutmatching. Lastly, we say that the cutmatching succeeded for $S$ if, among all $i$, at least a $c$ fraction of $A_S^{(i)}$ succeeded for the cut matching; i.e.\ $|I_S| \geq c L$. %are $h''$-length embedded. 
% \begin{itemize}
%     \item $\val(F_{Sij}) \geq .9 |A_S^{(i)}| = .9 |A_S^{(j)}|$ and
%     \item $A_S^{(i)}$ and $A_S^{(j)}$ are both $h''$-length embedded for the cut matching game.
% \end{itemize}
% We say that $A^{(i)}_S$ is $h''$-length embedded if, among all $j$, at least a $.9$ fraction of $(i,j)$ are $h''$-length embedded. 

Towards defining the node-weighting we will claim is length-constrained expanding, we define the node-weighting $\hat{B}_S$ for each $S \in \mcN[h'']$.
\begin{itemize}
    \item Specifically, for each cluster $S$  whose node-weighting we did not break up and which succeeded for the cut matching game, we let $\hat{B}_S := B_S^{(1)}$.
    \item If $S$ is a cluster  whose node-weighting $A_S$ we broke up into $A^{(1)}_S, A^{(2)}_S, \ldots$, which succeeded for both the cut matching game and cutmatching and $V_S$ are the vertices which succeeded for the cut matching games of $S$, then we let $\hat{B}_S$ be $A_S$ restricted to $V_S$.
    \item  If $S$ did not succeed for both the cut matching game and cutmatching, then we just let $\hat{B}_S$ be uniformly $0$.
\end{itemize}

% \enote{If $A_S$ is $h''$-length embedded node-weighting then we call $B_S$ its embedded part. For a given cluster $S$ whose node weighting $A_S$ we broke into $A_S^{(1)}, A_S^{(2)}, \ldots$, we let $I_S$ be the indices of the clusters in $A_S^{(1)}, A_S^{(2)}, \ldots$ that are $h''$-length embedded. Likewise, we let $\mcB_S := \bigcup_{i \in I_S} B_S^{(i)}$ be the embedded parts of the subset of $A_S^{(1)}, A_S^{(2)}, \ldots$ that are $h''$-length embedded. If $S$ is a cluster whose node-weighting we did not break up then we just let $\mcB_S = \{B_S\}$. Likewise, we let $\hat{B}_S := \sum_{B \in \mcB_S} B$.  Lastly, we let $\mcB := \bigcup_S \mcB_S$} %Finally, we let $B_{h''} :=\sum_{S} \hat{B}_S $ (where this sum is taken over clusters without multiplicity).

% We let $\mcA$ be all node-weightings which are $h''$-length gluing embedded, let $\mcB$ be their corresponding sub-node-weightings given by the cutmatching game and let $B_{h''} = \sum_{B \in \mcB} B$. 

We now argue that any demand $D$ that decomposes as $D = \sum_{S} D_{S}$ where $D_S$ is $\hat{B}_S$-respecting can be efficiently routed; here, we sum over clusters of $\mcN[h'']$ without multiplicity. To do so we will first route within routers given by our cut matching game (from \stepTwo), then route across clusters using the cutmatches (from \stepThree) and then again route according to our cut matching game (from \stepTwo). We describe this more formally below.

% Say that vertex $u$ is $h'$-length embedded for the pair $(i,j)$ if $u$ is in the matched part of the cutmatching for the pair $(i,j)$ (see \Cref{def:cutmatch} for a definition of the matched parts); that is, $u$ sends $A^{(i)}(u)$ flow according to $F_ij$ to nodes in $\supp(A^{(j)})$. Let $W_{ij}$ be all nodes of $W_0$ that are $h'$-length embedded for the pair $(i,j)$. Lastly, say that a pair $(i,j)$ is $h'$-length embedded if $\sum_{u \in W_{ij}} A^{(i)}(0) \geq A^{(i)}(0)$ the total mass according to $A^{(i)}$ of nodes matched for 

% Say that $u$ is $h'$-length gluing embedded for $S$ if for at least a $.9$ fraction of all $j$ we have that $u$ is $h'$-length embedded for the pair $(i,j)$. Lastly, say that $u$ is $h'$-length gluing embedded if for all $S$ whose node-weightings we broke up and which contain $u$, we have that $u$ is $h'$-length gluing embedded. 

% We let $W$ be all nodes of $W_0$ that are $h'$-length gluing embedded and let $A_{W}$ be $A$ restricted to nodes of $W$. 

\textbf{Routing Within Clusters.} We first describe how to route what we call a \emph{cut matching game demand}. Let $A_S^{(i)}$ be a node-weighting whose cut matching game succeeded and let $B_S^{(i)}$ be the corresponding expanding node-weighting returned by the cut matching game. Likewise, let $D_{Si}$ be a demand that is $B_S^{(i)}$-respecting. We call $D = \sum_{S,i} D_{Si}$ a cut matching game demand. By the guarantees of \Cref{thm:cutStrat} we know that the result of our cut matching game on $S$ is a $1/\eps$-step and $N^{O(\eps)}$-congestion router for $B_S^{(i)}$. Since each edge of the output of our cut matching game for cluster $S$ corresponds to a path of length $h''\cdot s$ in $G[S^+]$ and each of the $1/\eps$ cutmatches we compute has congestion $\tilde{O}(1/\phi')$ by \Cref{thm:multiCutmatch}, it follows that $D$ can be routed over $(h''s/\eps)$-length paths with congestion at most $\tilde{O}(\frac{N^{O(\eps)}}{\phi' \cdot \eps})$.

\textbf{Routing Across Clusters.} We next describe how to route what we call a \emph{cutmatching demand} between indices. Specifically, consider any function $D$ that decomposes as $\sum_S D_S$ (where our sum is over clusters whose node-weightings we broke up) such that for each $S$ we have:
\begin{enumerate}
    \item $(i,j) \in \supp(D_S)$ only if $A_S^{(i)}$ and $A_S^{(j)}$ succeeded for the succeeded for the respective cut matching games and the cutmatching for $S$ and;
    \item $\sum_j D_S(i, j) \leq |B_S^{(i)}|$ and $\sum_j D_S(j, i) \leq |B_S^{(i)}|$ for every $i$.
\end{enumerate}
Then, given any such $D$, we claim there is a $2h''s$-length flow $\hat{F}  = \sum_{S,i,j} \hat{F}_{Sij}$ with congestion $\tilde{O}(1/\phi')$ wherein each $\hat{F}_{Sij}$ routes $D_S(i,j)$ flow from $\supp(B_S^{(i)})$ to $\supp(B_S^{(j)})$ so that no vertex $u$ sends or receives more than $B_S^{(i)}(u)$ and $B_S^{(j)}(u)$ flow respectively according to $\hat{F}_{Sij}$.

To construct $\hat{F}$, first consider two pairs $(i, k)$ and $(k,j)$ that both succeed for the cutmatching where both $A_S^{(i)}$ and $A_S^{(j)}$ succeeded for their cut matching game. Observe that, by definition of a pair succeeding for the cutmatching, we know that $\val(F_{Sik}) \geq c| A_{S}^{(k)}|$ and $\val(F_{Skj}) \geq c| A_{S}^{(k)}|$. Likewise, we know that since both $A_S^{(i)}$ and $A_S^{(j)}$ succeeded for their cut matching game it holds that $B_S^{(i)} \geq c |A_S^{(i)}|$ and $B_S^{(j)} \geq c |A_S^{(j)}|$.

Thus, it follows by scaling and concatenating flow paths of $F_{Sik}$ and $F_{Skj}$ that it is possible to construct a flow $\hat{F}_{Sikj}$ which, for a fixed $S$, $i$ and $j$, is a $2h''s$-length flow that routes $D_S(i,j)$ flow from $i$ to $j$ and incurs congestion on edge $e$ at most $\sum_S \frac{D_S(i,j)}{| A_S^{(i)}|} \cdot O(F_{Sik}(e) + F_{Skj}(e))$. 

%Let $F_{Sik}'$ be $F_{Sik}$ but restricted to flow paths which start and end at vertices that succeeded for the cutmatching game for $A^{(i)}_S$ and $A^{(k)}_S$ respectively. Likewise, let $F_{Skj}'$ be $F_{Skj}$ restricted symmetrically. 
% Observe that we know that $\val(F_{Sik}') \geq .7 \val(F_{Sik}) \geq .6| A_S^{(i)}|$ and symmetrically $\val(F_{Skj}') \geq .7 \val(F_{Skj}) \geq .6| A_S^{(i)}|$. Lastly, by appropriately concatenating flow paths of $F_{Sik}'$ and $F_{Skj}'$ and scaling, we can construct, for a fixed $S$, $i$ and $j$, a $2h''s$-length flow $\hat{F}_{Sikj}$ that routes at least $D_S(i,j)$ flow from $i$ to $j$ and incurs congestion on edge $e$ at most $\sum_S \frac{D_S(i,j)}{| A_S^{(i)}|} \cdot O(F_{Sik}(e) + F_{Skj}(e))$.

Next, let $I_{Sij}$ be all $k$ such that the pair $(i, k)$ and the pair $(k, j)$ both succeeded for the cutmatching and $i$ and $j$ succeeded for the cut matching game. Likewise, let
\begin{align*}
    \hat{F}_{Sij} := \Theta(1) \sum_{k \in I_{Sij}} \hat{F}_{Sikj} / L
\end{align*}
for an appropriate hidden constant. Since $i$ and $j$ both succeeded for the cutmatching we know that $|I_{Sij}| \geq \Omega(L)$ so this is a $2h''s$-length flow that routes at least $D_S(i,j)$ from $i$ to $j$ and on edge $e$ incurs congestion at most 
\begin{align*}
    \frac{D_S(i,j)}{| A_S^{(i)}|} \sum_{k \in I_{Sij}}  O(F_{Sik}(e) + F_{Skj}(e)) / L.
\end{align*}
Let 
\begin{align*}
    \hat{F} := \sum_S \sum_{i,j} \hat{F}_{Sij}.
\end{align*}
be all such flows pairs for all $S$, $i$ and $j$.

This $2h''s$-length flow routes, simultaneously for every $S$, $i$ and $j$, $D_S(i,j)$ flow from $i$ to $j$ for cluster $S$ and on edge $e$ incurs congestion at most 
\begin{align*}
    \sum_{S, i,j} \frac{D_S(i,j)}{| A_S^{(i)}|} \sum_{k \in I_{ij}}  \frac{O(F_{Sik}(e) + F_{Skj}(e))}{L} &= \sum_{S,i,k} \frac{O(F_{Sik}(e))}{L} \sum_j \frac{D_S(i,j)}{| A_S^{(i)}|} + \sum_{S,j,k} \frac{O(F_{Skj}(e))}{L} \sum_i \frac{D_S(i,j)}{| A_S^{(j)}|}\\
    & \leq \sum_{S,i,j}\frac{O(F_{Sij}(e))}{L}\\
    & = \frac{F(e)}{L}\\
    & \leq \tilde{O}\left(\frac{1}{\phi'}\right)
\end{align*}
where, above, we use the fact that $F$ has congestion $\tilde{O}(L/\phi')$ as given in the definition of \stepThree.

\textbf{Routing Across and Within Clusters.}
Next, we describe how to route an arbitrary demand $D$ which decomposes as $D = \sum D_S$ where $D_S$ is $\hat{B}_S$ respecting. First, a minor technical detail to deal with the fact that a vertex can appear in multiple copies of a cluster: observe that since $D_S$ is $\hat{B}_S$-respecting, it is possible to decompose $D_S$ into $D_S = \sum_{ij} D_{Sij}$ where for each $i$ we have $\sum_j D_{Sij}$ and $\sum_j D_{Sji}$ are both $B_S^{(i)}$-respecting. Using this decomposition, we construct our demand $D_2$ between indices. Namely, we let $D_S'$ for indices $i,j \in I_S$ be
\begin{align*}
 D'_S(i,j) := \sum_{u,v} D_{Sij}(u,v)
\end{align*}
and let $D_2 := \sum_S D'_S$.

First, observe that $D_2$ is a cutmatching demand by construction and so by the above discussion can be routed over $2h''s$-length paths by a flow $\hat{F} = \sum_{S, i, j}\hat{F}_{Sij}$ with congestion at most $\tilde{O}(1/\phi')$ where each node $u$ sends and receives at most $\hat{B}_S^{(i)}(u)$ flow according to $\hat{F}_{Sij}$. For a given vertex $u$, we let $w_{Sij}(u)$ be the amount of flow that $u$ sends according to $\hat{F}_{Sij}$. We use these values to construct two cut matching game demands $D_1$ and $D_3$ such that concatenating the routing paths of $D_1$, $D_2$ and $D_3$ give a routing for $D$. 

We first describe $D_1$. Let 
\begin{align*}
    D_{Si}(u,j) := \sum_{v} D_{Sij}(u,v)
\end{align*}
be the amount of demand that vertex $u$ sends to $\hat{B}_S^{(j)}$ according to $D_S$ according to the portion of $u$'s node-weighting that is in $\hat{B}_S^{(i)}$. Next, consider the demand $\hat{D}_{Si}$ wherein vertex $u$ sends each of its demands $D_{Si}(u,j)$ to each node $v$ proportional to $w_{Sij}(v)$. Specifically, let
\begin{align*}
    \hat{D}_{Sij}(u,v) := D_{Si}(u,j) \cdot \frac{w_{Sij}(v)}{\val(\hat{F}_{Sij})}.
\end{align*}
and
\begin{align*}
    \hat{D}_{Si} := \sum_j \hat{D}_{Sij}.
\end{align*}
Lastly, let $D_1$
\begin{align*}
    D_1 := \sum_{S} \hat{D}_{Si}
\end{align*}

First, we claim that $D_1$ is a cut matching game demand. To do so, we must show that $\hat{D}_{Si}$ is $B_S^{(i)}$-respecting. To see this, observe that the demand that a vertex $u$ sends according to $\hat{D}_{Si}$ is 
\begin{align*}
\sum_v \hat{D}_{Si}(u,v) &= \sum_v \sum_j D_{Si}(u,j) \cdot \frac{w_{Sij}(v)}{\val(\hat{F}_{Sij})}\\
&= \sum_j D_{Si}(u,j) \sum_v \frac{w_{Sij}(v)}{\val(\hat{F}_{Sij})}\\
&= \sum_j \sum_{v} D_{Sij}(u,v)\\
& \leq B^{(i)}_S(u)
\end{align*}
where in the last step we used the fact that $\sum_j D_{Sij}$ is $B_S^{(i)}$-respecting. Symmetrically, one can show that $\sum_v \hat{D}_{Si}(v,u) \leq B^{(i)}_S(u)$ which shows that $D_1$ is indeed a cut matching game demand. It follows by the above discussion that we can route $D_1$ over $h''s/\eps$-length paths with congestion at most $\tilde{O}(\frac{N^{O(\eps)}}{\phi' \cdot \eps})$.

Next, we claim that it is possible to concatenate the routing paths of $D_1$ and $D_2$ to get a flow $F = \sum_{S,i,j} F_{Sij}$ in which $F_{Sij}$ routes from each vertex $u$ a flow of value $D_{Si}(u,j)$ to $B_S^{(j)}$.

We describe $F_{Sij}$. Recall that $D_1 = \sum_{i,j} \hat{D}_{Sij}$ and let $\hat{F} = \hat{F}_{Sij}$ be the aforementioned flow which routes $D_2$. We will construct $F_{Sij}$ by concatenating paths of the portion of the flow for $D_1$ which routes $\hat{D}_{Sij}$ and $F_{Sij}$. Specifically, notice that according to $\hat{D}_{Sij}$ the total flow from vertices to a vertex $v$ must be
\begin{align*}
    \frac{w_{Sij}(v)}{\val(\hat{F}_{Sij})} \sum_u D_{Si}(u,j) = \frac{w_{Sij}(v)}{\val(\hat{F}_{Sij})} \sum_{u,w} D_{Sij}(u,w) = w_{Sij}(v).
\end{align*}
and since by definition the flow from vertex $v$ according to $F_{Sij}$ is just $w_{Sij}(v)$, we have that we can concatenate the flow paths of these two flows at each vertex $v$. Next, observe that for a given vertex $u$ this flow sends
\begin{align*}
    \sum_{v} \hat{D}_{Sij}(u,v) =  \sum_v D_{Si}(u,j) \cdot \frac{w_{Sij}(v)}{\val(\hat{F}_{Sij})} = D_{Si}(u,j)
\end{align*}
flow from $u$ to $B_S^{(j)}$ as required. Lastly, observe that $F$ has length $O(h''s/\eps)$ and congestion at most $\tilde{O}(\frac{N^{O(\eps)}}{\phi' \cdot \eps})$. $D_3$ can be constructed symmetrically to $D_1$ and concatenated to $F$ for a flow with the same guarantees but one in which each vertex $u$ sends to vertex $v$ flow $\sum_{} D_{Sij}(u,v)$ flow. Summarizing, we shown how to route our initial demand $D$ that decomposes as $D = \sum_{S} D_{S}$ where $D_S$ is $\hat{B}_S$-with the above length and congestion.
\future{Unpack $D_3$ argument and make notation not so horrible}

\textbf{Constructing a Length-Constrained Expanding Node-Weighting.} Finally, we use the above routing to demonstrate the existence of a large length-constrained expanding subset. %Specifically, recall that if vertex $u$ is in a cluster $S$ whose node-weighting $A_S$ we split into $A_{S}^{(1)}, A_{S}^{(2)}, \ldots$, the value of the node-weighting $A_S(u)$ may be split among these node-weightings. We say that $u$ is \emph{$\alpha$-$h'$-length-embedded} if in at least an $\alpha$-fraction of the node-weightings

Specifically, say that a vertex $v$ is \emph{$h'$-length embedded} if for every cluster $S \in \mcN_{h'} \cap \mcN[h'']$ for all $h'' \leq h \cdot \frac{1}{\eps} \cdot (s)^{O(1/\eps)}$ which contains $v$, $v$ succeeds for the cut matching games of $S$ (of which there are either $1$ or $L$), $S$ succeeded for its cut matching games and cut matching. We let $B_{h'}$ be $A$ restricted to all \emph{$h'$-length embedded} vertices. Clearly $B_{h'}$ is $A$-respecting.

We claim that $B_{h'}$ has large $h'$-length expansion. Consider an $h'$-length $B_{h'}$-respecting demand $D$. Since $\mcN_{h'}$ is a neighborhood cover with covering radius $h'$, for each pair $(u,v) \in \supp(D)$, we know there must be some cluster $S \in \mcN_{h'}$ such that $u,v \in S$. It follows that we can decompose $D$ as $D = \sum_{h''} D_{h''}$ where every pair in the support of $D_{h''}$ is contained in some cluster in $S \in \mcN[h''] \cap \mcN_{h'}$ wherein both $u$ and $v$ succeeded for the cut matching games of $S$ and $S$'s cut matching games and cutmatching succeeded.

%$O(h''s / \eps) \leq h' \cdot (1/\eps)^{1+O(1/\eps')}$ 
Such a demand $D_{h''}$ is exactly the sort of demand we argued we can route above and so it follows that we can route each $D_{h''}$ and therefore $D$ (at an increase of $O(\log N)$ in congestion) over length  $O(h''s / \eps) \leq h' \cdot \frac{1}{\eps^2} \cdot (s)^{1+O(1/\eps)} $ paths with congestion at most $\tilde{O}(\frac{N^{O(\eps)}}{\phi' \cdot \eps})$. Thus, applying our choice of $\phi'$ (\Cref{eq:phiPrime}), we have $B_{h'}$ is $(h', \frac{1}{\eps^2} \cdot (s)^{1+O(1/\eps)})$-length $\tilde{O}(\phi \cdot \eps)$-expanding.

Letting $\bar{B}_{h'} := A - B_{h'}$, it follows by \Cref{lem:LDSCAtMostLEC} that
\begin{align*}
    \qLDSC\left(\phi \frac{\eps}{s^3 N^{O( 1/s)}}, \frac{h'}{2}, \frac{1}{\eps^2} \cdot (s)^{1+O(1/\eps)} \right) \leq \qLEC\left(\phi \eps, h',  \frac{1}{\eps^2} \cdot (s)^{1+O(1/\eps)}\right) \leq |\bar{B}_{h'}|.
\end{align*}
In other words, (up to appropriate slacks) the above upper bounds the demand-size of the demand-size-largest $(h',s)$-length $\phi$-sparse cut for any $h' \leq h$ in terms of $|\bar{B}_{h'}|$.

 Letting $\bar{B}^*$ be the $\bar{B}_{h'}$ of smallest size, we have the following upper bound for all $h'$ on the demand-size of the demand-size-largest length-constrained $\phi' \frac{\eps}{N^{O(\eps)}}$-sparse cut.
 \begin{align}\label{eq:LDSCBound}
     \qLDSC\left(\phi \frac{\eps}{s^3N^{O(1/s)}}, \frac{h'}{2}, \frac{1}{\eps^2} \cdot (s)^{1+O(1/\eps)}\right) \leq |\bar{B}^*|.
 \end{align}

\textbf{Correctness Analysis: Sparsity of Each of Our Candidate Cuts.}
Consider a fixed $h'' \leq h \cdot \frac{1}{\eps} \cdot (s)^{O(1/\eps)}$ that is a power of $2$. Let $(F,C)$ be one of the $1/\eps + 1$ cutmatches we compute (in either \stepTwo or \stepThree) in this iteration. We argue that $C$ has $(h'',s)$-length sparsity at most $\phi$.
% \begin{align}\label{eq:cutSparse}
%     \phi \cdot \tilde{O}\left(L^5 \cdot N^{O(\eps)} / \eps \right).
% \end{align}

To do so we begin by constructing a large $h''$-length demand $D_{h''}$ which is $h''s$-separated in $G + h' \cdot C_{h'}$. Let the pairs for $(F,C)$ be $\{(A_i, A_i')\}_{i}$ where the support of each such pair is in some cluster in $\mcN[h'']$ and with matched and unmatched parts $\{(M_i, M_i')\}_{i}$ and $\{(U_i, U_i')\}_{i}$ respectively. Throughout this proof we will assume without loss of generality that for all $i$ $A_i(U_i) \leq A_i'(U_i')$. For each pair $(A_i, A_i')$, let $D_{i}$ be an $h''$-length demand with $\supp(D_i) \subseteq \supp(U_i) \times \supp(U_i')$ where each $u \in U_i$ sends $A_i(u)$ to nodes in $U_i'$ so that no node $v \in U_i'$ receives more than $A_i'(v)$ demand. The definition of a cutmatch (\Cref{def:cutmatch}) ensures that in $G + h'' \cdot C_{h''}$ all pairs of this demand are at least $(h'' s)$-far.

Similarly, if we are in \stepTwo we let
\begin{align*}
    D_{h''} := \sum_{i} D_{i} / \tilde{O}(N^{O(\eps)})
\end{align*}
and if we are in \stepThree we let
\begin{align*}
    D_{h''} := \sum_{i} D_{i} / \tilde{O}(L \cdot N^{O(\eps)})
\end{align*}
be this demand summed and scaled appropriately across all pairs. Observe that $C_{h''}$ still clearly $(h''s)$-separates this demand.  Furthermore, observe that this demand is $h''$-length by since each pair in the support is contained in a cluster in $\mcN[h'']$. Also, observe that above demand is $A$-respecting by virtue of the width of each of our neighborhood covers being $N^{O(\eps)}$ and by definition of how we break up node-weightings; this holds regardless of whether the cutmatch is computed in \stepTwo or \stepThree. Lastly, observe that the size of this demand is 
\begin{align}\label{eq:demandLower}
    |D_{h''}| &= \sum_{i} A_i(U_i) / \tilde{O}(N^{O(\eps)})\nonumber\\
    & = \sum_i \frac{|A_i| - A_i(M_i)}{\tilde{O}(N^{O(\eps)})}\nonumber\\
    & \geq \frac{\sum_i |A_i| - \val(F)}{  \tilde{O}(N^{O(\eps)})}
\end{align}

On the other hand, by the definition of a cutmatch (\Cref{def:cutmatch}), we have that the size of $C$ is at most
\begin{align}\label{eq:cutUpper}
    |C|\leq \phi'' \cdot \left( \left(\sum_i |A_i| \right) - \val(F) \right).
\end{align}
where $\phi'' = \phi'/L$ if we are in \stepThree and $\phi'' = \phi'$ otherwise.

Combining Equations \ref{eq:demandLower} and \ref{eq:cutUpper}, the definition of $\phi'$ (\Cref{eq:phiPrime}), we have that $C_{h''}$ is an $h''$-length cut with sparsity at most 
\begin{align*}
    \phi' / \tilde{O}(N^{O(\eps)}) \leq \phi.
\end{align*}

% Concluding, we argue that the demand is large. Specifically, we have by the definition of a cutmatch that $C_{h''}$ has size

% Since $\mcN[h']$ is $\tilde{O}(L^3 \cdot N^{O(\eps)})$-batchable (regardless of whether we are in \stepTwo or \stepThree), it follows that this demand is also unit and $A$-respecting, as well as $h'$-length and $(h's)$-separated by $C$. Likewise, note that by the definition of a cutmatch we have
% \begin{align*}
%     \sum_{i} |U_i| &= \sum_i |B_i| - |M_i|\\
%     & \geq \left(\sum_i |B_i| \right) - \val(F)
% \end{align*}
% and so $D_{h'}$ has size
% \begin{align*}
%     |D_{h'}| &= \left(\sum_{i} |U_i| \right) \Big / \tilde{O}(L^3 \cdot N^{O(\eps)}) \\
%     &\geq \left(\left(\sum_{i} |B_i| \right) - \val(F) \right) \Big / \tilde{O}(L^3 \cdot N^{O(\eps)}).
% \end{align*}
% On the other hand, by the definition of a cutmatch, we know that the size of $C$ is at most
% \begin{align}
%     ||C|\leq \phi \cdot \left( \left(\sum_i |A_i| \right) - \val(F) \right).
% \end{align}
% Combining 

% It follows by the above that the $(h',s)$-length sparsity of $C$ with respect to $A$ is at most $\phi' \cdot \tilde{O}\left(L^3 \cdot N^{O(\eps)} \right)$. Since $C_{h'}$ is the sum of $1/\eps + O(L^2)$ many such cuts, it follows by scaling demands appropriately that $C_{h'}$ has $h'$-length sparsity at most $\phi' \cdot \tilde{O}(L^5  \cdot N^{O(\eps)} / \eps) \leq \phi$ by the definition of $\phi'$ (\Cref{eq:phiPrime}).

\textbf{Correctness Analysis: Output Cut is Demand-Size Large.} It remains to argue that the demand-size of the cut $C^*$ returned by our algorithm is sufficiently large. Recall that, if $C^*$ was computed when we were considering diameter $h'' \leq h \cdot \frac{1}{\eps} \cdot (s)^{O(1/\eps)}$, then we know that the $(h'', s)$-length sparsity of $C^*$ is at most $\phi$.

We claim that the $(h'',s)$-length demand-size of $C^*$ with respect to $A$ (\Cref{def:demandSize}) is 
\begin{align}\label{eq:cutLower}
    A_{(h'',s)}(C^*) \geq  |\bar{B}^*| \cdot \frac{\eps^3}{\tilde{O}(N^{O(\eps)})}.
\end{align}
Furthermore, by \Cref{eq:LDSCBound} for any $h' \leq h$, the largest $(\frac{h'}{2}, \frac{1}{\eps^2} \cdot (s)^{1+O(1/\eps)})$-length demand-size of a $(\frac{h'}{2}, \frac{1}{\eps^2} \cdot (s)^{1+O(1/\eps)})$-length $\phi \cdot \frac{\eps}{s^3N^{O(1/s)}}$-sparse cut is at most $|\bar{B}^*|$. It follows that $C_{h'}$ is a $\frac{\eps^3}{\tilde{O}(N^{O(\eps)})}$-approximate $(\leq h,s)$-length $\phi$-sparse cut with length slack approximation $\alpha_s = \max\left(2,\frac{1}{\eps^2} \cdot (s)^{1+O(1/\eps)}\right)$ and sparsity approximation $\alpha_\phi = \frac{\eps}{s^3 N^{O(1/s)}}$ (see, again, \Cref{dfn:apxLargeCut} for a definition of this notion of approximation).

It remains to argue that \Cref{eq:cutLower} holds. We let $\bar{B}_{h'} = \bar{B}^*$ for the remainder of this proof; that is, for the remainder of this proof we let $h'\leq h$ be the length with respect to which $B_{h'} = A - \bar{B}_{h'}$ is length-constrained expanding. %Recall that $B_{h'}$ is the sub-node-weighting of $A$ induced by all vertices that are $h'$-length embedded. 

A vertex $u$ can fail to be  $h'$-length embedded if there is a cluster $S \in \mcN_{h'} \cap \mcN[h'']$ for some $h'' \leq h' \cdot \frac{1}{\eps} \cdot (s)^{O(1/\eps)}$ which contains $u$ and with corresponding node-weighting $A_S$ for which
\begin{enumerate}
    \item \textbf{Cut Matching Games Fails for Vertex.} $u$ fails the cut matching game for $S$; in other words, there is some cluster $S \in \mcN[h'']$ containing $u$ such that $\sum_i B_S^{(i)}(u) < c A_S(u)$. Let $W^{(1)}_{h''}$ be all such nodes.
    \item \textbf{Cut Matching Game Fails for Cluster.} $u$ succeeds the cut matching game for $S$ but $S$ does not succeed the cut matching game; i.e\ it is not the case that a constant fraction of $A_S^{1}, A_S^{(2)}, \ldots$ succeed. Let $W^{(2)}_{h''}$ be all such nodes.
    \item \textbf{Cutmatching Fails.} $S$ is a cluster whose node-weighting we broke up which did not succeed for the cutmatching. Let $W^{(3)}_{h''}$ be all such nodes not in $W^{(2)}_{h''}$ or $W^{(1)}_{h''}$.
\end{enumerate}
Likewise, let $W^{(i)} := \bigcup_{h''} W^{(i)}_{h''}$ for $i \in [1,3]$. It follows that 
\begin{align*}
    |B_{h'}| = A - A(W^{(1)}) - A(W^{(2)}) - A(W^{(1)}).
\end{align*}
and so 
\begin{align*}
    |\bar{B}^*| = A(W^{(1)}) + A(W^{(2)}) + A(W^{(3)}).
\end{align*}
Likewise, by averaging there is some $h''$ such that 
\begin{align*}
A(W^{(1)}) + A(W^{(2)}) + A(W^{(3)})  \leq \tilde{O} \left(
    A(W^{(1)}_{h''}) + A(W^{(2)}_{h''}) + A(W^{(3)}_{h''}) \right).
\end{align*}

Furthermore, observe that by how defined what it means for a cluster to succeed we have $A(W^{(2)}_{h''}) \leq O(A(W_{h''}^{O(1)}))$ for an appropriate hidden constant. Thus, we have
\begin{align*}
    |\bar{B}^*|  \leq \tilde{O}\left(
    A(W^{(1)}_{h''})  + A(W^{(3)}_{h''}) \right).
\end{align*}
We case on which of $A(W^{(1)}_{h''})$ and $A(W^{(3)}_{h''})$ are larger.
\begin{enumerate}
    \item Suppose $A(W^{(1)}_{h''}) \geq A(W^{(3)}_{h''})$ so that $|\bar{B}^*|  \leq \tilde{O}\left(A(W^{(1)}_{h''}) \right)$. For a $u$ and $i$, say that $u$ fails the cut matching game for $S$ and $i$ if $B_S^{(i)}(u) \leq c'' \cdot A_S^{(i)}(u)$ for a fixed constant $c'' \in [0,1)$. Let $W^{(1)}_{Si}$ be all nodes that fail the cut matching game for $S$ and $i$. Observe that, by choosing $c''$ appropriately, we have
    \begin{align*}
       \sum_{S, i} A_S^{(i)}(W^{(1)}_{Si}) \geq \Omega\left(A(W^{(1)}_{h''})\right).
    \end{align*}
    and so it suffices to bound $\sum_{S,i} A_S^{(i)}(W^{(1)}_{Si})$.
    \future{write this out}
    
    % By definition of $W_{h''}^{(1)}$, we know that there must be some cluster $S_u \in \mcN[h'']$ which contains $u$ so that 
    % \begin{align*}
    %     \sum_i B_S^{(i)}(u) < c \cdot A_S(u).
    % \end{align*}
    % It follows that across all cutmatching games we run in \stepTwo, the total amount of node-weighting that succeeds is 
    % \begin{align*}
    %     \sum_{u} \sum_{S \ni u} \sum_i B_S^{(i)}(u) &\leq \sum_{u \not \in W_{h''}^{(1)}} \sum_{S \ni u} \sum_i A_S(u) + \sum_{u \in W_{h''}^{(1)}} \sum_{S \ni u} \sum_i c A_S(u)\\
    %     &\leq N^{O(\eps)}|A| + \sum_{u \in W_{h''}^{(1)}} \sum_{S \ni u} \sum_i c A_S(u).
    % \end{align*}

    % On the other hand, by the guarantees of \Cref{thm:cutStrat}, we have that if $|B_S^{(i)}| \geq (1-O(\frac{\alpha_{Si}}{\eps}))$

    % ============
    
    Let $(F^{(l)}, C^{(l)})$ be the cutmatch we compute in iteration $l \in [1/\eps]$ for pairs $\{(A_{Si,k}^{(l)},B_{Si,k}^{(l)})\}_{S,i,k}$ and let $F^{(l)} = \sum_{S,i,k} F_{Si, k}^{(l)}$ be the decomposition of this flow (one sub-flow for each pair).

    Let $\alpha_{Si}$ be the smallest matching played by the cut matching game we run on $A_S^{(i)}$. It follows that, for a fixed $S$ and $i$, there must be some iteration $l_{Si}$ among our $1/\eps$ cutmatches %and some index $k_{Si}$ such that
    \begin{align*}
        \sum_k \val(F_{Si,k}^{(l_{Si})}) = (1-\alpha_{Si})\cdot \sum_k |A_{Si,k}^{(l_{Si})}|.
    \end{align*}

    Likewise, we know by the guarantees of \Cref{thm:cutStrat} that
    \begin{align*}
        \left(1-\frac{\alpha_{Si}}{\eps}\right) \cdot |A_S^{(i)}| \leq |B_S^{(i)}| \leq |A_S^{(i)}| - \Omega\left(A_S^{(i)}(W_{Si}^{(1)}) \right).
    \end{align*}
    and so rearranging we have
    \begin{align*}
        \eps \cdot A_S^{(i)}(W_{Si}^{(1)}) \leq \alpha_{Si} \cdot |A_S^{(i)}|.
    \end{align*}

    % and so for a fixed $i$, using the fact that the batch size of our cut matching game is $N^{O(\eps)}$ by \Cref{thm:cutStrat}, we have
    % \begin{align*}
    %     \sum_{l,k} \val(F_{Si,k}^{(l)}) &\leq \frac{N^{O(\eps)}}{\eps} \cdot |A_S^{(i)}| - \alpha_{Si} \cdot |A_{Si,k_{Si}}^{(l_{Si})}|\\
    %     &\leq \frac{N^{O(\eps)}}{\eps} \cdot |A_S^{(i)}| - \alpha_{Si} \cdot |A_{S}^{(i)}|.
    % \end{align*}
    % where in the last line we used the fact that each pair played by the cut strategy has size at least $|A_{S}^{(i)}|/N^{O(\eps)}$ 

    Let $F = \sum_l F^{(l)}$, let $C = \sum_l C^{(l)}$ and let $D_{h''} = \eps \sum_l D_{h''}^{(l)}$ where $D_{h''}^{(l)}$ is the demand for this cutmatch as described in our sparsity analysis. Note that $C$ is the cut our algorithm considers from \stepTwo for this value of $h''$. Thus, $D_{h''}$ is $A$-respecting, $C$ $h''s$-separates $D_{h''}$ and $D_{h''}$ has size at least
    \begin{align*}
        |D_{h''}| \geq \eps \cdot \sum_l \sum_{S,i,k} \frac{ |A_{Si,k}| - \val(F_{Si,k}^{(l)})}{  \tilde{O}(N^{O(\eps)})} &\geq \eps \cdot\sum_{S, i} \frac{\alpha_{Si} \cdot \sum_k |A_{Si,k}^{(l_{Si})}|}{\tilde{O}(N^{(\eps)})}\geq \eps \cdot\sum_{S, i} \frac{\alpha_{Si} \cdot |A_{S}^{(i)}|}{\tilde{O}(N^{(\eps)})}%\\
    \end{align*}
    where in the last inequality we applied the fact that the total size of the node-weighting of all pairs played by a cut strategy on node-weighting $A$ is at least $|A|$ (by definition of a node-weighting).
    
    % \enote{Check with Bernhard that this is true}.

    Combining the above we get that this demand has size at least 
    \begin{align*}
        |D_{h''}| \geq \eps \cdot \sum_{S, i} \frac{\alpha_{Si} \cdot |A_{S}^{(i)}|}{\tilde{O}(N^{(\eps)})} \geq \eps^2 \cdot \sum_{S, i} \frac{A_S^{(i)}(W_{Si}^{1})}{\tilde{O}(N^{(\eps)})} \geq \frac{\eps^2 }{\tilde{O}(N^{(\eps)})} \cdot A(W^{(1)}_{h''}) \geq \frac{\eps^2 }{\tilde{O}(N^{(\eps)})} \cdot |\bar{B}^*|,
    \end{align*}
    demonstrating that in this case we have $|D_{h''}| \geq \frac{\eps^2 }{\tilde{O}(N^{(\eps)})} \cdot |\bar{B}^*|$. By an averaging argument, there must be some $C^{(l)}$ which separates at least an $\eps$ fraction of $D_{h''}$, demonstrating that one of the $C^{(l)}$ has demand-size at least $\frac{\eps^3}{\tilde{O}(N^{(\eps)})} \cdot |\bar{B}^*|$.

    \item Suppose $A(W^{(1)}_{h''}) < A(W^{(3)}_{h''})$ so that $|\bar{B}^*|  \leq \tilde{O}\left(
    A(W^{(3)}_{h''}) \right)$. Let $(F,C)$ be the cutmatch returned in \stepThree when we are using diameter $h''$. We claim that $A_{(h,s)}(C)$ is large. Namely, let $D_{h''}$ be the $A$-respecting demand which is $h''s$-separated and of size  as described above which by \Cref{eq:demandLower} has size at least 
    \begin{align*}
        \sum_{S,i,j} \frac{ |A_S^{(i)}| - \val(F_{Sij})}{  \tilde{O}(L \cdot N^{O(\eps)})}.
    \end{align*}
    Let $\mcS^{(3)}$ be all clusters of $\mcN[h'']$ that did not succeed for the cutmatching (without multiplicity). If a cluster $S \in \mcS^{(3)}$ then we know that, for at least a constant fraction of $i$, a constant fraction of the pairs $(i,j)$ did not succeed, i.e.\ $\val(F_{Sij}) < c |A_S^{(i)}|$. Thus, we have that, among all $L^2$ pairs, a constant fraction did not succeed for $S$ and so for any $S \in \mcS^{(3)}$ we have
    \begin{align*}
        \sum_{i,j} |A_S^{(i)}| - \val(F_{Sij}) \geq (1-c') L \cdot |A_S|
    \end{align*}
    for some constant $c' > 0$.

    Since each vertex of $W_{h''}^{(3)}$ appears in at least one cluster of $\mcS^{(3)}$, we have
    \begin{align*}
        \sum_{S \in \mcS^{(3)}} |A_S| \geq A(W_{h''}^{(3)})
    \end{align*}

    Thus, we have that this demand has size
    \begin{align*}
        \frac{\sum_i |A_i| - \val(F)}{  \tilde{O}(N^{O(\eps)})}
    \end{align*}

    Combining the above, we have that the demand $D_{h''}$ has size at least
    \begin{align*}
        \sum_{S,i,j} \frac{ |A_S^{(i)}| - \val(F_{Sij})}{  \tilde{O}(L \cdot N^{O(\eps)})} &\geq \sum_{S \in \mcS^{(3)}} \frac{ \sum_{i,j}|A_S^{(i)}| - \val(F_{Sij})}{  \tilde{O}(L \cdot N^{O(\eps)})}\\
    &\geq \frac{(1-c')}{N^{O(\eps)}}\sum_{S\in \mcS^{(3)}}|A_S|\\
    & \geq \frac{1}{\tilde{O}( N^{O(\eps)})) }\cdot A(W_{h''}^{(3)}).
    \end{align*}
    Thus, in this case we have $|D_{h''}| \geq |\bar{B}^*|/\tilde{O}(N^{O(\eps)})$ and so the demand-size of $C$ is at least this
\end{enumerate}
In either of the above cases we have that one of the cuts we compute for $h''$ has demand-size at least $|\bar{B}^*| \cdot \frac{\eps^3}{\tilde{O}(N^{O(\eps)})}$ as required.

\textbf{Correctness Analysis: Witnesses.} Lastly, we argue about the returned witness (as defined in \Cref{def:LCExpWitness}). 

\begin{itemize}
    \item \textbf{Neighborhood Cover.} Clearly, $\mcN$ is a neighborhood cover satisfying the required properties
    \item \textbf{Routers.} Recall that each $B_{h'}$ as described above is $(h', \frac{1}{\eps^2} \cdot (s)^{1+O(1/\eps)})$-length $\tilde{O}(\phi \cdot \eps)$-expanding. Since each cut considered by our algorithm has $(h'', s)$-length sparsity at most $\phi$ for some $h'' \leq h \cdot \frac{1}{\eps} \cdot s^{O(1/\eps)}$, it follows that if any of these cuts has non-zero size then our graph is not a $(h'', s)$-length $\phi$-expander. Thus, each of our cutmatches' cuts must always have a cut of size $0$ and so each router we compute for cluster $S$ must be a $1/\eps$-step and $N^{O(\eps)}$ congestion router. Thus, if $S$ is a cluster whose node-weighting we broke up into $A_S^{(1)}, A_S^{(2)}, \ldots$ then the union of the routers we compute for $A_S^{(1)}, A_S^{(2)}, \ldots$ using our cut matching game along with the corresponding matching edges from \stepThree is a $(s_0 = (2/\eps+1))$-step router with congestion $\kappa_0 = N^{O(\eps)}$.
    \item \textbf{Embedding of Routers.} Observe that the sum of the flows we compute across all cutmatches for clusters of $\mcN_{h'}$ have length at most $h' \cdot \frac{1}{\eps} \cdot s^{O(1/\eps)} = h' \cdot s_1$ and congestion at most $\kappa_1 = \tilde{O}(N^{O(\eps)}/(\phi \eps))$.
\end{itemize}
Lastly, observe that the overall for our witness we get
\begin{align*}
    s_0 \cdot s_1 = (2/\eps+1) \cdot \frac{1}{\eps} \cdot s^{O(1/\eps)} = O\left(\frac{1}{\eps^2} \cdot s^{O(1/\eps)} \right) = s_w
\end{align*}
and
\begin{align*}
    \kappa_0 \cdot \kappa_1 = N^{O(\eps)} \cdot \tilde{O}(N^{O(\eps)}/(\phi \eps)) = \tilde{O}(N^{O(\eps)}/(\phi \eps)) = 1/\phi_w
\end{align*}
\end{proof}

\future{Break the above into way more lemmas}

% \section{Algorithm: Demand-Size-Large Cuts from Expander Decompositions}

% \begin{lemma}
% For every $\eps > 0$ there is a cut strategy with cut batch size $N^{O(\eps)}$ which when used in a cut matching game with $1/\eps$ rounds of interaction against any $(1-\alpha)$-perfect matching strategy results in a $G_r$ that is a $1/\eps$-step and $N^{O(\eps)}$-congestion router for some $A' \preceq A$ of size $|A'| \geq (1-O(\frac{\alpha}{\eps})) \cdot |A|$. 

% This cut-strategy on a node-weighting $A$ in a graph with $m$ edges can be computed in time 
% \begin{align*}
%     \tcutStrat(A, m) \leq \tilde O(\tED(A,m, h,2^{1/\eps}, \phi, n^{\poly(\eps)}))/\poly(\eps)
% \end{align*}
% where $\tED(\ldots)$ is the time needed for computing an $(h,2^{1/\eps})$-length $\phi$-expander decomposition with cut slack $N^{\poly(\eps)}$ for node-weighting $A$ in an $m$-edge graph. Likewise, $A'$ can be computed in the same time and is vertex induced: for each vertex $u$ if $A'(u) \neq 0$ then $A'(u) = A(u)$ and can be computed in time.
% \end{lemma}

% \largeCutsFromCutStrats*

Combining \Cref{thm:cutStrat} and \Cref{lem: MBSC to CMG} immediately gives \Cref{lem:cutsFromEDs}, restated below for convenience.
\largeCutsFromEDs*

% \section{Expander Decomposition Algorithm \& Proof - Putting everything together}

\section{Algorithm: EDs from Demand-Size-Large Sparse Cuts}
We now show how to compute expander decompositions from demand-size large length-constrained sparse cuts. Specifically, we show the following. Recall the definition of an $\alpha$-approximate $(\leq h, s)$-length $\phi$-sparse cut with length approximation $\alpha_s$ and sparsity approximation $\alpha_\phi$ from \Cref{dfn:apxLargeCut}.

% Old Bernhard version:
% \begin{restatable}{lemma}{EDsFromCuts}
%     \label{lem: from MSC to ED}
%     Fix $\alpha, \alpha_{\phi},\alpha_s >1$ and let $\wsparseCut(A,m)$ be the time to compute an $\alpha$-approximate $(\leq h', s')$-length $\phi'$-sparse cut with sparsity approximation $\alpha_\phi$ and length approximation $\alpha_s$ w.r.t.\ node-weighting $A$ in a graph with $m'\leq m$ edges for $h' \leq h$, $s' \leq s$ and $\phi' \geq \phi$. 
    
%     Then, for every $\eps'>0$ such that $s\cdot (2\alpha_s)^{O(1/\eps')}=o(\log N/\log\log N)$, the work and depth to compute a $(\leq h,s\cdot (2\alpha_s)^{O(1/\eps')})$-length $(\phi,N^{O(1/s)}\cdot\alpha_{\phi}^{O(1/\eps')})$-expander decomposition for $A$ is
%     \begin{align*}
%         \wED(A,m) \leq O\bigg(\frac{\alpha_C\cdot N^{(1/s)+\eps'}}{\eps'}\bigg)\cdot \wsparseCut(A,m)
%     \end{align*}
%     and 
%     \begin{align*}
%         \dED(A,m) \leq O\bigg(\frac{\alpha_C\cdot N^{(1/s)+\eps'}}{\eps'}\bigg)\cdot \dsparseCut(A,m)
%     \end{align*}
    
% \end{restatable}

\future{These aren't our ED parameters}
\begin{restatable}{lemma}{EDsFromCuts}
    \label{lem: from MSC to ED}
    Fix $\alpha, \alpha_{\phi},\alpha_s >1$ and let $\wsparseCut(A,m)$ be the time to compute an $\alpha$-approximate $(\leq h', s')$-length $\phi'$-sparse cut with sparsity approximation $\alpha_\phi$ and length approximation $\alpha_s$ w.r.t.\ node-weighting $A$ in a graph with $m'\leq m$ edges for $h' \leq h$, $s' \leq s$ and $\phi' \geq \phi$. 
    
    Then, for every $\eps, \eps' >0$, one can compute a $(\leq h,s)$-length $\phi$-expander decomposition for $A$ with cut slack
    \begin{align*}
        \kappa =   \tilde{O}\left(\alpha \cdot N^{O(\eps')} \right) \cdot \tilde{O} \left(\alpha_\phi \cdot s \cdot N^{O(1/\sqrt{s})}  \ \right)^{O(1/\eps')},
    \end{align*}
    length slack
    \begin{align*}
        s = \alpha_s^{O(1/\eps')},
    \end{align*}
    work
    \begin{align*}
        \wED(A,m) \leq \tilde{O}\left(\frac{N^{O(\eps')}}{\eps'} \cdot \alpha \right) \cdot \wsparseCut(A,m)    
    \end{align*}
    and depth 
   \begin{align*}
        \dED(A,m) \leq \tilde{O}\left(\frac{N^{O(\eps')}}{\eps'} \cdot \alpha \right) \cdot \dsparseCut(A,m).    
    \end{align*}
    
    % \begin{align*}
    %     \wED(A,m) \leq O\bigg(\frac{\alpha_C\cdot N^{(1/s)+\eps'}}{\eps'}\bigg)\cdot \wsparseCut(A,m)
    % \end{align*}
    % and 
    % \begin{align*}
    %     \dED(A,m) \leq O\bigg(\frac{\alpha_C\cdot N^{(1/s)+\eps'}}{\eps'}\bigg)\cdot \dsparseCut(A,m)
    % \end{align*}
    
\end{restatable}

We begin by describing the algorithm we use to prove the above. We give pseudo-code in \Cref{alg:EDsfromCuts}. The algorithm runs in $\frac{1}{\eps'}$ top-level iterations we call epochs. In epoch $\epoch$ the algorithm repeatedly cuts an $\alpha$-approximate demand-size-largest $(\leq h_\epoch,s_\epoch)$-length $\phi_\epoch$-sparse cut for the target node weighting $A$ in the current graph with sparsity slack $\alpha_\phi$ and length slack $\alpha_s$. For the next epoch the algorithm adjusts its target values for sparsity and length by decreasing $\phi_{\epoch}$ by about $1/\alpha_\phi$, $h_{\epoch}$ by about $1/\alpha_s$ and increasing $s_{\epoch}$ by about $\alpha_s$.

The proof shows that, at the end of an epoch $\epoch$, no demand-size-large $(h_{\epoch}', s_{\epoch}')$-length $\phi_{\epoch}'$-sparse cuts exist for $A$ anymore for $h_{\epoch}'$, $s_{\epoch}'$ and $\phi_{\epoch}$ each slightly more relaxed than $h_{\epoch}$, $s_{\epoch}$ and $\phi_\epoch$. More specifically, for $h_\epoch' \leq h_\epoch / \alpha_s$, $s_{\epoch}' = s_{\epoch} \cdot \alpha_s$ and $\phi_{\epoch}' = \phi_{\epoch}/ \alpha_\phi$. Initially, no such cut of demand-size strictly more than $|A|$ exists (trivially) and in each epoch we improve the quality of this upper bound on the demand-size-largest cut by an $N^{\eps'}$ factor so that after $1/\eps'$ iterations no such cut exists (for appropriately relaxed length, length slack and sparsity). Generally speaking, the trick is to make sure that our upper bound on $|A|$ improves faster than we must relax our length, length slack and sparsity.

% The guarantee for the upper bound on the largest cut with sparsity $\phi_i$ improves, up to losses, by a factor of $\frac{1}{\gamma}$. Choosing a large enough $\gamma$ to ensure that this size bound improves by a factor of $N^{-\eps}$ guarantees that the number of epochs is $1/\eps'$ as desired which in turn ensures that the losses caused by adjusting $\phi,h$ and $s$ in each epoch remain under control. Generally increasing $\gamma$ strictly lowers the iteration depth and all approximation guarantees with the only downside being that the length and therefore computational complexity of the expander decomposition grows linearly in $\gamma$. 

\begin{algorithm}[h]
    \caption{Length-Constrained Expander Decompositions (from Large Sparse Cuts)}
    \label{alg:EDsfromCuts}
    \begin{algorithmic}[0] % The number tells where the line numbering should start
            \State \textbf{Input:} Edge-capacitated graph $G_0$, parameters $\eps, \eps' \in (0,1)$, node-weighting $A$ on $G$, length bound $h$, and conductance bound $\phi>0$, an algorithm for demand-size-largest sparse cut.
            \State \textbf{Output:} $(\leq h,s)$-length $\phi$-expander decomposition% with cut slack $\kappa$ for $A$ in $G$ \\ for $\kappa = \big(\frac{N^{\poly(\eps)}}{\poly(\eps)}\big)^{O(1/\eps')}$ and $s = (1/\eps)^{O(1/\eps')}$.
            \State \textbf{Initialize Graph and Number of Iterations:} $G = G_0$ and $l=\Theta(\log N \cdot \alpha \cdot N^{O(\eps')})$ 
            \State \textbf{Initialize Length:} $h_0 = \alpha_s \cdot (2\alpha_s)^{1/\eps'} \cdot h$
            \State \textbf{Initialize Length Slack:} $s_0 = 2 \cdot \alpha_s^{1/\eps'}$% $s_0 = 2$ %(and let $s_{1/\eps'}$ be the final length slack)
            \State \textbf{Initialize Sparsity:} $\phi_0 = \phi \cdot \alpha_\phi \cdot \prod_{\epoch \in [1/\eps']}\left( \tilde{O}(\alpha_\phi \cdot s_\epoch^3 \cdot N^{O(1/s_\epoch)}) \right)$
            \For{$\epoch = 1, 2, \ldots 1/\eps'$}
                \State \textbf{Update Length:} $h_\epoch = \frac{1}{2\alpha_s} \cdot h_{\epoch-1}$
                \State \textbf{Update Length Slack:}
                $s_\epoch = \alpha_s  \cdot s_{\epoch-1}$
                \State \textbf{Update Sparsity:} $\phi_\epoch = \frac{1}{\tilde{O}(\alpha_\phi \cdot s_\epoch^3 \cdot N^{O(1/s_\epoch)})} \cdot \phi_{\epoch-1} $
                \For{ $j =1,2,\ldots, l$ }
                    \State Let $C$ be an $\alpha$-approximate demand-size-largest $(\leq h_\epoch, s_{\epoch})$-length $\phi_{\epoch}$-sparse cut\\ \qquad \qquad  for $A$ and  length $h_\epoch''$  in $G$  with length and sparsity approximation $\alpha_s$ and $\alpha_\phi$.
                    \State \textbf{Update Graph:} $G = G + s_{\epoch} \cdot h''_\epoch \cdot C$.
                \EndFor
            \EndFor
            \State \Return $\sum_C C$
    \end{algorithmic}
\end{algorithm}

We will use the following notion of the demand-size of a sequence of cuts and the subsequent relation. Below, recall the definition of $\qLDSC$ from \Cref{def:LDSC}.
\begin{definition}[\qLDSCS] Fix a graph $G$, node-weighting $A$ and parameters $h$, $s$ and $\phi$. $\qLDSCS$ is the demand-size of the demand-size-largest sequence of $(h,s)$-length $\phi$-sparse moving cuts. Specifically,
    \begin{align*}
        \qLDSCS(\phi, h, s)  := \sum_i A_{(h,s)}(C_i)
    \end{align*}
    where above $A_{(h,s)}(C_i)$ is computed after applying all $C_j$ for $j < i$ and $(C_1, C_2, \ldots)$ is the $(h,s)$-length $\phi$-expanding moving cut sequence maximizing $\sum_i A_{(h,s)}(C_i)$.
\end{definition}

% \enote{
% \begin{restatable}[Union of Sparse Moving Cuts is a Sparse Moving Cut]{thm}{unionOfCuts}
% 	\label{thm:unionOfMovigCuts} Let $(C_1, \ldots, C_k)$ be a sequence of moving cuts where $C_i$ is $(h,s)$-length $\phi_i$-sparse cuts in $G - \sum_{j < i} C_j$ w.r.t.\ node-weighting $A$. Then the moving cut $\sum_i C_i$ is an $(h',s')$-length $\phi'$-sparse cut w.r.t.\ $A$ where $h' = 2h$, $s' = \frac{(s-2)}{2}$ and $\phi' = s^3 \cdot \log^3 N \cdot N^{O(1/s)} \cdot \frac{\sum_i |C_i|}{\sum_i |C_i|/\phi_i}$.
% \end{restatable}
% }

\begin{lemma}\label{lem:LDSCSAtMostLDSC}
    Given graph $G$ and node-weighting $A$, we have that 
    \begin{align*}
    \qLDSCS(\phi, h, s) \leq \qLDSC(\phi', h', s')
    \end{align*}
    where $h' = 2h$, $s' = \frac{(s-2)}{2}$ and $\phi' = s^3 \cdot \log^3 N \cdot N^{O(1/s)} \cdot \phi$.
\end{lemma}
\begin{proof}
    Let $(C_1, C_2, \ldots)$ be a sequence of $(h,s)$-length $\phi$-sparse cuts such that $\sum_i A_{(h,s)}(C_i) = \qLDSCS(\phi, h, s)$ where $\phi_i \leq \phi$ is the minimum value for which $C_i$ is $\phi_i$-sparse. Let $D_1, D_2, \ldots$ be the demands witnessing these cuts so that for all $i$ we have
    \begin{align*}
        |C_i|/\phi = |D_i|.
    \end{align*}
    Let $C = \sum_i C_i$. By \Cref{thm:unionOfMovingCuts} we know that $C$ is an $(h',s')$-length $\phi'$-sparse cut for $A$ and so there must be some $h'$-length demand $D$ which is $h's'$-separated by $C$ and $D$ has size at least
    \begin{align*}
        \frac{|C|}{\phi'} = |C|\frac{\sum_i |C_i|/\phi_i}{\sum_i |C_i|} = \sum_i |D_i|
    \end{align*}
    In other words, the demand-size of $C$ is at least $\sum_i |D_i|$, as required.
\end{proof}

\begin{lemma}\label{lem:cutUpperBound}
    At the end of the $\epoch$th epoch of \Cref{alg:EDsfromCuts}, we have that for every $h'_\epoch \leq h_{\epoch}/\alpha_s$ that the demand-size-largest $(h'_\epoch, s_\epoch)$-length $(\phi_\epoch/ \alpha_\phi)$-sparse cut has $(h'_\epoch, s_\epoch \cdot \alpha_s)$-length demand-size at most $|A| \cdot \left(\frac{1}{N^{O(\eps')}}\right)^{\epoch}$. In other words, we have for every $h'_\epoch \leq h_{\epoch}/\alpha_s$ that
\begin{align*}
    \qLDSC(\phi_{\epoch}/\alpha_\phi, h'_\epoch, s_{\epoch} \cdot \alpha_s) \leq |A| \cdot \left(\frac{1}{N^{O(\eps')}}\right)^{\epoch}.
\end{align*}
\end{lemma}
\begin{proof}
We prove this by induction. Let $\gamma := 1/ N^{O(\eps')}$ for convenience of notation. The base case when $\epoch = 0$ is trivial as the demand-size of any cut is trivially at most $|A|$. Next, suppose $\epoch > 0$ and assume for the sake of contradiction that our inductive hypothesis does not hold and so there is some $\hat{h}_\epoch \leq h_{\epoch}/\alpha_s$ such that at the end of the $\epoch$th epoch we have
\begin{align}\label{eq:cutLB}
    \qLDSC(\phi_{\epoch}/\alpha_\phi, \hat{h}_\epoch, s_{\epoch}  \cdot \alpha_s) > |A| / \gamma^{\epoch}.
\end{align}

Consider a cut $C$ we compute in an epoch, $\epoch$, with corresponding length $h''_\epoch$.
Unpacking the definition of approximate demand-size largest sparse cuts (\Cref{dfn:apxLargeCut}), we have that $C$ is an $(h''_{\epoch}, s_{\epoch})$-length $\phi_{\epoch}$-sparse cut for $A$ where $h''_{\epoch} \leq \alpha_s \cdot h_{\epoch}$ and, for any $h'_\epoch \leq h_{\epoch}/\alpha_s$ and, in particular, for $\hat{h}_\epoch$ we have 
\begin{align*}
    A_{(h''_\epoch,s_\epoch)}(C) \geq \frac{1}{\alpha} \cdot \qLDSC(\phi_\epoch/\alpha_\phi, \hat{h}_\epoch, s_\epoch \cdot \alpha_s)
    % A_{(h''_\epoch,\alpha_s \cdot s_\epoch)}(C) \geq \frac{1}{\alpha} \cdot \qLDSC(\phi_\epoch/\alpha_\phi, h_\epoch, s_\epoch).
\end{align*}
at the moment when $C$ is computed. Combining the above with \Cref{eq:cutLB} and the fact that $\qLDSC(\phi_\epoch/\alpha_\phi, \hat{h}_\epoch, s_\epoch \cdot \alpha_s)$ can only be smaller at the end of our epoch than in the middle of it \future{Why? Not totally clear. Need to update union of cuts to make this correct}, we get
\begin{align}
    A_{(h''_\epoch,s_\epoch)}(C) > \frac{1}{\alpha} \cdot |A|/ \gamma^\epoch \label{eq:demSizeLB}
\end{align}

Let $h''_{\epoch} \leq h_{\epoch} \cdot \alpha_s$ be a value that could correspond to a cut in the $\epoch$th iteration. Notice that, by our definition $h_{\epoch} = \frac{1}{2\alpha_s} \cdot h_{\epoch-1}$, we have that
\begin{align}\label{eq:hUpper}
    2h''_{\epoch} \leq 2\alpha_s \cdot h_{\epoch}= h_{\epoch-1}.
\end{align}
We may assume these values are, without loss of generality, powers of $2$ and so after we compute $\Theta(\log N \cdot \alpha \cdot \gamma)$-many cuts, we know that there must be some $h''_{\epoch}$ such that $\mcC_{h''_{\epoch} }$ contains at least $(\alpha \cdot \gamma)$-many cuts. We let $(C_1, C_2, \ldots)$ be these cuts. Applying \Cref{eq:demSizeLB} to each of our  $C_i$s, we get
\begin{align}\label{eq:demandSizeLB}
    \sum_i A_{(h''_\epoch, s_\epoch)}(C_i) > |A| / \gamma^{\epoch-1}.
\end{align}

$(C_1, C_2, \ldots)$ is an $(h_{\epoch}'', s_\epoch)$-length $\phi_\epoch$-sparse sequence of moving cuts and so forms a candidate for the demand-size-largest such sequence. \future{TODO: not totally clear because some things are getting cut; again update union of cuts}

It therefore follows that at the end of the $\epoch$th epoch we have
\begin{align*}
    \sum_i A_{(h''_{\epoch},  s_\epoch)}(C_i)
    &\leq \qLDSCS(\phi_\epoch, h''_{\epoch} , s_\epoch) \\
    &\leq \qLDSC\left(\phi_\epoch \cdot \tilde{O}(s_\epoch^3 \cdot N^{O(1/s_\epoch)}),2h''_{\epoch} , s_\epoch\right)\\
    &= \qLDSC\left(\phi_{\epoch-1},2h''_{\epoch} ,s_{\epoch-1} \cdot \alpha_s\right)\\
    & \leq \qLDSC\left( \phi_{\epoch-1}, h_{\epoch-1}, s_{\epoch-1} \right)\\
    & \leq |A|/ \gamma^{\epoch-1}.
\end{align*}
where, above, the second inequality follows from \Cref{lem:LDSCSAtMostLDSC}, the third from the definition of $\phi_{\epoch-1}$ and $s_{\epoch-1}$, the fourth from \Cref{eq:hUpper} and the fact that $\qLDSC(\phi, h, s)$ is monotone increasing in $h$ (as long as $s \geq 2$) and monotone decreasing in $s$ and the fifth from our inductive hypothesis. However, the above contradicts \Cref{eq:demandSizeLB}.
\end{proof}
    
We conclude with our proof of \Cref{lem: from MSC to ED}.
\EDsFromCuts*
\begin{proof}
We use \Cref{alg:EDsfromCuts}.

First, we claim that the returned cuts $\sum_C C$ are indeed an $(\leq h, s)$-length expander decomposition for sparsity $\phi$ with cut slack $\kappa$.

Let $\epoch$ be $1/\eps'$. By our choice of how we initialize $h_0$, $\phi_0$, observe that after our $1/\eps'$-many epochs we have that $h_\epoch = \alpha_s \cdot h$, $s_\epoch = 2 \cdot \alpha_s^{2/\eps'}$ and $\phi_\epoch = \alpha_\phi \cdot \phi$. Letting $s = 2 \cdot \alpha_s^{2/\eps'}$ and applying \Cref{lem:cutUpperBound} to the final epoch of our algorithm, we therefore have that, after applying $\sum_C C$ to our graph, for every $h' \leq h_\epoch / \alpha_s \leq h$ that the $(h,s)$-length demand-size of the demand-size-largest $(h,s)$-length $\phi_{\epoch} / \alpha_\phi = \phi$-sparse cut is at most 
\begin{align*}
|A| \cdot \left(\frac{1}{N^{O(\eps')}}^{1/\eps'} \right)    \leq m \cdot \left(\frac{1}{N^{O(\eps')}}\right)^{1/\eps'} < 1 
\end{align*}
where, above, we applied the fact that $|A| \leq m$ (since it must be degree-respecting). Since the demand-size of a cut is integral, it follows that for all $h' \leq h$ no $(h,s)$-length $\phi$-sparse cut exists and so $\sum_C C$ is indeed a $(\leq h, s)$-length $\phi$-expander decomposition (\Cref{def:LCED}).

We next consider the cut slack of this expander decomposition. Observe that each cut $C$ that we compute in the epoch $\epoch$ is, by construction, a ($(\leq h_{\epoch}, s_{\epoch})$-length) $\phi_\epoch$-sparse cut for $|A|$ and so has size at most 
\begin{align*}
    \phi_{\epoch} \cdot |A| & \leq \phi \cdot \alpha_\phi \cdot \left( \tilde{O}(\alpha_\phi \cdot s_\epoch^3 \cdot N^{O(1/s_\epoch)}) \right)^{1/\eps'} \cdot |A|\\
    &\leq \phi \cdot \alpha_\phi \cdot \left( \tilde{O}(\alpha_\phi \cdot s^{3} \cdot N^{O(1/s_0)}) \right)^{1/\eps'} \cdot |A|\\
    & \leq \phi \cdot \alpha_\phi^{O(1/\eps')} \cdot s^{O(1/\eps')} \cdot N^{O(1/(\eps' \cdot \sqrt{s}))} \cdot \left(\tilde{O}(1) \right)^{1/\eps'}\cdot |A|
\end{align*}
Furthermore, applying the fact that, although we compute $l = \tilde{O}(\alpha \cdot N^{O(\eps')})$ total such cuts in one iteration, the size of these cuts is geometrically increasing, so the entire size of $\sum_C C$ is dominated by the sum of the cuts we compute in the last iteration. Namely, we have
\begin{align*}
    \sum_C |C| \leq \phi |A| \cdot \left(  \tilde{O}\left(\alpha \cdot N^{O(\eps')} \right) \cdot \alpha_\phi^{O(1/\eps')} \cdot s^{O(1/\eps')} \cdot N^{O(1/(\eps' \cdot \sqrt{s}))} \cdot \left(\tilde{O}(1) \right)^{1/\eps'}  \right)
\end{align*}
giving our bound on the cut slack.

Lastly, the work and depth of our algorithm is trivial since we compute  $\tilde{O}(\frac{N^{O(\eps')}}{\eps'} \cdot \alpha)$-many cuts.

% Since overall we compute 

% the largest cut never has size more than $|A|$

% $(h, \alpha_s^{1/\eps'})$-length $\phi$-sparse cut is at most $|A| \cdot \eps^{1/\eps'}$. Since $|A| \leq n$ and node-weightings are integral, we have that this is $0$ as long as $\eps^{1/\eps'}$ is $o(n)$.

% Applying \Cref{lem:cutUpperBound} to the final  iteration, we know that after applying $\sum_C C$ to our input graph $G_0$ that for any $h'_l $

% $G$ does not contain any

% By \Cref{lem:cutUpperBound}

% Observe that it follows that \enote{Note result of end of induction}.

% Our proof will be by induction.  

% On the other hand, assuming that our induction succeeded for $\epoch-1$, we know that 

% % On the other hand, we claim that each $C_i$ is a candidate for the above $\qLDSC$. In particular, we claim for each $C_i$ that
% % \begin{align}
% %     A_{(h'', \alpha_s \cdot s_\epoch)}(C_i) \geq \frac{1}{\alpha} \cdot \qLDSC\left(\phi_\epoch \cdot \tilde{O}(s_\epoch^3 \cdot N^{O(1/s_\epoch)}),2h'',\alpha_s \cdot s_\epoch\right)
% % \end{align}
% \enote{Note when this is the case}

% % This follows because $C_i$ is $\alpha$-approximate and so for any 
% % \begin{align*}
% %     A_{(h'',\alpha_s \cdot s)}(C_i) \geq \frac{1}{\alpha} \cdot \qLDSC(\phi_\epoch/\alpha_\phi, h', s).
% % \end{align*}

%     \enote{TODO}
\end{proof}

\section{Algorithm: Length-Constrained EDs from ``The Spiral''}\label{sec:spiral}
% This is our final statement:

We conclude by combining our algorithm which computes length-constrained expander decompositions using demand-size-large length-constrained sparse cuts (\Cref{lem: from MSC to ED}) with our algorithm that computes large length-constrained sparse cuts using length-constrained expander decompositions (\Cref{lem:cutsFromEDs}). This forms a ``spiral'' of mutual recursion where each time we go around the spiral we make substantial progress on the size of the problem on which we are working (in terms of node-weighting size). We first state (and prove) our result with a maximally general tradeoff of parameters ($\eps$ and $\eps'$). We next simplify this presentation by choosing these parameters to get our final theorem.

% \enote{Simplify down to a single param where below are poly of each other.}

% \enote{Is the extra $\eps'$ ok? Still allows for constants, ignoring the tildes.}

\begin{theorem}
\label{thm:EDNoLink} There exists an algorithm that, given edge-capacitated graph $G$, parameters $\eps,\eps' \in (0,1)$, node-weighting $A$ on $G$, length bound
$h$, and conductance bound $\phi>0$, computes an $(\leq h,s)$-length witnessed $\phi$-expander
decomposition for $A$ in $G$
with cut and length slack 
\begin{align*}
    \kappa = N^{O(\eps')} \cdot \tilde{O}\left( \frac{N^{O(\eps)}}{\poly(\eps)} \right)^{O(1/\eps')} \qquad \qquad s=(1/\eps)^{O(1/\eps')}
\end{align*}
with work 
\begin{align*}
    \wED(A, m) \leq m \cdot  \tilde{O}\left(\frac{N^{O(\eps +\eps')}}{\poly(\eps, \eps')}\right) \left(  \tilde{O}(1/\eps)^{O(1/(\eps\eps'))}  +  \poly(h)  \right).
\end{align*}
and depth
\begin{align*}
    \dED(A, m) \leq \tilde{O}\left(\frac{N^{O(\eps +\eps')}}{\poly(\eps, \eps')}\right) \left(  \tilde{O}(1/\eps)^{O(1/(\eps\eps'))}  +  \poly(h)  \right).
\end{align*}
\end{theorem}
\begin{proof}
    Applying \Cref{lem:cutsFromEDs} with $L = N^{\eps'}$ we have that one can compute an $\alpha$-approximate $(\leq h, s)$-length demand-size largest cut with sparsity and length approximation $\alpha_\phi$ and $\alpha_s$ where
    \begin{align*}
        \alpha = \tilde{O}\left(\frac{N^{O(\eps)}}{\eps^3} \right) \qquad \qquad \alpha_\phi = \frac{s^3N^{O(1/s)}}{\eps} \qquad \qquad \alpha_s = \max\left(2, \frac{1}{\eps^2} \cdot (s)^{1 + O(1/\eps)} \right)
    \end{align*}
    with work at most
    \begin{align*}
        \wsparseCut(A,m)\leq m \cdot \tilde{O}\left( \frac{1}{\eps} \cdot (s)^{O(1/\eps)} \cdot N^{O(\eps)}   + \frac{1}{\eps} \cdot N^{O(\eps + \eps')} \cdot \poly(h) \right) +\frac{1}{\eps}  \sum_i  \wED(A_i, m_i)
    \end{align*}
    where $|A_i| \leq \frac{|A|}{N^{\eps'}}$  for all $i$ and $\{m_i\}_i$ are non-negative integers satisfying
    \begin{align}\label{eq:boundmi}
        \sum_i m_i \leq \frac{1}{\eps} \cdot \tilde{O}(m + n^{1 + O(\eps)} + N^{O(\eps)}).
    \end{align}
    
    Likewise, applying \Cref{lem: from MSC to ED}, we have that we can compute a $(\leq h,s)$-length $\phi$-expander decomposition for $A$ with cut slack
    \begin{align*}
        \kappa =  \tilde{O}\left(\alpha \cdot N^{O(\eps')} \right) \cdot \tilde{O} \left(\alpha_\phi \cdot s \cdot N^{O(1/\sqrt{s})}  \ \right)^{O(1/\eps')},
    \end{align*}
    length slack
    \begin{align*}
        s = \alpha_s^{O(1/\eps')},
    \end{align*}
    with work
    \begin{align*}
        \wED(A,m) \leq \tilde{O}\left(\frac{N^{O(\eps')}}{\eps'} \cdot \alpha \right) \cdot \wsparseCut(A,m)    
    \end{align*}
    where $\alpha$, $\alpha_\phi$ and $\alpha_s$ are the approximation factors of our approximate demand-size-largest sparse cut algorithm.

    We now combine the above recursions into a single recursion to compute a $(\leq h, s_i)$-length $\phi$-expander decomposition for $A$. We will index each level of our recursion by $i$. Observe that at the $i$th level of our recursion, we have that 
    \begin{align*}
        s_i = \alpha_{s,i}^{O(1/\eps')} = \frac{1}{\eps^2} \cdot s_{i-1}^{1 + O(1/\eps)}
    \end{align*}
    Furthermore, since each time we recurse we reduce our node-weighting's size by a multiplicative $N^{\eps'}$, the depth of our recursion is at most $1/\eps'$ and so we have that $s_i$ is always at most
    \begin{align*}
        s_i \leq \left(\frac{1}{\eps}\right)^{O(1/\eps')}
    \end{align*}
    It follows that our cut slack in the $i$th level of our recursion is
    \begin{align*}
        \kappa_i &= \tilde{O}\left(\alpha \cdot N^{O(\eps')} \right) \cdot \tilde{O} \left(\alpha_\phi \cdot s_i \cdot N^{O(1/\sqrt{s_i})}  \ \right)^{O(1/\eps')}  \\
        &= \tilde{O}\left(\frac{N^{O(\eps + \eps')}}{\poly(\eps)} \right) \cdot \tilde{O} \left(\frac{s_i^3 N^{O(1/s_i)}}{\eps} \cdot s_i \cdot N^{O(1/\sqrt{s_i})}  \ \right)^{O(1/\eps')}\\
        &= \tilde{O}\left(\frac{N^{O(\eps + \eps')}}{\poly(\eps)} \right) \cdot \tilde{O} \left(\frac{1}{\poly(\eps)} \cdot N^{O(1/s_i+1/\sqrt{s_i})} \right)^{O(1/\eps')}\\
        &= N^{O(\eps')} \cdot \tilde{O}\left( \frac{N^{O(\eps)}}{\poly(\eps)} \right)^{O(1/\eps')}.
        % &= \tilde{O}\left(\frac{N^{O(\eps)}}{\poly(\eps)} \right)^{O(1/\eps')} \cdot \tilde{O} \left(N^{O(1/s_i + 1/\sqrt{s_i})}  \right)^{O(1/\eps')}\\
        % &= \tilde{O}\left(\frac{N^{O(\eps)}}{\poly(\eps)} \right)^{O(1/\eps')} \cdot \left(N^{O(1/\sqrt{s_i})}\right)^{O(1/\eps')} \cdot N^{O(1/(\eps' \cdot \sqrt{s_i}))} \\
        % &= \tilde{O}\left(\frac{N^{O(\eps)}}{\poly(\eps)} \right)^{O(1/\eps')} \cdot \left(N^{O(1/\sqrt{s_i})}\right)^{O(1/\eps')}  \\
        % &\leq \tilde{O}\left(\frac{N^{O(\eps)}}{\poly(\eps)} \right)^{O(1/\eps')}
    \end{align*}
    where in the last line we used the fact that $1/\sqrt{s_i}+1/s_i \geq 1/\eps $ for every value of $s_i$ we consider. \future{Double check this}

    We next bound the work of our algorithm. Letting $m_j$ be the total number of edges at the $j$th level of our recursion and applying our bound on $\sum_i m_i$ from \Cref{eq:boundmi} we have
    \begin{align*}
        m_{j+1} \leq \frac{1}{\eps} \cdot \tilde{O}(m_j + n^{1 + O(\eps)} + N^{O(\eps)})
    \end{align*}
    Applying the fact that our recurrence has depth at most $1/\eps'$, we get that the maximum number of edges across an entire level of recursion $j$ is at most
    \begin{align}\label{eq:Saga}
        m_j &\leq  m \cdot \tilde{O}\left(\frac{1}{\eps}\right)^{1/\eps'} + \frac{1}{\eps'} \left( n^{1 + O(\eps)} + N^{O(\eps)} \right) \nonumber\\
        & =  m \cdot \tilde{O}\left(\frac{1}{\eps}\right)^{1/\eps'} + \frac{N^{O(\eps)}}{\eps'}
        % & =  \left(\frac{1}{\eps} \cdot \tilde{O}(m)\right)^{1/\eps'} + \frac{N^{O(\eps)}}{\eps'} 
    \end{align}
    where, above, $m$ is the number of edges in our original graph. On the other hand, we have that (apart from the recursive calls), the entire work of our algorithm at a single level of recursion is at most
    %\enote{What about the $s^{1/\eps}$? This becomes $(1/\eps)^{1/(\eps' \cdot \eps)}$ which is annoying. Ignoring this below.}
    \begin{align}\label{eq:gasfsa}
        &m_j \cdot \tilde{O}\left(\frac{\alpha \cdot N^{O(\eps')}}{\eps'}\right) \left( \frac{1}{\eps} \cdot (s)^{O(1/\eps)} \cdot N^{O(\eps)}   + \frac{1}{\eps} \cdot N^{O(\eps + \eps')} \cdot \poly(h)  \right) \nonumber\\
        &\leq m_j \cdot \tilde{O}\left(\frac{N^{O(\eps +\eps')}}{\poly(\eps, \eps')}\right) \left(  (s)^{O(1/\eps)}  +  \poly(h)  \right) \nonumber\\
        &\leq m_j \cdot \tilde{O}\left(\frac{N^{O(\eps +\eps')}}{\poly(\eps, \eps')}\right) \left(  (1/\eps)^{O(1/(\eps\eps'))}  +  \poly(h)  \right)
        % &\leq m_j \cdot N^{O(\eps + \eps')} \cdot \tilde{O}\left(\frac{ \cdot (1/\eps)^{O(1/\eps)}}{\poly(\eps, \eps')} \right)^{O(1/\eps')} \cdot  \poly(h).
    \end{align}
    Combining \Cref{eq:Saga} and \Cref{eq:gasfsa} we get that the total work on a single level of recursion (again, excluding recursive calls) is at most
    \begin{align*}
        \left( m \cdot \tilde{O}\left(\frac{1}{\eps}\right)^{1/\eps'} + \frac{N^{O(\eps)}}{\eps'}  \right) \cdot \tilde{O}\left(\frac{N^{O(\eps +\eps')}}{\poly(\eps, \eps')}\right) \left(  (1/\eps)^{O(1/(\eps\eps'))}  +  \poly(h)  \right)\\
        = m \cdot  \tilde{O}\left(\frac{N^{O(\eps +\eps')}}{\poly(\eps, \eps')}\right) \left(  \tilde{O}(1/\eps)^{O(1/(\eps\eps'))}  +  \poly(h)  \right)\\
        % = \tilde{O}\left(m \cdot \frac{N^{O(\eps)}\cdot (1/\eps)^{O(1/\eps)}}{\poly(\eps, \eps')} \right)^{O(1/\eps')} \cdot  \poly(h).
    \end{align*}
    Summing over our $1/\eps'$-many recursive levels of our algorithm then gives our work bound. The argument for depth is analogous (and, in fact, easier because we do not have to control the total number of edges over each level of recursion using \Cref{eq:boundmi}).
    \future{Go through this proof carefully / formalize it a bit more}

\end{proof}
Lastly, we give the simplified version of the above theorem with witnesses.

% \enote{Comment}
% \enote{Set $\eps = $ between a constant (1) and at least something like $1/\log^{1/3} N$. Then set $\eps' = \sqrt{\eps}$. Plug in. $\kappa = N^\sqrt{\eps} \cdot N^{\sqrt{\eps}}/\poly(\eps)$ but the $1/\poly(\eps)$ is just a poly-log so the whole thing becomes $\tilde{O}(N^\sqrt{\eps})$. $s$ will be $\exp(1/\sqrt{\eps})$. Reset so $s$ is like $\exp(1/\eps)$ and $\kappa$ is $\tilde{O}(N^{\poly(\eps)})$. Don't need $\tilde{O}$ bc of assumption on $\eps$}
% \begin{theorem}
% \label{thm:EDNoLinkSimple} There exists a constant $c >1$ such that given edge-capacitated graph $G$, parameter $\eps \in \left(\frac{1}{\log^{1/c} N},1 \right)$, node-weighting $A$ on $G$, length bound
% $h$, and conductance bound $\phi>0$, one can compute an $(\leq h,s)$-length $\phi$-expander
% decomposition for $A$ in $G$
% with cut and length slack 
% \begin{align*}
%     \kappa = N^{O(\poly(\eps))}   \qquad \qquad s= \exp(1/\eps)
% \end{align*}
% work and depth
% \begin{align*}
%     \wED(A, m) \leq m \cdot  \tilde{O}\left(N^{O(\poly(\eps))}\cdot \poly(h) \right) \qquad  \dED(A, m) \leq \tilde{O}\left(N^{O(\poly(\eps))}\cdot \poly(h) \right).
% \end{align*}
% % and depth
% % \begin{align*}
% %     \dED(A, m) \leq \tilde{O}\left(N^{O(\poly(\eps))}\cdot \poly(h) \right).
% % \end{align*}
% \end{theorem}
\mainAlgThm*
\begin{proof}
    We first apply \Cref{thm:EDNoLink} with $\eps'_0 = \eps$ and $\eps_0 = \eps^2$ where $\eps'_0$ and $\eps_0$ are the parameters described in \Cref{thm:EDNoLink}. Likewise, let $c$ be an upper-bound on the exponent of the poly-log in the $\tilde{O}$ notations and the exponent of all $\poly$ notation of \Cref{thm:EDNoLink}. It follows that we can compute a $(\leq h,s)$-length $\phi$-expander decomposition for $A$ in $G$ with cut and length slack 
\begin{align*}
    \kappa = \frac{N^{O(\eps)}}{(\eps)^{c \cdot O(1/\eps)}} \cdot \log^{c \cdot O(1/\eps)} N \qquad \qquad s=(1/\eps^2)^{O(1/\eps)}
\end{align*}
with work 
\begin{align*}
    \wED(A, m) \leq m \cdot  \tilde{O}\left(\frac{N^{O(\poly(\eps))}}{\poly(\eps)}\right) \left( (1/\eps)^{O(1/\eps^3)} \cdot \log^{c \cdot O(1/\eps)} N +  \poly(h)  \right).
\end{align*}
and depth
\begin{align*}
    \dED(A, m) \leq \tilde{O}\left(\frac{N^{O(\poly(\eps))}}{\poly(\eps)}\right) \left( (1/\eps)^{O(1/\eps^3)} \cdot \log^{c \cdot O(1/\eps)}N  +  \poly(h)  \right).
\end{align*}

We begin by reasoning about our cut slack. Notice that, for a suitable large hidden constant in the $\Omega$ we have that if
\begin{align*}
 \eps \geq \Omega\left(\sqrt{\frac{c \cdot \log \log n}{\log N}} \right)
\end{align*}
then we have that 
\begin{align}
    \frac{c \cdot O(1/\eps) \cdot \log \log N}{ \log N} \leq \eps
\end{align}
Similarly, we have that if
\begin{align*}
    \eps \geq \Omega\left( \left({\frac{c}{\log n}}\right)^{1/3}\right) 
\end{align*}
then 
\begin{align*}
    \frac{\log(1/\eps)\cdot c \cdot O(1/\eps)}{\log N} \leq \frac{c \cdot O(1/\eps^2)}{\log N} \leq \eps
\end{align*}
Thus, our cut slack is
\begin{align*}
    \kappa = N^{O(\eps) + \frac{c \cdot O(1/\eps) \cdot \log \log N}{ \log N} + \frac{\log(1/\eps)\cdot c \cdot O(1/\eps)}{\log N}} \leq N^{O(\eps)}
\end{align*}
Likewise, our length slack is
\begin{align*}
    s = (1/\eps^2)^{O(1/\eps)} = \exp(O(1/\eps) \cdot \log(1/\eps^2)) \leq \exp(O(1/\eps^2)) = \exp(\poly(1/\eps)).
\end{align*}
Lastly, observe that if
\begin{align*}
    \eps \geq \Omega\left(\frac{1}{\log N}\right)^{(1/5)}
\end{align*}
and
\begin{align*}
    \eps \geq \Omega\left(\log \log n / \log n \right)
\end{align*}
then we have
\begin{align*}
    \frac{O(1/\eps^3) \cdot \log(1/\eps)}{\log N} + \frac{c \cdot O(1/\eps) \cdot \log \log N}{\log N } \leq \frac{O(1/\eps^4)}{\log N} + \frac{O(1/\eps) \cdot \log \log N}{ \log N} \leq O(\eps)
\end{align*}
and so our work is 
\begin{align*}
    \wED(A, m) & \leq m \cdot  \tilde{O}\left(\frac{N^{O(\poly(\eps))}}{\poly(\eps)}\right) \left( (1/\eps)^{O(1/\eps^3)} \cdot \log^{c \cdot O(1/\eps)} N +  \poly(h)  \right)\\
    &= m \cdot  \tilde{O}\left(\frac{N^{O(\poly(\eps))}}{\poly(\eps)}\right) \left( N^{\frac{O(1/\eps^3) \cdot \log(1/\eps)}{\log N} + \frac{c \cdot O(1/\eps) \cdot \log \log N}{\log N }} +  \poly(h)  \right)\\
    &= m \cdot  \tilde{O}\left(\frac{N^{O(\poly(\eps))}}{\poly(\eps)}\right) \left( N^{O(\eps)} +  \poly(h)  \right)\\
    &= m \cdot  \tilde{O}\left(\frac{N^{O(\poly(\eps))}}{\poly(\eps)}\right) \cdot \poly(h)\\
    & \leq m \cdot  \tilde{O}\left(N^{O(\poly(\eps))}\cdot \poly(h) \right) 
\end{align*}
Where in the last line we used the fact that $\eps \geq 1/ \log ^{O(1)} N$. The final result comes from letting $\eps$ above be smaller by a suitable large constant (to get the cut slack from $N^{O(\eps)}$ to $n^{\eps}$). Our depth bound is symmetric.

Lastly, we discuss how to compute our witness. The basic idea is to slightly strengthen the expander decomposition we compute to deal with the slacks from \Cref{lem:cutsFromEDs}. In particular, recall that by \Cref{lem:cutsFromEDs} if the graph is already $( \leq h \cdot \frac{1}{\eps} \cdot s^{O(1/\eps)}, s)$-length $\phi$-expanding then when we apply \Cref{lem:cutsFromEDs} we get a $(\leq h, s_w)$-length $\phi_w$-expansion witness where $s_w = \frac{1}{\eps^2} \cdot s^{O(1/\eps)}$ and $\phi_w = \tilde{O}(\phi \eps/N^{O(\eps)})$. 

Fix an $\eps$ and let $\eps_0 = \poly(\eps)$ so that if $s_0 = \exp(\poly(1/\eps_0))$ we have that 
    \begin{align*}
        \frac{1}{\eps_0^2} \cdot s_0^{O(1/\eps_0)} \leq \exp(\poly(1/\eps))
    \end{align*}
    and we let $\phi_0 = \phi \cdot \tilde{O}\left(N^{O(\eps_0)}\right) / \eps_0$ so that
    \begin{align*}
        \phi \geq \phi_0 \cdot \eps / \tilde{O}(N^{O(\eps)})
    \end{align*}
    Thus, if our graph is $(\leq h_0, s_0)$-length $\phi_0$-expanding and we apply \Cref{lem:cutsFromEDs} with these parameters then we get back a $(\leq h, s)$-length $\phi$-expansion witness. 

    Next, apply our algorithm for computing length-constrained expander decompositions using $\eps_0$ to compute a $(\leq h_0, s_0)$-length $\phi_0$-expander decomposition. Since $h_0 \geq h$ and $\eps_0 = \poly(\eps)$, such a decomposition is a $(\leq h, \exp(\poly(\eps)))$-length $\phi$-expander decomposition but with a multiplicative cut slack increase of $\phi_0/\phi = N^{O(\eps_0)}/\eps_0$ for a total cut slack of (assuming $\eps_0 \geq \log \log N / \log N$)
    \begin{align*}
        N^{\eps_0} \cdot N^{O(\eps_0)}/\eps_0 =        N^{O(\eps_0) \log(\eps_0)}/\eps_0 \leq N^{O(\eps_0)}
    \end{align*}
    Lastly, letting $\eps_0$ be smaller by an appropriate polynomial gives a length slack of $\exp(\poly(\eps))$ and cut slack of $n^{\eps}$. Furthermore, the time to compute this expander decomposition is as described above since our parameters have only changed by a polynomial. Furthermore, if we apply \Cref{lem:cutsFromEDs} after applying the above decomposition we get back the desired witness. Finally, the work to invoke \Cref{lem:cutsFromEDs} with $L = N^\eps_0$ is
    \begin{align*}
        &m \cdot \tilde{O}\left( \frac{1}{\eps_0} \cdot s_0^{O(1/\eps_0)} \cdot N^{O(\eps_0)}   + \frac{1}{\eps_0} \cdot N^{\eps_0} \cdot N^{O(\eps_0)} \cdot \poly(h) \right) +\frac{\tilde{O}(1)}{\poly(\eps_0)}  \sum_i  \wED(A_i,m_i)\\
        &\leq m \cdot \tilde{O}\left( \frac{1}{\eps_0} \cdot s_0^{O(1/\eps_0)} \cdot N^{O(\eps_0)}   + \frac{1}{\eps} \cdot N^{O(\eps_0)} \cdot \poly(h) \right) \\& \qquad +\frac{\tilde{O}(1)}{\poly(\eps_0)}  \sum_i  m_i \cdot  \tilde{O}\left(N^{O(\poly(\eps_0))}\cdot \poly(h) \right) \\
        &\leq m \cdot \tilde{O}\left( \frac{1}{\eps_0} \cdot s_0^{O(1/\eps_0)} \cdot N^{O(\eps_0)}   + \frac{1}{\eps} \cdot N^{O(\eps_0)} \cdot \poly(h) \right) \\& \qquad +\frac{\tilde{O}(1)}{\poly(\eps_0)}  \frac{1}{\eps_0} \cdot \tilde{O}(m + n^{1 + O(\eps_0)} + N^{O(\eps_0)}) \cdot  \tilde{O}\left(N^{O(\poly(\eps_0))}\cdot \poly(h) \right)\\
        &\leq m \cdot \tilde{O}\left(N^{O(\poly(\eps))}\cdot \poly(h) \right)
    \end{align*}
    where above we applied the previously described algorithm for expander decompositions and in the last line we again assumed $\eps \geq 1/\log^{1/c}N$ for a suitable large constant $c$. Note that the above is asymptotically the same as the time to compute our expander decompositions. The depth calculation is analogous.
\end{proof}

In \Cref{sec:graphProps} we show how to achieve the above algorithm with the linkedness property. Likewise, we note that the above algorithm immediately gives ``witnessed'' expander decompositions. Specifically, if we run the above algorithm after already applying our expander decomposition, then we find no sparse cuts and the output of our algorithm (and, in particular, \Cref{lem:cutsFromEDs}) is a collection of routers, one for each cluster of our neighborhood cover, along with an embedding of a router into each cluster of the neighborhood cover. We also note that plugging the above result into \Cref{lem:cutsFromEDs} gives efficient approximation algorithms for the  demand-size-largest cut problem.

\section{Maintaining Extra Properties (\Cref{thm:expdecomp exist} with Linkedness)}\label{sec:graphProps}
In the previous sections we showed how repeatedly cutting large length-constrained sparse cuts gives length-constrained expander decompositions. In this section, we discuss how to incorporate additional graph properties---in particular, the linkedness property---into these decompositions.

% Lastly, we provide a general framework for efficiently incorporating various properties into length-constrained expander decompositions (in fact, our framework extends to maintaining arbitrary graph properties). Informally, we consider ``low-recourse'' properties, namely properties that can be easily restored to a graph after a small number of edge deletions.

% \paragraph*{Result.}  We show that one can always provide a length-constrained expander decomposition with a collection of any such properties, provided they together require low recourse. As our main application of this framework, we show how to achieve the well-studied property of \emph{linkedness} in length-constrained expander decompositions. Roughly, a linked (length-constrained) expander decomposition is a (length-constrained) expander decomposition even after many self-loops are added to the graph. Self-loops only make it harder to make a graph a (length-constrained) expander. This notion of linkedness is key to the role of expander decompositions in facilitating routings \cite{goranci2021expander,haeupler2022expander}.

\paragraph*{Techniques.} The basic idea for achieving said properties is to use the robustness of length-constrained expanders. In particular, start with a length-constrained expander decomposition. The resulting graph is a length-constrained expander but may not have our desired properties. We then force our graph to have our properties, possibly at the expense of length-constrained expansion. By the robustness properties of length-constrained expanders, we can easily restore length-constrained expansion, now possibly at the cost of our properties. Going back and forth between restoring our properties and length-constrained expansion eventually gives us a length-constrained expander with our desired properties.
% \enote{Add overview of what this section is and how it relates to previous}
%\textbf{ Ellis: this should be moved to main paper}

%\newcommand{\pset}{\mathcal{P}}
%\newcommand{\alg}{\mathcal{A}}
\newcommand{\del}{\textsf{Del}}
\newcommand{\inidel}{I}
\newcommand{\maintdel}{M_{\del}}

We call a $\pset$ set of (possibly infinitely many) graphs a \emph{property}. We say that a graph $G$ has property $\pset$ iff $G\in \pset$. 
%An \emph{update} of a graph $G$ refers to either the deletion of an edge $e\in E(G)$ or the insertion of an edge $e'\notin E(G)$, where both $e,e'$ can be self-loops.
%
We say that an algorithm $\alg$ \emph{maintains property $\pset$} with {initial rate $I: \mathbb{R}^+\to \mathbb{R}^+$} and maintenance rate {$M: \mathbb{R}^+\to \mathbb{R}^+$}, iff for every graph $G$ on $n$ vertices,
\begin{itemize}
    \item the algorithm $\alg$ can first find a subset $E'\subseteq E(G)$ of at most $\inidel(n)\cdot |E(G)|$ edges such that the graph $G'=G\setminus E'$ has property $\pset$; and
    \item upon any online sequence of $k$ edge deletions from $G'$, the algorithm $\alg$ maintains a pruned set of at most $M(n)\cdot k$ edges of $G'$, such that when the edges in the pruned set are removed from $G'$, the remaining graph has property $\pset$.
\end{itemize}

\noindent In this case, we also say that the property $\pset$ has (or can be maintained with) rate $(\inidel(\cdot), M(\cdot))$.

The main result in this section is the following theorem.
We will first prove this theorem, and then use it to achieve an additional linkedness property in \Cref{thm:expdecomp exist}.

\begin{theorem}
\label{thm: rates}
Let $\pset_1,\ldots,\pset_k$ be graph properties with rates $(I_1,M_1),\ldots,(I_k,M_k)$ respectively where $M_1\ge \cdots\ge M_k$. 
Assume further that
\begin{itemize}
    \item $M_1\cdot \sum_{2\le i\le k}M_i<1/4$; and
    \item $\sum_{2\le i\le k}M_i<1/2$.
\end{itemize}
%Let $\pset=\bigcap_{1\le  i\le k}\pset_i$.
Then property $\pset=\bigcap_{1\le  i\le k}\pset_i$ can be maintained with rates $(O\big(\sum_{1\le i\le k}I_i),O(\sum_{1\le i\le k}M_i))$.
\end{theorem}

\subsection{Proof of \Cref{thm: rates}}

The remainder of this section is dedicated to the proof of \Cref{thm: rates}.
We first prove the following lemma, which is the special case of \Cref{thm: rates} where $k=2$.

\begin{lemma}
\label{lem: rates}
Let $\pset_1,\pset_2$ be graph properties with rates $(I_1,M_1)$ and $(I_2,M_2)$.
If $M_1\cdot M_2<1/2$, then the property $\pset_1\cap \pset_2$ can be maintained with rate $(I_1+(2+2M_1)\cdot I_2,M_1+2(1+M_1)^2\cdot M_2)$.
\end{lemma}
\begin{proof}
Denote by $\alg_1$ and $\alg_2$ the algorithms for maintaining $\pset_1$ and $\pset_2$ with rates $(I_1,M_1)$ and $(I_2,M_2)$, respectively.
We now describe an algorithm that maintains property $\pset_1\cap \pset_2$. 
We first apply $\alg_1$ to $G$ and compute a set $E^1_0\subseteq E(G)$ of edges such that $|E^1_0|\le I_1\cdot |E(G)|$ and $G\setminus E^1_0$ has property $\pset_1$.
We then apply $\alg_2$ to $G\setminus E^1_0$ and compute a set $E^2_0\subseteq E(G\setminus E^1_0)$ of edges such that $|E^2_0|\le I_1\cdot |E(G\setminus (E^1_0\cup E^2_0))|$ and $G\setminus(E^1_0\cup E^2_0)$ has property $\pset_2$.
We then perform iterations. In the $i$th iteration, we iteratively maintain properties $\pset_1$ and $\pset_2$, starting with a graph $G_{i-1}$ (where $G_0=G\setminus (E^1_0\cup E^2_0)$). We apply $\alg_1$ to compute a set $E^1_i\subseteq E(G_{i-1})$ such that $|E^1_i|\le M_1\cdot |E^2_{i-1}|$ $G_{i-1}\setminus E^1_i$ has property $\pset_1$, and then apply $\alg_2$ to compute a set $E^2_i\subseteq E(G_{i-1}\setminus E^1_i)$ such that $G_{i-1}\setminus (E^1_i\cup E^2_i)$ has property $\pset_1$. Define $G_i=G_{i-1}\setminus (E^1_i\cup E^2_i)$ and continue to the next iteration. Whenever in some iteration $i^*$ we have $E^1_{i^*}=E^2_{i^*}=\emptyset$, we terminate the algorithm and return the graph $G_i$. In other words, the set of edges that we have removed from $G$ is $E'=\bigcup_{0\le i\le i^*}(E^1_i\cup E^2_i)$.

We now describe an ball-moving abstract process and use it to estimate the size of set $E'$. We have two boxes denoted by $B_1$ and $B_2$, respectively. Initially, $B_1$ contains $I_1\cdot |E(G)|$ \emph{inactive} balls and $B_2$ contains $I_2\cdot |E(G)|$ \emph{active} balls. In each iteration, we either deactivate $t$ balls in $B_1$ (if $B_1$ contains at least $t$ active balls at the moment) and add $t\cdot M_2$ new active balls to $B_2$, or deactivate $t$ balls from $B_2$ (if $B_2$ contains at least $t$ active balls at the moment) and add $t\cdot M_1$ new active balls to $B_1$, for an arbitrary $t$. The process can continue for arbitrarily many iterations and may terminate at any point.

On the one hand, we show our algorithm can be modelled a realization of the process, such that the size of $E'$ is bounded by the number of balls in $B_1\cup B_2$ at the terminating point. 
Recall that our algorithm starts by computing a set $E_0^1$ and then a set $E_0^2$ to initiate properties $\pset_1$ and $\pset_2$, where $|E_0^1|\le I_1\cdot |E(G)|$ and $|E_0^2|\le I_2\cdot |E(G)|$. 
The reason that the initial balls in $B_1$ are inactive while the initial balls in $B_2$ are active is because the set $E_0^2$ is computed in graph $G\setminus E^0_1$, which means that the edges in $E^0_1$ are not viewed as the online updates for property $\pset_2$. 
In the first iteration after the initialization, we compute a new prune set $|E^1_1|\le M_1\cdot |E_{0}^2|$ to maintain $\pset_1$. This can be viewed as deactivating $|E_{1}^1|/M_1\le |E_{0}^2|$ balls in $B_2$ and add $|E_{1}^1|$ new active balls in $B_1$.
Similarly, in the algorithm we then compute a new prune set $E^2_1$ to maintain $\pset_1$, and this can be viewed as deactivating $|E_{1}^2|/M_1$ balls in $B_2$ and add $|E_{1}^2|$ new active balls in $B_2$. Later iterations can be simulated in a similar way. The algorithm ends whenever $E^1_{i^*}=E^2_{i^*}=\emptyset$ for some $i^*$, and we let the process end at the same time.

From the definition of maintenance rate, upon $k$ online updates, property $\pset_i$ can be maintained by a pruned set of at most $M_i\cdot k$ edges. This is exactly modelled by our process. In particular, every time we compute a set $E_i^1$, in our process we deactivate $|E_i^1|/M_2$ balls from $B_2$, so we are guaranteed that $r$ active balls in one $B_2$ can give birth to at most $M_2\cdot r$ new balls in $B_1$. Therefore, at any time, we are never required to deactivate more balls than the current active balls in a box. As every edge in the pruned set $E'$ corresponds to a distinct ball in $B_1$ or $B_2$, the size of $E'$ is bounded by the number of balls in $B_1\cup B_2$ at the end of the process.

On the other hand, we show that, no matter how the process proceeds, the number of balls in $B_1\cup B_2$ at any time is at most $(I_1+(2+2M_1)\cdot I_2)$. In fact, for each initial active ball in $B_2$, it will eventually gives birth to at most $1+M_1+M_2M_1+M_1M_2M_1+\cdots$ balls in the final set. Therefore, the number of balls in $B_1\cup B_2$ at any time is bounded by
\[
\begin{split}
& |E(G)|\cdot 
\bigg(
I_1+I_2+M_1I_2+M_2(M_1I_2)+M_1(M_2M_1I_2)+\cdots
\bigg)\\
& \le |E(G)|\cdot 
\bigg(
I_1+I_2\cdot (1+M_1)\cdot \frac{1}{1-M_1M_2}
\bigg)\\
& \le |E(G)|\cdot (I_1+(2+2M_1)\cdot I_2),
\end{split}
\]
as $M_1M_2\le 1/2$. %And it is also easy to verify that after $i^*=O(\log |E(G)|)$ iterations we must have $E^1_{i^*}=E^2_{i^*}=\emptyset$.
Therefore, our algorithm initializes property $\pset_1\cap \pset_2$ with rate $I\le I_1+(2+2M_1)\cdot I_2$.

Now upon an update, we perform iterations as exactly described before to maintain properties $\pset_1$ and $\pset_2$. Via a similar ball-moving process, we can show that the pruned set upon every update has size at most
\[
\begin{split}
M & \le 
M_1+(1+M_1)M_2+M_1(1+M_1)M_2+M_2(M_1(1+M_1)M_2)+\cdots
\\
& \le 
M_1+M_2\cdot (1+M_1)^2\cdot \frac{1}{1-M_1M_2}\\
& \le M_1+2(1+M_1)^2\cdot M_2.
\end{split}
\]
\end{proof}

We now prove \Cref{thm: rates} using \Cref{lem: rates}. Denote $I=\sum_{1\le i\le k}I_i$ and $M=\sum_{1\le i\le k}M_i$. 

%\paragraph{Case 1. $\sum_{1\le i\le k}M_i<1/10$.} %Denote $I=\sum_{1\le i\le k}I_i$ and $M=\sum_{1\le i\le k}M_i$.

We first show that property $\pset'=\bigcap_{2\le i\le k}\pset_i$ can be maintained with initial rate $I'=O(\sum_{2\le i\le k}I_i)$ and maintenance rate $M'=O(\sum_{2\le i\le k}M_i)$.
We describe an algorithm similar to (but simpler than) the one in the proof of \Cref{lem: rates} that maintains property $\pset$. Denote by $\alg_i$ the algorithm for maintaining property $\pset_i$. We perform iterations. In the first iteration, we start with a graph $G_0=G$, and apply $\alg_i$ to compute a set $E^1_i\subseteq E(G)$ such that $G\setminus E^1_i$ has property $\pset_i$. Define $E^1=\bigcup_{1\le i\le k}E^1_i$, so $|E^1|\le \big(\sum_{2\le i\le k}I_i\big)\cdot |E(G)|$. We remove edges in $E^1$ from graph $G_0$ and denote by $G_1$ the remaining graph as the outcome of iteration $1$.
We now describe the second iteration.
Observe that, for each $1\le i\le k$, the graph $G_1$ produced by iteration $1$ can be seen as obtained from $G\setminus E^1_i$ by removing edges in $E^1\setminus E^1_i$. Note that $G\setminus E^1_i$ has property $\pset_i$. From the definition of maintenance rate, we can apply $\alg_i$ to compute another set $E^2_i\subseteq E(G_1)$ with $|E^2_i|\le M_i\cdot |E^1\setminus E^1_i|\le M_i\cdot |E^1|$, such that $G_1\setminus E^2_i$ has property $\pset_i$.
Define $E^2=\bigcup_{1\le i\le k}E^2_i$, so $|E^2|\le \big(\sum_{2\le i\le k}M_i\big)\cdot |E^1|$.
We sequentially perform iteration $j$ for $j=3,4,\ldots$ similarly and compute sets $E^3,E^4,\ldots$.
%Via similar arguments, we can show that $|E^j|\le M |E^{j-1}|$. Therefore, after $\Omega(\log_{1/M}n)$ iterations, the algorithm will return an empty $E^j$. It is easy to verify that when this happens, the outcome graph $G^j$ of this iteration has all properties $\pset_1,\ldots,\pset_k$. 
Whenever for some iteration the set $E^{i^*}$ is empty, we terminate the algorithm and return $E'$ as the union of all sets $E^1,E^2,\ldots,E^{i^*}$ computed in these iterations. 
Via a similar ball-moving process as in the proof of \Cref{lem: rates}, we can show that
\[
\begin{split}
|E'|\le & \sum_{2\le i\le k}I_i\cdot |E(G)|+\sum_{j\ge 1}\bigg(\sum_{2\le i\le k}M_i\bigg)^j\cdot \sum_{2\le i\le k}I_i\cdot |E(G)|\\
\le & \frac{\sum_{2\le i\le k}I_i}{1-\sum_{2\le i\le k}M_i}\cdot |E(G)|\le 2\cdot \sum_{2\le i\le k}I_i \cdot |E(G)|.
\end{split}
\]
In other words, the initial rate of property $\pset'$ is at most $I'\le 2\cdot \sum_{2\le i\le k}I_i$. The same algorithm can also be used to deal with online update sequence. Via similar arguments, we can show that the maintenance rate is at most
\[
M'=\bigg(\sum_{2\le i\le k}M_i\bigg)+\bigg(\sum_{2\le i\le k}M_i\bigg)^2+\bigg(\sum_{2\le i\le k}M_i\bigg)^3+\cdots\le \frac{\sum_{2\le i\le k}M_i}{1-\sum_{2\le i\le k}M_i}\le 2\cdot \sum_{2\le i\le k}M_i.
\]

%\paragraph{Case 2. $M_1>1/10$ and $M_1\cdot \sum_{2\le i\le k}M_i<1/100$.}
%Note that in this case $\sum_{2\le i\le k}M_i<1/100$ must hold, so according to Case 1, we can maintain the property $\pset'=\bigcap_{2\le i\le k}\pset_i$ with initial rate $I'=O(\sum_{2\le i\le k}I_i)$ and maintenance rate $M'=O(\sum_{2\le i\le k}M_i)$.

We then show that property $\pset_1\cap \pset'$ can be maintained with rate $(O(I),O(M))$.
Note that $M_1\cdot M'\le M_1\cdot 2\cdot \sum_{2\le i\le k}M_i\le 1/2$.
Applying \Cref{lem: rates} to properties $\pset'$ and $\pset_1$, we get that property $\pset$ can be maintained with initial rate
\[
I'+(2+2M')\cdot I_1\le I'+ (2+2\cdot 2\cdot 1/2)I_1\le 5\cdot\sum_{1\le i\le k}I_i,
\]
and maintenance rate
\[
M'+2(1+M')^2\cdot M_1\le M'+ 2(1+ 2\cdot 1/2)^2 M_1\le 9\cdot \sum_{1\le i\le k}M_i.
\]

\subsection{Achieving Linkedness in \Cref{thm:expdecomp exist}}

\future{Go through this more carefully}

Let $G$ be the input graph. %Let $C$ be some moving cut that is applied to $G$.
We define the property $\pset(G,\beta)$ as all graphs $G'$ with $G'=G-C+L_C^{\beta}$ for some moving cut $C$.
Clearly, for graph $G$, its initial rate is $0$, and its maintenance rate is $\beta$ (as removing an edge will cause $\beta$ self-loops to be added).

%Essentially, a $\beta$-linked $(h,s)$-length $\phi$-expander decomposition is a cut $C$ with $G-C+L_C^{\beta}$ being a $(h,s)$-length $\phi$-expander. 
Set $s= \exp(\poly(1/\eps))$ as in \Cref{thm:expdecomp exist}.
Let $\pset(h,s,\phi)$ be the set of all $( h,s)$-length $\phi$-expanders. 
%Then from \Cref{thm:expdecomp exist}, its initial rate is $\phi\cdot \big(\frac{N^{\poly(\eps)}}{\poly(\eps)}\big)^{O(1/\eps')}$. 

The algorithm for preserving both properties $\pset(G,\beta)$ and $\pset(h,s,\phi)$ works as follows.
Let
\[
k=\frac{\poly(\eps)}{\eps},\quad\and k'=1/\poly(\eps).
\]
We first apply a $(hk',s/\poly(\eps))$-length $\big(\phi \cdot 2^{O(k)} \cdot N^{O(1/k')} \cdot \log^{2} N\big)$-expander decomposition. This has initial rate $\phi\cdot 2^{O(k)} \cdot N^{O(1/k')} \cdot \log^{2} N\cdot N^\eps$.
From \Cref{thm:HCExpPru} and \Cref{thm: equivalence}, its maintenance rate is $\phi\cdot\big(N^\eps \cdot 2^{O(k)}\cdot N^{O(1/k')} \big)\log N=\phi\cdot N^{\poly(\eps)}$.
Therefore, as long as $\beta\le \frac{1}{2\phi}\cdot N^{\poly(\eps)}$, from \Cref{lem: rates}, properties $\pset(G,\beta)$ and $\pset(h,s,\phi)$ can be achieved by alternately applying the algorithms for maintaining $\pset(G,\beta)$ and $\pset(h,s,\phi)$, with the running time dominated by the initialization running time. %Therefore, we obtain the following stronger version of \Cref{thm:expdecomp exist}.

\bibliographystyle{alpha}
\bibliography{main}

\end{document}